\documentclass[a4paper,11pt]{article}
\usepackage[T1]{fontenc} 

\makeatletter
\gdef\@fpheader{ }
\gdef\@journal{ }
\makeatother
\RequirePackage{amsmath}
\RequirePackage{amssymb}
\RequirePackage{epsfig}
\RequirePackage{graphicx}
\RequirePackage[numbers,sort&compress]{natbib}
\RequirePackage{color}
\RequirePackage[colorlinks=true
,urlcolor=blue
,anchorcolor=blue
,citecolor=blue
,filecolor=blue
,linkcolor=blue
,menucolor=blue
,pagecolor=blue
,linktocpage=true
,pdfproducer=medialab
,pdfa=true
]{hyperref}

\newif\ifnotoc\notocfalse
\newif\ifemailadd\emailaddfalse
\newif\iftoccontinuous\toccontinuousfalse
\makeatletter
\def\@subheader{\@empty}
\def\@keywords{\@empty}
\def\@abstract{\@empty}
\def\@xtum{\@empty}
\def\@dedicated{\@empty}
\def\@arxivnumber{\@empty}
\def\@collaboration{\@empty}
\def\@collaborationImg{\@empty}
\def\@proceeding{\@empty}
\def\@preprint{\@empty}

\newcommand{\subheader}[1]{\gdef\@subheader{#1}}
\newcommand{\keywords}[1]{\if!\@keywords!\gdef\@keywords{#1}\else%
\PackageWarningNoLine{\jname}{Keywords already defined.\MessageBreak Ignoring last definition.}\fi}
\renewcommand{\abstract}[1]{\gdef\@abstract{#1}}
\newcommand{\dedicated}[1]{\gdef\@dedicated{#1}}
\newcommand{\arxivnumber}[1]{\gdef\@arxivnumber{#1}}
\newcommand{\proceeding}[1]{\gdef\@proceeding{#1}}
\newcommand{\xtumfont}[1]{\textsc{#1}}
\newcommand{\correctionref}[3]{\gdef\@xtum{\xtumfont{#1} \href{#2}{#3}}}
\newcommand\jname{JHEP}
\newcommand\acknowledgments{\section*{Acknowledgments}}

\newcommand\preprint[1]{\gdef\@preprint{\hfill #1}}

\newtheorem{theorem}{Theorem}

\newtheorem{corollary}[theorem]{Corollary}

\newenvironment{proof}[1][Proof]{\noindent\textbf{#1.} }{\ \rule{0.5em}{0.5em}}

\makeatother

\newcommand\note[2][]{%
\if!#1!%
\stepcounter{footnote}\footnotetext{#2}%
\else%
{\renewcommand\thefootnote{#1}%
\footnotetext{#2}}%
\fi}

\makeatletter
\newtoks\auth@toks
\renewcommand{\author}[2][]{%
  \if!#1!%
    \auth@toks=\expandafter{\the\auth@toks#2\ }%
  \else
    \auth@toks=\expandafter{\the\auth@toks#2$^{#1}$\ }%
  \fi
}
\makeatother
\makeatletter
\newtoks\affil@toks\newif\ifaffil\affilfalse
\newcommand{\affiliation}[2][]{%
\affiltrue
  \if!#1!%
    \affil@toks=\expandafter{\the\affil@toks{\item[]#2}}%
  \else
    \affil@toks=\expandafter{\the\affil@toks{\item[$^{#1}$]#2}}%
  \fi
}
\makeatother
\makeatletter
\newtoks\email@toks\newcounter{email@counter}%
\newcommand{\emailAdd}[1]{%
\emailaddtrue%
\ifnum\theemail@counter>0\email@toks=\expandafter{\the\email@toks, \@email{#1}}%
\else\email@toks=\expandafter{\the\email@toks\@email{#1}}%
\fi\stepcounter{email@counter}}
\newcommand{\@email}[1]{\href{mailto:#1}{\tt #1}}
\makeatother

\makeatletter
\newcommand*\collaboration[1]{\gdef\@collaboration{#1}}
\newcommand*\collaborationImg[2][]{\gdef\@collaborationImg{#2}}
\makeatletter
\newcommand\afterLogoSpace{\smallskip}
\newcommand\afterSubheaderSpace{\vskip3pt plus 2pt minus 1pt}
\newcommand\afterProceedingsSpace{\vskip21pt plus0.4fil minus15pt}
\newcommand\afterTitleSpace{\vskip23pt plus0.06fil minus13pt}
\newcommand\afterRuleSpace{\vskip23pt plus0.06fil minus13pt}
\newcommand\afterCollaborationSpace{\vskip3pt plus 2pt minus 1pt}
\newcommand\afterCollaborationImgSpace{\vskip3pt plus 2pt minus 1pt}
\newcommand\afterAuthorSpace{\vskip5pt plus4pt minus4pt}
\newcommand\afterAffiliationSpace{\vskip3pt plus3pt}
\newcommand\afterEmailSpace{\vskip16pt plus9pt minus10pt\filbreak}
\newcommand\afterXtumSpace{\par\bigskip}
\newcommand\afterAbstractSpace{\vskip16pt plus9pt minus13pt}
\newcommand\afterKeywordsSpace{\vskip16pt plus9pt minus13pt}
\newcommand\afterArxivSpace{\vskip3pt plus0.01fil minus10pt}
\newcommand\afterDedicatedSpace{\vskip0pt plus0.01fil}
\newcommand\afterTocSpace{\bigskip\medskip}
\newcommand\afterTocRuleSpace{\bigskip\bigskip}
\newlength{\affiliationsSep}
\newcommand\beforetochook{\pagestyle{myplain}\pagenumbering{roman}}

\DeclareFixedFont\trfont{OT1}{phv}{b}{sc}{11}

\renewcommand\maketitle{
\pagestyle{empty}
\thispagestyle{titlepage}
\setcounter{page}{0}
\noindent{\small\scshape\@fpheader}\@preprint\par

\afterLogoSpace
\if!\@subheader!\else\noindent{\trfont{\@subheader}}\fi
\afterSubheaderSpace
\if!\@proceeding!\else\noindent{\sc\@proceeding}\fi
\afterProceedingsSpace
{\LARGE\flushleft\sffamily\bfseries\@title\par}
\afterTitleSpace
\hrule height 1.5\p@%
\afterRuleSpace
\if!\@collaboration!\else
{\Large\bfseries\sffamily\raggedright\@collaboration}\par
\afterCollaborationSpace
\fi
\if!\@collaborationImg!\else
{\normalsize\bfseries\sffamily\raggedright\@collaborationImg}\par
\afterCollaborationImgSpace
\fi
{\bfseries\raggedright\sffamily\the\auth@toks\par}
\afterAuthorSpace
\ifaffil\begin{list}{}{%
\setlength{\leftmargin}{0.28cm}%
\setlength{\labelsep}{0pt}%
\setlength{\itemsep}{\affiliationsSep}%
\setlength{\topsep}{-\parskip}}
\itshape\small%
\the\affil@toks
\end{list}\fi
\afterAffiliationSpace
\ifemailadd 
\noindent\hspace{0.28cm}\begin{minipage}[l]{.9\textwidth}
\begin{flushleft}
\textit{E-mail:} \the\email@toks
\end{flushleft}
\end{minipage}
\else 
\PackageWarningNoLine{\jname}{E-mails are missing.\MessageBreak Plese use \protect\emailAdd\space macro to provide e-mails.}
\fi
\afterEmailSpace
\if!\@xtum!\else\noindent{\@xtum}\afterXtumSpace\fi
\if!\@abstract!\else\noindent{\renewcommand\baselinestretch{.9}\textsc{Abstract:}}\ \@abstract\afterAbstractSpace\fi
\if!\@keywords!\else\noindent{\textsc{Keywords:}} \@keywords\afterKeywordsSpace\fi
\if!\@arxivnumber!\else\noindent{\textsc{ArXiv ePrint:}} \href{http://arxiv.org/abs/\@arxivnumber}{\@arxivnumber}\afterArxivSpace\fi
\if!\@dedicated!\else\vbox{\small\it\raggedleft\@dedicated}\afterDedicatedSpace\fi
\ifnotoc\else
\iftoccontinuous\else\newpage\fi
\beforetochook\hrule
\tableofcontents
\afterTocSpace
\hrule
\afterTocRuleSpace
\fi
\setcounter{footnote}{0}
\pagestyle{myplain}\pagenumbering{arabic}
} 

\renewcommand{\baselinestretch}{1.1}\normalsize
\divide\@tempdima\baselineskip
\@tempcnta=\@tempdima
\skip\@mpfootins = \skip\footins
\renewcommand{\@dotsep}{10000}

\newcommand\ps@myplain{
\pagenumbering{arabic}
\renewcommand\@oddfoot{\hfill-- \thepage\ --\hfill}
\renewcommand\@oddhead{}}
\let\ps@plain=\ps@myplain

\newcommand\ps@titlepage{\renewcommand\@oddfoot{}\renewcommand\@oddhead{}}

\numberwithin{equation}{section}

\renewcommand\section{\@startsection{section}{1}{\z@}%
                                   {-3.5ex \@plus -1.3ex \@minus -.7ex}%
                                   {2.3ex \@plus.4ex \@minus .4ex}%
                                   {\normalfont\large\bfseries}}
\renewcommand\subsection{\@startsection{subsection}{2}{\z@}%
                                   {-2.3ex\@plus -1ex \@minus -.5ex}%
                                   {1.2ex \@plus .3ex \@minus .3ex}%
                                   {\normalfont\normalsize\bfseries}}
\renewcommand\subsubsection{\@startsection{subsubsection}{3}{\z@}%
                                   {-2.3ex\@plus -1ex \@minus -.5ex}%
                                   {1ex \@plus .2ex \@minus .2ex}%
                                   {\normalfont\normalsize\bfseries}}
\renewcommand\paragraph{\@startsection{paragraph}{4}{\z@}%
                                   {1.75ex \@plus1ex \@minus.2ex}%
                                   {-1em}%
                                   {\normalfont\normalsize\bfseries}}
\renewcommand\subparagraph{\@startsection{subparagraph}{5}{\parindent}%
                                   {1.75ex \@plus1ex \@minus .2ex}%
                                   {-1em}%
                                   {\normalfont\normalsize\bfseries}}

\def\fnum@figure{\textbf{\figurename\nobreakspace\thefigure}}
\def\fnum@table{\textbf{\tablename\nobreakspace\thetable}}

\long\def\@makecaption#1#2{%
  \vskip\abovecaptionskip
  \sbox\@tempboxa{\small #1. #2}%
  \ifdim \wd\@tempboxa >\hsize
    \small #1. #2\par
  \else
    \global \@minipagefalse
    \hb@xt@\hsize{\hfil\box\@tempboxa\hfil}%
  \fi
  \vskip\belowcaptionskip}

\renewenvironment{thebibliography}[1]{%
\begin{oldthebibliography}{#1}%
\small%
\raggedright%
\setlength{\itemsep}{5pt plus 0.2ex minus 0.05ex}%
}%
{%
\end{oldthebibliography}%
}

\begin{document}


\title{\boldmath Quantum Newton duality}

\author[a,b]{Wen-Du Li}
\author[b,1]{and Wu-Sheng Dai}\note{daiwusheng@tju.edu.cn.}

\affiliation[a]{Theoretical Physics Division, Chern Institute of Mathematics, Nankai University, Tianjin, 300071, P. R. China}
\affiliation[b]{Department of Physics, Tianjin University, Tianjin 300350, P.R. China}








\abstract{Newton revealed an underlying duality relation between power potentials in
classical mechanics. In this paper, we establish the quantum version of the
Newton duality. The main aim of this paper is threefold: (1) first
generalizing the original Newton duality to more general potentials, including
general polynomial potentials and transcendental-function potentials, 2)
constructing a quantum version of the Newton duality, including power
potentials, general polynomial potentials, transcendental-function potentials,
and power potentials in different spatial dimensions, and 3) suggesting a
method for solving eigenproblems in quantum mechanics based on the quantum
Newton duality provided in the paper. The classical Newton duality is a
duality among orbits of classical dynamical systems. Our result shows that the
Newton duality is not only limited to power potentials, but a\ more universal
duality relation among dynamical systems with various potentials. The key task
of this paper is to construct a quantum Newton duality, the quantum version of
the classical Newton duality. The quantum Newton duality provides a duality
relations among wave functions and eigenvalues. As applications, we suggest a
method for solving potentials from their Newtonianly dual potential: once the
solution of a potential is known, the solution of all its dual potentials can
be obtained by the duality transformation directly. Using this method, we
obtain a series of exact solutions of various potentials. In appendices, as
preparations, we solve the potentials which is solved by the Newton duality
method in this paper by directly solving the eigenequation.}


\maketitle
\flushbottom


\section{Introduction}

In classical mechanics, Newton revealed a duality between power potentials in
his \textit{Principia\ }(Corollary III of Proposition VII)\textit{
}\cite{chandrasekhar1995newton}. Newton's duality relates the orbits of two
power potentials $U\left(  r\right)  =\xi r^{a+1}$ and $V\left(  \rho\right)
=\eta\rho^{A+1}$ by the duality relation $\frac{a+3}{2}\leftrightarrow\frac
{2}{A+3}$, $E\leftrightarrow-\eta$, and $\xi\leftrightarrow-\mathcal{E}$,
where $E$ and $\mathcal{E}$ are the energies of the orbits of $U\left(
r\right)  $ and $V\left(  \rho\right)  $, respectively. That is, for these two
power potentials $U\left(  r\right)  $ and $V\left(  \rho\right)  $, when
$\frac{a+3}{2}=\frac{2}{A+3}$, the energy of a system equals the negative
coupling constant of its dual system, and then the orbit of a potential can be
obtained from the orbit of its Newtonianly dual potential.

The Newton duality reveals profound dynamical nature of mechanical systems.
The original question Newton asked is that given a power law of centripetal
attraction, does there exist a dual law for which a body with the same
constant of areas will describe the same orbit \cite{chandrasekhar1995newton}.
The modern formulation of the Newton duality is given by some authors
\cite{arnold1990huygens,needham1993newton,hall2000planetary}. Newton pointed
out that between\ the harmonic-oscillator potential ($r^{2}$-potential)\ and
the Coulomb potential ($1/r$-potential), there exists such a duality
\cite{chandrasekhar1995newton}.

We will construct a quantum version for the Newton duality. The quantum Newton
duality constructed in the present paper includes not only power-potentials,
but also more general potentials.

(1) Newton himself found the classical Newton duality between power
potentials. In this paper, we generalize the Newton duality to a more general
case --- general polynomial potentials. The general polynomial potential is a
linear combination of power potentials with arbitrary real number powers.

(2) A quantum version of the classical Newton duality is a main result of this
paper. The quantum Newton duality provides a duality relation between wave
functions and eigenvalues, just like\ that the classical Newton duality
provides a duality relation between orbits. We will show that by the quantum
Newton duality relation between two dual potentials, one can transform the
wave function and eigenvalue of one potential to those of its dual potential.
Besides power potentials, we provide duality relations for more general kinds
of potentials, including general polynomial potentials and transcendental
function potentials. Moreover, for power potentials, beyond three dimensions,
we also consider the quantum Newton duality of two power potentials in
different spatial dimensions. It should be pointed out that if two potentials
are classical Newtonianly dual, they are also quantum Newtonianly dual.

(3) Another important issue of the present paper is to make the quantum Newton
duality serve as a method of solving eigenproblems in quantum mechanics. The
quantum Newton duality allows us to transform the eigenfunction and eigenvalue
of one potential to the eigenfunction and eigenvalue of its dual potential.
That is, if the solution of one potential is known, one can obtain the
solution of its dual potential through the duality relation directly.
Concretely, the original Newton duality is a duality between two one-term
power potentials. A one-term potential has one dual potential; in other words,
a potential $V_{1}$ with its duality $V_{2}$ constitute a dual pair $\left(
V_{1},V_{2}\right)  $. In this paper, we generalize the Newton duality to
general polynomial potentials with arbitrary terms. It will be shown that an
$N$-term general polynomial potential $V_{1}$ has $N$ dual potentials $V_{2},$
$V_{3},$ $\ldots,V_{N+1}$. These $N+1$ potentials, which are Newtonianly dual
to each other, constitute a dual set $\left(  V_{1},V_{2},\ldots
,V_{N+1}\right)  $. The eigenfunction and eigenvalue of the potentials in a
dual set are related by the duality relation. In a dual set $\left(
V_{1},V_{2},\ldots,V_{N+1}\right)  $, if a potential, say, $V_{i}$, is solved,
then the solution of the other potentials in the dual set can be immediately
obtained by the duality relation. The Newton duality provides us an efficient
tool for seeking the solution in mechanics.

To illustrate how this works, we solve some potentials from their exactly
solved dual potentials by the duality relation. The potentials solved in this
paper belong to the following dual sets: ($\alpha r^{2/3}$, $\frac{\alpha
}{\sqrt{r}}$), ($\alpha r^{6}$, $\frac{\alpha}{r^{3/2}}$), ($\beta e^{\alpha
r}$, $\frac{\beta}{\left(  r\ln\alpha r\right)  ^{2}}$), ($\alpha r^{2}%
+\frac{\beta}{r}$, $\frac{\alpha}{r}+\frac{\beta}{r^{3/2}}$, $\alpha
r^{2}+\beta r^{6}$), ($\alpha r^{2}+\beta r$,$\frac{\alpha}{r}+\frac{\beta
}{\sqrt{r}}$, $\frac{\alpha}{r^{2/3}}+\beta r^{2/3}$), ($\alpha r^{2}+\beta
r$, $\frac{\alpha}{r}+\frac{\beta}{\sqrt{r}}$, $\frac{\alpha}{r^{2/3}}+\beta
r^{2/3}$), ($\frac{\alpha}{\sqrt{r}}+\frac{\beta}{r^{3/2}}$, $\alpha
r^{2/3}+\frac{\beta}{r^{4/3}}$, $\alpha r^{6}+\beta r^{4}$), ($\alpha
r^{6}+\beta r^{2}$, $\frac{\alpha}{r^{3/2}}+\frac{\beta}{r}$), and ($\alpha
r^{2}+\frac{\beta}{r}+\sigma r$, $\frac{\alpha}{r}+\frac{\beta}{r^{3/2}}%
+\frac{\sigma}{r^{1/2}}$, $\alpha r^{2}+\beta r^{6}+\sigma r^{4}$,
$\frac{\alpha}{r^{2/3}}+\beta r^{2/3}+\frac{\sigma}{r^{4/3}}$). Especially,
besides solving such potentials with the help of the quantum Newton duality
relation, as verifications, we also solve them by the conventional method,
i.e., directly solving the eigenequation, in Appendices (\ref{Vm32}%
)-(\ref{Vm2323m43}).

In section \ref{classical}, we give a brief review for the original classical
Newton duality. In sections \ref{classicalGD} and \ref{classicalGPP}, we
discuss the generalization of the classical Newton duality. In section
\ref{classicalGD}, we give a general discussion on the classical Newton
duality, which shows the possibility of the generalization of Newton's
original result. In section \ref{classicalGPP}, the original classical Newton
duality which is only a duality between power potentials is generalized to
general polynomial potentials. In sections \ref{general}-\ref{eln}, we provide
a main result of this paper, the quantum Newton duality. A general discussion
of the quantum Newton duality is given in section \ref{general}. In section
\ref{3D}, we give a quantum version of the original Newton duality between
power potentials. In section \ref{arbitraryDIM}, we provide the quantum Newton
duality between power potentials in different spatial dimensions, i.e., a
duality between an $n$-dimensional potential and its $m$-dimensional dual
potential. In section \ref{polynomialnterms}, the quantum Newton duality for
general polynomial potentials is given. In section \ref{polynomialntermsmany},
we discuss the problem of the dual set. In section \ref{eln}, we give an
example of the Newton duality between two transcendental-function potentials.
In sections \ref{SolvingGD}-\ref{Solvingeln}, we show how to solve the
eigenproblem by the duality relation. As examples, we consider some sets of
dual potentials. Examples of dual sets of one-term potentials, including
three-dimensional harmonic-oscillator potentials and three-dimensional Coulomb
potentials, $n$-dimensional harmonic-oscillator potentials and $m$-dimensional
Coulomb potentials, ($\alpha r^{2/3}$, $\frac{\alpha}{\sqrt{r}}$), ($\alpha
r^{2/3}$, $\frac{\alpha}{\sqrt{r}}$), and ($\alpha r^{6}$, $\frac{\alpha
}{r^{3/2}}$) are given in section \ref{SolvingPP}. Examples of dual sets of
two-term potentials, ($\alpha r^{2}+\frac{\beta}{r}$, $\frac{\alpha}{r}%
+\frac{\beta}{r^{3/2}}$, $\alpha r^{2}+\beta r^{6}$), ($\alpha r^{2}+\beta
r$,$\frac{\alpha}{r}+\frac{\beta}{\sqrt{r}}$, $\frac{\alpha}{r^{2/3}}+\beta
r^{2/3}$), ($\alpha r^{2}+\beta r$, $\frac{\alpha}{r}+\frac{\beta}{\sqrt{r}}$,
$\frac{\alpha}{r^{2/3}}+\beta r^{2/3}$), ($\frac{\alpha}{\sqrt{r}}+\frac
{\beta}{r^{3/2}}$, $\alpha r^{2/3}+\frac{\beta}{r^{4/3}}$, $\alpha r^{6}+\beta
r^{4}$), and ($\alpha r^{6}+\beta r^{2}$, $\frac{\alpha}{r^{3/2}}+\frac{\beta
}{r}$) are given in section \ref{Solving2T}. An example of a dual set of
three-term potentials ($\alpha r^{2}+\frac{\beta}{r}+\sigma r$, $\frac{\alpha
}{r}+\frac{\beta}{r^{3/2}}+\frac{\sigma}{r^{1/2}}$, $\alpha r^{2}+\beta
r^{6}+\sigma r^{4}$, $\frac{\alpha}{r^{2/3}}+\beta r^{2/3}+\frac{\sigma
}{r^{4/3}}$) is given in section \ref{Solving3T}. A dual set of
transcendental-function potentials ($\frac{\xi}{\left(  r\ln\beta r\right)
^{2}}$, $\eta e^{\sigma\rho}$) is given in \ref{eln}. The conclusion is given
in section \ref{Conclusion}. Moreover, in Appendices (\ref{Vm32}%
)-(\ref{Vm2323m43}), as examples and verifications, we solve exact solutions
for some potentials.

\section{The classical Newton duality}

\subsection{The original Newton duality of power potentials: Revisit
\label{classical}}

The original Newton duality revealed by Newton himself is a duality between
two power potentials in classical mechanics. A modern expression of the
original\ classical Newton duality can be found in Refs.
\cite{arnold1990huygens,needham1993newton,hall2000planetary}. For
completeness, in this section we give a brief review on the original classical
Newton duality. Our expression and proof are somewhat different from the
original version.

\begin{theorem}
Two power potentials%
\begin{align}
U\left(  r\right)   &  =\xi r^{a+1},\\
V\left(  \rho\right)   &  =\eta\rho^{A+1},
\end{align}
are classically Newtonianly dual to each other, if
\begin{equation}
\frac{a+3}{2}=\frac{2}{A+3}. \label{aA}%
\end{equation}
The orbit with the energy $\mathcal{E}$ of the potential $V\left(
\rho\right)  $ can be obtained by performing the replacements%
\begin{align}
E  &  \rightarrow-\eta,\label{Eh}\\
\xi &  \rightarrow-\mathcal{E} \label{XiE}%
\end{align}
to the orbit with the energy $E$ of the potential $U\left(  r\right)  $.
\textit{As a result, the corresponding transformation of coordinates is}%
\begin{align}
r  &  \rightarrow\rho^{\left(  A+3\right)  /2}\text{ \ \ or \ \ }r^{\left(
a+3\right)  /2}\rightarrow\rho,\\
\theta &  \rightarrow\frac{A+3}{2}\phi\text{ \ \ or \ \ }\frac{a+3}{2}%
\theta\rightarrow\phi.
\end{align}
\textit{ }
\end{theorem}

\begin{proof}
For the potential $U\left(  r\right)  =\xi r^{a+1}$, the equation of the orbit
$r=r\left(  \theta\right)  $ of the energy $E$ reads
\cite{goldstein2002classical}%
\begin{equation}
\frac{d\theta}{dr}=\frac{L/r^{2}}{\sqrt{2\left[  E-L^{2}/\left(
2r^{2}\right)  -\xi r^{a+1}\right]  }}, \label{dthetadr}%
\end{equation}
where $L$ is the angular momentum.

Substituting the replacements (\ref{aA}), (\ref{Eh}), and (\ref{XiE}) into the
orbit equation (\ref{dthetadr}), we arrive at%
\begin{equation}
\frac{d\theta}{dr}=\frac{L/r^{2}}{\sqrt{2\left[  \left(  -\eta\right)
-L^{2}/\left(  2r^{2}\right)  -\left(  -\mathcal{E}\right)  r^{\frac{4}%
{A+3}-2}\right]  }}.
\end{equation}
Rearranging the equation as%
\begin{equation}
\frac{d\left(  \frac{2}{A+3}\theta\right)  }{d\left(  r^{\frac{2}{A+3}%
}\right)  }=\frac{L/\left(  r^{\frac{2}{A+3}}\right)  ^{2}}{\sqrt{2\left\{
\mathcal{E}-L^{2}/\left[  2\left(  r^{\frac{2}{A+3}}\right)  ^{2}\right]
-\eta\left(  r^{\frac{2}{A+3}}\right)  ^{A+1}\right\}  }},
\end{equation}
we can see that this is just the equation of the orbit $\rho\left(
\phi\right)  $ with the energy $\mathcal{E}$ of the potential $V\left(
\rho\right)  =\eta\rho^{A+1}$ with a transformation of coordinates%
\begin{align}
r^{\frac{2}{A+3}}  &  \rightarrow\rho,\\
\frac{2}{A+3}\theta &  \rightarrow\phi.
\end{align}

\end{proof}

It can be seen that the energy of an orbit of a system equals the negative
coupling constant of its dual system.

It is worth mentioning that in classical mechanics the angular momenta of
these two Newtonianly dual systems are the same.

\subsection{General discussions \label{classicalGD}}

The original Newton duality is only for power potentials. In this paper, we
will generalize Newton's result to more general cases. To show the possibility
of the generalization of the original classical Newton duality, we first give
a general discussion on the duality in classical mechanics.

In classical mechanics, the orbit equation of the energy $E$ of the potential
$U\left(  r\right)  $ reads \cite{goldstein2002classical}%
\begin{equation}
\frac{d\theta\left(  r\right)  }{dr}=\frac{L/r^{2}}{\sqrt{2\left[
E-L^{2}/\left(  2r^{2}\right)  -U\left(  r\right)  \right]  }}. \label{eqr}%
\end{equation}
Writing the transformation of coordinates on the orbital plane as%
\begin{align}
r  &  \rightarrow g\left(  \rho\right)  ,\\
\theta\left(  r\right)   &  \rightarrow f\left(  \rho\right)  \phi\left(
\rho\right)
\end{align}
gives
\begin{equation}
\frac{d\phi\left(  \rho\right)  }{d\rho}=\frac{L/\rho^{2}}{\sqrt{2\left[
E\left(  \frac{g^{2}\left(  \rho\right)  f\left(  \rho\right)  }{g^{\prime
}\left(  \rho\right)  \rho^{2}}\right)  ^{2}-\frac{L^{2}}{2g^{2}\left(
\rho\right)  }\left(  \frac{g^{2}\left(  \rho\right)  f\left(  \rho\right)
}{g^{\prime}\left(  \rho\right)  \rho^{2}}\right)  ^{2}-U\left(  \rho\right)
\left(  \frac{g^{2}\left(  \rho\right)  f\left(  \rho\right)  }{g^{\prime
}\left(  \rho\right)  \rho^{2}}\right)  ^{2}\right]  }}-\frac{f^{\prime
}\left(  \rho\right)  }{f\left(  \rho\right)  }\phi\left(  \rho\right)  .
\label{eqphirho1}%
\end{equation}
Here the only change is a transformation of coordinates, so, of course, Eq.
(\ref{eqphirho1}) must also be an orbit equation.

In order to make Eq. (\ref{eqphirho1}) still serve as an orbit equation, it is
required that%
\begin{equation}
\frac{f^{\prime}\left(  \rho\right)  }{f\left(  \rho\right)  }=0,
\end{equation}
i.e.,%
\begin{equation}
f\left(  \rho\right)  =\beta\label{clafbeta}%
\end{equation}
with $\beta$ a constant. Substituting Eq. (\ref{clafbeta}) into Eq.
(\ref{eqphirho1}) gives%
\begin{equation}
\frac{d\phi\left(  \rho\right)  }{d\rho}=\frac{L/\rho^{2}}{\sqrt{2\left[
\beta^{2}E\left(  \frac{g^{2}\left(  \rho\right)  }{g^{\prime}\left(
\rho\right)  \rho^{2}}\right)  ^{2}-\beta^{2}\frac{L^{2}}{2g^{2}\left(
\rho\right)  }\left(  \frac{g^{2}\left(  \rho\right)  }{g^{\prime}\left(
\rho\right)  \rho^{2}}\right)  ^{2}-\beta^{2}U\left(  \rho\right)  \left(
\frac{g^{2}\left(  \rho\right)  }{g^{\prime}\left(  \rho\right)  \rho^{2}%
}\right)  ^{2}\right]  }}. \label{eqphirho2}%
\end{equation}
Insisting that Eq. (\ref{eqphirho2}) is still an orbit equation, we must
requires that in
\begin{equation}
\beta^{2}E\left(  \frac{g^{2}\left(  \rho\right)  }{g^{\prime}\left(
\rho\right)  \rho^{2}}\right)  ^{2}-\beta^{2}\frac{L^{2}}{2g^{2}\left(
\rho\right)  }\left(  \frac{g^{2}\left(  \rho\right)  }{g^{\prime}\left(
\rho\right)  \rho^{2}}\right)  ^{2}-\beta^{2}U\left(  \rho\right)  \left(
\frac{g^{2}\left(  \rho\right)  }{g^{\prime}\left(  \rho\right)  \rho^{2}%
}\right)  ^{2}, \label{transeq}%
\end{equation}%
\begin{align}
&  \text{serving as the centrifugal potential, there must exist a term being
proportional to }1/\rho^{2}\text{,}\label{classicalcd1}\\
&  \text{serving as the energy, there must exist a constant term which is
independent of }\rho\text{. } \label{classicalcd2}%
\end{align}

The choice is not unique, and different choices lead to different dualities.

\subsection{General polynomial potentials \label{classicalGPP}}

In this section, we generalize the classical Newton duality to general
polynomial potentials. By \textit{general polynomial }here we mean a
superposition power series containing arbitrary real-number powers.
The\ original Newton duality, a duality between power potentials, is a special
case of this general duality since a power potential is nothing but a one-term
general polynomial potential.

\subsubsection{Generalizing the classical Newton duality to general polynomial
potentials}

\begin{theorem}
\label{cndgpp} \textit{Two general polynomial potentials}%
\begin{equation}
U\left(  r\right)  =\xi r^{a+1}+\sum_{n}\mu_{n}r^{b_{n}+1}\text{ \ and
\ }V\left(  \rho\right)  =\eta\rho^{A+1}+\sum_{n}\lambda_{n}\rho^{B_{n}+1}%
\end{equation}
\textit{are classically Newtonianly dual to each other, if}%
\begin{align}
\frac{a+3}{2}  &  =\frac{2}{A+3},\label{claAanda}\\
\sqrt{\frac{2}{a+3}}\left(  b_{n}+3\right)   &  =\sqrt{\frac{2}{A+3}}\left(
B_{n}+3\right)  . \label{claBandb}%
\end{align}
The orbit with the energy $\mathcal{E}$ of the potential $V\left(
\rho\right)  $ can be obtained by performing the replacements%
\begin{align}
E  &  \rightarrow-\eta,\label{claEta}\\
\xi &  \rightarrow-\mathcal{E},\label{claxieta}\\
\mu_{n}  &  \rightarrow\lambda_{n} \label{clamulamb}%
\end{align}
to the orbit with the energy $E$ of the potential $U\left(  r\right)  $.
\textit{As a result, the corresponding transformation of coordinates is}%
\begin{align}
r  &  \rightarrow\rho^{\left(  A+3\right)  /2}\text{ \ \ or \ \ }r^{\left(
a+3\right)  /2}\rightarrow\rho,\\
\theta &  \rightarrow\frac{A+3}{2}\phi\text{ \ \ or \ \ }\frac{a+3}{2}%
\theta\rightarrow\phi.
\end{align}

\end{theorem}

\begin{proof}
For the potential $U\left(  r\right)  =\xi r^{a+1}+\sum_{n}\mu_{n}r^{b_{n}+1}%
$, the orbit $r=r\left(  \theta\right)  $ of the energy $E$ reads%
\begin{equation}
\frac{d\theta}{dr}=\frac{L/r^{2}}{\sqrt{2\left[  E-L^{2}/\left(
2r^{2}\right)  -\left(  \xi r^{a+1}+\sum_{n}\mu_{n}r^{b_{n}+1}\right)
\right]  }}. \label{claorb}%
\end{equation}

Substituting the replacements (\ref{claAanda}), (\ref{claEta}),
(\ref{claxieta}), and (\ref{clamulamb}) into the orbit equation (\ref{claorb}%
), we arrive at%
\begin{equation}
\frac{d\theta}{dr}=\frac{L/r^{2}}{\sqrt{2\left[  \left(  -\eta\right)
-L^{2}/\left(  2r^{2}\right)  -\left(  -\mathcal{E}\right)  r^{\frac{4}%
{A+3}-2}-\sum_{n}\lambda_{n}r^{\frac{2}{A+3}\left(  B_{n}+3\right)
-2}\right]  }},
\end{equation}
Rearranging the equation as%
\begin{equation}
\frac{d\left(  \frac{2}{A+3}\theta\right)  }{d\left(  r^{\frac{2}{A+3}%
}\right)  }=\frac{L/\left(  r^{\frac{2}{A+3}}\right)  ^{2}}{\sqrt{2\left\{
\mathcal{E}-L^{2}/\left[  2\left(  r^{\frac{2}{A+3}}\right)  ^{2}\right]
-\eta\left(  r^{\frac{2}{A+3}}\right)  ^{A+1}-\sum_{n}\lambda_{n}\left(
r^{\frac{2}{A+3}}\right)  ^{B_{n}+1}\right\}  }},
\end{equation}
we can see that this is just an orbit equation of the energy $\mathcal{E}$ for
the potential $V\left(  \rho\right)  =\eta\rho^{A+1}+\sum\lambda_{n}%
\rho^{B_{n}+1}$ with the transformation of coordinates%
\begin{align}
r^{\frac{2}{A+3}}  &  \rightarrow\rho,\\
\frac{2}{A+3}\theta &  \rightarrow\phi.
\end{align}

\end{proof}

It is worth showing how the result given by Theorem \ref{cndgpp} is reached.

The classical Newton duality for general polynomial potentials given by
Theorem \ref{cndgpp} can be obtained by choosing the term $\beta^{2}%
\frac{L^{2}}{2g^{2}\left(  \rho\right)  }\left(  \frac{g^{2}\left(
\rho\right)  }{g^{\prime}\left(  \rho\right)  \rho^{2}}\right)  ^{2}$ in Eq.
(\ref{transeq}) as the centrifugal potential term, i.e.,%
\begin{equation}
\beta^{2}\frac{L^{2}}{2g^{2}\left(  \rho\right)  }\left(  \frac{g^{2}\left(
\rho\right)  }{g^{\prime}\left(  \rho\right)  \rho^{2}}\right)  ^{2}%
=\frac{L^{2}}{2\rho^{2}} \label{cgppL}%
\end{equation}
to satisfy Condition \ref{classicalcd1}. Solving Eq. (\ref{cgppL}) gives
\begin{equation}
g\left(  \rho\right)  =\rho^{\beta},
\end{equation}
i.e., the transformation of coordinates is%
\begin{equation}
r\rightarrow\rho^{\beta}\text{ \ or \ }r^{1/\beta}\rightarrow\rho.
\label{rrho}%
\end{equation}
With the transformation of coordinates (\ref{rrho}), Eq. (\ref{transeq})
becomes%
\begin{equation}
E\rho^{2\beta-2}-\frac{L^{2}}{2\rho^{2}}-\rho^{2\beta-2}U\left(  \rho\right)
. \label{rrho1}%
\end{equation}
To satisfy Condition \ref{classicalcd2}, we choose%
\begin{equation}
-\rho^{2\beta-2}U\left(  \rho\right)  =\mathcal{E}-\sum_{n}\kappa_{n}\rho^{n}.
\end{equation}
Then we have%
\begin{equation}
U\left(  \rho\right)  =-\mathcal{E}\rho^{2-2\beta}+\sum_{n}\kappa_{n}%
\rho^{n+2-2\beta},
\end{equation}
and Eq. (\ref{transeq}) becomes%
\begin{equation}
\mathcal{E}-\frac{L^{2}}{2\rho^{2}}+E\rho^{2\beta-2}-\sum_{n}\kappa_{n}%
\rho^{n}.
\end{equation}
Applying the transformation of coordinates (\ref{rrho}), we obtain a pair of
dual potentials%
\begin{align}
U\left(  r\right)   &  =-\mathcal{E}r^{\left(  2-2\beta\right)  /\beta}%
+\sum_{n}\kappa_{n}r^{\left(  n+2-2\beta\right)  /\beta},\label{Urcgp}\\
V\left(  \rho\right)   &  =-E\rho^{2\beta-2}+\sum_{n}\kappa_{n}\rho^{n}.
\label{Vrhocgp}%
\end{align}

Rewrite the potential $U\left(  r\right)  $ given by Eq. (\ref{Urcgp}) as%
\begin{equation}
U\left(  r\right)  =\xi r^{a+1}+\sum_{n}\mu_{n}r^{b_{n}+1},
\end{equation}
with%
\begin{align}
\xi &  \rightarrow-\mathcal{E},\label{claxieta2}\\
\beta &  \rightarrow\frac{2}{a+3},\label{clabetaa}\\
b_{n}  &  \rightarrow\frac{\left(  2+n\right)  \left(  a+3\right)  }%
{2}-3,\label{claba}\\
\kappa_{n}  &  \rightarrow\mu_{n}; \label{kappamu}%
\end{align}
rewrite the potential $V\left(  \rho\right)  $ given by Eq. (\ref{Vrhocgp})
as
\begin{equation}
V\left(  \rho\right)  =\eta\rho^{A+1}+\sum_{n}\lambda_{n}\rho^{B_{n}+1},
\end{equation}
with%
\begin{align}
E  &  \rightarrow-\eta,\label{claeeta}\\
\beta &  \rightarrow\frac{A+3}{2},\label{clabetaA}\\
B_{n}  &  \rightarrow n-1,\label{claB}\\
\kappa_{n}  &  \rightarrow\lambda_{n}. \label{kappalambda}%
\end{align}
Then comparing Eq. (\ref{claxieta2}) with Eq. (\ref{claeeta}), Eq.
(\ref{clabetaa}) with Eq. (\ref{clabetaA}), Eq. (\ref{claba}) with\ Eq.
(\ref{claB}), and Eq. (\ref{kappamu}) with Eq. (\ref{kappalambda}), we arrive
at the duality relation given by Theorem \ref{cndgpp}.

\subsubsection{General polynomial potentials consist of two Newton dual
potentials \label{classical polynomial}}

In this section, we consider an interesting special case of the Newton duality
of general polynomial potentials. In this case, $U_{1}\left(  r\right)  $ and
$V_{1}\left(  \rho\right)  $ are Newtonianly dual, $U_{2}\left(  r\right)  $
and $V_{2}\left(  \rho\right)  $ are Newtonianly dual, and their sum $U\left(
r\right)  =U_{1}\left(  r\right)  +U_{2}\left(  r\right)  $ and $V\left(
\rho\right)  =V_{1}\left(  \rho\right)  +V_{2}\left(  \rho\right)  $ are also
Newtonianly dual.

\begin{theorem}
Two general polynomial potentials\newline%
\begin{align}
U\left(  r\right)   &  =\xi r^{a+1}+\mu r^{2\left(  \sqrt{\frac{a+3}{2}%
}-1\right)  },\\
V\left(  \rho\right)   &  =\eta\rho^{A+1}+\mu\rho^{2\left(  \sqrt{\frac
{A+3}{2}}-1\right)  },
\end{align}
are Newtonianly dual, while $\xi r^{a+1}$ and $\eta\rho^{A+1}$ are Newtonianly
dual and $\mu r^{2\left(  \sqrt{\frac{a+3}{2}}-1\right)  }$ and $\mu
\rho^{2\left(  \sqrt{\frac{A+3}{2}}-1\right)  }$ are Newtonianly dual, if
\begin{equation}
\frac{a+3}{2}=\frac{2}{A+3}.
\end{equation}

\end{theorem}

The proof is similar to the proof of Theorem \ref{cndgpp}.

\section{The quantum Newton duality}

\subsection{General discussions \label{general}}

The main aim of this paper is to provide a quantum version of the classical
Newton duality, including power potentials, general polynomial potentials, and
some other potentials.

In quantum mechanics, instead of the orbit, the wave function is used to
describe a mechanical system. Therefore, the quantum Newton duality is a
duality relation between wave functions. Moreover, the energy of
quantum-mechanical bound states is discrete, so we also need a duality
relation between discrete bound-state eigenvalue spectra.

Now we first give a general discussion on the quantum Newton duality.

Generally speaking, the existence of the Newton duality between two potentials
means that there exists a duality relation which can transform the
eigenfunction and eigenvalue of a potential to those of the other potential.
In other words, the eigenfunction and eigenvalue of a potential can be
achieved from the solution of its dual potential with the help of the duality relation.

Consider two potentials $U\left(  r\right)  $ and $V\left(  \rho\right)  $.
The radial equation of $U\left(  r\right)  $ is
\begin{equation}
\frac{d^{2}u\left(  r\right)  }{dr^{2}}+\left[  E-\frac{l\left(  l+1\right)
}{r^{2}}-U\left(  r\right)  \right]  u\left(  r\right)  =0 \label{1}%
\end{equation}
with the eigenfunction $u\left(  r\right)  $ and eigenvalue $E$; the radial
equation of $V\left(  \rho\right)  $ is
\begin{equation}
\frac{d^{2}v\left(  \rho\right)  }{d\rho^{2}}+\left[  \mathcal{E}-\frac
{\ell\left(  \ell+1\right)  }{\rho^{2}}-V\left(  \rho\right)  \right]
v\left(  \rho\right)  =0,
\end{equation}
with the eigenfunction $v\left(  \rho\right)  $ and eigenvalue $\mathcal{E}$.
Two potentials $U\left(  r\right)  $ and $V\left(  \rho\right)  $ are called a
pair of dual potentials, if there exists a duality transform which can
transform $u\left(  r\right)  $ and $E$ to $v\left(  \rho\right)  $ and
$\mathcal{E}$.

Generally, write the duality transformations of coordinates and eigenfunctions
as
\begin{align}
r  &  \rightarrow g\left(  \rho\right)  ,\label{t1}\\
u\left(  r\right)   &  \rightarrow f\left(  \rho\right)  v\left(  \rho\right)
, \label{t2}%
\end{align}
where the duality transform is described by the functions $g\left(
\rho\right)  $ and $f\left(  \rho\right)  $. Substituting Eqs. (\ref{t1}) and
(\ref{t2}) into the radial equation of $U\left(  r\right)  $, Eq. (\ref{1}),
gives
\begin{equation}
v^{\prime\prime}\left(  \rho\right)  +\left(  2\frac{f^{\prime}\left(
\rho\right)  }{f\left(  \rho\right)  }-\frac{g^{\prime\prime}\left(
\rho\right)  }{g^{\prime}\left(  \rho\right)  }\right)  v^{\prime}\left(
\rho\right)  +\left(  \frac{f^{\prime\prime}\left(  \rho\right)  }{f\left(
\rho\right)  }-\frac{f^{\prime}\left(  \rho\right)  }{f\left(  \rho\right)
}\frac{g^{\prime\prime}\left(  \rho\right)  }{g^{\prime}\left(  \rho\right)
}\right)  v\left(  \rho\right)  +g^{\prime}\left(  \rho\right)  ^{2}\left[
E-\frac{l\left(  l+1\right)  }{g^{2}\left(  \rho\right)  }-U\left(  g\left(
\rho\right)  \right)  \right]  v\left(  \rho\right)  =0. \label{2}%
\end{equation}

If requiring that Eq. (\ref{2}) is the radial equation of the potential
$V\left(  \rho\right)  $ with the eigenfunction $v\left(  \rho\right)  $, the
first-order derivative term in Eq. (\ref{2}) must vanish, i.e.,%
\begin{equation}
2\frac{f^{\prime}\left(  \rho\right)  }{f\left(  \rho\right)  }-\frac
{g^{\prime\prime}\left(  \rho\right)  }{g^{\prime}\left(  \rho\right)  }=0,
\label{fdgdd}%
\end{equation}
and the solution of Eq. (\ref{fdgdd}) gives a relation between $f\left(
\rho\right)  $ and $g\left(  \rho\right)  $,%
\begin{equation}
f\left(  \rho\right)  =\sqrt{g^{\prime}\left(  \rho\right)  }. \label{fg}%
\end{equation}

Then substituting Eq. (\ref{fg}) into Eq. (\ref{2}) gives%
\begin{equation}
v^{\prime\prime}\left(  \rho\right)  +\left[  g^{\prime}\left(  \rho\right)
^{2}E+\frac{g^{\prime\prime\prime}\left(  \rho\right)  }{2g^{\prime}\left(
\rho\right)  }-\frac{3g^{\prime\prime}\left(  \rho\right)  ^{2}}{4g^{\prime
}\left(  \rho\right)  ^{2}}-l\left(  l+1\right)  \frac{g^{\prime}\left(
\rho\right)  ^{2}}{g^{2}\left(  \rho\right)  }-g^{\prime}\left(  \rho\right)
^{2}U\left(  g\left(  \rho\right)  \right)  \right]  v\left(  \rho\right)  =0.
\label{3}%
\end{equation}
\ \ \ \ \ \ \ \ \ \ \ \ \ \

If insisting that Eq. (\ref{3}) is still a radial equation, then it must
requires that%
\begin{align}
&  \text{there must exist a term in the coefficient of }v\left(  \rho\right)
\text{ playing the role of the centrifugal potential,}\nonumber\\
&  \text{which should be porportional to }1/\rho^{2}\text{,}\label{cd1}\\
&  \text{there must exist a term in the coefficient of }v\left(  \rho\right)
\text{ playing the role of the eigenvalue,}\nonumber\\
&  \text{which should be a constant }\mathcal{E}\text{.} \label{cd2}%
\end{align}

Different choices lead different duality relations.

\subsection{Three-dimensional power potentials \label{3D}}

The original Newton duality in \textit{Principia} is the duality between power
potentials. In this section, we present a quantum version of the Newton
duality for three-dimensional power potentials.

\begin{theorem}
\label{Twf3D} Two power potentials%
\begin{equation}
U\left(  r\right)  =\xi r^{a+1}\text{ \ and }V\left(  \rho\right)  =\eta
\rho^{A+1} \label{UV}%
\end{equation}
are quantum Newtonianly dual to each other, if%
\begin{equation}
\frac{a+3}{2}=\frac{2}{A+3}. \label{aA3D}%
\end{equation}
The bound-state eigenfunction of the energy $\mathcal{E}$ and the angular
quantum number $\ell$ of the potential $V\left(  \rho\right)  $\ can be
obtained by performing the replacements%
\begin{align}
E  &  \rightarrow-\eta\left(  \frac{2}{A+3}\right)  ^{2},\label{coup}\\
\xi &  \rightarrow-\mathcal{E}\left(  \frac{2}{A+3}\right)  ^{2} \label{en}%
\end{align}
to the bound-state eigenfunction of the energy $E$ and the angular quantum
number $l$ of the potential $U\left(  r\right)  $. As a result, the duality
relation of the angular momentum between these two dual systems is
\begin{equation}
l+\frac{1}{2}\rightarrow\frac{2}{A+3}\left(  \ell+\frac{1}{2}\right)  ,
\label{ltrans}%
\end{equation}
the transformation of coordinates is%
\begin{equation}
r\rightarrow\rho^{\left(  A+3\right)  /2}\text{ \ or \ }r^{\left(  a+3\right)
/2}\rightarrow\rho, \label{coorTpower}%
\end{equation}
and the transformation of the eigenfunction is%
\begin{equation}
u\left(  r\right)  \rightarrow\rho^{\left(  A+1\right)  /4}v\left(
\rho\right)  \text{ \ or \ }r^{\left(  a+1\right)  /4}u\left(  r\right)
\rightarrow v\left(  \rho\right)  . \label{wTpower}%
\end{equation}

\end{theorem}

\begin{proof}
The radial equation of the potential $U\left(  r\right)  =\xi r^{a+1}$ reads%
\begin{equation}
\frac{d^{2}u\left(  r\right)  }{dr^{2}}+\left[  E-\frac{l\left(  l+1\right)
}{r^{2}}-\xi r^{a+1}\right]  u\left(  r\right)  =0, \label{radailequr}%
\end{equation}
where $u\left(  r\right)  =R\left(  r\right)  /r$ with $R\left(  r\right)  $
the radial wave function. Substituting the duality relations (\ref{aA3D}),
(\ref{coup}), and (\ref{en}) into the radial equation (\ref{radailequr}) gives%
\begin{equation}
\frac{d^{2}u\left(  r\right)  }{dr^{2}}+\left\{  \left[  -\eta\left(  \frac
{2}{A+3}\right)  ^{2}\right]  -\frac{l\left(  l+1\right)  }{r^{2}}-\left[
-\mathcal{E}\left(  \frac{2}{A+3}\right)  ^{2}\right]  r^{4/\left(
A+3\right)  -2}\right\}  u\left(  r\right)  =0.
\end{equation}
Rearrange the equation as%
\begin{equation}
\frac{d^{2}u\left(  r\right)  }{dr^{2}}+\left(  \frac{2}{A+3}\right)
^{2}\left[  \mathcal{E}-\frac{\left(  \frac{A+3}{2}\right)  ^{2}l\left(
l+1\right)  }{\left(  r^{\frac{2}{A+3}}\right)  ^{2}}-\eta\left(  r^{\frac
{2}{A+3}}\right)  ^{A+1}\right]  \frac{1}{\left(  r^{\frac{2}{A+3}}\right)
^{A+1}}u\left(  r\right)  =0. \label{ur}%
\end{equation}
It can be seen that this is just the radial equation of the potential
$V\left(  \rho\right)  =\eta\rho^{A+1}$,%
\begin{equation}
\frac{d^{2}v\left(  \rho\right)  }{d\rho^{2}}+\left[  \mathcal{E}-\frac
{\ell\left(  \ell+1\right)  }{\rho^{2}}-\eta\rho^{A+1}\right]  v\left(
\rho\right)  =0,
\end{equation}
with the duality relations (\ref{ltrans}), (\ref{coorTpower}), and
(\ref{wTpower}).
\end{proof}

The duality relation between the eigenvalues of two dual potentials can be
obtained directly.

The positive power attractive potential $U\left(  r\right)  =\xi r^{a+1}$,
i.e., $a+1>0$ and\ $\xi>0$, has only bound states with positive eigenvalues.
The duality of $U\left(  r\right)  $ is the potential $V\left(  \rho\right)
=\eta\rho^{A+1}$. The potential $V\left(  \rho\right)  $ is a negative power
attractive potential, i.e., $A+1<0$ and $\eta<0$, with negative bound-state eigenvalues.

\begin{corollary}
\label{Tev3D} If the eigenvalue of the power potential $U\left(  r\right)
=\xi r^{a+1}$ is%
\begin{equation}
E=f\left(  n_{r},l,\xi\right)  , \label{En}%
\end{equation}
where $n_{r}$ is the radial quantum number and $l$ is the angular quantum
number, then the eigenvalue of its Newton duality, the negative power
potential $V\left(  \rho\right)  =\eta\rho^{A+1}$, is%
\begin{equation}
\mathcal{E}=-\left(  \frac{A+3}{2}\right)  ^{2}f^{-1}\left(  n_{r},\frac
{2}{A+3}\left(  \ell+\frac{1}{2}\right)  -\frac{1}{2},-\left(  \frac{2}%
{A+3}\right)  ^{2}\eta\right)  , \label{specf3}%
\end{equation}
where $f^{-1}$ denotes the inverse function of $f$ and $\ell$ is the angular
quantum number of the system of $V\left(  \rho\right)  $.
\end{corollary}

\begin{proof}
Substituting the relation between the two dual potentials, Eqs. (\ref{coup}),
(\ref{en}), and (\ref{ltrans}), into the expression of the eigenvalue of the
potential $U\left(  r\right)  $, Eq. (\ref{En}), gives%
\begin{equation}
-\eta\left(  \frac{2}{A+3}\right)  ^{2}=f\left(  n_{r},\frac{2}{A+3}\left(
\ell+\frac{1}{2}\right)  -\frac{1}{2},-\mathcal{E}\left(  \frac{2}%
{A+3}\right)  ^{2}\right)  . \label{specf}%
\end{equation}

The eigenvalue $\mathcal{E}$ of the potential $V\left(  \rho\right)  $ given
by Eq. (\ref{specf3}) then can be solved directly.
\end{proof}

In the following, we show how the result given by Theorem \ref{Twf3D} is reached.

The Newton duality between power potentials is indeed obtained by choosing
the\ term $g^{\prime}\left(  \rho\right)  ^{2}/g^{2}\left(  \rho\right)  $ in
Eq. (\ref{3}) to be proportional to $1/\rho^{2}$, which will finally become
the centrifugal-potential term or a part of the centrifugal-potential term.

Concretely, first, according to Condition \ref{cd1}, we need to choose a term
in the coefficient of $v\left(  \rho\right)  $ being proportional to
$1/\rho^{2}$ and serving as a centrifugal potential or a part of a centrifugal
potential. For the original Newton duality between power potentials, choose%
\begin{equation}
\frac{g^{\prime}\left(  \rho\right)  ^{2}}{g^{2}\left(  \rho\right)  }%
=\frac{\gamma^{2}}{\rho^{2}} \label{gd2g2}%
\end{equation}
in Eq. (\ref{3}), where $\gamma$ is a constant. Solving Eq. (\ref{gd2g2})
gives%
\begin{equation}
g\left(  \rho\right)  =\rho^{\gamma}. \label{grho1}%
\end{equation}
By Eq. (\ref{grho1}), the radial equation Eq. (\ref{3}) becomes
\begin{equation}
\frac{d^{2}v\left(  \rho\right)  }{d\rho^{2}}+\left[  -\frac{\gamma^{2}}%
{\rho^{2\left(  1-\gamma\right)  }}U\left(  \rho\right)  -\frac{l\left(
l+1\right)  \gamma^{2}-\left(  1-\gamma^{2}\right)  /4}{\rho^{2}}%
+\frac{E\gamma^{2}}{\rho^{2\left(  1-\gamma\right)  }}\right]  v\left(
\rho\right)  =0. \label{vrho2}%
\end{equation}
Clearly, the term $\left[  l\left(  l+1\right)  \gamma^{2}-\frac{1-\gamma^{2}%
}{4}\right]  /\rho^{2}$ is the centrifugal-potential term.

Second, we need to consider Condition \ref{cd2}, which requires us to choose a
term in the coefficient of $v\left(  \rho\right)  $ as a constant term and
serve as the eigenvalue $\mathcal{E}$. Now there is only one possible choice:%
\begin{equation}
-\frac{\gamma^{2}}{\rho^{2\left(  1-\gamma\right)  }}U\left(  \rho\right)
=\mathcal{E},
\end{equation}
which gives%
\begin{equation}
U\left(  \rho\right)  =-\frac{\mathcal{E}}{\gamma^{2}}\rho^{2\left(
1-\gamma\right)  }. \label{Urho1}%
\end{equation}
By Eqs. (\ref{t1})\ and (\ref{grho1}), we can see that the transformation of
coordinates between two dual potentials is%
\begin{equation}
\rho^{\gamma}\rightarrow r. \label{grho11}%
\end{equation}
By this transformation of coordinates, we arrive at the potential%
\begin{equation}
U\left(  r\right)  =-\frac{\mathcal{E}}{\gamma^{2}}r^{2\left(  1-\gamma
\right)  /\gamma}. \label{urpower}%
\end{equation}
Now by Eqs. (\ref{vrho2}) and (\ref{Urho1}), Eq. (\ref{3}) becomes%
\begin{equation}
\frac{d^{2}v\left(  \rho\right)  }{d\rho^{2}}+\left[  \mathcal{E}%
-\frac{l\left(  l+1\right)  \gamma^{2}-\left(  1-\gamma^{2}\right)  /4}%
{\rho^{2}}+\frac{E\gamma^{2}}{\rho^{2\left(  1-\gamma\right)  }}\right]
v\left(  \rho\right)  =0. \label{vrho1}%
\end{equation}
This is just the radial equation of the potential%
\begin{equation}
V\left(  \rho\right)  =-\frac{E\gamma^{2}}{\rho^{2\left(  1-\gamma\right)  }}.
\label{vrhopower}%
\end{equation}
The potential $V\left(  \rho\right)  $ given by Eq. (\ref{vrhopower}) is the
Newton duality of the potential $U\left(  r\right)  $ given by Eq.
(\ref{urpower}).

Rewriting these two potentials as%
\begin{equation}
U\left(  r\right)  =\xi r^{a+1}%
\end{equation}
with%
\begin{align}
\gamma &  =\frac{2}{a+3},\label{gammaa1}\\
\mathcal{E}  &  =\mathcal{-}\left(  \frac{2}{a+3}\right)  ^{2}\xi\label{eipxi}%
\end{align}
and%
\begin{equation}
V\left(  \rho\right)  =\eta\rho^{A+1}%
\end{equation}
with%
\begin{align}
\gamma &  =\frac{A+3}{2},\label{gammaA1}\\
E  &  =-\left(  \frac{2}{A+3}\right)  ^{2}\eta, \label{Eeta}%
\end{align}
we then arrive at the duality relations (\ref{aA3D}), (\ref{coup}), and
(\ref{en}).

The term which is proportional to $1/\rho^{2}$ play a role of the centrifugal
potential, so we have%
\begin{align}
\ell\left(  \ell+1\right)   &  =l\left(  l+1\right)  \gamma^{2}-\frac{1}%
{4}\left(  1-\gamma^{2}\right) \nonumber\\
&  =l\left(  l+1\right)  \left(  \frac{A+3}{2}\right)  ^{2}-\frac{1}{4}\left[
1-\left(  \frac{A+3}{2}\right)  ^{2}\right]  .
\end{align}
Then we achieve the replacement of the angular momentum given by Eq.
(\ref{ltrans}).

The transformation of coordinates, Eq. (\ref{coorTpower}), can be obtained by
substituting Eq. (\ref{gammaA1}) into Eq. (\ref{grho1}). Then$\ $the
transformation of the eigenfunction, by using (\ref{fg}), is%
\begin{equation}
u\left(  r\right)  \rightarrow\left(  \frac{A+3}{2}\right)  ^{1/2}%
\rho^{\left(  A+1\right)  /4}v\left(  \rho\right)  .
\end{equation}
Dropping the constant\ $\left(  \frac{A+3}{2}\right)  ^{1/2}$ which does not
influence the result of the wave function, we arrive at the transformation
(\ref{wTpower}).

By the Newton duality relation, we can obtain the condition of the existence
of bound states. The power potential with a positive power, $r^{a+1}$
($a>-1$), has only bound states. The Newton duality of $r^{a+1}$, by Eq.
(\ref{aA3D}), is $r^{A+1}$ with $A=\frac{4}{a+3}-3$. The upper bound of $A$ is
$A^{upper}\overset{a=-1}{=}-1$ and the lower bound of $A$ is $A^{lower}%
\overset{a\rightarrow\infty}{=}-3$. Therefore, the condition of the existence
of bound states of the potential $r^{A+1}$ is $-3<A<-1$, or, equivalently, the
condition of the existence of bound states of the potential $r^{\beta}$ is
$-2<\beta<0$. This agrees with the usual result \cite{brau2004sufficient}.

It can be seen that if two potentials are classical Newton duality, they are
also quantum Newton duality. In classical mechanics, the orbits of two dual
potentials can be simply transformed from one to the other. Similarly, in
quantum mechanics, the wave functions of two dual potentials can be simply
transformed from one to the other.

\subsection{Arbitrary-dimensional power potentials \label{arbitraryDIM}}

In section \ref{3D}, we consider the quantum Newton duality between two
three-dimensional power potentials. In this section, we consider the quantum
Newton duality between two power potentials in different dimensions:
$n$-dimensional potential $U\left(  r\right)  =\xi r^{a+1}$ and $m$%
-dimensional potential $V\left(  \rho\right)  =\eta\rho^{A+1}$.

\begin{theorem}
\label{TwfnmD} Two power potentials in different dimensions,%
\begin{align}
U\left(  r\right)   &  =\xi r^{a+1}\text{, \ \ \ }n\text{-dimensional, }\\
V\left(  \rho\right)   &  =\eta\rho^{A+1}\text{, \ \ \ }m\text{-dimensional,}%
\end{align}
are quantum Newtonianly dual to each other, if%
\begin{equation}
\frac{a+3}{2}=\frac{2}{A+3}. \label{Aandanm}%
\end{equation}
The bound-state eigenfunction of\ the energy $\mathcal{E}$ and the angular
quantum number $\ell$ of the $m$-dimensional potential $V\left(  \rho\right)
$ can be obtained by performing the replacements%
\begin{align}
E  &  \rightarrow-\eta\left(  \frac{2}{A+3}\right)  ^{2},\label{nm-Eh}\\
\xi &  \rightarrow-\mathcal{E}\left(  \frac{2}{A+3}\right)  ^{2} \label{nm-XE}%
\end{align}
to the bound-state eigenfunction of the energy $E$ and the angular quantum
number $l$ of the $n$-dimensional potential $U\left(  r\right)  $. As a
result, the duality relation of the angular momentum between these two dual
systems is%
\begin{equation}
l+\frac{n}{2}-1=\frac{2}{A+3}\left(  \ell+\frac{m}{2}-1\right)  , \label{lell}%
\end{equation}
the transformation of coordinates is%
\begin{equation}
r\rightarrow\rho^{\left(  A+3\right)  /2}\text{ \ or \ }r^{\left(  a+3\right)
/2}\rightarrow\rho, \label{nmDcoorTpower}%
\end{equation}
and the transformation of the eigenfunction is%
\begin{equation}
u\left(  r\right)  \rightarrow\rho^{\left(  A+1\right)  /4}v\left(
\rho\right)  \text{ \ or \ }r^{\left(  a+1\right)  /4}u\left(  r\right)
\rightarrow v\left(  \rho\right)  . \label{wftrans}%
\end{equation}
\newline
\end{theorem}

\begin{proof}
In $n$ dimensions, the radial equation of the potential $U\left(  r\right)
=\xi r^{a+1}$ reads \cite{graham2009spectral}%
\begin{equation}
\frac{d^{2}u\left(  r\right)  }{dr^{2}}+\left[  E-\frac{\left(  l-\frac{3}%
{2}+\frac{n}{2}\right)  \left(  l-\frac{1}{2}+\frac{n}{2}\right)  }{r^{2}}-\xi
r^{a+1}\right]  u\left(  r\right)  =0, \label{radialeqn}%
\end{equation}
where $u\left(  r\right)  =R\left(  r\right)  /r$ with $R\left(  r\right)  $
the radial wave function. Performing the transformations (\ref{Aandanm}),
(\ref{nm-Eh}), and (\ref{nm-XE}), we arrive at%
\begin{equation}
\frac{d^{2}u\left(  r\right)  }{dr^{2}}+\left\{  \left[  -\eta\left(  \frac
{2}{A+3}\right)  ^{2}\right]  -\frac{\left(  l-\frac{3}{2}+\frac{n}{2}\right)
\left(  l-\frac{1}{2}+\frac{n}{2}\right)  }{r^{2}}-\left[  -\mathcal{E}\left(
\frac{2}{A+3}\right)  ^{2}\right]  r^{\frac{4}{A+3}-2}\right\}  u\left(
r\right)  =0.
\end{equation}
Rearrange the equation as%
\begin{equation}
\frac{d^{2}u\left(  r\right)  }{dr^{2}}+\left(  \frac{2}{A+3}\right)
^{2}\left[  \mathcal{E}-\frac{\left(  \frac{A+3}{2}\right)  ^{2}\left(
l-\frac{3}{2}+\frac{n}{2}\right)  \left(  l-\frac{1}{2}+\frac{n}{2}\right)
}{\left(  r^{\frac{2}{A+3}}\right)  ^{2}}-\eta\left(  r^{\frac{2}{A+3}%
}\right)  ^{A+1}\right]  \frac{1}{\left(  r^{\frac{2}{A+3}}\right)  ^{A+1}%
}u\left(  r\right)  =0. \label{eqmn}%
\end{equation}
This is just the radial equation of the central potential $V\left(
\rho\right)  =\eta\rho^{A+1}$ in $m$ dimensions:%
\begin{equation}
\frac{d^{2}v\left(  \rho\right)  }{d\rho^{2}}+\left[  \mathcal{E}%
-\frac{\left(  \ell-\frac{3}{2}+\frac{m}{2}\right)  \left(  \ell-\frac{1}%
{2}+\frac{m}{2}\right)  }{\rho^{2}}-\eta\rho^{A+1}\right]  v\left(
\rho\right)  =0
\end{equation}
with the replacement of the angular momentum, Eq. (\ref{lell}), the
transformation of coordinates (\ref{nmDcoorTpower}), and the transformation of
eigenfunctions, Eq. (\ref{wftrans}).
\end{proof}

Note that the influence of the value of the spatial dimension appears only in
the centrifugal potential term, because the angular quantum number is
different in different dimensions.

In the following, we present a relation between the eigenvalues of two dual
potentials in different dimensions.

The $n$-dimensional positive-power attractive potential $U\left(  r\right)
=\xi r^{a+1}$ ($a+1>0$ and\ $\xi>0$) has only bound states and positive
eigenvalues, i.e., $E>0$. Its $m$-dimensional Newton duality $V\left(
\rho\right)  =\eta\rho^{A+1}$ is a negative-power attractive potential
($A+1<0$ and $\eta<0$) and has bound-state eigenvalue, i.e., $\mathcal{E}<$
$0$.

\begin{corollary}
\label{TevnmD} If the eigenvalue of the\ $n$-dimensional positive-power
potential $U\left(  r\right)  =\xi r^{a+1}$ is%
\begin{equation}
E=f\left(  n_{r},l,\xi,n\right)  , \label{Ennm}%
\end{equation}
where $n_{r}$ is the radial quantum number and $l$ the angular quantum number,
then the eigenvalue of the its $m$-dimensional Newton duality, the
negative-power potential $V\left(  \rho\right)  =\eta\rho^{A+1}$, is%
\begin{equation}
\mathcal{E}=-\frac{4}{\left(  a+3\right)  ^{2}}f^{-1}\left(  n_{r},\frac
{2}{A+3}\left(  \ell+\frac{m}{2}-1\right)  +1+\frac{n}{2},-\frac{\left(
a+3\right)  ^{2}}{4}\eta,m\right)  , \label{specfnm}%
\end{equation}
where $f^{-1}$ denotes the inverse function of $f$ and $\ell$ is the angular
quantum number.
\end{corollary}

\begin{proof}
Substituting the relation between the two dual potentials, Eqs. (\ref{nm-Eh}),
(\ref{nm-XE}), and (\ref{lell}), into the expression of the eigenvalue of
the\ $m$-dimensional potential $U\left(  r\right)  $, Eq. (\ref{Ennm}), gives%
\begin{equation}
-\frac{\left(  a+3\right)  ^{2}}{4}\eta=f\left(  n_{r},\frac{2}{A+3}\left(
\ell+\frac{m}{2}-1\right)  +1+\frac{n}{2},-\mathcal{E}\frac{\left(
a+3\right)  ^{2}}{4},n\right)  . \label{specfnm3}%
\end{equation}

The eigenvalue $\mathcal{E}$ of the potential $V\left(  \rho\right)  $ given
by Eq. (\ref{specfnm}) then can be solved directly.
\end{proof}

\subsection{General polynomial potentials \label{polynomialnterms}}

The original Newton duality is only for power potentials in classical
mechanics. In section \ref{classicalGPP}, we generalize Newton's original
result\ to the case of general polynomial potentials in classical mechanics.
In this section, we provide the quantum version of the generalized Newton
duality for general polynomial potentials.

In the following, we will show that the duality among general polynomial
potentials. It allows us to solve more than one potentials from one known potential.

\begin{theorem}
\label{gpp} Two general polynomial potentials%
\begin{equation}
U\left(  r\right)  =\xi r^{a+1}+\sum_{n}\mu_{n}r^{b_{n}+1}\text{ \ and
\ }V\left(  \rho\right)  =\eta\rho^{A+1}+\sum_{n}\lambda_{n}\rho^{B_{n}+1}
\label{UVgp}%
\end{equation}
are quantum Newtonianly dual to each other, if%
\begin{align}
\frac{a+3}{2}  &  =\frac{2}{A+3},\label{Aanda6}\\
\sqrt{\frac{2}{a+3}}\left(  b_{n}+3\right)   &  =\sqrt{\frac{2}{A+3}}\left(
B_{n}+3\right)  . \label{Bandb6}%
\end{align}
The bound-state eigenfunction of the energy $\mathcal{E}$ and the angular
quantum number $\ell$ of the potential $V\left(  \rho\right)  $\ can be
obtained by performing the replacements%
\begin{align}
E  &  \rightarrow-\eta\left(  \frac{2}{A+3}\right)  ^{2},\label{Eta6}\\
\xi &  \rightarrow-\mathcal{E}\left(  \frac{2}{A+3}\right)  ^{2}%
,\label{xieta6}\\
\mu_{n}  &  \rightarrow\left(  \frac{2}{A+3}\right)  ^{2}\lambda_{n}
\label{mulamb6}%
\end{align}
to the bound-state eigenfunction of the energy $E$ and the angular quantum
number $l$ of the potential $U\left(  r\right)  $. As a result, the duality
relation of the angular momentum between these two dual systems is%
\begin{equation}
l+\frac{1}{2}\rightarrow\frac{2}{A+3}\left(  \ell+\frac{1}{2}\right)  ,
\label{ltrans6}%
\end{equation}
the transformation of coordinates is%
\begin{equation}
r\rightarrow\rho^{\frac{A+3}{2}}\text{\ \ or \ }r^{\frac{a+3}{2}}%
\rightarrow\rho, \label{rtrans6}%
\end{equation}
and the transformation of the eigenfunction is%
\begin{equation}
u\left(  r\right)  \rightarrow\rho^{\left(  A+1\right)  /4}v\left(
\rho\right)  \text{\ \ \ or \ \ }r^{\left(  a+1\right)  /4}u\left(  r\right)
\rightarrow v\left(  \rho\right)  . \label{wftrans6}%
\end{equation}

\end{theorem}

\begin{proof}
The radial equation of the general polynomial potential $U\left(  r\right)
=\xi r^{a+1}+\sum\mu_{n}r^{b_{n}+1}$ reads%
\begin{equation}
\frac{d^{2}u\left(  r\right)  }{dr^{2}}+\left[  E-\frac{l\left(  l+1\right)
}{r^{2}}-\xi r^{a+1}-\sum_{n}\mu_{n}r^{b_{n}+1}\right]  u\left(  r\right)  =0.
\label{raeq6}%
\end{equation}
Substituting the duality relations (\ref{Aanda6}), (\ref{Bandb6}),
(\ref{Eta6}), (\ref{xieta6}), and (\ref{mulamb6}) into the radial equation of
$U\left(  r\right)  $, Eq. (\ref{raeq6}), gives%
\begin{equation}
\frac{d^{2}u\left(  r\right)  }{dr^{2}}+\left\{  -\left(  \frac{2}%
{A+3}\right)  ^{2}\eta-\frac{l\left(  l+1\right)  }{r^{2}}-\left[  -\left(
\frac{2}{A+3}\right)  ^{2}\mathcal{E}\right]  r^{\frac{4}{A+3}-2}-\sum\left(
\frac{2}{A+3}\right)  ^{2}\lambda_{n}r^{\frac{2}{A+3}\left(  3+B_{n}\right)
-2}\right\}  u\left(  r\right)  =0.
\end{equation}
Rearranging the equation as%
\begin{equation}
\frac{d^{2}u\left(  r\right)  }{dr^{2}}+\left(  \frac{2}{A+3}\right)
^{2}\left[  \mathcal{E}-\frac{\left(  \frac{A+3}{2}\right)  ^{2}l\left(
l+1\right)  }{\left(  r^{\frac{2}{A+3}}\right)  ^{2}}-\eta\left(  r^{\frac
{2}{A+3}}\right)  ^{A+1}-\sum\lambda_{n}\left(  r^{\frac{2}{A+3}}\right)
^{B_{n}+1}\right]  \frac{1}{\left(  r^{\frac{2}{A+3}}\right)  ^{A+1}}u\left(
r\right)  =0, \label{raeq66}%
\end{equation}
\newline we can see that this is just the radial equation of the potential
$V\left(  \rho\right)  $ with the replacements (\ref{ltrans6}), (\ref{rtrans6}%
), and (\ref{wftrans6}):%
\begin{equation}
\frac{d^{2}v\left(  \rho\right)  }{d\rho^{2}}+\left[  \mathcal{E}-\frac
{\ell\left(  \ell+1\right)  }{\rho^{2}}-\eta\rho^{A+1}-\sum\lambda_{n}%
\rho^{B_{n}+1}\right]  v\left(  \rho\right)  =0.
\end{equation}

\end{proof}

By Theorem \ref{gpp}, we can obtained the eigenvalue of $V\left(  \rho\right)
$.

\begin{corollary}
\label{Tevpoly} If the eigenvalue of the general polynomial potential
$U\left(  r\right)  =\xi r^{a+1}+\sum_{n}\mu_{n}r^{b_{n}+1}$ is%
\begin{equation}
E=f\left(  n_{r},l,\xi,\mu_{n}\right)  , \label{Enpoly}%
\end{equation}
where $n_{r}$ is the radial quantum number and $l$ is the angular quantum
number, then the eigenvalue of its Newtonianly dual potential $V\left(
\rho\right)  =\eta\rho^{A+1}+\sum_{n}\lambda_{n}\rho^{B_{n}+1}$ is%
\begin{equation}
\mathcal{E}=-\left(  \frac{A+3}{2}\right)  ^{2}f^{-1}\left(  n_{r},\frac
{2}{A+3}\left(  \ell+\frac{1}{2}\right)  -\frac{1}{2},-\eta\left(  \frac
{2}{A+3}\right)  ^{2},\left(  \frac{2}{A+3}\right)  ^{2}\lambda_{n}\right)  ,
\label{specfpoly}%
\end{equation}
where $f^{-1}$ denotes the inverse function of $f$ and $\ell$ is the angular
quantum number of the system of $V\left(  \rho\right)  $.
\end{corollary}

\begin{proof}
Substituting the relation between the two dual potentials, Eqs. (\ref{Eta6}),
(\ref{xieta6}), (\ref{mulamb6}) and (\ref{ltrans6}), into the expression of
the eigenvalue of the potential $U\left(  r\right)  $, Eq. (\ref{Enpoly}),
gives%
\begin{equation}
-\eta\left(  \frac{2}{A+3}\right)  ^{2}=f\left(  n_{r},\frac{2}{A+3}\left(
\ell+\frac{1}{2}\right)  -\frac{1}{2},-\mathcal{E}\left(  \frac{2}%
{A+3}\right)  ^{2}\right)  . \label{specfpoly1}%
\end{equation}
The eigenvalue $\mathcal{E}$ of the potential $V\left(  \rho\right)  $ given
by Eq. (\ref{specfpoly}) then can be solved directly.
\end{proof}

In the following, we show how the result given by Theorem \ref{gpp} is reached.

The duality between two general polynomial potentials, according to
Condition\ \ref{cd1}, is obtained by choosing the term $g^{\prime}\left(
\rho\right)  ^{2}/g^{2}\left(  \rho\right)  $ in Eq. (\ref{3}) as the
centrifugal-potential term or a part of the centrifugal-potential term, i.e.,%
\begin{equation}
\frac{g^{\prime}\left(  \rho\right)  ^{2}}{g^{2}\left(  \rho\right)  }%
=\frac{\gamma^{2}}{\rho^{2}} \label{gdggp}%
\end{equation}
with $\gamma$ a constant. Solving Eq. (\ref{gdggp}) gives%
\begin{equation}
g\left(  \rho\right)  =\rho^{\gamma}. \label{grho6}%
\end{equation}
Substituting Eq. (\ref{grho6}) into the radial equation (\ref{3}) gives%
\begin{equation}
\frac{d^{2}v\left(  \rho\right)  }{d\rho^{2}}+\left[  -\frac{\gamma^{2}}%
{\rho^{2\left(  1-\gamma\right)  }}U\left(  \rho\right)  -\frac{l\left(
l+1\right)  \gamma^{2}-\left(  1-\gamma^{2}\right)  /4}{\rho^{2}}%
+\frac{E\gamma^{2}}{\rho^{2\left(  1-\gamma\right)  }}\right]  v\left(
\rho\right)  =0. \label{vrho6}%
\end{equation}
Here the term $\left[  l\left(  l+1\right)  \gamma^{2}-\left(  1-\gamma
^{2}\right)  /4\right]  /\rho^{2}$ is the centrifugal-potential term.

According to Condition \ref{cd2}, we need a constant term to serve as the
eigenvalue $\mathcal{E}$. Choosing
\begin{equation}
-\frac{\gamma^{2}}{\rho^{2\left(  1-\gamma\right)  }}U\left(  \rho\right)
=\mathcal{E}-\sum\kappa_{n}\rho^{n},
\end{equation}
where $\mathcal{E}$ is the eigenvalue, gives%
\begin{equation}
U\left(  \rho\right)  =-\frac{\mathcal{E}}{\gamma^{2}}\rho^{2\left(
1-\gamma\right)  }+\sum\frac{\kappa_{n}}{\gamma^{2}}\rho^{2\left(
1-\gamma\right)  +n}. \label{Urho6}%
\end{equation}

The transformation of coordinates between two dual potentials, by Eqs.
(\ref{t1})\ and (\ref{grho6}), is%
\begin{equation}
r\rightarrow\rho^{\gamma}. \label{grho66}%
\end{equation}
By this transformation of coordinates, we arrive at the potential%
\begin{equation}
U\left(  r\right)  =-\frac{\mathcal{E}}{\gamma^{2}}r^{2\left(  1-\gamma
\right)  /\gamma}+\sum_{n}\frac{\kappa_{n}}{\gamma^{2}}r^{\left[  2\left(
1-\gamma\right)  +n\right]  /\gamma}.
\end{equation}

By Eqs. (\ref{vrho6}) and (\ref{Urho6}),\ Eq. (\ref{3}) becomes%
\begin{equation}
\frac{d^{2}v\left(  \rho\right)  }{dr^{2}}+\left[  \mathcal{E}-\frac{l\left(
l+1\right)  \gamma^{2}-\left(  1-\gamma^{2}\right)  /4}{\rho^{2}}%
+\frac{E\gamma^{2}}{\rho^{2\left(  1-\gamma\right)  }}-\sum\kappa_{n}\rho
^{n}\right]  v\left(  \rho\right)  =0. \label{vrho66}%
\end{equation}
This is just the radial equation of the potential%
\begin{equation}
V\left(  \rho\right)  =-\frac{E\gamma^{2}}{\rho^{2\left(  1-\gamma\right)  }%
}+\sum\kappa_{n}\rho^{n}.
\end{equation}

Rewriting these two potentials as%
\begin{equation}
U\left(  r\right)  =\xi r^{a+1}+\sum\mu_{n}r^{b_{n}+1}%
\end{equation}
with%
\begin{align}
\gamma &  =\frac{2}{a+3},\label{gammaa16}\\
\mathcal{E}  &  =\mathcal{-}\left(  \frac{2}{a+3}\right)  ^{2}\xi
,\label{eipxi6}\\
b_{n}  &  =\frac{\left(  2+n\right)  \left(  a+3\right)  }{2}-3,\label{bbeta6}%
\\
\kappa_{n}  &  =\left(  \frac{2}{a+3}\right)  ^{2}\mu_{n} \label{kappamu6}%
\end{align}
and%
\begin{equation}
V\left(  \rho\right)  =\eta\rho^{A+1}+\sum_{n}\lambda_{n}\rho^{B_{n}+1}%
\end{equation}
with%
\begin{align}
\gamma &  =\frac{A+3}{2},\label{gammaA16}\\
E  &  =-\eta\left(  \frac{2}{A+3}\right)  ^{2},\label{Eeta6}\\
B_{n}  &  =n-1,\label{Bbeta6}\\
\lambda_{n}  &  =\kappa_{n}, \label{lambdakappa6}%
\end{align}
we then arrive at the duality relations (\ref{Aanda6}), (\ref{Bandb6}),
(\ref{Eta6}), (\ref{xieta6}), and (\ref{mulamb6}).

The term which is proportional to $1/\rho^{2}$ plays a role of the centrifugal
potential, so%
\begin{align}
\ell\left(  \ell+1\right)   &  =l\left(  l+1\right)  \gamma^{2}-\frac{1}%
{4}\left(  1-\gamma^{2}\right) \nonumber\\
&  =l\left(  l+1\right)  \left(  \frac{A+3}{2}\right)  ^{2}-\frac{1}{4}\left[
1-\left(  \frac{A+3}{2}\right)  ^{2}\right]  .
\end{align}
Then we arrive at the replacement of the angular momentum given by Eq.
(\ref{ltrans6}).

The transformation of coordinates (\ref{rtrans6}) can be obtained by
substituting Eq. (\ref{gammaA16}) into Eq. (\ref{grho6}). Then$\ $the
transformation of the eigenfunction, by using (\ref{fg}), is%
\begin{equation}
u\left(  r\right)  \rightarrow\left(  \frac{A+3}{2}\right)  ^{1/2}%
\rho^{\left(  A+1\right)  /4}v\left(  \rho\right)  . \label{ur6}%
\end{equation}

The constant factor $\left(  \frac{A+3}{2}\right)  ^{1/2}$ does not change the
wave function, so the transformation (\ref{ur6}) can be written as the
transformation (\ref{wftrans6}).

In the above, we provide a quantum Newton duality between general polynomial
potentials. The general polynomial potential is a very general potential. In
fact, the general polynomial potential contains all potentials that can be
expressed as a series of power potentials with arbitrary real number powers.
Therefore, many potentials can be considered as a special case of general
polynomial potentials and can be treated by use of the Newton duality of
general polynomial potentials.

\subsection{An $N$-term potential and its $N$ dual potentials
\label{polynomialntermsmany}}

Inspecting the result of the quantum Newton duality among general polynomial
potentials, we can see that an $N$-term potential has $N$ dual potentials.

\subsubsection{An $N$-term potential and its $N$ dual potentials}

As a direct result of Theorem \ref{gpp}, we have the following conclusion,
which is useful in solving various potentials.

$N+1$\textit{ }$N$\textit{-term general polynomial potentials}%
\begin{equation}
\left\{  V_{1}\left(  r\right)  ,\cdots,V_{P}\left(  r\right)  ,\cdots
,V_{Q}\left(  r\right)  ,\cdots,V_{N+1}\left(  r\right)  \right\}
\end{equation}
\textit{are quantum Newtonianly dual to each other, if any two potentials of
them, say }%
\begin{align}
V_{P}\left(  r\right)   &  =A_{1}r^{\alpha_{1}+1}+\cdots+A_{i}r^{\alpha_{i}%
+1}+\cdots+A_{N}r^{\alpha_{N}+1},\\
V_{Q}\left(  r\right)   &  =B_{1}r^{\beta_{1}+1}+\cdots+B_{j}r^{\beta_{j}%
+1}+\cdots+B_{N}r^{\beta_{N}+1},
\end{align}
\textit{satisfy the relations}%
\begin{align}
\frac{\alpha_{i}+3}{2}  &  =\frac{2}{\beta_{j}+3},\text{ \ }i,j=1,\cdots,N,\\
\sqrt{\frac{2}{\alpha_{i}+3}}\left(  \alpha_{m}+3\right)   &  =\sqrt{\frac
{2}{\beta_{j}+3}}\left(  \beta_{n}+3\right)  ,\text{ \ }m,n\neq i,j,
\end{align}
\textit{where }$P,Q=1,\cdots,N+1$\textit{ and }$m,n=1,\cdots,N$\textit{. The
bound-state eigenfunction of the energy }$E_{Q}$\textit{ and the angular
quantum number }$\ell$\textit{ of the potential }$V_{Q}\left(  r\right)
$\textit{\ can be obtained by performing the replacements}%
\begin{align}
E_{P}  &  \rightarrow-B_{j}\left(  \frac{2}{\beta_{j}+3}\right)  ^{2},\\
A_{i}  &  \rightarrow-E_{Q}\left(  \frac{2}{\beta_{j}+3}\right)  ^{2},\\
A_{m}  &  \rightarrow B_{n}\left(  \frac{2}{\beta_{j}+3}\right)  ^{2}%
\end{align}
\textit{to the bound-state eigenfunction of the energy }$E_{P}$\textit{ and
the angular quantum number }$l$\textit{ of the potential }$V_{P}\left(
r\right)  $\textit{. As a result, the duality relation of the angular momentum
between the two dual potentials }$V_{P}\left(  r\right)  $\textit{ and }%
$V_{Q}\left(  r\right)  $\textit{ is}%
\begin{equation}
l_{P}+\frac{1}{2}=\frac{2}{\beta_{j}+3}\left(  l_{Q}+\frac{1}{2}\right)  ,
\end{equation}
\textit{the transformation of coordinates is}%
\begin{equation}
r_{P}\rightarrow r_{Q}^{\frac{\beta_{j}+3}{2}}\ \ \text{or }\ r_{P}%
^{\frac{A_{i}+3}{2}}\rightarrow r_{Q},
\end{equation}
\textit{and the transformation of the eigenfunction is}%
\begin{equation}
u_{P}\left(  r_{P}\right)  \rightarrow r_{Q}^{\left(  \beta_{j}+1\right)
/4}u_{Q}\left(  r_{Q}\right)  \ \ \text{or }\ r_{p}^{\left(  A_{i}+1\right)
/4}u_{P}\left(  r_{P}\right)  \rightarrow u_{Q}\left(  r_{Q}\right)  .
\end{equation}

This result is in fact a corollary of Theorem \ref{gpp}.

In the following, we illustrate the dualities among $N+1$ dual potentials.

\subsubsection{A $2$-term potential and its two dual potentials: examples}

Three $2$-term general polynomial potentials $\left\{  V_{1}\left(  r\right)
,V_{2}\left(  r\right)  ,V_{3}\left(  r\right)  \right\}  $,
\begin{align}
V_{1}\left(  r\right)   &  =A_{1}r^{\alpha_{1}+1}+B_{1}r^{\beta_{1}+1},\\
V_{2}\left(  r\right)   &  =A_{2}r^{\alpha_{2}+1}+B_{2}r^{\beta_{2}+1},\\
V_{3}\left(  r\right)   &  =A_{3}r^{\alpha_{3}+1}+B_{3}r^{\beta_{3}+1},
\end{align}
are dual to each other, if they satisfy the relations
\begin{align}
\frac{\alpha_{1}+3}{2}  &  =\frac{2}{\alpha_{2}+3},\text{ }\frac{\beta_{1}%
+3}{2}=\frac{2}{\beta_{3}+3},\text{ }\frac{\beta_{2}+3}{2}=\frac{2}{\alpha
_{3}+3},\nonumber\\
\sqrt{\frac{2}{\alpha_{1}+3}}\left(  \beta_{1}+3\right)   &  =\sqrt{\frac
{2}{\alpha_{2}+3}}\left(  \beta_{2}+3\right)  ,\nonumber\\
\sqrt{\frac{2}{\beta_{1}+3}}\left(  \alpha_{1}+3\right)   &  =\sqrt{\frac
{2}{\beta_{3}+3}}\left(  \alpha_{3}+3\right)  ,\nonumber\\
\sqrt{\frac{2}{\beta_{2}+3}}\left(  \alpha_{2}+3\right)   &  =\sqrt{\frac
{2}{\alpha_{3}+3}}\left(  \beta_{3}+3\right)  . \label{2term3}%
\end{align}
Their eigenvalues of bound states and the coupling constants have the
relation
\begin{align}
E_{1}  &  =-A_{2}\left(  \frac{2}{\alpha_{2}+3}\right)  ^{2},\text{ \ }%
A_{1}=-E_{2}\left(  \frac{2}{\alpha_{2}+3}\right)  ^{2},\text{ \ }B_{1}%
=B_{2}\left(  \frac{2}{\alpha_{2}+3}\right)  ^{2},\nonumber\\
E_{1}  &  =-B_{3}\left(  \frac{2}{\beta_{3}+3}\right)  ^{2},\text{ \ }%
B_{1}=-E_{3}\left(  \frac{2}{\beta_{3}+3}\right)  ^{2},\text{ \ }A_{1}%
=A_{3}\left(  \frac{2}{\beta_{3}+3}\right)  ^{2},\nonumber\\
E_{2}  &  =-A_{3}\left(  \frac{2}{\alpha_{3}+3}\right)  ^{2},\text{ \ }%
B_{2}=-E_{3}\left(  \frac{2}{\alpha_{3}+3}\right)  ^{2},\text{ \ }A_{2}%
=B_{3}\left(  \frac{2}{\alpha_{3}+3}\right)  ^{2}%
\end{align}
and their angular momenta have the relation
\begin{align}
l_{1}+\frac{1}{2}  &  =\frac{2}{\alpha_{2}+3}\left(  l_{2}+\frac{1}{2}\right)
,\nonumber\\
l_{1}+\frac{1}{2}  &  =\frac{2}{\beta_{3}+3}\left(  l_{3}+\frac{1}{2}\right)
,\nonumber\\
l_{2}+\frac{1}{2}  &  =\frac{2}{\alpha_{3}+3}\left(  l_{3}+\frac{1}{2}\right)
.
\end{align}
Their eigenfunctions are related by the transformations
\begin{align}
r_{1}  &  \rightarrow r_{2}^{\frac{\alpha_{2}+3}{2}}\ \ \text{or \ }%
r_{1}^{\frac{\alpha_{1}+3}{2}}\rightarrow r_{2},\nonumber\\
r_{1}  &  \rightarrow r_{3}^{\frac{\beta_{3}+3}{2}}\text{\ \ or \ }%
r_{1}^{\frac{\beta_{1}+3}{2}}\rightarrow r_{3},\nonumber\\
r_{2}  &  \rightarrow r_{3}^{\frac{\alpha_{3}+3}{2}}\text{\ \ or \ }%
r_{2}^{\frac{\beta_{2}+3}{2}}\rightarrow r_{3}%
\end{align}
and%
\begin{align}
u_{1}\left(  r_{1}\right)   &  \rightarrow r_{2}^{\left(  \alpha_{2}+1\right)
/4}u_{2}\left(  r_{2}\right)  \ \ \text{or \ }r_{1}^{\left(  \alpha
_{1}+1\right)  /4}u_{1}\left(  r_{1}\right)  \rightarrow u_{2}\left(
r_{2}\right)  ,\nonumber\\
u_{1}\left(  r_{1}\right)   &  \rightarrow r_{3}^{\left(  \beta_{3}+1\right)
/4}u_{2}\left(  r_{3}\right)  \ \ \text{or \ }r_{1}^{\left(  \beta
_{1}+1\right)  /4}u_{1}\left(  r_{1}\right)  \rightarrow u_{3}\left(
r_{3}\right)  ,\nonumber\\
u_{2}\left(  r_{2}\right)   &  \rightarrow r_{3}^{\left(  \alpha_{3}+1\right)
/4}u_{3}\left(  r_{3}\right)  \ \ \text{or \ }r_{2}^{\left(  \beta
_{2}+1\right)  /4}u_{2}\left(  r_{2}\right)  \rightarrow u_{3}\left(
r_{3}\right)  .
\end{align}

\subsubsection{A $3$-term potential and its three dual potentials}

Four general polynomial potentials $\left\{  V_{1}\left(  r\right)
,V_{2}\left(  r\right)  ,V_{3}\left(  r\right)  ,V_{4}\left(  r\right)
\right\}  $,%
\begin{align}
V_{1}\left(  r\right)   &  =A_{1}r^{\alpha_{1}+1}+B_{1}r^{\beta_{1}+1}%
+C_{1}r^{\gamma_{1}+1},\nonumber\\
V_{2}\left(  r\right)   &  =A_{2}r^{\alpha_{2}+1}+B_{2}r^{\beta_{2}+1}%
+C_{2}r^{\gamma_{2}+1},\nonumber\\
V_{3}\left(  r\right)   &  =A_{3}r^{\alpha_{3}+1}+B_{3}r^{\beta_{3}+1}%
+C_{3}r^{\gamma_{3}+1},\nonumber\\
V_{4}\left(  r\right)   &  =A_{4}r^{\alpha_{4}+1}+B_{4}r^{\beta_{4}+1}%
+C_{4}r^{\gamma_{4}+1},
\end{align}
are dual to each other, if their powers satisfy the relation%
\begin{align}
\frac{\alpha_{1}+3}{2}  &  =\frac{2}{\alpha_{2}+3},\text{ }\frac{\beta_{1}%
+3}{2}=\frac{2}{\beta_{3}+3},\text{ }\frac{\gamma_{1}+3}{2}=\frac{2}%
{\gamma_{4}+3},\nonumber\\
\frac{\beta_{2}+3}{2}  &  =\frac{2}{\alpha_{3}+3},\text{ }\frac{\gamma_{2}%
+3}{2}=\frac{2}{\alpha_{4}+3},\text{ }\frac{\gamma_{3}+3}{2}=\frac{2}%
{\beta_{4}+3},\nonumber\\
\sqrt{\frac{2}{\alpha_{1}+3}}\left(  \beta_{1}+3\right)   &  =\sqrt{\frac
{2}{\alpha_{2}+3}}\left(  \beta_{2}+3\right)  ,\text{ }\sqrt{\frac{2}%
{\alpha_{1}+3}}\left(  \gamma_{1}+3\right)  =\sqrt{\frac{2}{\alpha_{2}+3}%
}\left(  \gamma_{2}+3\right)  ,\nonumber\\
\sqrt{\frac{2}{\beta_{1}+3}}\left(  \alpha_{1}+3\right)   &  =\sqrt{\frac
{2}{\beta_{3}+3}}\left(  \alpha_{3}+3\right)  ,\text{ }\sqrt{\frac{2}%
{\beta_{1}+3}}\left(  \gamma_{1}+3\right)  =\sqrt{\frac{2}{\beta_{3}+3}%
}\left(  \gamma_{3}+3\right)  ,\nonumber\\
\sqrt{\frac{2}{\gamma_{1}+3}}\left(  \alpha_{1}+3\right)   &  =\sqrt{\frac
{2}{\gamma_{4}+3}}\left(  \alpha_{4}+3\right)  ,\text{ }\sqrt{\frac{2}%
{\gamma_{1}+3}}\left(  \beta_{1}+3\right)  =\sqrt{\frac{2}{\gamma_{4}+3}%
}\left(  \beta_{4}+3\right)  ,\nonumber\\
\sqrt{\frac{2}{\beta_{2}+3}}\left(  \alpha_{2}+3\right)   &  =\sqrt{\frac
{2}{\alpha_{3}+3}}\left(  \beta_{3}+3\right)  ,\text{ }\sqrt{\frac{2}%
{\beta_{2}+3}}\left(  \gamma_{2}+3\right)  =\sqrt{\frac{2}{\alpha_{3}+3}%
}\left(  \gamma_{3}+3\right)  ,\nonumber\\
\sqrt{\frac{2}{\gamma_{2}+3}}\left(  \alpha_{2}+3\right)   &  =\sqrt{\frac
{2}{\alpha_{4}+3}}\left(  \gamma_{4}+3\right)  ,\text{ }\sqrt{\frac{2}%
{\gamma_{2}+3}}\left(  \beta_{2}+3\right)  =\sqrt{\frac{2}{\alpha_{4}+3}%
}\left(  \beta_{4}+3\right)  ,\nonumber\\
\sqrt{\frac{2}{\gamma_{3}+3}}\left(  \beta_{3}+3\right)   &  =\sqrt{\frac
{2}{\beta_{4}+3}}\left(  \gamma_{4}+3\right)  ,\text{ }\sqrt{\frac{2}%
{\gamma_{3}+3}}\left(  \alpha_{3}+3\right)  =\sqrt{\frac{2}{\beta_{4}+3}%
}\left(  \alpha_{4}+3\right)  .
\end{align}
Their eigenvalues of bound states and the coupling constants have the relation%
\begin{align}
E_{1}  &  =-A_{2}\left(  \frac{2}{\alpha_{2}+3}\right)  ^{2},\text{ \ }%
A_{1}=-E_{2}\left(  \frac{2}{\alpha_{2}+3}\right)  ^{2},\text{ \ }B_{1}%
=B_{2}\left(  \frac{2}{\alpha_{2}+3}\right)  ^{2},\text{ \ }C_{1}=C_{2}\left(
\frac{2}{\alpha_{2}+3}\right)  ^{2},\nonumber\\
E_{1}  &  =-B_{3}\left(  \frac{2}{\beta_{3}+3}\right)  ^{2},\text{ \ }%
B_{1}=-E_{3}\left(  \frac{2}{\beta_{3}+3}\right)  ^{2},\text{ \ }A_{1}%
=A_{3}\left(  \frac{2}{\beta_{3}+3}\right)  ^{2},\text{ \ }C_{1}=C_{3}\left(
\frac{2}{\beta_{3}+3}\right)  ^{2},\nonumber\\
E_{1}  &  =-C_{3}\left(  \frac{2}{\gamma_{4}+3}\right)  ^{2},\text{ \ }%
C_{1}=-C_{4}\left(  \frac{2}{\gamma_{4}+3}\right)  ^{2},\text{ \ }A_{1}%
=A_{4}\left(  \frac{2}{\gamma_{4}+3}\right)  ^{2},\text{ \ }B_{1}=B_{4}\left(
\frac{2}{\gamma_{4}+3}\right)  ^{2},\nonumber\\
E_{2}  &  =-A_{3}\left(  \frac{2}{\alpha_{3}+3}\right)  ^{2},\text{ \ }%
B_{2}=-E_{3}\left(  \frac{2}{\alpha_{3}+3}\right)  ^{2},\text{ \ }A_{2}%
=B_{3}\left(  \frac{2}{\alpha_{3}+3}\right)  ^{2},\text{ \ }C_{2}=C_{3}\left(
\frac{2}{\alpha_{3}+3}\right)  ^{2},\nonumber\\
E_{2}  &  =-A_{3}\left(  \frac{2}{\alpha_{4}+3}\right)  ^{2},\text{ \ }%
C_{2}=-E_{4}\left(  \frac{2}{\alpha_{4}+3}\right)  ^{2},\text{ \ }A_{2}%
=C_{4}\left(  \frac{2}{\alpha_{3}+3}\right)  ^{2},\text{ \ }B_{2}=B_{4}\left(
\frac{2}{\alpha_{3}+3}\right)  ^{2},\nonumber\\
E_{3}  &  =-B_{4}\left(  \frac{2}{\beta_{4}+3}\right)  ^{2},\text{ \ }%
C_{3}=-E_{4}\left(  \frac{2}{\beta_{4}+3}\right)  ^{2},\text{ \ }B_{3}%
=C_{4}\left(  \frac{2}{\beta_{4}+3}\right)  ^{2},\text{ \ }A_{3}=A_{4}\left(
\frac{2}{\beta_{4}+3}\right)  ^{2}%
\end{align}
and their angular momenta have the relation%
\begin{align}
l_{1}+\frac{1}{2}  &  =\frac{2}{\alpha_{2}+3}\left(  l_{2}+\frac{1}{2}\right)
,\nonumber\\
l_{1}+\frac{1}{2}  &  =\frac{2}{\beta_{3}+3}\left(  l_{3}+\frac{1}{2}\right)
,\nonumber\\
l_{1}+\frac{1}{2}  &  =\frac{2}{\gamma_{4}+3}\left(  l_{4}+\frac{1}{2}\right)
,\nonumber\\
l_{2}+\frac{1}{2}  &  =\frac{2}{\alpha_{3}+3}\left(  l_{3}+\frac{1}{2}\right)
,\nonumber\\
l_{2}+\frac{1}{2}  &  =\frac{2}{\alpha_{4}+3}\left(  l_{4}+\frac{1}{2}\right)
,\nonumber\\
l_{3}+\frac{1}{2}  &  =\frac{2}{\beta_{4}+3}\left(  l_{4}+\frac{1}{2}\right)
.
\end{align}
Their eigenfunctions are related by the transformations%
\begin{align}
r_{1}  &  \rightarrow r_{2}^{\frac{\alpha_{2}+3}{2}}\ \ \text{or \ }%
\ r_{1}^{\frac{\alpha_{1}+3}{2}}\rightarrow r_{2},\nonumber\\
r_{1}  &  \rightarrow r_{3}^{\frac{\beta_{3}+3}{2}}\text{\ \ or \ \ }%
r_{1}^{\frac{\beta_{1}+3}{2}}\rightarrow r_{3},\nonumber\\
r_{1}  &  \rightarrow r_{4}^{\frac{\gamma_{4}+3}{2}}\text{\ \ or \ \ }%
r_{1}^{\frac{\gamma_{1}+3}{2}}\rightarrow r_{4},\nonumber\\
r_{2}  &  \rightarrow r_{3}^{\frac{\alpha_{3}+3}{2}}\text{\ \ or \ \ }%
r_{2}^{\frac{\beta_{2}+3}{2}}\rightarrow r_{3},\nonumber\\
r_{2}  &  \rightarrow r_{4}^{\frac{\alpha_{4}+3}{2}}\text{\ \ or \ \ }%
r_{2}^{\frac{\gamma_{2}+3}{2}}\rightarrow r_{4},\nonumber\\
r_{3}  &  \rightarrow r_{4}^{\frac{\beta_{4}+3}{2}}\text{\ \ or \ \ }%
r_{3}^{\frac{\gamma_{3}+3}{2}}\rightarrow r_{4}%
\end{align}
and%
\begin{align}
u_{1}\left(  r_{1}\right)   &  \rightarrow r_{2}^{\left(  \alpha_{2}+1\right)
/4}u_{2}\left(  r_{2}\right)  \ \ \text{or \ }\ r_{1}^{\left(  \alpha
_{1}+1\right)  /4}u_{1}\left(  r_{1}\right)  \rightarrow u_{2}\left(
r_{2}\right)  ,\nonumber\\
u_{1}\left(  r_{1}\right)   &  \rightarrow r_{3}^{\left(  \beta_{3}+1\right)
/4}u_{2}\left(  r_{3}\right)  \ \ \text{or \ }\ r_{1}^{\left(  \beta
_{1}+1\right)  /4}u_{1}\left(  r_{1}\right)  \rightarrow u_{2}\left(
r_{3}\right)  ,\nonumber\\
u_{1}\left(  r_{1}\right)   &  \rightarrow r_{4}^{\left(  \gamma_{4}+1\right)
/4}u_{4}\left(  r_{4}\right)  \ \ \text{or \ }\ r_{1}^{\left(  \gamma
_{1}+1\right)  /4}u_{1}\left(  r_{1}\right)  \rightarrow u_{4}\left(
r_{4}\right)  ,\nonumber\\
u_{2}\left(  r_{2}\right)   &  \rightarrow r_{3}^{\left(  \alpha_{3}+1\right)
/4}u_{3}\left(  r_{3}\right)  \ \ \text{or \ }\ r_{2}^{\left(  \beta
_{2}+1\right)  /4}u_{2}\left(  r_{2}\right)  \rightarrow u_{3}\left(
r_{3}\right)  ,\nonumber\\
u_{2}\left(  r_{2}\right)   &  \rightarrow r_{4}^{\left(  \alpha_{4}+1\right)
/4}u_{4}\left(  r_{4}\right)  \ \ \text{or \ }\ r_{2}^{\left(  \gamma
_{2}+1\right)  /4}u_{2}\left(  r_{2}\right)  \rightarrow u_{4}\left(
r_{4}\right)  ,\nonumber\\
u_{3}\left(  r_{3}\right)   &  \rightarrow r_{4}^{\left(  \beta_{4}+1\right)
/4}u_{4}\left(  r_{4}\right)  \ \ \text{or \ }\ r_{3}^{\left(  \gamma
_{3}+1\right)  /4}u_{3}\left(  r_{3}\right)  \rightarrow u_{4}\left(
r_{4}\right)  .
\end{align}

\subsection{General polynomial potentials consist of two Newton dual
potentials \label{polynomial SC}}

In section \ref{classical polynomial}, we discuss a special case of the
classical Newton duality, in which each term of a potential is Newtonianly
dual to the corresponding term of its Newton dual potential. In this section,
we consider the quantum version of such a case.

\begin{theorem}
Two Newtonianly dual general polynomial potentials%
\begin{align}
U\left(  r\right)   &  =\xi r^{a+1}+\mu r^{2\left(  \sqrt{\frac{a+3}{2}%
}-1\right)  },\\
V\left(  \rho\right)   &  =\eta\rho^{A+1}+\mu\left(  \frac{A+3}{2}\right)
^{2}\rho^{2\left(  \sqrt{\frac{A+3}{2}}-1\right)  },
\end{align}
in which $\xi r^{a+1}$ is the Newton duality of $\eta\rho^{A+1}$ and $\mu
r^{2\left(  \sqrt{\frac{a+3}{2}}-1\right)  }$ is the Newton duality of
$\mu\left(  \frac{A+3}{2}\right)  ^{2}\rho^{2\left(  \sqrt{\frac{A+3}{2}%
}-1\right)  }$, if%
\begin{equation}
\frac{a+3}{2}=\frac{2}{A+3}.
\end{equation}

\end{theorem}

The proof is straightforward by the duality relation given above. An example
will be given in Sec. \ref{Solvingeln}.

\subsection{The quantum Newton duality between $U\left(  r\right)  =\xi
e^{\sigma r}$ and $V\left(  \rho\right)  =\frac{\eta}{\left(  \rho\ln
\alpha\rho\right)  ^{2}}$ \label{eln}}

In this section, we give an example of the Newton duality between two
transcendental-function potentials.

\begin{theorem}
\label{expln} Two potentials%
\begin{equation}
U\left(  r\right)  =\xi e^{\sigma r}\text{ \ and }V\left(  \rho\right)
=\frac{\eta}{\left(  \rho\ln\alpha\rho\right)  ^{2}}%
\end{equation}
are quantum Newtonianly dual to each other. The duality relations are%
\begin{align}
E  &  \rightarrow-\left(  \frac{\sigma}{2}\right)  ^{2}\left[  \ell\left(
\ell+1\right)  +\frac{1}{4}\right]  ,\label{Et}\\
\xi &  \rightarrow-\left(  \frac{\sigma}{2}\right)  ^{2}\frac{\mathcal{E}%
}{\alpha^{2}},\label{xit}\\
l\left(  l+1\right)   &  \rightarrow\eta, \label{lt}%
\end{align}
or, equivalently,%
\begin{align}
\mathcal{E}  &  \rightarrow-\left(  \frac{2}{\sigma}\right)  ^{2}\alpha^{2}%
\xi,\label{Et2}\\
\eta &  \rightarrow l\left(  l+1\right)  ,\label{xit2}\\
\ell\left(  \ell+1\right)   &  \rightarrow-\left(  \frac{2}{\sigma}\right)
^{2}E-\frac{1}{4}, \label{lt2}%
\end{align}
\newline Here $E$, $\xi$, and $l$ are the eigenvalues, the coupling constants,
and the angular momentum of the $U\left(  r\right)  $ system; $\mathcal{E}$,
$\eta$, and $\ell$ are those of the $V\left(  \rho\right)  $ system. As a
result, the corresponding transformation of coordinates is%
\begin{equation}
r\rightarrow\frac{2}{\sigma}\ln\alpha\rho\text{\ \ or \ }\frac{1}{\alpha
}e^{\frac{\sigma}{2}r}\rightarrow\rho, \label{rt1}%
\end{equation}
and the transformation of the eigenfunction is%
\begin{equation}
u\left(  r\right)  \rightarrow\rho^{-1/2}v\left(  \rho\right)  \text{\ \ or
\ \ }e^{\frac{\sigma}{4}r}u\left(  r\right)  \rightarrow v\left(  \rho\right)
. \label{wft1}%
\end{equation}

\end{theorem}

\begin{proof}
First, we prove the dual transformation from $U\left(  r\right)  $ to
$V\left(  \rho\right)  $. \newline The radial equation of the central
potential $U\left(  r\right)  =\xi e^{\sigma r}$ reads%
\begin{equation}
\frac{d^{2}u\left(  r\right)  }{dr^{2}}+\left[  E-\frac{l\left(  l+1\right)
}{r^{2}}-\xi e^{\sigma r}\right]  u\left(  r\right)  =0. \label{raeq1}%
\end{equation}
Substituting the replacements (\ref{rt1}) and (\ref{wft1}) into Eq.
(\ref{raeq1}) gives%
\begin{equation}
\frac{d}{\frac{2}{\sigma\rho}d\rho}\frac{d\left(  \rho^{-1/2}v\left(
\rho\right)  \right)  }{\frac{2}{\sigma\rho}d\rho}+\left[  E-\frac{l\left(
l+1\right)  }{\left(  \frac{2}{\sigma}\ln\alpha\rho\right)  ^{2}}-\xi
\alpha^{2}\rho^{2}\right]  \rho^{-1/2}v\left(  \rho\right)  =0.
\end{equation}
Rearrange the equation as%
\begin{equation}
\frac{d^{2}v\left(  \rho\right)  }{d\rho^{2}}+\left[  -\left(  \frac{2}%
{\sigma}\right)  ^{2}\alpha^{2}\xi-\frac{-\left(  \frac{2}{\sigma}\right)
^{2}E-1/4}{\rho^{2}}-\frac{l\left(  l+1\right)  }{\left(  \rho\ln\alpha
\rho\right)  ^{2}}\right]  v\left(  \rho\right)  =0. \label{vr1}%
\end{equation}
Substituting the replacements (\ref{Et}), (\ref{xit}), and (\ref{lt}) into Eq.
(\ref{vr1}) gives the radial equation of the potential $V\left(  \rho\right)
$,
\begin{equation}
\frac{d^{2}v\left(  \rho\right)  }{d\rho^{2}}+\left[  \mathcal{E}-\frac
{\ell\left(  \ell+1\right)  }{\rho^{2}}-\frac{\eta}{\left(  \rho\ln\alpha
\rho\right)  ^{2}}\right]  v\left(  \rho\right)  =0.
\end{equation}
\newline

Second, we prove the dual transformation from $V\left(  \rho\right)  $ to
$U\left(  r\right)  $.\newline The radial equation of the central potential
$U\left(  r\right)  =\frac{\xi}{\left(  r\ln\beta r\right)  ^{2}}\ $reads%
\begin{equation}
\frac{d^{2}u\left(  r\right)  }{dr^{2}}+\left[  E-\frac{l\left(  l+1\right)
}{r^{2}}-\frac{\xi}{\left(  r\ln\beta r\right)  ^{2}}\right]  u\left(
r\right)  =0. \label{rqeq2}%
\end{equation}
Substituting the replacements (\ref{rt1}) and (\ref{wft1}) into (\ref{rqeq2})
gives%
\begin{equation}
\frac{d}{\frac{\sigma}{2\beta}e^{\frac{\sigma}{2}\rho}d\rho}\frac
{de^{\frac{\sigma}{4}\rho}v\left(  \rho\right)  }{\frac{\sigma}{2\beta
}e^{\frac{\sigma}{2}\rho}d\rho}+\left[  E-\frac{l\left(  l+1\right)  }{\left(
\frac{1}{\beta}e^{\frac{\sigma}{2}\rho}\right)  ^{2}}-\frac{\xi}{\left(
\frac{1}{\beta}e^{\frac{\sigma}{2}\rho}\frac{\sigma}{2}\rho\right)  ^{2}%
}\right]  e^{\frac{\sigma}{4}\rho}v\left(  \rho\right)  =0.
\end{equation}
Rearrange the equation as%
\begin{equation}
\frac{d^{2}v\left(  \rho\right)  }{d\rho^{2}}+\left[  -\left(  \frac{\sigma
}{2}\right)  ^{2}l\left(  l+1\right)  -\left(  \frac{\sigma}{2}\right)
^{2}\frac{1}{4}-\frac{\xi}{\rho^{2}}+\left(  \frac{\sigma}{2}\right)
^{2}\frac{E}{\beta^{2}}e^{\sigma\rho}\right]  v\left(  \rho\right)  =0.
\label{vr2}%
\end{equation}
Substituting the replacements (\ref{Et2}), (\ref{xit2}), and (\ref{lt2}) into
Eq. (\ref{vr2}) gives the radial equation of the potential $V\left(
\rho\right)  $,%
\begin{equation}
\frac{d^{2}v\left(  \rho\right)  }{d\rho^{2}}+\left[  \mathcal{E}-\frac
{\ell\left(  \ell+1\right)  }{\rho^{2}}-\eta e^{\sigma\rho}\right]  v\left(
\rho\right)  =0.
\end{equation}

\end{proof}

In the following, we show how the result given by Theorem \ref{expln} is reached.

\textit{The duality transformation from }$U\left(  r\right)  $\textit{ to
}$V\left(  \rho\right)  $\textit{.} The duality transformation from $U\left(
r\right)  $ to $V\left(  \rho\right)  $ is obtained by choosing the
term$\ Eg^{\prime}\left(  \rho\right)  ^{2}$ in Eq. (\ref{3}) to be a part of
the centrifugal potential term in Condition\ \ref{cd1}, i.e.,%
\begin{equation}
g^{\prime}\left(  \rho\right)  ^{2}=\frac{\gamma^{2}}{\rho^{2}}, \label{g3}%
\end{equation}
so%
\begin{equation}
g\left(  \rho\right)  =\gamma\ln\alpha\rho\label{grho2}%
\end{equation}
with $\alpha$ an arbitrary constant.

Substituting Eq. (\ref{grho2}) into Eq. (\ref{3}) gives%
\begin{equation}
\frac{d^{2}v\left(  \rho\right)  }{d\rho^{2}}+\left[  -\frac{\gamma^{2}}%
{\rho^{2}}U\left(  \rho\right)  +\frac{E\gamma^{2}+1/4}{\rho^{2}}%
-\frac{l\left(  l+1\right)  }{\left(  \rho\ln\alpha\rho\right)  ^{2}}\right]
v\left(  \rho\right)  =0. \label{vrho3}%
\end{equation}
Here the term $\left(  E\gamma^{2}+1/4\right)  /\rho^{2}$ is the centrifugal
potential term.

Now we need a term serving as the term of the eigenvalue $\mathcal{E}$
according to Condition\ \ref{cd2}. Choosing%
\begin{equation}
-\frac{\gamma^{2}}{\rho^{2}}U\left(  \rho\right)  =\mathcal{E}%
\end{equation}
gives%
\begin{equation}
U\left(  \rho\right)  =-\frac{\rho^{2}}{\gamma^{2}}\mathcal{E}. \label{Urho2}%
\end{equation}

By\ Eq. (\ref{grho2}), we obtain the transformation of coordinates between two
dual potentials:%
\begin{equation}
r\rightarrow\gamma\ln\alpha\rho. \label{grho22}%
\end{equation}
Then we have%
\begin{equation}
U\left(  r\right)  =-\frac{\mathcal{E}}{\alpha^{2}\gamma^{2}}e^{2r/\gamma}.
\end{equation}

By Eqs. (\ref{vrho3}) and (\ref{Urho2}), Eq. (\ref{3}) becomes%
\begin{equation}
\frac{d^{2}v\left(  \rho\right)  }{d\rho^{2}}+\left[  \mathcal{E}%
-\frac{-E\gamma^{2}-1/4}{\rho^{2}}-\frac{l\left(  l+1\right)  }{\left(
\rho\ln\alpha\rho\right)  ^{2}}\right]  v\left(  \rho\right)  =0.
\label{vrho22}%
\end{equation}
This is just the radial equation of the potential%
\begin{equation}
V\left(  \rho\right)  =\frac{l\left(  l+1\right)  }{\left(  \rho\ln\alpha
\rho\right)  ^{2}}.
\end{equation}

Rewrite these two potentials $U\left(  r\right)  $ and $V\left(  \rho\right)
$ as%
\begin{equation}
U\left(  r\right)  =\xi e^{\sigma r}%
\end{equation}
with%
\begin{align}
\gamma &  =\frac{2}{\sigma},\label{gammasigam}\\
\xi &  =-\left(  \frac{\sigma}{2}\right)  ^{2}\frac{\mathcal{E}}{\alpha^{2}}
\label{xieps}%
\end{align}
and%
\begin{equation}
V\left(  \rho\right)  =\frac{\eta}{\left(  \rho\ln\alpha\rho\right)  ^{2}}%
\end{equation}
with%
\begin{equation}
\eta=l\left(  l+1\right)  . \label{etal}%
\end{equation}
The term being proportional to $1/\rho^{2}$ in Eq. (\ref{vrho22}) is the
centrifugal potential, so%
\begin{equation}
\ell\left(  \ell+1\right)  =-E\gamma^{2}-\frac{1}{4}. \label{ellE}%
\end{equation}
By Eqs. (\ref{gammasigam}), (\ref{xieps}), (\ref{etal}), and (\ref{ellE}), we
obtain the duality relations (\ref{Et}), (\ref{xit}), and (\ref{lt}). By Eqs.
(\ref{grho22}) and (\ref{gammasigam}), we obtain the transformation of
coordinates (\ref{rt1}). Then by Eq. (\ref{g3}), we obtain the transformation
of the eigenfunction,%
\begin{equation}
u\left(  r\right)  \rightarrow\sqrt{\frac{2}{\rho\sigma}}v\left(  \rho\right)
.
\end{equation}
Dropping the constant factor, we arrive at the replacement (\ref{wft1}).

\textit{The duality transformation from }$V\left(  \rho\right)  $\textit{ to
}$U\left(  r\right)  $\textit{.} First, the duality transformation from
$V\left(  \rho\right)  $ to $U\left(  r\right)  $ is obtained by choosing%
\begin{equation}
g^{\prime}\left(  \rho\right)  ^{2}U\left(  \rho\right)  =\frac{\gamma^{2}%
}{\rho^{2}}%
\end{equation}
in Eq. (\ref{3}) to be a part of the centrifugal potential term in Condition
\ref{cd1}. This gives%
\begin{equation}
g\left(  \rho\right)  =\int\frac{\gamma}{\rho\sqrt{U\left(  \rho\right)  }%
}d\rho. \label{grho3}%
\end{equation}
Substituting Eq. (\ref{grho3}) into Eq. (\ref{3}), we have%
\begin{equation}
\frac{d^{2}v\left(  \rho\right)  }{d\rho^{2}}+\left[  \frac{\gamma^{2}E}%
{\rho^{2}U\left(  \rho\right)  }+\frac{1/4-\gamma^{2}}{\rho^{2}}+\frac{3}%
{16}\frac{U^{\prime}\left(  \rho\right)  ^{2}}{U^{2}\left(  \rho\right)
}-\frac{1}{4\rho}\frac{U^{\prime}\left(  \rho\right)  }{U\left(  \rho\right)
}-\frac{U^{\prime\prime}\left(  \rho\right)  }{4U\left(  \rho\right)  }%
-\frac{l\left(  l+1\right)  }{\rho^{2}U\left(  \rho\right)  \left(  \int%
\frac{1}{\rho\sqrt{U\left(  \rho\right)  }}d\rho\right)  ^{2}}\right]
v\left(  \rho\right)  =0.
\end{equation}
Here the term $\left(  1/4-\gamma^{2}\right)  /\rho^{2}$ is a part of the
centrifugal-potential term.

Second, let us choose a term serving as the term of the eigenvalue
$\mathcal{E}$ according to Condition\ \ref{cd2}. Choosing%
\begin{equation}
\frac{1}{\rho^{2}U\left(  \rho\right)  \left(  \int\frac{1}{\rho\sqrt{U\left(
\rho\right)  }}d\rho\right)  ^{2}}=\alpha^{2}%
\end{equation}
gives%
\begin{equation}
U\left(  \rho\right)  =\frac{e^{-2\alpha\rho}}{\rho^{2}}. \label{Urho3}%
\end{equation}
Then by Eq. (\ref{grho3}) we obtain%
\begin{equation}
g\left(  \rho\right)  =\frac{\gamma}{\alpha}e^{\alpha\rho}, \label{grho4}%
\end{equation}
so the replacement of coordinates between two dual potentials is%
\begin{equation}
r\rightarrow\frac{\gamma}{\alpha}e^{\alpha\rho}. \label{grho44}%
\end{equation}
Then we have%
\begin{equation}
U\left(  r\right)  =\frac{\gamma^{2}}{\left(  r\ln\frac{\alpha}{\gamma
}r\right)  ^{2}}.
\end{equation}
By Eqs. (\ref{grho3}) and (\ref{Urho3}), Eq. (\ref{3}) becomes%
\begin{equation}
\frac{d^{2}v\left(  \rho\right)  }{d\rho^{2}}+\left\{  -\left[  l\left(
l+1\right)  +\frac{1}{4}\right]  \alpha^{2}-\frac{\gamma^{2}}{\rho^{2}}%
+\gamma^{2}Ee^{2\alpha\rho}\right\}  v\left(  \rho\right)  =0
\end{equation}
This is just the radial equation of the potential%
\begin{equation}
V\left(  \rho\right)  =-\gamma^{2}Ee^{2\alpha\rho}.
\end{equation}

Rewrite these two potentials $U\left(  r\right)  $ and $V\left(  \rho\right)
$ as%
\begin{equation}
U\left(  r\right)  =\frac{\xi}{\left(  r\ln\beta r\right)  ^{2}}%
\end{equation}
with%
\begin{equation}
\gamma=\frac{\alpha}{\beta} \label{gammaab}%
\end{equation}
and%
\begin{equation}
V\left(  \rho\right)  =\eta e^{\sigma\rho}%
\end{equation}
with%
\begin{align}
\alpha &  =\frac{\sigma}{2},\label{alphasigma}\\
\gamma^{2}  &  =-\frac{\eta}{E}. \label{gammaetaE}%
\end{align}
The term being proportional to $1/\rho^{2}$ in Eq. (\ref{vrho22}) is the
centrifugal potential, so%
\begin{equation}
\ell\left(  \ell+1\right)  =\gamma^{2}. \label{ellgamma}%
\end{equation}

By Eqs. (\ref{gammaab}), (\ref{alphasigma}), (\ref{gammaetaE}) and
(\ref{ellgamma}), we obtain the duality relations (\ref{Et2}), (\ref{xit2}),
and (\ref{lt2}). By Eqs. (\ref{gammaab}) and (\ref{alphasigma}), we obtain the
transformation of coordinates (\ref{rt1}). Then by Eq. (\ref{grho4}), we
obtain the replacement of the eigenfunction,%
\begin{equation}
u\left(  r\right)  \rightarrow\xi^{1/4}e^{\frac{\sigma}{4}\rho}v\left(
\rho\right)  .
\end{equation}
Dropping the constant factor, we arrive at the replacement (\ref{wft1}).

\section{The quantum Newton duality as a method for solving eigenproblems}

\subsection{General discussions \label{SolvingGD}}

The Newton duality reveals an internal connection among various potentials.
Once a solution of a potential is known, one can immediately obtain the
solution of its dual potentials by the duality relation. That is to say, the
Newton duality can serve as a method for solving eigenproblems.

Concretely, as shown in section \ref{polynomialntermsmany}, $N+1$ $N$-term
general polynomial potentials can form a dual-potential set
\begin{equation}
\left\{  V_{1}\left(  r\right)  ,V_{2}\left(  r\right)  ,\cdots,V_{N+1}\left(
r\right)  \right\}  ,
\end{equation}
in which any two potentials are dual to each other. That is to say, once a
solution of a potential in a dual-potential set is known, the solution of the
other $N$ potentials in the set can be immediately obtained by the duality relation.

For one-term potentials, the duality set is ($ar^{a+1}$, $Ar^{\frac{4}{a+3}%
-2}$), for two-term potentials, the duality set is ($ar^{\alpha}+br^{\beta}$,
$A_{1}r^{\frac{4}{\alpha+2}-2}+B_{1}r^{\frac{2}{\alpha+2}\left(
\beta+2\right)  -2}$, $A_{2}r^{\frac{2}{\beta+2}\left(  \alpha+2\right)
-2}+B_{2}r^{\frac{4}{\beta+2}-2}$), for three-term potentials, the duality set
is ($ar^{\alpha}+br^{\beta}+cr^{\gamma}$, $A_{1}r^{\frac{4}{\alpha+2}-2}%
+B_{1}r^{\frac{2}{\alpha+2}\left(  \beta+2\right)  -2}+C_{1}r^{\frac{2}%
{\alpha+2}\left(  \gamma+2\right)  -2}$, $A_{2}r^{\frac{2}{\beta+2}\left(
\alpha+2\right)  -2}+B_{2}r^{\frac{4}{\beta+2}-2}+C_{2}r^{\frac{2}{\beta
+2}\left(  \gamma+2\right)  -2}$ $A_{2}r^{\frac{2}{\beta+2}\left(
\alpha+2\right)  -2}+B_{2}r^{\frac{4}{\beta+2}-2}$, $A_{3}r^{\frac{2}%
{\gamma+2}\left(  \alpha+2\right)  -2}+B_{3}r^{\frac{2}{\gamma+2}\left(
\beta+2\right)  -2}+C_{3}r^{\frac{4}{\gamma+2}-2}$), and so on.

In the following, we show how this works.

\subsection{Power potentials \label{SolvingPP}}

First consider some power potentials, including three-dimensional and
arbitrary-dimensional harmonic-oscillator potentials\ and Coulomb potentials,
$r^{2/3}$-potential and $1/\sqrt{r}$-potential, and $1/r^{3/2}$-potential and
$r^{6}$-potential.

\subsubsection{The harmonic-oscillator potential\ and the Coulomb potential:
three dimensions and arbitrary dimensions \label{h-C}}

To show how to solve a potential from its Newton duality, we first use the
harmonic-oscillator potential\ and the Coulomb potential in three dimensions
and in arbitrary dimensions as examples. The harmonic-oscillator
potential\ and the Coulomb potential are Newtonianly dual to each other
\cite{arnold1990huygens}, so we can solve one from another.

\paragraph{The harmonic-oscillator potential and the Coulomb potential: three
dimensions}

The harmonic-oscillator potential and the Coulomb potential form a duality set%
\begin{equation}
(\xi r^{2},\frac{\eta}{\rho}),
\end{equation}
i.e., the harmonic-oscillator potential
\begin{equation}
U\left(  r\right)  =\xi r^{2}%
\end{equation}
and the Coulomb potential%
\begin{equation}
V\left(  \rho\right)  =\frac{\eta}{\rho}%
\end{equation}
are Newtonianly dual to each other.

In the following, we solve the Coulomb potential from the harmonic-oscillator potential.

\textit{The harmonic-oscillator potential.} For the harmonic-oscillator
potential $U\left(  r\right)  =\xi r^{2}$, i.e., the power potential $U\left(
r\right)  =\xi r^{a+1}$ with $a=1$, the radial equation reads%
\begin{equation}
\frac{d^{2}u\left(  r\right)  }{dr^{2}}+\left[  E-\frac{l\left(  l+1\right)
}{r^{2}}-\xi r^{2}\right]  u\left(  r\right)  =0.
\end{equation}
The radial eigenfunction is
\cite{ronveaux1995heun,slavyanov2000special,li2016exact} (see Appendix
\ref{Harmonic})%
\begin{equation}
u_{l}\left(  r\right)  =A_{l}e^{-\frac{\sqrt{\xi}}{2}r^{2}}\xi^{\left(
l+1\right)  /4}r^{l+1}N\left(  2l+1,0,\frac{E}{\xi^{1/2}},0,\xi^{1/4}r\right)
, \label{radwfheun}%
\end{equation}
where $N\left(  \alpha,\beta,\gamma,\delta,z\right)  $ is the Heun biconfluent
function \cite{ronveaux1995heun,slavyanov2000special}. To be in accordance
with the form of the solution of other potentials discussed later, we first
express the solution of the harmonic-oscillator potential by the Heun
function. Usually, the eigenfunction of the harmonic-oscillator potential is
represented by the hypergeometric function. By the relation between the Heun
function and the hypergeometric function
\cite{ronveaux1995heun,slavyanov2000special},%
\begin{equation}
N\left(  \alpha,0,\gamma,0,z\right)  =\text{ }_{1}F_{1}\left(  \frac{1}%
{2}+\frac{\alpha}{4}-\frac{\gamma}{4},{1+}\frac{\alpha}{2},{{z}^{2}}\right)
{,} \label{Heun1F1}%
\end{equation}
the eigenfunction (\ref{radwfheun}) reduces to%
\begin{equation}
u_{l}\left(  r\right)  =A_{l}e^{-\frac{\sqrt{\xi}r^{2}}{2}}\xi^{\left(
l+1\right)  /4}r^{l+1}\text{ }_{1}F_{1}\left(  \frac{l}{2}+\frac{3}{4}%
-\frac{E}{4\sqrt{\xi}},\frac{3}{2}+l,\sqrt{\xi}r^{2}\right)  .
\end{equation}
In this approach, the eigenvalue is the zero of the function $K_{2}\left(
2l+1,0,\frac{\mathcal{E}}{\xi^{1/2}},0\right)  $, i.e., \cite{li2016exact}%
\begin{equation}
K_{2}\left(  2l+1,0,\frac{E}{\xi^{1/2}},0\right)  =0,
\end{equation}
where%
\begin{equation}
K_{2}\left(  \alpha,\beta,\gamma,\delta\right)  =\frac{\Gamma\left(
1+\alpha\right)  }{\Gamma\left(  \left(  \alpha-\gamma\right)  /2\right)
\Gamma\left(  1+\left(  \alpha+\gamma\right)  /2\right)  }J_{1+\left(
\alpha+\gamma\right)  /2}\left(  \frac{\alpha+\gamma}{2},\beta,\frac{3}%
{2}\alpha-\frac{1}{2}\gamma,\delta-\frac{\alpha-\gamma}{2}\beta\right)
\label{K2}%
\end{equation}
with%
\begin{equation}
J_{\lambda}\left(  \alpha,\beta,\gamma,\delta\right)  =\int_{0}^{\infty
}x^{\lambda-1}e^{-\beta x-x^{2}}N\left(  \alpha,\beta,\gamma,\delta,x\right)
dx.
\end{equation}

In the case of harmonic-oscillator potentials, $K_{2}\left(  2l+1,0,\frac
{E}{\xi^{1/2}},0\right)  $ reduces to \cite{batola1982quelques}%
\begin{equation}
K_{2}\left(  2l+1,0,\frac{E}{\xi^{1/2}},0\right)  =\frac{\Gamma\left(
l+\frac{3}{2}\right)  }{\Gamma\left(  \frac{3}{4}+\frac{l}{2}-\frac{E}%
{4\xi^{1/2}}\right)  }.
\end{equation}
Therefore, the eigenvalue is the singularities of $\Gamma\left(  \frac{3}%
{4}+\frac{l}{2}-\frac{E}{4\xi^{1/2}}\right)  $:%
\begin{equation}
\frac{3}{4}+\frac{l}{2}-\frac{E}{4\xi^{1/2}}=-n_{r},\text{ \ }n_{r}%
=0,1,2\ldots.
\end{equation}
Then%
\begin{equation}
E=2\sqrt{\xi}\left(  2n_{r}+l+\frac{3}{2}\right)  ,\text{ \ }n_{r}%
=0,1,2,\ldots, \label{spech}%
\end{equation}
where $n_{r}$ is the radial quantum number.

The reason why, instead of the Hermite polynomial, we use the Heun function to
express the eigenfunction of the harmonic-oscillator potential is that the
eigenfunction of the Coulomb potential can also be expressed as a Heun function.

\textit{The Coulomb potential. }The Coulomb potential is a dual potential of
the harmonic-oscillator potential. Then the solution of the Coulomb potential
can be obtained from the solution of the harmonic-oscillator potential by the
duality relation given by Theorems (\ref{Twf3D}) and (\ref{Tev3D}).

The radial equation of the Coulomb potential is
\begin{equation}
\frac{d^{2}v\left(  \rho\right)  }{d\rho^{2}}+\left[  \mathcal{E}-\frac
{\ell\left(  \ell+1\right)  }{\rho^{2}}-\frac{\eta}{\rho}\right]  v\left(
\rho\right)  =0.
\end{equation}

By the duality relations (\ref{aA3D}), (\ref{coup}), (\ref{en}), and
(\ref{ltrans}) given by Theorem (\ref{Twf3D}), we have the following
replacements:%
\begin{align}
\frac{a+3}{2}  &  \rightarrow\frac{2}{A+3}\text{ with }a=1,\label{aAa2}\\
E  &  \rightarrow-4\eta,\label{eeta}\\
\xi &  \rightarrow-4\mathcal{E},\label{xie}\\
l+\frac{1}{2}  &  \rightarrow2\left(  \ell+\frac{1}{2}\right)  .
\label{lellhc}%
\end{align}
By these duality relations, we can obtain the eigenvalue and eigenfunction of
the Coulomb potential from the eigenvalue and eigenfunction of the
harmonic-oscillator potential, Eqs (\ref{radwfheun}) and (\ref{spech}).

The duality transformation of coordinates, Eq. (\ref{coorTpower}), and the
duality relation of radial eigenfunctions, Eq. (\ref{wTpower}), in this case
become%
\begin{align}
r  &  \rightarrow\rho^{1/2}\text{ \ \ or \ }r^{2}\rightarrow\rho
,\label{coorhacou}\\
u\left(  r\right)   &  \rightarrow\rho^{-1/4}v\left(  \rho\right)  .
\label{wfhacou}%
\end{align}

\textit{Eigenvalue.} Substituting the duality relation of eigenvalues, Eq.
(\ref{specf}), into the eigenvalue of the harmonic-oscillator potential, Eq.
(\ref{spech}), we have%
\begin{equation}
-4\eta=2\sqrt{-4\mathcal{E}}\left[  2n_{r}+2\left(  \ell+\frac{1}{2}\right)
-\frac{1}{2}+\frac{3}{2}\right]  ,\text{ \ }n_{r}=0,1,2,\ldots.
\end{equation}
Solving $\mathcal{E}$ gives%
\begin{equation}
\mathcal{E}=-\frac{\eta^{2}}{4}\frac{1}{\left(  n_{r}+\ell+1\right)  ^{2}%
},\text{ \ }n_{r}=0,1,2,\ldots.
\end{equation}
This is just the eigenvalue of the Coulomb potential
\cite{flugge1994practical}.

\textit{Eigenfunction}. Starting from the eigenfunction of the
harmonic-oscillator potential (\ref{radwfheun}), by the duality relations,
including the replacements (\ref{eeta}), (\ref{xie}), (\ref{lellhc}), the
transformation of coordinates (\ref{coorhacou}), and the transformation of
eigenfunctions (\ref{wfhacou}), we arrive at%
\begin{equation}
\frac{1}{\rho^{1/4}}v\left(  \rho\right)  =A_{\ell}e^{-\frac{\sqrt
{-4\mathcal{E}}}{2}\sqrt{\rho}^{2}}\left[  \left(  -4\mathcal{E}\right)
^{1/4}\sqrt{\rho}\right]  ^{2\left(  \ell+1/2\right)  +1/2}N\left(  2\left[
2\left(  \ell+\frac{1}{2}\right)  \right]  ,0,\frac{-4\eta}{\left(
-4\mathcal{E}\right)  ^{1/2}},0,\left(  -4\mathcal{E}\right)  ^{1/4}\sqrt
{\rho}\right)  .
\end{equation}
Then we obtain the eigenfunction of the Coulomb potential%
\begin{equation}
v\left(  \rho\right)  =A_{\ell}e^{-\sqrt{-\mathcal{E}}\rho}\left(
2\sqrt{-\mathcal{E}}\rho\right)  ^{\ell+1}N\left(  4\ell+2,0,-\frac{2\eta
}{\sqrt{-\mathcal{E}}},0,\left(  2\rho\right)  ^{1/2}\left(  -\mathcal{E}%
\right)  ^{1/4}\right)  . \label{nurhoCoulomb}%
\end{equation}

By the relation between the hypergeometric function and the Heun function, Eq.
(\ref{Heun1F1}), the eigenfunction of the Coulomb potential, Eq.
(\ref{nurhoCoulomb}), becomes%
\begin{equation}
v\left(  \rho\right)  =A_{\ell}e^{-\sqrt{-\mathcal{E}}\rho}\left(
2\sqrt{-\mathcal{E}}\right)  ^{\ell+1}\rho^{\ell+1}\text{ }_{1}F_{1}\left(
\ell+1+\frac{\eta}{2\sqrt{-\mathcal{E}}},2\left(  \ell+1\right)
,2\sqrt{-\mathcal{E}}\rho\right)  .
\end{equation}
This is just the familiar form of the eigenfunction of the Coulomb potential
\cite{flugge1994practical}.

It should be noted here that the same procedure can also be applied to solve
the harmonic-oscillator potential\ from the Coulomb potential.

\paragraph{$n$-dimensional harmonic-oscillator potentials and $m$-dimensional
Coulomb potentials}

In this section, by the Newton duality relation, we calculate the eigenproblem
of the $m$-dimensional Coulomb potential from the solution of the
$n$-dimensional harmonic-oscillator potential.

\textit{The} $n$\textit{-dimensional harmonic-oscillator potential.} In $n$
dimensions, the radial equation for the harmonic-oscillator potential
$U\left(  r\right)  =\xi r^{2}$, i.e., the power potential $U\left(  r\right)
=\xi r^{a+1}$ with $a=1$, reads%
\begin{equation}
\frac{d^{2}u\left(  r\right)  }{dr^{2}}+\left[  E-\frac{\left(
l-3/2+n/2\right)  \left(  l-1/2+n/2\right)  }{r^{2}}-\xi r^{2}\right]
u\left(  r\right)  =0.
\end{equation}
The $n$-dimensional harmonic-oscillator radial eigenfunction is
\cite{ronveaux1995heun,slavyanov2000special,li2016exact}%
\begin{equation}
u\left(  r\right)  =A_{l}\xi^{l/4+\left(  n-1\right)  /8}{r}^{l+n/2-1/2}%
e^{-\sqrt{\xi}{r}^{2}/2}N\left(  2l+n-2,0,{\frac{E}{\xi^{1/2}}},0,\xi
^{1/4}r\right)  . \label{radwfhn}%
\end{equation}
By the relation between the hypergeometric function and the Heun function, Eq.
(\ref{Heun1F1}), $n$-dimensional radial harmonic-oscillator eigenfunction, Eq.
(\ref{radwfhn}), becomes%
\begin{equation}
u_{l}\left(  r\right)  =A_{l}\xi^{l/4+\left(  n-1\right)  /8}{r}%
^{l+n/2-1/2}e^{-\sqrt{\xi}r^{2}/2}\text{ }_{1}F_{1}\left(  \frac{l}{2}%
+\frac{n}{4}-\frac{E}{4\sqrt{\xi}};\,l+\frac{n}{2};\sqrt{\xi}r^{2}\right)  .
\end{equation}

The eigenvalue is the zero of the function $K_{2}\left(  2l+n-2,0,{\frac
{E}{\xi^{1/2}}},0\right)  $, i.e., \cite{li2016exact}%
\begin{equation}
K_{2}\left(  2l+n-2,0,{\frac{E}{\xi^{1/2}}},0\right)  =0,
\end{equation}
In this case, $K_{2}\left(  2l+n-2,0,{\frac{E}{\xi^{1/2}}},0\right)  $ reduces
to%
\begin{equation}
K_{2}\left(  2l+n-2,0,{\frac{E}{\xi^{1/2}}},0\right)  =\frac{\Gamma\left(
l+\frac{n}{2}-3\right)  }{\Gamma\left(  \frac{l}{2}+\frac{n}{4}-\frac{E}%
{4\xi^{1/2}}\right)  }.
\end{equation}
Then the eigenvalue of the bound state is the singularities of $\Gamma\left(
\frac{3}{4}+\frac{l}{2}-\frac{E}{4\xi^{1/2}}\right)  $:%
\begin{equation}
E=2\sqrt{\xi}\left(  2n_{r}+l+\frac{n}{2}\right)  ,\text{ }\ \text{\ }%
n_{r}=0,1,2,\ldots, \label{spechn}%
\end{equation}
where $n_{r}$ is the radial quantum number.

\textit{The }$m$\textit{-dimensional Coulomb potential. }The $m$-dimensional
Coulomb potential is a dual potential of the $n$-dimensional
harmonic-oscillator potential. In the following, we solve the solution of the
$m$-dimensional Coulomb potential from the solution of the $n$-dimensional
harmonic-oscillator potential.

The radial equation of the $m$-dimensional Coulomb potential is
\begin{equation}
\frac{d^{2}v\left(  \rho\right)  }{d\rho^{2}}+\left[  \mathcal{E}%
-\frac{\left(  \ell-3/2+m/2\right)  \left(  \ell-1/2+m/2\right)  }{\rho^{2}%
}-\frac{\eta}{\rho}\right]  v\left(  \rho\right)  =0.
\end{equation}

By the duality relation between potentials in different dimensions, Eqs.
(\ref{nm-Eh}), (\ref{nm-XE}), and (\ref{lell}), we have the following
replacements:%
\begin{align}
E  &  \rightarrow-4\eta,\label{coupnm}\\
\xi &  \rightarrow-4\mathcal{E},\label{spectransnm}\\
l+\frac{n}{2}-1  &  \rightarrow2\left(  \ell+\frac{m}{2}-1\right)  .
\label{lnm}%
\end{align}

Substituting the duality relation, Eq. (\ref{specfnm3}), into the bound-state
eigenvalue of the $n$-dimensional harmonic-oscillator potential, Eq.
(\ref{spechn}), we have%
\begin{equation}
-4\eta=2\sqrt{-4\mathcal{E}}\left(  2n_{r}+2\left(  \ell+\frac{m}{2}-1\right)
-\frac{n}{2}+1+\frac{n}{2}\right)  ,\text{ \ }n_{r}=0,1,2,\ldots.
\end{equation}
Then the eigenvalue of the $m$-dimensional Coulomb potential reads%
\begin{equation}
\mathcal{E}=-\frac{\eta^{2}/4}{n_{r}+\ell+m/2-1/2},\text{ \ }n_{r}%
=0,1,2,\ldots.
\end{equation}

The duality transformation of coordinates, Eq. (\ref{nmDcoorTpower}), and the
duality relation of radial eigenfunctions, Eq. (\ref{wftrans}), in this case
become%
\begin{align}
r  &  \rightarrow\rho^{1/2}\text{ \ \ or \ }r^{2}\rightarrow\rho
,\label{nmcoorhamcou}\\
u\left(  r\right)   &  \rightarrow\rho^{1/4}v\left(  \rho\right)  .
\label{nmwfhamcou}%
\end{align}

Substituting the duality relations (\ref{coupnm}), (\ref{spectransnm}),
(\ref{lnm}), (\ref{nmcoorhamcou}), and (\ref{nmwfhamcou}) into the
eigenfunction of the $n$-dimensional harmonic-oscillator, we have%
\begin{align}
\rho^{-1/4}v\left(  \rho\right)   &  =A_{\ell}\left[  \left(  -4\mathcal{E}%
\right)  ^{1/4}\right]  ^{2\left(  \ell+m/2-1\right)  +1/2}\sqrt{\rho
}^{2\left(  \ell+m/2-1\right)  +1/2}e^{-\frac{\sqrt{-4\mathcal{E}}}{2}%
\sqrt{\rho}^{2}}\nonumber\\
&  \times N\left(  2\left[  2\left(  \ell+\frac{m}{2}-1\right)  \right]
,0,{\frac{-4\eta}{\left(  -4\mathcal{E}\right)  ^{1/2}}},0,\left(
-4\mathcal{E}\right)  ^{1/4}\sqrt{\rho}\right)  .
\end{align}
Then the eigenfunction of the $m$-dimensional Coulomb potential reads%
\begin{equation}
v\left(  \rho\right)  =A_{\ell}\left(  2\sqrt{-\mathcal{E}}\right)
^{\ell+m/2-1/2}\rho^{\ell+m/2-1/2}e^{-\sqrt{-\mathcal{E}}\rho}N\left(
4\ell+2m-4,0,-{\frac{2\eta}{\sqrt{-\mathcal{E}}}},0,\left(  -\mathcal{E}%
\right)  ^{1/4}\left(  2\rho\right)  ^{1/2}\right)  . \label{mdCoulomb}%
\end{equation}
By the relation between the hypergeometric function and the Heun function, Eq.
(\ref{Heun1F1}), the eigenfunction of the $m$-dimensional Coulomb potential,
Eq. (\ref{mdCoulomb}), becomes%
\begin{equation}
v\left(  \rho\right)  =A_{\ell}\left(  2\sqrt{-\mathcal{E}}\right)
^{\ell-1/2+m/2}\rho^{\ell-1/2+m/2}e{^{-\sqrt{-\mathcal{E}}\rho}}\text{ }%
_{1}F_{1}\left(  \ell+\frac{m}{2}-\frac{1}{2}+{{\frac{\eta}{2\sqrt
{-\mathcal{E}}};\,2}}\ell+{{{m}}}-{{1;\,2\sqrt{-\mathcal{E}}\rho}}\right)  .
\end{equation}
This agrees with the result obtained by directly solving the eigenequation of
the $m$-dimensional Coulomb potential \cite{dong2011wave}.

\subsubsection{$r^{2/3}$-potential and $1/\sqrt{r}$-potential \label{2312}}

In this section, we solve the eigenproblem of the $r^{2/3}$-potential from the
solution of its duality $1/\sqrt{r}$-potential whose eigenproblem has been
exactly solved \cite{li2016exact}.

The duality of the $r^{2/3}$-potential
\begin{equation}
V\left(  \rho\right)  =\eta\rho^{2/3} \label{Vr23}%
\end{equation}
according to the duality relation given above, is the inverse-square-root
potential
\begin{equation}
U\left(  r\right)  =\frac{\xi}{\sqrt{r}}. \label{Ur12}%
\end{equation}

$1/\sqrt{r}$\textit{-potential.} The radial equation for the
inverse-square-root potential (\ref{Ur12})\ reads%
\begin{equation}
\frac{d^{2}u\left(  r\right)  }{dr^{2}}+\left[  E-\frac{l\left(  l+1\right)
}{r^{2}}-\frac{\xi}{\sqrt{r}}\right]  u\left(  r\right)  =0.
\end{equation}
The eigenproblem is solved exactly in Ref. \cite{li2016exact}. The
eigenfunction is%
\begin{equation}
u\left(  r\right)  =A_{l}\exp\left(  -\sqrt{-E}r+\frac{\xi}{\sqrt{-E}}\sqrt
{r}\right)  \left[  2\left(  -E\right)  ^{1/2}r\right]  ^{l+1}N\left(
4l+2,-\frac{\sqrt{2}\xi}{\left(  -E\right)  ^{3/4}},\frac{\xi^{2}}{2\left(
-E\right)  ^{3/2}},0,\sqrt{2\left(  -E\right)  ^{1/2}r}\right)  .
\label{radwfsq}%
\end{equation}
The eigenvalue is the zero of the function $K_{2}\left(  4l+2,-\frac{\sqrt
{2}\xi}{\left(  -E\right)  ^{3/4}},\frac{\xi^{2}}{2\left(  -E\right)  ^{3/2}%
},0\right)  $, i.e., \cite{li2016exact}
\begin{equation}
K_{2}\left(  4l+2,-\frac{\sqrt{2}\xi}{\left(  -E\right)  ^{3/4}},\frac{\xi
^{2}}{2\left(  -E\right)  ^{3/2}},0\right)  =0, \label{specsq}%
\end{equation}
where \cite{ronveaux1995heun,li2016exact}%
\begin{align}
&  K_{2}\left(  4l+2,-\frac{\sqrt{2}\xi}{\left(  -E\right)  ^{3/4}},\frac
{\xi^{2}}{2\left(  -E\right)  ^{3/2}},0\right)  =\frac{\Gamma\left(
4l+3\right)  }{\Gamma\left(  2l+1-\frac{\xi^{2}}{4\left(  -E\right)  ^{3/2}%
}\right)  \Gamma\left(  2l+2+\frac{\xi^{2}}{4\left(  -E\right)  ^{3/2}%
}\right)  }\nonumber\\
&  \times J_{2l+2+\frac{\xi^{2}}{4\left(  -E\right)  ^{3/2}}}\left(
2l+1+\frac{\xi^{2}}{4\left(  -E\right)  ^{3/2}},-\frac{\sqrt{2}\xi}{\left(
-E\right)  ^{3/4}},6l+3-\frac{\xi^{2}}{4\left(  -E\right)  ^{3/2}}%
,-\frac{\sqrt{2}\xi}{2\left(  -E\right)  ^{3/4}}\left[  \frac{\xi^{2}%
}{2\left(  -E\right)  ^{3/2}}-4l-2\right]  \right)
\end{align}
and \cite{ronveaux1995heun}%
\begin{align}
&  J_{2l+2+\frac{\xi^{2}}{4\left(  -E\right)  ^{3/2}}}\left(  2l+1+\frac
{\xi^{2}}{4\left(  -E\right)  ^{3/2}},-\frac{\sqrt{2}\xi}{\left(  -E\right)
^{3/4}},6l+3-\frac{\xi^{2}}{4\left(  -E\right)  ^{3/2}},-\frac{\sqrt{2}\xi
}{2\left(  -E\right)  ^{3/4}}\left[  \frac{\xi^{2}}{2\left(  -E\right)
^{3/2}}-4l-2\right]  \right) \nonumber\\
&  =\int_{0}^{\infty}dxx^{2l+1+\frac{\xi^{2}}{4\left(  -E\right)  ^{3/2}}%
}e^{\frac{\sqrt{2}\xi}{\left(  -E\right)  ^{3/4}}x-x^{2}}\nonumber\\
&  \times N\left(  2l+1+\frac{\xi^{2}}{4\left(  -E\right)  ^{3/2}}%
,-\frac{\sqrt{2}\xi}{\left(  -E\right)  ^{3/4}},6l+3-\frac{\xi^{2}}{4\left(
-E\right)  ^{3/2}},-\frac{\sqrt{2}\xi}{2\left(  -E\right)  ^{3/4}}\left[
\frac{\xi^{2}}{2\left(  -E\right)  ^{3/2}}-4l-2\right]  ,x\right)  .
\end{align}

$r^{2/3}$\textit{-potential.} The $r^{2/3}$-potential $V\left(  \rho\right)
=\eta\rho^{2/3}$ is the dual potential of the $1/\sqrt{r}$-potential $U\left(
r\right)  =\frac{\xi}{\sqrt{r}}$\textit{.}

The radial equation of $r^{2/3}$-potential\textit{ }is
\begin{equation}
\frac{d^{2}v\left(  \rho\right)  }{d\rho^{2}}+\left[  \mathcal{E}-\frac
{\ell\left(  \ell+1\right)  }{\rho^{2}}-\eta\rho^{2/3}\right]  v\left(
\rho\right)  =0.
\end{equation}
The duality relation between these two dual potentials,\ by Eqs. (\ref{coup}),
(\ref{en}), and (\ref{ltrans}), are%
\begin{align}
E  &  \rightarrow-\frac{9}{16}\eta,\label{eeta2}\\
\xi &  \rightarrow-\frac{9}{16}\mathcal{E},\label{xie2}\\
l+\frac{1}{2}  &  \rightarrow\frac{3}{4}\left(  \ell+\frac{1}{2}\right)  .
\label{lell2}%
\end{align}

Substituting the duality relation between the eigenvalues into the implicit
expression of the eigenvalue of the $r^{2/3}$-potential, Eq. (\ref{Vr23}),
gives%
\begin{equation}
K_{2}\left(  4\left[  \frac{3}{4}\left(  \ell+\frac{1}{2}\right)  -\frac{1}%
{2}\right]  +2,-\frac{\sqrt{2}\left(  -\frac{9}{16}\mathcal{E}\right)
}{\left(  \frac{9}{16}\eta\right)  ^{3/4}},\frac{\left(  -\frac{9}%
{16}\mathcal{E}\right)  ^{2}}{2\left(  \frac{9}{16}\eta\right)  ^{3/2}%
},0\right)  =0.
\end{equation}
Then we arrive at an implicit expression of the eigenvalue of the $r^{2/3}%
$-potential (\ref{Ur12}):%
\begin{equation}
K_{2}\left(  3\ell+\frac{3}{2},\frac{\sqrt{6}\mathcal{E}}{2\eta^{3/4}}%
,\frac{3\mathcal{E}^{2}}{8\eta^{3/2}},0\right)  =0.
\end{equation}

The duality transformation of coordinates, Eq. (\ref{coorTpower}), and the
duality relation of radial eigenfunctions, Eq. (\ref{wTpower}), in this case
become%
\begin{align}
r^{3/4}  &  \rightarrow\rho\text{\ \ or \ }r\rightarrow\rho^{4/3}%
,\label{rtrsq}\\
u\left(  r\right)   &  \rightarrow\rho^{1/6}v\left(  \rho\right)  .
\label{radeqtranssq}%
\end{align}

Substituting the duality relations (\ref{eeta2}), (\ref{xie2}), (\ref{lell2}),
(\ref{rtrsq}), and (\ref{radeqtranssq}) into the radial eigenfunction of the
$1/\sqrt{r}$-potential, Eq. (\ref{radwfsq}), gives%
\begin{align}
\rho^{1/6}v\left(  \rho\right)   &  =A_{l}\exp\left(  -\sqrt{\frac{9}{16}\eta
}\rho^{4/3}+\frac{-\frac{9}{16}\mathcal{E}}{\sqrt{\frac{9}{16}\eta}}\sqrt
{\rho^{4/3}}\right)  \left[  2\left(  \frac{9}{16}\eta\right)  ^{1/2}%
\rho^{4/3}\right]  ^{\frac{3}{4}\left(  \ell+1/2\right)  -1/2+1}\nonumber\\
&  \times N\left(  4\left[  \frac{3}{4}\left(  \ell+\frac{1}{2}\right)
-\frac{1}{2}\right]  +2,-\frac{-\frac{9}{16}\mathcal{E}\sqrt{2}}{\left(
\frac{9}{16}\eta\right)  ^{3/4}},\frac{\left(  -\frac{9}{16}\mathcal{E}%
\right)  ^{2}}{2\left(  \frac{9}{16}\eta\right)  ^{3/2}},0,\sqrt{2\left(
\frac{9}{16}\eta\right)  ^{1/2}\rho^{4/3}}\right)  .
\end{align}
Then we arrive at the radial eigenfunction of the $r^{2/3}$-potential:%
\begin{align}
v\left(  \rho\right)   &  =A_{l}\exp\left(  -\frac{3}{4}\eta^{1/2}\rho
^{4/3}-\frac{3\mathcal{E}}{4\eta^{1/2}}\rho^{2/3}\right)  \left(  \frac{3}%
{2}\eta^{1/2}\right)  ^{3\left(  \ell+1/2\right)  /4+1/2}\rho^{\ell
+1}\nonumber\\
&  \times N\left(  3\ell+\frac{3}{2},-\frac{\sqrt{6}\mathcal{E}}{2\eta^{3/4}%
},\frac{3\mathcal{E}^{2}}{8\eta^{3/2}},0,\frac{\sqrt{6}}{2}\eta^{1/4}%
\rho^{2/3}\right)  .
\end{align}

Additionally, in order to verify the solution of\ the $r^{2/3}$-potential
obtained by the duality relation given in this section, in Appendix \ref{V23}
we solve the $r^{2/3}$-potential by solving the radial eigenequation directly.

\subsubsection{$1/r^{3/2}$-potential and $r^{6}$-potential \label{326}}

In this section, we solve the eigenproblem of the $r^{6}$-potential $V\left(
\rho\right)  =\eta\rho^{6}$ from its duality $1/r^{3/2}$-potential $U\left(
r\right)  =\xi/r^{3/2}$. Nevertheless, these two potentials have not been
solved in literature. To show how to solve an eigenproblem by the duality
relation, additionally, we solve the exact solutions of these two potentials
in Appendices \ref{Vm32} and \ref{V6}.

The duality of the $r^{6}$-potential
\begin{equation}
V\left(  \rho\right)  =\eta\rho^{6}, \label{Vr6}%
\end{equation}
according to the duality relation (\ref{aA3D}), is the $1/r^{3/2}$-potential
\begin{equation}
U\left(  r\right)  =\frac{\xi}{r^{3/2}}. \label{Ur23}%
\end{equation}

$1/r^{3/2}$\textit{-potential.} The radial equation of the $1/r^{3/2}%
$-potential (\ref{Ur12})\ reads%

\begin{equation}
\frac{d^{2}u\left(  r\right)  }{dr^{2}}+\left[  E-\frac{l\left(  l+1\right)
}{r^{2}}-\frac{\xi}{r^{3/2}}\right]  u\left(  r\right)  =0.
\end{equation}
From Appendix \ref{Vm32}, the eigenfunction of the bound state of the
$1/r^{3/2}$-potential%
\begin{equation}
u\left(  r\right)  =A_{l}e^{-\sqrt{-E}r}r^{l+1}N\left(  4l+2,0,0,-\frac
{4\sqrt{2}\xi}{\left(  -E\right)  ^{1/4}},\sqrt{2\left(  -E\right)  ^{1/2}%
r}\right)  \label{radwf32}%
\end{equation}
and the eigenvalue of the bound state is given by an implicit expression%
\begin{equation}
K_{2}\left(  2\left(  2l+1\right)  ,0,0,-\frac{4\sqrt{2}\xi}{\left(
-E\right)  ^{1/4}}\right)  =0. \label{spec32}%
\end{equation}

$r^{6}$\textit{-potential.} The solution of the $r^{6}$-potential can be
obtained by its dual potential, the $1/r^{3/2}$-potential.

The radial equation of $r^{6}$-potential (\ref{Ur12})\textit{ }is
\begin{equation}
\frac{d^{2}v\left(  \rho\right)  }{d\rho^{2}}+\left[  \mathcal{E}-\frac
{\ell\left(  \ell+1\right)  }{\rho^{2}}-\eta\rho^{6}\right]  v\left(
\rho\right)  =0.
\end{equation}

The duality relations between these two dual potentials,\ by Eqs.
(\ref{coup}), (\ref{en}), and (\ref{ltrans}), are%

\begin{align}
E  &  \rightarrow-\frac{\eta}{16},\label{eeta3}\\
\xi &  \rightarrow-\frac{\mathcal{E}}{16},\label{xie3}\\
l+\frac{1}{2}  &  =\frac{1}{4}\left(  \ell+\frac{1}{2}\right)  . \label{lell3}%
\end{align}
Substituting the duality relation (\ref{specf}) between the eigenvalues of
these two dual potentials into the implicit expression of the eigenvalue of
the $1/r^{3/2}$-potential, Eq. (\ref{spec32}), gives%

\begin{equation}
K_{2}\left(  4\left[  \frac{1}{4}\left(  \ell+\frac{1}{2}\right)  -\frac{1}%
{2}\right]  +2,0,0,-\frac{4\sqrt{2}\left(  -\mathcal{E}/16\right)  }{\left(
\eta/16\right)  ^{1/4}}\right)  =0.
\end{equation}
An implicit expression of the eigenvalue of the $r^{6}$-potential then reads%
\begin{equation}
K_{2}\left(  \ell+\frac{1}{2},0,0,\frac{\sqrt{2}\mathcal{E}}{2\eta^{1/4}%
}\right)  =0.
\end{equation}

The duality transformation of coordinates, Eq. (\ref{coorTpower}), and the
duality relation of radial eigenfunctions, Eq. (\ref{wTpower}), in this case
become%
\begin{align}
r^{1/4}  &  \rightarrow\rho\text{\ \ or \ }r\rightarrow\rho^{4}%
,\label{coor32and6}\\
u\left(  r\right)   &  \rightarrow\rho^{3/2}v\left(  \rho\right)  .
\label{wft32and6}%
\end{align}

For eigenfunctions, substituting the duality relations (\ref{eeta3}),
(\ref{xie3}), (\ref{lell3}), (\ref{coor32and6}), and (\ref{wft32and6}) into
the radial eigenfunction of the $1/r^{3/2}$-potential, Eq. (\ref{radwf32}),
gives
\begin{equation}
\rho^{3/2}v\left(  \rho\right)  =A_{l}e^{-\sqrt{\frac{\eta}{16}}\rho^{4}}%
\rho^{4\left[  \frac{1}{4}\left(  \ell+1/2\right)  -1/2+1\right]  }N\left(
4\left[  \frac{1}{4}\left(  \ell+\frac{1}{2}\right)  -\frac{1}{2}\right]
+2,0,0,\frac{4\sqrt{2}\frac{\mathcal{E}}{16}}{\left(  \eta/16\right)  ^{1/4}%
},\sqrt{2\left(  \frac{\eta}{16}\right)  ^{1/2}\rho^{4}}\right)  .
\label{vr32}%
\end{equation}
The radial wave function of the $r^{6}$-potential then reads%
\begin{equation}
v\left(  \rho\right)  =A_{l}e^{-\frac{\sqrt{\eta}}{4}\rho^{4}}\rho^{\ell
+1}N\left(  \ell+\frac{1}{2},0,0,\frac{\sqrt{2}\mathcal{E}}{2\eta^{1/4}}%
,\frac{\sqrt{2}}{2}\eta^{1/4}\rho^{2}\right)  .
\end{equation}

The result agrees with the exact solution of the $r^{6}$-potential given in
Appendix \ref{V6}.

\subsection{Two-term general polynomial potentials \label{Solving2T}}

A two-term general polynomial potential has two dual potentials. That is to
say, one can solve two potentials from their dual potential. In this section,
we show how to solve two potentials from their two-term dual potential.

\subsubsection{Solving $V\left(  \rho\right)  =\frac{\eta}{\rho}+\frac
{\lambda}{\rho^{3/2}}$ and $V\left(  \rho\right)  =\eta\rho^{2}+\lambda
\rho^{6}$ from their duality $U\left(  r\right)  =\xi r^{2}+\frac{\mu}{r}$}

In this section, taking the\ dual set of potentials
\begin{equation}
\left(  \xi r^{2}+\frac{\mu}{r},\frac{\eta}{r}+\frac{\lambda}{r^{3/2}}%
,\eta\rho^{2}+\lambda\rho^{6}\right)
\end{equation}
as an example, we show how to solve not only one potential from a solved
potential. To show how to solve an eigenproblem by the duality relation,
additionally, we solve the exact solutions of these three potentials in
Appendices \ref{V2m1}, \ref{Vm1m32}, and \ref{V62}.

The two-term general polynomial potential
\begin{equation}
U\left(  r\right)  =\xi r^{2}+\frac{\mu}{r} \label{Urtwoterms}%
\end{equation}
has two dual potentials:
\begin{align}
V\left(  \rho\right)   &  =\frac{\eta}{\rho}+\frac{\lambda}{\rho^{3/2}%
},\label{Vrho2m1}\\
V\left(  \rho\right)   &  =\eta\rho^{2}+\lambda\rho^{6}. \label{Vrho26}%
\end{align}

The potential $U\left(  r\right)  =\xi r^{2}+\frac{\mu}{r}$ can be exactly
solved (see Appendix \ref{V2m1}).

The radial equation of $U\left(  r\right)  $\ is%
\begin{equation}
\frac{d^{2}u\left(  r\right)  }{dr^{2}}+\left[  E-\frac{l\left(  l+1\right)
}{r^{2}}-\xi r^{2}-\frac{\mu}{r}\right]  u\left(  r\right)  =0.
\end{equation}
The radial eigenfunction is
\begin{equation}
u\left(  r\right)  =A_{l}{r}^{l+1}e{^{\sqrt{\xi}{r}^{2}/2}}N\left(
2l+1,0,-{\frac{E}{\sqrt{\xi}}},-{\frac{2i\mu}{\xi^{1/4}}},i\xi^{1/4}r\right)
\label{radwf1}%
\end{equation}
and the eigenvalue of bound states can be expressed by an implicit expression:%
\begin{equation}
K_{2}\left(  2l+1,0,-{\frac{E}{\sqrt{\xi}}},-{\frac{2i\mu}{\xi^{1/4}}}\right)
=0, \label{specsq1}%
\end{equation}
i.e., the eigenvalue is the zero of $K_{2}\left(  2l+1,0,-{\frac{E}{\sqrt{\xi
}}},-{\frac{2i\mu}{\xi^{1/4}}}\right)  $.

\paragraph{Solving $V\left(  \rho\right)  =\frac{\eta}{\rho}+\frac{\lambda
}{\rho^{3/2}}$ from $U\left(  r\right)  =\xi r^{2}+\frac{\mu}{r}$}

The potential (\ref{Vrho2m1}) is one of the two dual potentials of the
potential (\ref{Urtwoterms}) with $a=1$, $b=-2$ and $A=-2$, $B=-5/2$ in the
duality relations (\ref{Aanda6}) and (\ref{Bandb6}). We can solve it from the
solution of the potential $U\left(  r\right)  $.

The radial equation of $V\left(  \rho\right)  =\frac{\eta}{\rho}+\frac
{\lambda}{\rho^{3/2}}$ is
\begin{equation}
\frac{d^{2}v\left(  \rho\right)  }{d\rho^{2}}+\left[  \mathcal{E}-\frac
{\ell\left(  \ell+1\right)  }{\rho^{2}}-\frac{\eta}{\rho}-\frac{\lambda}%
{\rho^{3/2}}\right]  v\left(  \rho\right)  =0.
\end{equation}

The dual relations (\ref{Eta6}), (\ref{xieta6}), (\ref{mulamb6}),
(\ref{ltrans6}), (\ref{rtrans6}), and (\ref{wftrans6}) with $a=1$ and $b=-2$
read%
\begin{align}
E  &  \rightarrow-4\eta,\\
\xi &  \rightarrow-4\mathcal{E},\\
\mu &  \rightarrow4\lambda,\\
l+\frac{1}{2}  &  \rightarrow2\left(  \ell+\frac{1}{2}\right)  ,
\end{align}
and%
\begin{align}
r^{2}  &  \rightarrow\rho\text{\ \ or \ }r\rightarrow\sqrt{\rho},\\
u\left(  r\right)   &  \rightarrow\rho^{-1/4}v\left(  \rho\right)  .
\end{align}
Substituting the above dual relations into the eigenfunction of the potential
(\ref{Urtwoterms}), Eq. (\ref{radwf1}), gives%
\begin{equation}
\rho^{-1/4}v\left(  \rho\right)  =A_{\ell}\left(  \sqrt{\rho}\right)
^{2\left(  \ell+1/2\right)  +1/2}e{^{\frac{\sqrt{-4\mathcal{E}}}{2}\left(
\sqrt{\rho}\right)  ^{2}}}N\left(  4\left(  \ell+\frac{1}{2}\right)
,0,-{\frac{-4\eta}{\sqrt{-4\mathcal{E}}}},-{\frac{2i\left(  4\lambda\right)
}{\left(  -4\mathcal{E}\right)  ^{1/4}}},i\left(  -4\mathcal{E}\right)
^{1/4}\sqrt{\rho}\right)  . \label{nurhotwoterms1}%
\end{equation}
The radial eigenfunction of the potential (\ref{Vrho2m1}) can be obtained by
rearranging Eq. (\ref{nurhotwoterms1}):%
\begin{equation}
v\left(  \rho\right)  =A_{\ell}\rho^{\ell+1}e{^{\sqrt{-\mathcal{E}}\rho}%
}N\left(  4\ell+2,0,{\frac{2\eta}{\sqrt{-\mathcal{E}}}},-{\frac{4\sqrt
{2}i\lambda}{\left(  -\mathcal{E}\right)  ^{1/4}}},i\left(  -\mathcal{E}%
\right)  ^{1/4}\sqrt{2\rho}\right)  . \label{vwf1}%
\end{equation}
Substituting the above replacements into Eq. (\ref{specsq1}) gives an implicit
expression of the eigenvalue of the potential (\ref{Vrho2m1}):%
\begin{equation}
K_{2}\left(  4\left(  \ell+\frac{1}{2}\right)  ,0,-{\frac{-4\eta}%
{\sqrt{-4\mathcal{E}}}},-{\frac{2i\left(  4\lambda\right)  }{\left(
-4\mathcal{E}\right)  ^{1/4}}}\right)  =0,
\end{equation}
i.e.,%
\begin{equation}
K_{2}\left(  4\ell+2,0,{\frac{2\eta}{\sqrt{-\mathcal{E}}}},-{\frac{4\sqrt
{2}i\lambda}{\left(  -\mathcal{E}\right)  ^{1/4}}}\right)  =0. \label{vsp1}%
\end{equation}

In Appendix \ref{Vm1m32}, as a verification, we solve the potential
(\ref{Vrho2m1}) by directly solving the radial equation. The results obtained
by the duality relation agrees with the result given in Appendix \ref{Vm1m32}.

\paragraph{Solving $V\left(  \rho\right)  =\eta\rho^{2}+\lambda\rho^{6}$ from
$U\left(  r\right)  =\xi r^{2}+\frac{\mu}{r}$ \label{V62sol}}

The potential (\ref{Vrho26}) is another dual potential of the potential
(\ref{Urtwoterms}) with $a=-2$, $b=1$ and $A=1$, $B=5$ in\ the duality
relations (\ref{Aanda6}) and (\ref{Bandb6}). We can also solve it from the
solution of the potential (\ref{Urtwoterms}).

The radial equation of $V\left(  \rho\right)  =\eta\rho^{2}+\lambda\rho^{6}$
is
\begin{equation}
\frac{d^{2}v\left(  \rho\right)  }{d\rho^{2}}+\left[  \mathcal{E}-\frac
{\ell\left(  \ell+1\right)  }{\rho^{2}}-\eta\rho^{2}-\lambda\rho^{6}\right]
v\left(  \rho\right)  =0.
\end{equation}

The dual relations (\ref{Eta6}), (\ref{xieta6}), (\ref{mulamb6}),
(\ref{ltrans6}), (\ref{rtrans6}), and (\ref{wftrans6}) with $a=-2$ and $b=1$
read%
\begin{align}
E  &  \rightarrow-\frac{\eta}{4},\\
\mu &  \rightarrow-\frac{\mathcal{E}}{4},\\
\xi &  \rightarrow\frac{\lambda}{4},\\
l+\frac{1}{2}  &  \rightarrow\frac{1}{2}\left(  \ell+\frac{1}{2}\right)  ,
\end{align}
and%
\begin{align}
r^{1/2}  &  \rightarrow\rho\text{\ \ or \ }r\rightarrow\rho^{2},\\
u\left(  r\right)   &  \rightarrow\rho^{1/2}v\left(  \rho\right)  .
\end{align}
Substituting the dual relations into the eigenfunction of the potential
(\ref{Urtwoterms}), Eq. (\ref{radwf1}), gives%
\begin{equation}
\rho^{1/2}v\left(  \rho\right)  =A_{\ell}\left(  \rho^{2}\right)  ^{\left(
1/2\right)  \left(  \ell+1/2\right)  +1/2}e{^{\frac{\sqrt{\frac{\lambda}{4}}%
}{2}\left(  \rho^{2}\right)  ^{2}}}N\left(  2\left[  \frac{1}{2}\left(
\ell+\frac{1}{2}\right)  \right]  ,0,-{\frac{\left(  -\frac{\eta}{4}\right)
}{\sqrt{\frac{\lambda}{4}}}},-{\frac{2i\left(  -\frac{\mathcal{E}}{4}\right)
}{\left(  \frac{\lambda}{4}\right)  ^{1/4}}},i\left(  \frac{\lambda}%
{4}\right)  ^{1/4}r\right)  .
\end{equation}
Then the radial eigenfunction of the potential (\ref{Vrho26}) reads%
\begin{equation}
v\left(  \rho\right)  =A_{\ell}{\rho}^{\ell+1}e{^{\frac{\sqrt{\lambda}}%
{4}{\rho}^{4}}}N\left(  \ell+\frac{1}{2},0,{\frac{\eta}{2\sqrt{\lambda}}%
},{\frac{i\mathcal{E}}{\sqrt{2}\lambda^{1/4}}},\frac{i}{\sqrt{2}}\lambda
^{1/4}\rho^{2}\right)  . \label{vwf11}%
\end{equation}
Substituting the replacements into Eq. (\ref{specsq1}) gives an implicit
expression of the eigenvalue of the potential (\ref{Vrho26}):%
\begin{equation}
K_{2}\left(  2\left[  \frac{1}{2}\left(  \ell+\frac{1}{2}\right)  \right]
,0,-{\frac{-\eta/4}{\sqrt{\lambda/4}}},-{\frac{2i\left(  -\mathcal{E}%
/4\right)  }{\left(  \lambda/4\right)  ^{1/4}}}\right)  =0,
\end{equation}
i.e.,
\begin{equation}
K_{2}\left(  \ell+\frac{1}{2},0,{\frac{\eta}{2\sqrt{\lambda}}},{\frac
{i\mathcal{E}}{\sqrt{2}\lambda^{1/4}}}\right)  =0. \label{vsp11}%
\end{equation}

This result agrees with the solution given in Appendix \ref{V62}, which is
obtained by solving the radial equation directly.

\subsubsection{Solving $V\left(  \rho\right)  =\frac{\eta}{\rho}+\frac
{\lambda}{\sqrt{\rho}}$ and $V\left(  \rho\right)  =\frac{\eta}{\rho^{2/3}%
}+\lambda\rho^{2/3}$ from their duality $U\left(  r\right)  =\xi r^{2}+\mu r$}

In this section, we consider a\ dual set of\ two-term potentials%
\[
\left(  \xi r^{2}+\mu r,\frac{\eta}{\rho}+\frac{\lambda}{\sqrt{\rho}}%
,\frac{\eta}{\rho^{2/3}}+\lambda\rho^{2/3}\right)  ,
\]
in which the solution of $U\left(  r\right)  =\xi r^{2}+\mu r$ is already
known. To show how to solve an eigenproblem by the duality relation,
additionally, we solve the exact solutions of these three potentials in
Appendices \ref{V21}, \ref{Vm1m12}, and \ref{Vm2323}.

The two-term general polynomial potential
\begin{equation}
U\left(  r\right)  =\xi r^{2}+\mu r \label{Ur21}%
\end{equation}
has two dual potentials:
\begin{align}
V\left(  \rho\right)   &  =\frac{\eta}{\rho}+\frac{\lambda}{\sqrt{\rho}%
},\label{Vrhom1m12}\\
V\left(  \rho\right)   &  =\frac{\eta}{\rho^{2/3}}+\lambda\rho^{2/3}.
\label{Vrhom5323}%
\end{align}

The radial equation of the potential (\ref{Ur21}) is%
\begin{equation}
\frac{d^{2}u\left(  r\right)  }{dr^{2}}+\left[  E-\frac{l\left(  l+1\right)
}{r^{2}}-\xi r^{2}-\mu r\right]  u\left(  r\right)  =0.
\end{equation}

The radial eigenfunction is (see Appendix \ref{V21})%
\begin{equation}
u\left(  r\right)  =A_{l}{r}^{l+1}\exp\left(  \frac{\sqrt{\xi}}{2}r^{2}%
+\frac{\mu}{2\sqrt{\xi}}r\right)  N\left(  2l+1,{\frac{i\mu}{{\xi}^{3/4}}%
},-{\frac{E}{\sqrt{{\xi}}}}-{\frac{{\mu}^{2}}{4{\xi}^{3/2}}},0,i\xi
^{1/4}r\right)  . \label{radwf2}%
\end{equation}

The eigenvalue of the bound state can be expressed by the implicit expression%
\begin{equation}
K_{2}\left(  2l+1,{\frac{i\mu}{{\xi}^{3/4}}},-{\frac{E}{\sqrt{{\xi}}}}%
-{\frac{{\mu}^{2}}{4{\xi}^{3/2}}},0\right)  =0, \label{specsq2}%
\end{equation}
i.e., the eigenvalue of $U\left(  r\right)  $ is the zero of $K_{2}\left(
2l+1,{\frac{i\mu}{{\xi}^{3/4}}},-{\frac{E}{\sqrt{{\xi}}}}-{\frac{{\mu}^{2}%
}{4{\xi}^{3/2}}},0\right)  $.

\paragraph{Solving $V\left(  \rho\right)  =\frac{\eta}{\rho}+\frac{\lambda
}{\sqrt{\rho}}$ from $U\left(  r\right)  =\xi r^{2}+\mu r$}

The potential (\ref{Vrhom1m12}) is a dual potential of the potential
(\ref{Ur21}) with $a=1$, $b=0$ and $A=-2$, $B=-3/2$ in the duality relations
(\ref{Aanda6}) and (\ref{Bandb6}).

The radial equation of $V\left(  \rho\right)  =\frac{\eta}{\rho}+\frac
{\lambda}{\sqrt{\rho}}$ is
\begin{equation}
\frac{d^{2}v\left(  \rho\right)  }{d\rho^{2}}+\left[  \mathcal{E}-\frac
{\ell\left(  \ell+1\right)  }{\rho^{2}}-\frac{\eta}{\rho}-\frac{\lambda}%
{\sqrt{\rho}}\right]  v\left(  \rho\right)  =0.
\end{equation}

The dual relations (\ref{Eta6}), (\ref{xieta6}), (\ref{mulamb6}),
(\ref{ltrans6}), (\ref{rtrans6}), and (\ref{wftrans6}) with $a=1$ and $b=0$
read%
\begin{align}
E  &  \rightarrow-4\eta,\\
\xi &  \rightarrow-4\mathcal{E},\\
\mu &  \rightarrow4\lambda,\\
l+\frac{1}{2}  &  \rightarrow2\left(  \ell+\frac{1}{2}\right)  ,
\end{align}
and%
\begin{align}
r^{2}  &  \rightarrow\rho\text{ \ \ or \ }r\rightarrow\sqrt{\rho},\\
u\left(  r\right)   &  \rightarrow\rho^{-1/4}v\left(  \rho\right)  .
\end{align}
Substituting the dual relations into the eigenfunction of the potential
(\ref{Ur21}), Eq. (\ref{radwf1}), gives
\begin{align}
\rho^{-1/4}v\left(  \rho\right)   &  =A_{\ell}\sqrt{\rho}^{2\left(
\ell+1/2\right)  +1/2}\exp\left(  \frac{\sqrt{-4\mathcal{E}}}{2}\sqrt{\rho
}^{2}+\frac{4\lambda}{2\sqrt{-4\mathcal{E}}}\sqrt{\rho}\right) \nonumber\\
&  \times N\left(  4\left(  \ell+\frac{1}{2}\right)  ,{\frac{i\left(
4\lambda\right)  }{{\left(  -4\mathcal{E}\right)  }^{3/4}}},-{\frac{-4\eta
}{\sqrt{-4\mathcal{E}}}}-{\frac{\left(  4\lambda\right)  ^{2}}{4{\left(
-4\mathcal{E}\right)  }^{3/2}}},0,i\left(  -4\mathcal{E}\right)  ^{1/4}%
\sqrt{\rho}\right)  .
\end{align}
Then the eigenfunction of the potential (\ref{Vrhom1m12}) reads%
\begin{equation}
v\left(  \rho\right)  =A_{\ell}\rho^{\ell+1}\exp\left(  \sqrt{-\mathcal{E}%
}\rho+\frac{\lambda}{\sqrt{-\mathcal{E}}}\sqrt{\rho}\right)  N\left(
4\ell+2,{\frac{\sqrt{2}i\lambda}{{\left(  -\mathcal{E}\right)  }^{3/4}}%
},{\frac{2\eta}{\sqrt{-\mathcal{E}}}}-{\frac{\lambda^{2}}{2{\left(
-\mathcal{E}\right)  }^{3/2}}},0,i\left(  -\mathcal{E}\right)  ^{1/4}%
\sqrt{2\rho}\right)  .
\end{equation}
Substituting the dual relations into (\ref{specsq2}) gives an implicit
expression of the eigenvalue of the potential (\ref{Vrhom1m12}):
\begin{equation}
K_{2}\left(  4\left(  \ell+\frac{1}{2}\right)  ,{\frac{i\left(  4\lambda
\right)  }{{\left(  -4\mathcal{E}\right)  }^{3/4}}},-{\frac{-4\eta}%
{\sqrt{-4\mathcal{E}}}}-{\frac{\left(  4\lambda\right)  ^{2}}{4{\left(
-4\mathcal{E}\right)  }^{3/2}}},0\right)  =0,
\end{equation}
i.e.,
\begin{equation}
K_{2}\left(  4\ell+2,{\frac{\sqrt{2}i\lambda}{{\left(  -\mathcal{E}\right)
}^{3/4}}},{\frac{2\eta}{\sqrt{-\mathcal{E}}}}-{\frac{\lambda^{2}}{2{\left(
-\mathcal{E}\right)  }^{3/2}}},0\right)  =0. \label{vsp2}%
\end{equation}

This result agrees with the result obtained by directly solving the radial
equation, which is given in Appendix \ref{Vm1m12} for verification.

\paragraph{Solving $V\left(  \rho\right)  =\frac{\eta}{\rho^{2/3}}+\lambda
\rho^{2/3}$ from $U\left(  r\right)  =\xi r^{2}+\mu r$}

The potential (\ref{Vrhom5323}) is another dual potential of the potential
(\ref{Ur21}) with $a=0$, $b=1$ and $A=-5/3$, $B=-1/3$ in\ the duality
relations (\ref{Aanda6}) and (\ref{Bandb6}).

The radial equation of $V\left(  \rho\right)  =\frac{\eta}{\rho^{2/3}}%
+\lambda\rho^{2/3}$ is
\begin{equation}
\frac{d^{2}v\left(  \rho\right)  }{d\rho^{2}}+\left[  \mathcal{E}-\frac
{\ell\left(  \ell+1\right)  }{\rho^{2}}-\frac{\eta}{\rho^{2/3}}-\lambda
\rho^{2/3}\right]  v\left(  \rho\right)  =0.
\end{equation}

The dual relations (\ref{Eta6}), (\ref{xieta6}), (\ref{mulamb6}),
(\ref{ltrans6}), (\ref{rtrans6}), and (\ref{wftrans6}) with $a=0$ and $b=1$
read%
\begin{align}
E  &  \rightarrow-\frac{9}{4}\eta,\\
\mu &  \rightarrow-\frac{9}{4}\mathcal{E},\\
\xi &  \rightarrow\frac{9}{4}\lambda,\\
l+\frac{1}{2}  &  \rightarrow\frac{3}{2}\left(  \ell+\frac{1}{2}\right)  ,
\end{align}
and%
\begin{align}
r^{3/2}  &  \rightarrow\rho\text{\ \ or \ }r\rightarrow\rho^{2/3},\\
u\left(  r\right)   &  \rightarrow\rho^{-1/6}v\left(  \rho\right)  .
\end{align}
Substituting the dual relations into the eigenfunction of the potential
(\ref{Ur21}), Eq. (\ref{radwf2}), gives%
\begin{align}
\rho^{-1/6}v\left(  \rho\right)   &  =A_{\ell}\left(  \rho^{2/3}\right)
^{\left(  3/2\right)  \left(  \ell+1/2\right)  +1/2}\exp\left(  \frac{\left(
\frac{9}{4}\lambda\right)  ^{1/2}}{2}\left(  \rho^{\frac{2}{3}}\right)
^{2}+\frac{-\frac{9}{4}\mathcal{E}}{2\left(  \frac{9}{4}\lambda\right)
^{1/2}}\rho^{\frac{2}{3}}\right) \nonumber\\
&  \times N\left(  2\left[  \frac{3}{2}\left(  \ell+\frac{1}{2}\right)
\right]  ,{\frac{i\left(  -\frac{9}{4}\mathcal{E}\right)  }{\left(  \frac
{9}{4}\lambda\right)  ^{3/4}}},-{\frac{-\frac{9}{4}\eta}{\left(  \frac{9}%
{4}\lambda\right)  ^{1/2}}}-{\frac{\left(  -\frac{9}{4}\mathcal{E}\right)
^{2}}{4\left(  \frac{9}{4}\lambda\right)  ^{3/2}}},0,i\left(  \frac{9}%
{4}\lambda\right)  ^{1/4}\rho^{\frac{2}{3}}\right)  .
\end{align}
Then the radial eigenfunction of the potential (\ref{Vrhom5323}) reads
\begin{equation}
v\left(  \rho\right)  =A_{\ell}\rho^{\ell+1}\exp\left(  \frac{\left(  \frac
{9}{4}\lambda\right)  ^{1/2}}{2}\rho^{4/3}+\frac{-\frac{9}{4}\mathcal{E}%
}{2\left(  \frac{9}{4}\lambda\right)  ^{1/2}}\rho^{2/3}\right)  N\left(
3\ell+\frac{3}{2},-{\frac{i\sqrt{6}\mathcal{E}}{2\lambda^{3/4}}},{\frac{3\eta
}{2\sqrt{\lambda}}}-{\frac{3\mathcal{E}^{2}}{8\lambda^{3/2}}},0,i\frac
{\sqrt{6}}{2}\lambda^{1/4}\rho^{2/3}\right)  . \label{vwf22}%
\end{equation}
Substituting the dual relations into Eq. (\ref{specsq2}) gives an implicit
expression of the eigenvalue of the potential (\ref{Vrhom5323}):%
\begin{equation}
K_{2}\left(  2\left[  \frac{3}{2}\left(  \ell+\frac{1}{2}\right)  \right]
,{\frac{i\left(  -\frac{9}{4}\mathcal{E}\right)  }{\left(  \frac{9}{4}%
\lambda\right)  ^{3/4}}},-{\frac{-\frac{9}{4}\eta}{\left(  \frac{9}{4}%
\lambda\right)  ^{1/2}}}-{\frac{\left(  -\frac{9}{4}\mathcal{E}\right)  ^{2}%
}{4\left(  \frac{9}{4}\lambda\right)  ^{3/2}}},0\right)  =0,
\end{equation}
i.e.,
\begin{equation}
K_{2}\left(  3\ell+\frac{3}{2},-{\frac{i\sqrt{6}\mathcal{E}}{2\lambda^{3/4}}%
},{\frac{3\eta}{2\sqrt{\lambda}}}-{\frac{3\mathcal{E}^{2}}{8\lambda^{3/2}}%
},0\right)  =0. \label{vsp22}%
\end{equation}

This result agrees with the result obtained by directly solving the radial
equation, which is given in Appendix \ref{Vm2323} for verification.

\subsubsection{Solving $V\left(  \rho\right)  =\eta\rho^{2/3}+\frac{\lambda
}{\rho^{4/3}}$ and $V\left(  \rho\right)  =\eta\rho^{6}+\lambda\rho^{4}$ from
their duality $U\left(  r\right)  =\frac{\xi}{\sqrt{r}}+\frac{\mu}{r^{3/2}}$}

In this section, we consider the\ dual set of\ two-term potentials
\begin{equation}
\left(  \frac{\xi}{\sqrt{r}}+\frac{\mu}{r^{3/2}},\eta\rho^{2/3}+\frac{\lambda
}{\rho^{4/3}},\eta\rho^{6}+\lambda\rho^{4}\right)  ,
\end{equation}
in which the solution of $U\left(  r\right)  =\frac{\xi}{\sqrt{r}}+\frac{\mu
}{r^{3/2}}$ is already known. To show how to solve an eigenproblem by the
duality relation, additionally, we solve exact solutions of these three
potentials in Appendices \ref{Vm12m32}, \ref{V23m43}, and \ref{V64}.

The two-term general polynomial potential
\begin{equation}
U\left(  r\right)  =\frac{\xi}{\sqrt{r}}+\frac{\mu}{r^{3/2}} \label{Urm12m32}%
\end{equation}
has two dual potentials:
\begin{align}
V\left(  \rho\right)   &  =\eta\rho^{2/3}+\frac{\lambda}{\rho^{4/3}%
},\label{Vrho23m43}\\
V\left(  \rho\right)   &  =\eta\rho^{6}+\lambda\rho^{4}. \label{Vrho64}%
\end{align}

The radial equation of the potential (\ref{Urm12m32}) is%
\begin{equation}
\frac{d^{2}u\left(  r\right)  }{dr^{2}}+\left[  E-\frac{l\left(  l+1\right)
}{r^{2}}-\frac{\xi}{\sqrt{r}}-\frac{\mu}{r^{3/2}}\right]  u\left(  r\right)
=0
\end{equation}
The radial eigenfunction is (see Appendix \ref{Vm12m32})%
\begin{align}
u\left(  r\right)   &  =A_{l}{r}^{l+1}\exp\left(  {\sqrt{-E}r+}\frac{\xi
}{\sqrt{-E}}\sqrt{r}\right) \nonumber\\
&  \times N\left(  4l+2,{\frac{i\xi\sqrt{2}}{\left(  -E\right)  ^{3/4}}%
},-{\frac{{\xi}^{2}}{2\left(  -E\right)  ^{3/2}}},{\frac{-4\sqrt{2}i\mu
}{\left(  -E\right)  ^{1/4}}},i\sqrt{2}\left(  -E\right)  ^{1/4}\sqrt
{r}\right)  . \label{radwf3}%
\end{align}
The eigenvalue of the bound state can be expressed by the implicit expression
\cite{li2016scattering}
\begin{equation}
K_{2}\left(  4l+2,{\frac{i\xi\sqrt{2}}{\left(  -E\right)  ^{3/4}}}%
,-{\frac{{\xi}^{2}}{2\left(  -E\right)  ^{3/2}}},{\frac{-4\sqrt{2}i\mu
}{\left(  -E\right)  ^{1/4}}}\right)  =0, \label{specsq3}%
\end{equation}
i.e., the eigenvalue is the zero of $K_{2}\left(  4l+2,{\frac{i\xi\sqrt{2}%
}{\left(  -E\right)  ^{3/4}}},-{\frac{{\xi}^{2}}{2\left(  -E\right)  ^{3/2}}%
},{\frac{-4\sqrt{2}i\mu}{\left(  -E\right)  ^{1/4}}}\right)  $.

\paragraph{Solving $V\left(  \rho\right)  =\eta\rho^{2/3}+\frac{\lambda}%
{\rho^{4/3}}$ form $U\left(  r\right)  =\frac{\xi}{\sqrt{r}}+\frac{\mu
}{r^{3/2}}$}

The potential (\ref{Vrho23m43}) is a dual potential of the potential
(\ref{Urm12m32}) with $a=-3/2$, $b=-5/2$ and $A=-1/3$, $B=-7/3$ in the duality
relations (\ref{Aanda6}) and (\ref{Bandb6}).

The radial equation of the potential (\ref{Vrho23m43}) is
\begin{equation}
\frac{d^{2}v\left(  \rho\right)  }{d\rho^{2}}+\left[  \mathcal{E}-\frac
{\ell\left(  \ell+1\right)  }{\rho^{2}}-\eta\rho^{2/3}-\frac{\lambda}%
{\rho^{4/3}}\right]  v\left(  \rho\right)  =0.
\end{equation}

The dual relations (\ref{Eta6}), (\ref{xieta6}), (\ref{mulamb6}),
(\ref{ltrans6}), (\ref{rtrans6}), and (\ref{wftrans6}) with $a=-3/2$ and
$b=-5/2$ read%
\begin{align}
E  &  \rightarrow-\frac{9}{16}\eta,\\
\xi &  \rightarrow-\frac{9}{16}\mathcal{E},\\
\mu &  \rightarrow\frac{9}{16}\lambda,\\
l+\frac{1}{2}  &  \rightarrow\frac{3}{4}\left(  \ell+\frac{1}{2}\right)  ,
\end{align}
and%
\begin{align}
r^{3/4}  &  \rightarrow\rho\text{ \ \ or \ \ }r\rightarrow\rho^{4/3},\\
u\left(  r\right)   &  \rightarrow\rho^{1/6}v\left(  \rho\right)  .
\end{align}
Substituting the dual relations into the eigenfunction of the potential
(\ref{Urm12m32}), Eq. (\ref{radwf3}), gives the eigenfunction of the potential
(\ref{Vrho23m43}):%
\begin{align}
&  \rho^{1/6}v\left(  \rho\right)  =A_{\ell}\left(  \rho^{4/3}\right)
^{\left(  3/4\right)  \left(  \ell+1/2\right)  +1/2}\exp\left(  {\sqrt
{-\left(  -\frac{9}{16}\eta\right)  }\rho^{4/3}+}\frac{-\frac{9}%
{16}\mathcal{E}}{\sqrt{-\left(  -\frac{9}{16}\eta\right)  }}\sqrt{\rho^{4/3}%
}\right) \nonumber\\
&  \times N\left(  4\left[  \frac{3}{4}\left(  \ell+\frac{1}{2}\right)
\right]  ,{\frac{i\left(  -\frac{9}{16}\mathcal{E}\right)  \sqrt{2}}{\left[
-\left(  -\frac{9}{16}\eta\right)  \right]  ^{3/4}}},-{\frac{\left(  -\frac
{9}{16}\mathcal{E}\right)  ^{2}}{2\left[  -\left(  -\frac{9}{16}\eta\right)
\right]  ^{3/2}}},{\frac{-4\sqrt{2}i\left(  \frac{9}{16}\lambda\right)
}{\left[  -\left(  -\frac{9}{16}\eta\right)  \right]  ^{1/4}}},i\sqrt
{2}\left[  -\left(  -\frac{9}{16}\eta\right)  \right]  ^{1/4}\sqrt{\rho^{4/3}%
}\right)  ,
\end{align}
Then
\begin{equation}
v\left(  \rho\right)  =A_{\ell}\rho^{\ell+1}\exp\left(  \frac{3\sqrt{\eta}}%
{4}{\rho^{4/3}-}\frac{3\mathcal{E}}{4\sqrt{\eta}}\rho^{2/3}\right)  N\left(
3\ell+\frac{3}{2},-{\frac{i\sqrt{6}\mathcal{E}}{2\eta^{3/4}}},-{\frac
{3\mathcal{E}^{2}}{8\eta^{3/2}}},-{\frac{i3\sqrt{6}\lambda}{2\eta^{1/4}}%
},i\frac{\sqrt{6}}{2}\eta^{1/4}\rho^{2/3}\right)  . \label{vwf3}%
\end{equation}
Substituting the dual relations into Eq. (\ref{specsq3}) gives an implicit
expression of the eigenvalue of the potential (\ref{Vrho23m43}):%
\begin{equation}
K_{2}\left(  4\left[  \frac{3}{4}\left(  \ell+\frac{1}{2}\right)  \right]
,{\frac{i\left(  -\frac{9}{16}\mathcal{E}\right)  \sqrt{2}}{\left[  -\left(
-\frac{9}{16}\eta\right)  \right]  ^{3/4}}},-{\frac{\left(  -\frac{9}%
{16}\mathcal{E}\right)  ^{2}}{2\left[  -\left(  -\frac{9}{16}\eta\right)
\right]  ^{3/2}}},{\frac{-4\sqrt{2}i\left(  \frac{9}{16}\lambda\right)
}{\left[  -\left(  -\frac{9}{16}\eta\right)  \right]  ^{1/4}}}\right)  =0,
\end{equation}
i.e.,
\begin{equation}
K_{2}\left(  3\ell+\frac{3}{2},-{\frac{i\sqrt{6}\mathcal{E}}{2\eta^{3/4}}%
},-{\frac{3\mathcal{E}^{2}}{8\eta^{3/2}}},-{\frac{i3\sqrt{6}\lambda}%
{2\eta^{1/4}}}\right)  =0. \label{vsp3}%
\end{equation}
This result agrees with the result obtained by directly solving the radial
equation, which is given in Appendix \ref{V23m43} for verification.

\paragraph{Solving $V\left(  \rho\right)  =\eta\rho^{6}+\lambda\rho^{4}$ form
$U\left(  r\right)  =\frac{\xi}{\sqrt{r}}+\frac{\mu}{r^{3/2}}$}

The potential (\ref{Vrho64})
\[
V\left(  \rho\right)  =\eta\rho^{6}+\lambda\rho^{4}%
\]
is a dual potential of the potential (\ref{Urm12m32}) with $a=-5/2$, $b=-3/2$
and $A=5$, $B=3$ in the duality relations (\ref{Aanda6}) and (\ref{Bandb6}).

The radial equation of the potential (\ref{Vrho64}) is
\begin{equation}
\frac{d^{2}v\left(  \rho\right)  }{d\rho^{2}}+\left[  \mathcal{E}-\frac
{\ell\left(  \ell+1\right)  }{\rho^{2}}-\eta\rho^{6}-\lambda\rho^{4}\right]
v\left(  \rho\right)  =0.
\end{equation}

The dual relations (\ref{Eta6}), (\ref{xieta6}), (\ref{mulamb6}),
(\ref{ltrans6}), (\ref{rtrans6}), and (\ref{wftrans6}) with $a=-5/2$ and
$b=-3/2$ read%
\begin{align}
E  &  \rightarrow-\frac{1}{16}\eta,\\
\mu &  \rightarrow-\frac{1}{16}\mathcal{E},\\
\xi &  \rightarrow\frac{1}{16}\lambda,\\
l+\frac{1}{2}  &  \rightarrow\frac{1}{4}\left(  \ell+\frac{1}{2}\right)  ,
\end{align}
and%
\begin{align}
r^{1/4}  &  \rightarrow\rho\text{ \ \ or \ \ }r\rightarrow\rho^{4},\\
u\left(  r\right)   &  \rightarrow\rho^{3/2}v\left(  \rho\right)  .
\end{align}
Substituting the dual relations into the eigenfunction of the potential
(\ref{Urm12m32}), Eq. (\ref{radwf3}), gives%
\begin{align}
&  \rho^{3/2}v\left(  \rho\right)  =A_{l}\left(  \rho^{4}\right)  ^{\left(
1/4\right)  \left(  \ell+1/2\right)  +1/2}\exp\left(  \left[  -\left(
-\frac{1}{16}\eta\right)  \right]  ^{1/2}{\rho^{4}+}\frac{\frac{1}{16}\lambda
}{\left[  -\left(  -\frac{1}{16}\eta\right)  \right]  ^{1/2}}\sqrt{\rho^{4}%
}\right) \nonumber\\
&  \times N\left(  4\left[  \frac{1}{4}\left(  \ell+\frac{1}{2}\right)
\right]  ,{\frac{i\left(  \frac{1}{16}\lambda\right)  \sqrt{2}}{\left[
-\left(  -\frac{1}{16}\eta\right)  \right]  ^{3/4}}},-{\frac{\left(  \frac
{1}{16}\lambda\right)  ^{2}}{2\left[  -\left(  -\frac{1}{16}\eta\right)
\right]  ^{3/2}}},{\frac{-4\sqrt{2}i\left(  -\frac{1}{16}\mathcal{E}\right)
}{\left[  -\left(  -\frac{1}{16}\eta\right)  \right]  ^{1/4}}},i\sqrt
{2}\left[  -\left(  -\frac{1}{16}\eta\right)  \right]  ^{1/4}\sqrt{\rho^{4}%
}\right)  .
\end{align}
Then we arrive at the eigenfunction of the potential (\ref{Vrho64}):%
\begin{equation}
v\left(  \rho\right)  =A_{l}{\rho}^{\ell+1}\exp\left(  \frac{\sqrt{\eta}}%
{4}{\rho^{4}+}\frac{\lambda}{4\sqrt{\eta}}\rho^{2}\right)  N\left(  \ell
+\frac{1}{2},{\frac{i\lambda\sqrt{2}}{2\eta^{3/4}}},-{\frac{\lambda^{2}}%
{8\eta^{3/2}}},{\frac{i\sqrt{2}\mathcal{E}}{2\eta^{1/4}}},i\frac{\sqrt{2}}%
{2}\eta^{1/4}\rho^{2}\right)  . \label{vwf33}%
\end{equation}
Substituting the dual relations into (\ref{specsq3}) gives%
\begin{equation}
K_{2}\left(  4\left[  \frac{1}{4}\left(  \ell+\frac{1}{2}\right)  \right]
,{\frac{i\left(  \frac{1}{16}\lambda\right)  \sqrt{2}}{\left[  -\left(
-\frac{1}{16}\eta\right)  \right]  ^{3/4}}},-{\frac{\left(  \frac{1}%
{16}\lambda\right)  ^{2}}{2\left[  -\left(  -\frac{1}{16}\eta\right)  \right]
^{3/2}}},{\frac{-4\sqrt{2}i\left(  -\frac{1}{16}\mathcal{E}\right)  }{\left[
-\left(  -\frac{1}{16}\eta\right)  \right]  ^{1/4}}}\right)  =0.
\end{equation}
Then we arrive at an implicit expression of the eigenvalue of the potential
(\ref{Vrho64}):%
\begin{equation}
K_{2}\left(  \ell+\frac{1}{2},{\frac{i\lambda\sqrt{2}}{2\eta^{3/4}}}%
,-{\frac{\lambda^{2}}{8\eta^{3/2}}},{\frac{i\sqrt{2}\mathcal{E}}{2\eta^{1/4}}%
}\right)  =0. \label{vsp33}%
\end{equation}

This result agrees with the result obtained by directly solving the radial
equation, which is given in Appendix \ref{V64} for verification.

\subsubsection{Solving $V\left(  \rho\right)  =\frac{\eta}{\rho^{3/2}}%
+\frac{\mu}{\rho}$ from $U\left(  r\right)  =\xi r^{6}+\lambda r^{2}$
\label{Solvingeln}}

This is an example of the Newton duality discussed in section
\ref{polynomial SC}. In this case, each term of one potential is Newtonianly
dual to the corresponding term of its Newton dual potential. That is,
($\frac{\eta}{\rho^{3/2}}+\frac{\mu}{\rho}$, $\xi r^{6}+\lambda r^{2}$),
($\frac{\eta}{\rho^{3/2}}$, $\xi r^{6}$), and ($\frac{\mu}{\rho}$,$\lambda
r^{2}$) are all dual sets. In this case, though the potential is a two-term
potential, it has only one dual potential, i.e., the dual set contains only
two two-term potentials: ($\frac{\eta}{\rho^{3/2}}+\frac{\mu}{\rho}$, $\xi
r^{6}+\lambda r^{2}$).

The $r^{6}$-potential is the Newton duality of $r^{-2/3}$-potential and the
$r^{2}$-potential is the Newton duality of $r^{-1}$-potential, while
the\ linear combination of $r^{6}$-potential and $r^{2}$-potential is also the
Newton duality of the linear combination of $r^{-2/3}$-potential and $r^{-1}%
$-potential, i.e., the potential
\begin{equation}
U\left(  r\right)  =\xi r^{6}+\lambda r^{2} \label{Ur62}%
\end{equation}
is the Newton duality of%
\begin{equation}
V\left(  \rho\right)  =\frac{\eta}{\rho^{3/2}}+\frac{\mu}{\rho}.
\label{Vrho32m1}%
\end{equation}

\textit{The potential }$U\left(  r\right)  =\xi r^{6}+\lambda r^{2}$\textit{.}
The potential (\ref{Ur62}) can be exactly solved (see Appendix \ref{V62}).

The radial equation of $U\left(  r\right)  $\ is%
\begin{equation}
\frac{d^{2}u\left(  r\right)  }{dr^{2}}+\left[  E-\frac{l\left(  l+1\right)
}{r^{2}}-\xi r^{6}-\lambda r^{2}\right]  u\left(  r\right)  =0. \label{req1}%
\end{equation}
The radial eigenfunction is (see Appendix \ref{V62})%
\begin{equation}
u\left(  r\right)  =e^{\sqrt{\xi}{r}^{4}/4}N\left(  l+\frac{1}{2}%
,0,{\frac{\lambda}{2\sqrt{\xi}}},i{\frac{\sqrt{2}E}{2\xi^{1/4}}},i\frac
{\sqrt{2}}{2}\xi^{1/4}{r}^{2}\right)  r^{l+1}. \label{wf1}%
\end{equation}
The eigenvalue of the bound state can be expressed by the implicit expression%
\begin{equation}
K_{2}\left(  l+\frac{1}{2},0,{\frac{\lambda}{2\sqrt{\xi}}},i{\frac{\sqrt{2}%
E}{2\xi^{1/4}}}\right)  =0, \label{spctrans}%
\end{equation}
i.e., the eigenvalue is the zero of $K_{2}\left(  l+\frac{1}{2},0,{\frac
{\lambda}{2\sqrt{\xi}}},i{\frac{\sqrt{2}E}{2\xi^{1/4}}}\right)  $.

\textit{The potential }$V\left(  \rho\right)  =\frac{\eta}{\rho^{3/2}}%
+\frac{\mu}{\rho}$\textit{. }The radial equation of the potential
(\ref{Vrho32m1}) is
\begin{equation}
\frac{d^{2}v\left(  \rho\right)  }{d\rho^{2}}+\left[  \mathcal{E}-\frac
{\ell\left(  \ell+1\right)  }{\rho^{2}}-\frac{\eta}{\rho^{3/2}}-\frac{\mu
}{\rho}\right]  v\left(  \rho\right)  =0.
\end{equation}

The dual relations (\ref{Eta6}), (\ref{xieta6}), (\ref{mulamb6}),
(\ref{ltrans6}), (\ref{rtrans6}), and (\ref{wftrans6}) with $a=$ $5$ and $b=1$
read%
\begin{align}
E  &  \rightarrow-16\eta,\\
\xi &  \rightarrow-16\mathcal{E},\\
\lambda &  \rightarrow16\mu,\\
l+\frac{1}{2}  &  \rightarrow4\left(  \ell+\frac{1}{2}\right)  ,
\end{align}
and%
\begin{align}
r^{4}  &  \rightarrow\rho\text{ \ \ or \ \ }r\rightarrow\rho^{1/4}\\
u\left(  r\right)   &  \rightarrow\rho^{-3/8}v\left(  \rho\right)  .
\end{align}
Substituting the dual relations into the eigenfunction of the potential
(\ref{Ur62}), Eq. (\ref{wf1}), gives%
\begin{align}
\rho^{-3/8}v\left(  \rho\right)   &  =e^{\frac{\sqrt{-16\mathcal{E}}}%
{4}\left(  \rho^{1/4}\right)  ^{4}}\left(  \rho^{1/4}\right)  ^{4\left(
\ell+1/2\right)  +1/2}\nonumber\\
&  \times N\left(  4\left(  \ell+\frac{1}{2}\right)  ,0,{\frac{16\mu}%
{2\sqrt{-16\mathcal{E}}}},i{\frac{\sqrt{2}\left(  -16\eta\right)  }{2\left(
-16\mathcal{E}\right)  ^{1/4}}},i\frac{\sqrt{2}}{2}\left(  -16\mathcal{E}%
\right)  ^{1/4}\left(  \rho^{1/4}\right)  ^{2}\right)  .
\end{align}
Then the eigenfunction
\begin{equation}
v\left(  \rho\right)  =e^{\sqrt{-\mathcal{E}}\rho}N\left(  4\ell
+2,0,{\frac{2\mu}{\sqrt{-\mathcal{E}}}},-i{\frac{4\sqrt{2}\eta}{\left(
-\mathcal{E}\right)  ^{1/4}}},i\sqrt{2}\left(  -\mathcal{E}\right)
^{1/4}\sqrt{\rho}\right)  \rho{r}^{\ell+1}.
\end{equation}
Substituting the dual relations into the eigenvalue of the potential
(\ref{Ur62}), Eq. (\ref{spctrans}), gives%
\begin{equation}
K_{2}\left(  4\ell+2,0,{\frac{2\mu}{\sqrt{-\mathcal{E}}}},-i{\frac{4\sqrt
{2}\eta}{\left(  -\mathcal{E}\right)  ^{1/4}}}\right)  =0.
\end{equation}
This result agrees with the result obtained by directly solving the radial
equation, which is given in Appendix \ref{Vm1m32} for verification.

\subsection{Three-term general polynomial potentials \label{Solving3T}}

As mentioned above, a three-term general polynomial potentials has three dual
potentials. Once the solution of a three-term general polynomial potentials is
solved, the solutions of its three dual potentials can be immediately obtained
by the duality relation given in Section \ref{polynomialnterms}.

In this section, we consider a dual set of three-term potentials
\[
\left(  \xi r^{2}+\frac{\mu}{r}+\kappa r,\frac{\eta}{r}+\frac{\nu}{r^{3/2}%
}+\frac{\lambda}{r^{1/2}},\eta r^{2}+\nu r^{6}+\lambda r^{4},\frac{\eta
}{r^{2/3}}+\nu r^{2/3}+\frac{\lambda}{r^{4/3}}\right)  .
\]

The potential
\begin{equation}
U\left(  r\right)  =\xi r^{2}+\frac{\mu}{r}+\kappa r \label{Ur2m11}%
\end{equation}
has three dual potentials:
\begin{align}
V\left(  \rho\right)   &  =\frac{\eta}{\rho}+\frac{\nu}{\rho^{3/2}}%
+\frac{\lambda}{\rho^{1/2}},\label{Vrhom1m32m12}\\
V\left(  \rho\right)   &  =\eta\rho^{2}+\nu\rho^{6}+\lambda\rho^{4}%
,\label{Vrho264}\\
V\left(  \rho\right)   &  =\frac{\eta}{\rho^{2/3}}+\nu\rho^{2/3}+\frac
{\lambda}{\rho^{4/3}}. \label{Vrhom232343}%
\end{align}

The radial equation of the potential (\ref{Ur2m11}) is%
\begin{equation}
\frac{d^{2}u\left(  r\right)  }{dr^{2}}+\left[  E-\frac{l\left(  l+1\right)
}{r^{2}}-\xi r^{2}-\frac{\mu}{r}-\kappa r\right]  u\left(  r\right)  =0.
\end{equation}
The radial eigenfunction is (see Appendix \ref{V2m11})%
\begin{equation}
u\left(  r\right)  =A_{l}r^{l+1}\exp\left(  \frac{\sqrt{\xi}}{2}{r}^{2}%
+\frac{\kappa}{2\sqrt{\xi}}r\right)  N\left(  2l+1,{\frac{i\kappa}{{\xi}%
^{3/4}}},-{\frac{\kappa^{2}}{4{\xi}^{3/2}}}-{\frac{E}{{\xi}^{1/2}}}%
,-i{\frac{2\mu}{\xi^{1/4}}},i\xi^{1/4}r\right)  . \label{radwf61}%
\end{equation}
The eigenvalue of the bound states can be expressed by the implicit expression
\cite{li2016scattering}%
\begin{equation}
K_{2}\left(  2l+1,{\frac{i\kappa}{{\xi}^{3/4}}},-{\frac{\kappa^{2}}{4{\xi
}^{3/2}}}-{\frac{E}{{\xi}^{1/2}}},-i{\frac{2\mu}{\xi^{1/4}}}\right)  =0.
\label{specsq61}%
\end{equation}

\subsubsection{Solving $V\left(  \rho\right)  =\frac{\eta}{\rho}+\frac{\nu
}{\rho^{3/2}}+\frac{\lambda}{\rho^{1/2}}$ from $U\left(  r\right)  =\xi
r^{2}+\frac{\mu}{r}+\kappa r$}

The potential (\ref{Vrhom1m32m12}) is a dual potential of the potential
(\ref{Ur2m11}) with $a=1$, $b_{1}=-2$, $b_{2}=0$ and $A=-2$, $B_{1}%
=-5/2$,\ $B_{2}=-3/2$ in the duality relations (\ref{Aanda6}) and
(\ref{Bandb6}).

The radial equation is
\begin{equation}
\frac{d^{2}v\left(  \rho\right)  }{d\rho^{2}}+\left[  \mathcal{E}-\frac
{\ell\left(  \ell+1\right)  }{\rho^{2}}-\frac{\eta}{\rho}-\frac{\nu}%
{\rho^{3/2}}-\frac{\lambda}{\rho^{1/2}}\right]  v\left(  \rho\right)  =0.
\end{equation}

The dual relations (\ref{Eta6}), (\ref{xieta6}), (\ref{mulamb6}),
(\ref{ltrans6}), (\ref{rtrans6}), and (\ref{wftrans6}) with $a=1$, $b_{1}=-2$,
and $b_{2}=0$ read%
\begin{align}
E  &  \rightarrow-4\eta,\\
\xi &  \rightarrow-4\mathcal{E},\\
\mu &  \rightarrow4\nu,\\
\kappa &  \rightarrow4\lambda,\\
l+\frac{1}{2}  &  \rightarrow2\left(  \ell+\frac{1}{2}\right)  ,
\end{align}
and%
\begin{align}
r^{2}  &  \rightarrow\rho\text{ \ \ or \ \ }r\rightarrow\sqrt{\rho},\\
u\left(  r\right)   &  \rightarrow\rho^{-1/4}v\left(  \rho\right)  .
\end{align}
Substituting the dual relations into the eigenfunction of the potential
(\ref{Ur2m11}), Eq. (\ref{radwf61}), gives%
\begin{align}
\rho^{-1/4}v\left(  \rho\right)   &  =A_{\ell}\left(  \sqrt{\rho}\right)
^{2\left(  \ell+1/2\right)  +1/2}\exp\left(  \frac{\sqrt{-4\mathcal{E}}}%
{2}\sqrt{\rho}^{2}+\frac{4\lambda}{2\sqrt{-4\mathcal{E}}}\sqrt{\rho}\right)
\nonumber\\
&  \times N\left(  2\left[  2\left(  \ell+\frac{1}{2}\right)  \right]
,{\frac{i\left(  4\lambda\right)  }{\left(  -4\mathcal{E}\right)  ^{3/4}}%
},-{\frac{\left(  4\lambda\right)  ^{2}}{4\left(  -4\mathcal{E}\right)
^{3/2}}}-{\frac{-4\eta}{\left(  -4\mathcal{E}\right)  ^{1/2}}},-i{\frac
{2\left(  4\nu\right)  }{\left(  -4\mathcal{E}\right)  ^{1/4}}},i\left(
-4\mathcal{E}\right)  ^{1/4}\sqrt{\rho}\right)  .
\end{align}
Then the eigenfunction%
\begin{align}
v\left(  \rho\right)   &  =A_{\ell}\rho^{\ell+1}\exp\left(  \sqrt
{-\mathcal{E}}\rho+\frac{\lambda}{\sqrt{-\mathcal{E}}}\sqrt{\rho}\right)
\nonumber\\
&  \times N\left(  4\ell+2,{\frac{i\sqrt{2}\lambda}{\left(  -\mathcal{E}%
\right)  ^{3/4}}},-{\frac{\lambda^{2}}{2\left(  -\mathcal{E}\right)  ^{3/2}}%
}+{\frac{2\eta}{\left(  -\mathcal{E}\right)  ^{1/2}}},-i{\frac{4\sqrt{2}\nu
}{\left(  -\mathcal{E}\right)  ^{1/4}}},i\left(  -\mathcal{E}\right)
^{1/4}\sqrt{2\rho}\right)  . \label{vwf61}%
\end{align}
Substituting the dual relations into (\ref{specsq61}) gives an implicit
expression of the eigenvalue of the potential (\ref{Vrhom1m32m12}):%
\begin{equation}
K_{2}\left(  2\left[  2\left(  \ell+\frac{1}{2}\right)  \right]
,{\frac{i\left(  4\lambda\right)  }{\left(  -4\mathcal{E}\right)  ^{3/4}}%
},-{\frac{\left(  4\lambda\right)  ^{2}}{4\left(  -4\mathcal{E}\right)
^{3/2}}}-{\frac{-4\eta}{\left(  -4\mathcal{E}\right)  ^{1/2}}},-i{\frac
{2\left(  4\nu\right)  }{\left(  -4\mathcal{E}\right)  ^{1/4}}}\right)  =0,
\end{equation}
i.e.,
\begin{equation}
K_{2}\left(  4\ell+2,{\frac{i\sqrt{2}\lambda}{\left(  -\mathcal{E}\right)
^{3/4}}},-{\frac{\lambda^{2}}{2\left(  -\mathcal{E}\right)  ^{3/2}}}%
+{\frac{2\eta}{\left(  -\mathcal{E}\right)  ^{1/2}}},-i{\frac{4\sqrt{2}\nu
}{\left(  -\mathcal{E}\right)  ^{1/4}}}\right)  =0. \label{vsp61}%
\end{equation}
This result agrees with the result obtained by directly solving the radial
equation, which is given in Appendix \ref{Vm1m32m12} for verification.

\subsubsection{Solving $V\left(  \rho\right)  =\eta\rho^{2}+\nu\rho
^{6}+\lambda\rho^{4}$ from $U\left(  r\right)  =\xi r^{2}+\frac{\mu}{r}+\kappa
r$}

The potential (\ref{Vrho264}) is a dual potential of the potential
(\ref{Ur2m11}) with $a=-2$, $b_{1}=1$, and $b_{2}=0$ and $A=1$, $B_{1}=5$,
$B_{2}=3$ in the duality relations (\ref{Aanda6}) and (\ref{Bandb6}).

The radial equation of the potential (\ref{Vrho264}) is
\begin{equation}
\frac{d^{2}v\left(  \rho\right)  }{d\rho^{2}}+\left[  \mathcal{E}-\frac
{\ell\left(  \ell+1\right)  }{\rho^{2}}-\eta\rho^{2}-\nu\rho^{6}-\lambda
\rho^{4}\right]  v\left(  \rho\right)  =0.
\end{equation}

The dual relations (\ref{Eta6}), (\ref{xieta6}), (\ref{mulamb6}),
(\ref{ltrans6}), (\ref{rtrans6}), and (\ref{wftrans6}) with $a=-2$, $b_{1}=1$,
and $b_{2}=0$ read%
\begin{align}
E  &  \rightarrow-\frac{\eta}{4},\\
\mu &  \rightarrow-\frac{\mathcal{E}}{4},\\
\xi &  \rightarrow\frac{\nu}{4},\\
\kappa &  \rightarrow\frac{\lambda}{4},\\
l+\frac{1}{2}  &  \rightarrow\frac{1}{2}\left(  \ell+\frac{1}{2}\right)  ,
\end{align}
and%
\begin{align}
r  &  \rightarrow\rho^{2}\text{ \ \ or \ \ }r^{1/2}\rightarrow\rho,\\
u\left(  r\right)   &  \rightarrow\rho^{1/2}v\left(  \rho\right)  .
\end{align}
Substituting the dual relations into the eigenfunction of the potential
(\ref{Ur2m11}), Eq. (\ref{radwf61}), gives%
\begin{align}
\rho^{1/2}v\left(  \rho\right)   &  =A_{\ell}\left(  \rho^{2}\right)
^{\left(  1/2\right)  \left(  \ell+1/2\right)  +1/2}\exp\left(  \frac{\left(
\frac{\nu}{4}\right)  ^{1/2}}{2}\left(  \rho^{2}\right)  ^{2}+\frac
{\frac{\lambda}{4}}{2\left(  \frac{\nu}{4}\right)  ^{1/2}}\rho^{2}\right)
\nonumber\\
&  \times N\left(  2\left[  \frac{1}{2}\left(  \ell+\frac{1}{2}\right)
\right]  ,{\frac{i\frac{\lambda}{4}}{\left(  \frac{\nu}{4}\right)  ^{3/4}}%
},-{\frac{\left(  -\frac{\eta}{4}\right)  }{\left(  \frac{\nu}{4}\right)
^{1/2}}}-{\frac{\left(  \frac{\lambda}{4}\right)  ^{2}}{4\left(  \frac{\nu}%
{4}\right)  ^{3/2}}},-i{\frac{2\left(  -\frac{\mathcal{E}}{4}\right)
}{\left(  \frac{\nu}{4}\right)  ^{1/4}}},i\left(  \frac{\nu}{4}\right)
^{1/4}\rho^{2}\right)  .
\end{align}
Then we obtain the eigenfunction of the potential (\ref{Vrho264}):%
\begin{equation}
v\left(  \rho\right)  =A_{\ell}\rho^{\ell+1}\exp\left(  \frac{\sqrt{\nu}}%
{4}\rho^{4}+\frac{\lambda}{4\sqrt{\nu}}\rho^{2}\right)  N\left(  \ell+\frac
{1}{2},{\frac{i\sqrt{2}\lambda}{2\nu^{3/4}}},{\frac{\eta}{2\nu^{1/2}}}%
-{\frac{\lambda^{2}}{8\nu^{3/2}}},i{\frac{\sqrt{2}\mathcal{E}}{2\nu^{1/4}}%
},i\frac{\sqrt{2}}{2}\nu^{1/4}\rho^{2}\right)  . \label{vwf62}%
\end{equation}
Substituting the dual relations into (\ref{specsq61}) gives an implicit
expression of the eigenvalue of the potential (\ref{Vrho264}):%
\begin{equation}
K_{2}\left(  2\left[  \frac{1}{2}\left(  \ell+\frac{1}{2}\right)  \right]
,{\frac{i\frac{\lambda}{4}}{\left(  \frac{\nu}{4}\right)  ^{3/4}}}%
,-{\frac{\left(  -\frac{\eta}{4}\right)  }{\left(  \frac{\nu}{4}\right)
^{1/2}}}-{\frac{\left(  \frac{\lambda}{4}\right)  ^{2}}{4\left(  \frac{\nu}%
{4}\right)  ^{3/2}}},-i{\frac{2\left(  -\frac{\mathcal{E}}{4}\right)
}{\left(  \frac{\nu}{4}\right)  ^{1/4}}}\right)  =0,
\end{equation}
i.e.,
\begin{equation}
K_{2}\left(  \ell+\frac{1}{2},{\frac{i\sqrt{2}\lambda}{2\nu^{3/4}}}%
,{\frac{\eta}{2\nu^{1/2}}}-{\frac{\lambda^{2}}{8\nu^{3/2}}},i{\frac{\sqrt
{2}\mathcal{E}}{2\nu^{1/4}}}\right)  =0. \label{vsp62}%
\end{equation}
This result agrees with the result obtained by directly solving the radial
equation, which is given in Appendix \ref{V264} for verification.

\subsubsection{Solving $V\left(  \rho\right)  =\frac{\eta}{\rho^{2/3}}+\nu
\rho^{2/3}+\frac{\lambda}{\rho^{4/3}}$ from $U\left(  r\right)  =\xi
r^{2}+\frac{\mu}{r}+\kappa r$}

The potential (\ref{Vrhom232343}) is a dual potential of the potential
(\ref{Ur2m11}) with $a=0$, $b_{1}=1$, $b_{2}=-2$ and $A=-5/3$, $B_{1}=-1/3$,
$B_{2}=-7/3$ in the duality relations (\ref{Aanda6}) and (\ref{Bandb6}).

The radial equation of the potential (\ref{Vrhom232343}) is
\begin{equation}
\frac{d^{2}v\left(  \rho\right)  }{d\rho^{2}}+\left[  \mathcal{E}-\frac
{\ell\left(  \ell+1\right)  }{\rho^{2}}-\frac{\eta}{\rho^{2/3}}-\nu\rho
^{2/3}-\frac{\lambda}{\rho^{4/3}}\right]  v\left(  \rho\right)  =0.
\end{equation}

The dual relations (\ref{Eta6}), (\ref{xieta6}), (\ref{mulamb6}),
(\ref{ltrans6}), (\ref{rtrans6}), and (\ref{wftrans6}) with $a=0$, $b_{1}=1$,
and $b_{2}=-2$ read%
\begin{align}
E  &  \rightarrow-\frac{9}{4}\eta,\\
\kappa &  \rightarrow-\frac{9}{4}\mathcal{E},\\
\xi &  \rightarrow\frac{9}{4}\nu,\\
\mu &  \rightarrow\frac{9}{4}\lambda,\\
l+\frac{1}{2}  &  \rightarrow\frac{3}{2}\left(  \ell+\frac{1}{2}\right)  ,
\end{align}
and%
\begin{align}
r  &  \rightarrow\rho^{2/3}\text{ \ or \ }r^{3/2}\rightarrow\rho,\\
u\left(  r\right)   &  \rightarrow\rho^{-1/6}v\left(  \rho\right)  .
\end{align}
Substituting the dual relations into the eigenfunction of the potential
(\ref{Ur2m11}), Eq. (\ref{radwf61}), gives%
\begin{align}
\rho^{-1/6}v\left(  \rho\right)   &  =A_{\ell}\left(  \rho^{2/3}\right)
^{\left(  3/2\right)  \left(  \ell+1/2\right)  +1/2}\exp\left(  \frac{\left(
\frac{9}{4}\nu\right)  ^{1/2}}{2}\left(  \rho^{2/3}\right)  ^{2}+\frac
{-\frac{9}{4}\mathcal{E}}{2\sqrt{\frac{9}{4}\nu}}\rho^{2/3}\right) \nonumber\\
&  \times N\left(  2\left(  \frac{3}{2}\left(  \ell+\frac{1}{2}\right)
\right)  ,{\frac{i\left(  -\frac{9}{4}\mathcal{E}\right)  }{{\left(  \frac
{9}{4}\nu\right)  }^{3/4}}},-{\frac{\left(  -\frac{9}{4}\mathcal{E}\right)
^{2}}{4\left(  \frac{9}{4}\nu\right)  ^{3/2}}}-{\frac{-\frac{9}{4}\eta
}{\left(  \frac{9}{4}\nu\right)  ^{1/2}}},-{\frac{2i\left(  \frac{9}{4}%
\lambda\right)  }{\left(  \frac{9}{4}\nu\right)  ^{1/4}}},i\left(  \frac{9}%
{4}\nu\right)  ^{1/4}\rho^{2/3}\right)  .
\end{align}
Then we obtain the eigenfunction of the potential (\ref{Vrhom232343}):%
\begin{align}
v\left(  \rho\right)   &  =A_{\ell}\rho^{\ell+1}\exp\left(  \frac{3\nu^{1/2}%
}{4}\rho^{4/3}-\frac{3\mathcal{E}}{4\nu^{1/2}}\rho^{2/3}\right) \nonumber\\
&  \times N\left(  3\ell+\frac{3}{2},-{\frac{i\sqrt{6}\mathcal{E}}{2\nu^{3/4}%
}},{\frac{3\eta}{2\nu^{1/2}}}-{\frac{3\mathcal{E}^{2}}{8\nu^{3/2}}}%
,-i{\frac{3\sqrt{6}\lambda}{2\nu^{1/4}}},i\frac{\sqrt{6}}{2}\nu^{1/4}%
\rho^{2/3}\right)  . \label{vwf63}%
\end{align}
Substituting the dual relations into (\ref{specsq61}) gives an implicit
expression of the eigenvalue of the potential (\ref{Vrhom232343}):%
\begin{equation}
K_{2}\left(  2\left[  \frac{3}{2}\left(  \ell+\frac{1}{2}\right)  \right]
,{\frac{i\left(  -\frac{9}{4}\mathcal{E}\right)  }{{\left(  \frac{9}{4}%
\nu\right)  }^{3/4}}},-{\frac{\left(  -\frac{9}{4}\mathcal{E}\right)  ^{2}%
}{4\left(  \frac{9}{4}\nu\right)  ^{3/2}}}-{\frac{-\frac{9}{4}\eta}{\left(
\frac{9}{4}\nu\right)  ^{1/2}}},-{\frac{2i\left(  \frac{9}{4}\lambda\right)
}{\left(  \frac{9}{4}\nu\right)  ^{1/4}}}\right)  =0,
\end{equation}
i.e.,
\begin{equation}
K_{2}\left(  3\ell+\frac{3}{2},-{\frac{i\sqrt{6}\mathcal{E}}{2\nu^{3/4}}%
},{\frac{3\eta}{2\nu^{1/2}}}-{\frac{3\mathcal{E}^{2}}{8\nu^{3/2}}}%
,-i{\frac{3\sqrt{6}\lambda}{2\nu^{1/4}}}\right)  =0. \label{vsp63}%
\end{equation}
This result agrees with the result obtained by directly solving the radial
equation, which is given in Appendix \ref{Vm2323m43} for verification.

\section{Conclusions and Outlook \label{Conclusion}}

The Newton duality reveals the underlying nature of mechanics. In this paper,
we provide three results.

(1) In classical mechanics, we generalize Newton's original duality which
involves only power potentials to more general kinds of potentials --- general
polynomial potentials.

(2) In quantum mechanics, we provide a quantum Newton duality for power
potentials, general polynomial potentials, transcendental-function potentials,
and power potentials in different spatial dimensions.

(3) Based on the quantum Newton duality, we develop a method for solving
eigenproblems of a quantum mechanical system and solve some potentials with
this method.

Moreover, in appendices, as preparations and verifications, we solve the
potentials which is solved by the method of the Newton duality in the main
body of the paper by directly solving the eigenequation. The potentials solved
in appendices are $\alpha r^{2/3}$, $\alpha r^{6}$, $\frac{\alpha}{r^{3/2}}$,
$\alpha r^{2}+\frac{\beta}{r}$, $\frac{\alpha}{r}+\frac{\beta}{r^{3/2}}$,
$\alpha r^{2}+\beta r^{6}$, $\alpha r^{2}+\beta r$,$\frac{\alpha}{r}%
+\frac{\beta}{\sqrt{r}}$, $\frac{\alpha}{r^{2/3}}+\beta r^{2/3}$, $\alpha
r^{2}+\beta r$, $\frac{\alpha}{r}+\frac{\beta}{\sqrt{r}}$, $\frac{\alpha
}{r^{2/3}}+\beta r^{2/3}$, $\frac{\alpha}{\sqrt{r}}+\frac{\beta}{r^{3/2}}$,
$\alpha r^{2/3}+\frac{\beta}{r^{4/3}}$, $\alpha r^{6}+\beta r^{4}$, $\alpha
r^{6}+\beta r^{2}$, $\frac{\alpha}{r^{3/2}}+\frac{\beta}{r}$, $\alpha
r^{2}+\frac{\beta}{r}+\sigma r$, $\frac{\alpha}{r}+\frac{\beta}{r^{3/2}}%
+\frac{\sigma}{r^{1/2}}$, $\alpha r^{2}+\beta r^{6}+\sigma r^{4}$, and
$\frac{\alpha}{r^{2/3}}+\beta r^{2/3}+\frac{\sigma}{r^{4/3}}$.

The original Newton duality includes only power potentials. We show that the
Newton duality is in fact a more general duality in dynamics: such a duality
exists in more general potentials, including general polynomial potentials and
transcendental-function potentials. This inspires us to make a conjecture that
the Newton duality exists in all kinds of potentials and every potential has
Newton dualities.

These potentials are all long-range potentials. The exact solution is
important for understanding long-range potentials. The long-range potential is
the most important model in quantum mechanics, e.g., the Coulomb potential and
the harmonic oscillator potential. Beyond the Coulomb potential and the
harmonic oscillator potential, recently, some other long-range power
potentials are exactly solved, e.g., the one-dimensional inverse square root
potential \cite{ishkhanyan2015exact} and the three-dimensional inverse square
root potential \cite{li2016exact}. The long-range-potential scattering is
often a difficult problem. One reason is that different long-range potentials
often have different scattering boundary conditions which are determined by
the large-distance asymptotic solution. The long-range-potential scattering
has been studied for many years
\cite{enss1979asymptotic,levy1963low,hinckelmann1971low,barford2003renormalization,hod2013scattering,yafaev1998scattering,romo1998study}%
. Besides scattering of potentials in quantum mechanics, there are also other
kinds of long-range scattering, e.g., the scattering on black holes
\cite{stadnik2013resonant,flambaum2012dense}. In a scattering problem, a key
task is to seek scattering phase shifts \cite{pang2012relation,li2015heat}.
Many solutions mentioned in the paper are the Heun function, there are also
researches on the Heun function
\cite{karayer2015extension,birkandan2007examples,ciftci2010physical,hortacsu2011heun}%
.

In this paper, we discuss the Newton duality in classical mechanics and in
quantum mechanics. The quantum version of the Newton duality is essentially
the duality for the Schr\"{o}dinger equation. In future work, we will consider
the Newton duality for other dynamical equations, such as the Klein--Gordon
equation and the Dirac equation, etc. The Einstein equation of gravity is also
a field equation. By considering the Newton duality between Einstein
equations, we can reveal the duality relation between spacetime manifolds.

In this paper, we also suggest a method for solving eigenequations. We here
only focus on how to obtain an exact solution, in future works, we will try to
construct an approximate method based on the Newton duality for
eigenequations. Furthermore, we can also consider the Newton duality between
field equations, and then consider the Newton duality in quantum field theory.

\acknowledgments

We are very indebted to Dr G. Zeitrauman for his encouragement. This work is supported in part by NSF of China
under Grant No. 11575125 and No. 11675119.

\appendix
\section{The exact solution of $U\left(  r\right)  =-\frac{\alpha}{r^{3/2}}$
\label{Vm32}}

The inverse fractional power potential $U\left(  r\right)  =-\alpha/r^{3/2}$
is a long-range potential, which has both bound states and scattering states.
In this Appendix, we present the exact solutions of both bound and scattering
states for the $1/r^{3/2}$-potential.

In particular, the study of exact solutions is important for the long-range
potential scattering. The scattering boundary conditions for different
long-range potentials are different, while in short-range potential
scattering, the scattering boundary conditions is the same for all short-range
potentials \cite{liu2014scattering,li2016scattering}. To solve a scattering
problem, one first needs to determine the scattering boundary condition which
is determined by the long-distance asymptotic behavior of the solution. The
$1/r^{3/2}$-potential is a long-range potential since the inverse-power
potential $1/r^{s}$ with\ $0<s<2$ is identified as a long-range potential
\cite{burke2011r}. Though $1/r^{s}$ with\ $0<s\leq1$ and with\ $1<s<2$ are
both long-range potentials, they are quite different in scattering: the
scattering boundary for $1/r^{s}$-potential with\ $0<s\leq1$ differs from one
potential to another, but the scattering boundary for $1/r^{s}$-potential
with\ $1<s\leq2$ is the same as short-range potentials.

\subsection{The radial equation \label{RE}}

The radial wave function $R_{l}\left(  r\right)  =u_{l}\left(  r\right)  /r$
of the potential%
\begin{equation}
U\left(  r\right)  =-\frac{\alpha}{r^{3/2}}%
\end{equation}
is determined by the radial equation%
\begin{equation}
\frac{d^{2}u_{l}\left(  r\right)  }{dr^{2}}+\left[  k^{2}-\frac{l\left(
l+1\right)  }{r^{2}}+\frac{\alpha}{r^{3/2}}\right]  u_{l}\left(  r\right)  =0
\label{radialeq32}%
\end{equation}
with the boundary conditions at $r=0$ and at $r\rightarrow\infty$.

At $r=0$, both for scattering states and bound states, the boundary condition
is $u_{l}\left(  0\right)  =0$, or, more stronger, $\lim_{r\rightarrow0}%
u_{l}\left(  r\right)  /r^{l+1}=1$ \cite{romo1998study}, since the radial wave
function $R_{l}\left(  r\right)  $ at $r=0$ must be finite.

At $r\rightarrow\infty$, the boundary condition of scattering states is that
the wave function equals the asymptotic solution at $r\rightarrow\infty$,
i.e., $u_{l}\left(  r\right)  \overset{r\rightarrow\infty}{\sim}u_{l}^{\infty
}\left(  r\right)  $; the boundary condition of bound states is that the wave
function equals zero, i.e., $u_{l}\left(  r\right)  \overset{r\rightarrow
\infty}{=}0$.

In order to solve the radial equation (\ref{radialeq32}), we introduce%
\begin{equation}
u_{l}\left(  z\right)  =A_{l}e^{-z^{2}/2}z^{2\left(  l+1\right)  }f_{l}\left(
z\right)  \label{uf32}%
\end{equation}
with%
\begin{equation}
z=-\sqrt{-2ikr}, \label{z}%
\end{equation}
where $A_{l}$ is a constant. Substituting Eq. (\ref{uf32}) into the radial
equation (\ref{radialeq32}) gives an equation of $f_{l}\left(  z\right)  $:%
\begin{equation}
zf_{l}^{\prime\prime}\left(  z\right)  +\left(  -2{z}^{2}+4l+3\right)
f_{l}^{\prime}\left(  z\right)  +\left[  -4\left(  l+1\right)  z-4\lambda
\right]  f_{l}\left(  z\right)  =0, \label{eqfscat32}%
\end{equation}
where $\lambda=\alpha/\sqrt{-2ik}$. This equation is just the Biconfluent Heun
equation \cite{ronveaux1995heun}.

Next, we solve Eq. (\ref{eqfscat32}) with boundary conditions.

\subsection{The regular solution \label{RegularSolution}}

The solution satisfying the boundary condition at $r=0$ is called the regular solution.

At $r=0$, as mentioned above, the scattering state and the bound state has the
same boundary condition \cite{romo1998study},%
\begin{equation}
\lim_{r\rightarrow0}\frac{u_{l}\left(  r\right)  }{r^{l+1}}=1, \label{bdcond}%
\end{equation}
because the asymptotic solution of the radial equation (\ref{radialeq32}) at
$r=0$ is $u_{l}\left(  r\right)  \sim r^{l+1}$ (there is also another
asymptotic solution $u_{l}\left(  r\right)  \sim r^{-l}$ but it diverges at
$r=0$) \cite{ballentine1998quantum}.

The equation of $f_{l}\left(  z\right)  $, Eq. (\ref{eqfscat32}), has two
linearly independent solutions \cite{ronveaux1995heun},%
\begin{align}
y_{l}^{\left(  1\right)  }\left(  z\right)   &  =N\left(  4l+2,0,0,8\lambda
,z\right)  ,\\
y_{l}^{\left(  2\right)  }\left(  z\right)   &  =cN\left(  4l+2,0,0,8\lambda
,z\right)  \ln z+\sum_{n\geq0}d_{n}z^{n-2\left(  2l+1\right)  },
\end{align}
where $N\left(  \alpha,\beta,\gamma,\delta,z\right)  $ is the Heun Biconfluent
function \cite{ronveaux1995heun,slavyanov2000special}. The constant $c$ is%
\begin{equation}
c=\frac{2}{2l+1}\left(  \lambda d_{4l+1}+ld_{4l}\right)
\end{equation}
and $d_{\nu}$ is determined by the recurrence relation \cite{ronveaux1995heun}%
\begin{align}
&  d_{-1}=0,\text{ }d_{0}=1,\nonumber\\
&  \left(  \nu+2\right)  \left(  \nu-4l\right)  d_{\nu+2}-4\lambda d_{\nu
+1}+\left[  -2\left(  \nu+1\right)  +4l+2\right]  d_{\nu}=0.
\end{align}

To determine the solution with the boundary condition (\ref{bdcond}), we use
the expansion of the Heun Biconfluent function near $z=0$
\cite{ronveaux1995heun},%
\begin{equation}
N\left(  4l+2,0,0,8\lambda,z\right)  =\sum_{n\geq0}\frac{A_{n}}{\left(
4l+3\right)  _{n}}\frac{z^{n}}{n!},
\end{equation}
where $\left(  a\right)  _{n}=\Gamma\left(  a+n\right)  /\Gamma\left(
a\right)  $ is Pochhammer's symbol and the coefficient $A_{n}$ is determined
by the recurrence relation \cite{ronveaux1995heun}:%
\begin{align}
A_{0}  &  =1,\text{ \ }A_{1}=4\lambda,\nonumber\\
A_{n+2}  &  =4\lambda A_{n+1}+2\left(  n+1\right)  \left(  n+4l+3\right)
\left(  2l+2+n\right)  A_{n}.
\end{align}

It is clear that only $y_{l}^{\left(  1\right)  }\left(  z\right)  $ satisfies
the boundary condition at $r=0$, Eq. (\ref{bdcond}). Then, by Eqs.
(\ref{uf32}) and (\ref{z}), we obtain the regular solution:
\begin{equation}
u_{l}\left(  r\right)  =A_{l}e^{ikr}\left(  -2ikr\right)  ^{l+1}N\left(
4l+2,0,0,8\lambda,-\sqrt{-2ikr}\right)  .
\end{equation}

\subsection{The irregular solution \label{IrregularSolution}}

The solution satisfying the boundary condition at $r\rightarrow\infty$ is
called the irregular solution. Scattering states and bound states have
different boundary conditions at $r\rightarrow\infty$.

\subsubsection{The Scattering boundary condition}

The scattering boundary condition requires that the wave function equals the
asymptotic solution of the radial equation (\ref{radialeq32}) at
$r\rightarrow\infty$. The asymptotic behavior of the wave function at
$r\rightarrow\infty$ is dominated by the external potential, $-\alpha/r^{3/2}%
$, rather than the centrifugal potential $l\left(  l+1\right)  /r^{2}$ like
that in the case of short-range potentials. That is to say, the scattering
boundary condition is determined by the external potential $-\alpha/r^{3/2}$.

To determine the scattering boundary condition, we first solve the asymptotic
solution of the radial equation (\ref{radialeq32}).

Letting%
\begin{equation}
u_{l}\left(  r\right)  =e^{h\left(  r\right)  }e^{\pm ikr} \label{exp}%
\end{equation}
and substituting into the radial equation (\ref{radialeq32}) give an equation
of $h\left(  r\right)  $:%
\begin{equation}
h^{\prime\prime}\left(  r\right)  +h^{\prime}\left(  r\right)  ^{2}%
\pm2ikh^{\prime}\left(  r\right)  =\frac{l\left(  l+1\right)  }{r^{2}}%
-\frac{\alpha}{r^{3/2}}.
\end{equation}
Only taking the leading contribution into account for we only concentrate on
the asymptotic behavior, we have%
\begin{equation}
\pm2ikh^{\prime}\left(  r\right)  \sim-\frac{\alpha}{r^{3/2}}.
\end{equation}
Then we have%
\begin{equation}
h\left(  r\right)  \sim\pm\frac{i\alpha}{k\sqrt{r}}.
\end{equation}
Obviously, this contribution vanishes when $r\rightarrow\infty$ and can be
dropped out in the asymptotic solution. Then, the large-distance asymptotic
solution of Eq. (\ref{radialeq32}) reads%
\begin{equation}
u_{l}\left(  r\right)  =e^{\pm ikr}.
\end{equation}
It can be seen that the asymptotic behavior of the potential $U\left(
r\right)  =-\alpha/r^{3/2}$ is the same as that of the short-range potential.

The scattering boundary condition, then, can be written as%
\begin{equation}
\lim_{r\rightarrow\infty}e^{\pm ikr}u_{l}\left(  r\right)  =1.
\label{bdcondir}%
\end{equation}

\subsubsection{The irregular solution}

Now we can determine the irregular solution by the scattering boundary
condition given above.

Eq. (\ref{eqfscat32}) with the scattering boundary condition (\ref{bdcondir})
has two linearly independent solutions \cite{ronveaux1995heun}%
\begin{align}
B_{l}^{+}\left(  4l+2,0,0,8\lambda,z\right)   &  =z^{-2\left(  l+1\right)
}\sum_{n\geq0}\frac{a_{n}}{z^{n}},\label{f132}\\
H^{+}\left(  4l+2,0,0,8\lambda,z\right)   &  =z^{-2\left(  l+1\right)
}e^{z^{2}}\sum_{n\geq0}\frac{e_{n}}{z^{n}}, \label{f232}%
\end{align}
where $B_{l}^{+}\left(  \alpha,\beta,\gamma,\delta,z\right)  $ and $H_{l}%
^{+}\left(  \alpha,\beta,\gamma,\delta,z\right)  $ are two other biconfluent
Heun functions \cite{ronveaux1995heun}, different from the Heun\ functions
mentioned above. The coefficients in Eqs. (\ref{f132}) and (\ref{f232}) are
given by%
\begin{align}
&  a_{0}=1,\text{ \ }a_{1}=2\lambda,\nonumber\\
&  2\left(  n+2\right)  a_{n+2}-2\lambda a_{n+1}+\left[  n\left(  n+2\right)
-\left(  2l+1\right)  ^{2}+1\right]  a_{n}=0,
\end{align}
and%
\begin{align}
&  e_{0}=1,\text{ \ }e_{1}=-2\lambda,\nonumber\\
&  2\left(  n+2\right)  e_{n+2}+2\lambda e_{n+1}-\left[  n\left(  n+2\right)
-\left(  2l+1\right)  ^{2}+1\right]  e_{n}=0.
\end{align}
By Eq. (\ref{uf32}), one can directly check that $B_{l}^{+}\left(
\alpha,\beta,\gamma,\delta,z\right)  $ and $H_{l}^{+}\left(  \alpha
,\beta,\gamma,\delta,z\right)  $ satisfy the scattering boundary condition
(\ref{bdcondir}).

\subsection{Bound states and scattering states \label{ExactSolution}}

With the regular solution and the irregular solution, we can construct the
solution of bound states and scattering states.

First express the regular solution as a linear combination of the two
irregular solutions \cite{ronveaux1995heun}:%
\begin{align}
N\left(  4l+2,0,0,8\lambda,z\right)   &  =K_{1}\left(  4l+2,0,0,8\lambda
\right)  B_{l}^{+}\left(  4l+2,0,0,8\lambda,z\right) \nonumber\\
&  +K_{2}\left(  4l+2,0,0,8\lambda\right)  H_{l}^{+}\left(  4l+2,0,0,8\lambda
,z\right)  ,
\end{align}
where $K_{1}\left(  4l+2,0,0,8\lambda\right)  $ and $K_{2}\left(
4l+2,0,0,8\lambda\right)  $ are the coefficients of combination
\cite{ronveaux1995heun}. Using the expansions (\ref{f132}) and (\ref{f232}),
we can rewrite (\ref{uf32}) as%
\begin{equation}
u_{l}\left(  r\right)  =A_{l}K_{1}\left(  4l+2,0,0,8\lambda\right)
e^{ikr}\sum_{n\geq0}\frac{a_{n}}{\left(  -\sqrt{-2ikr}\right)  ^{n}}%
+A_{l}K_{2}\left(  4l+2,0,0,8\lambda\right)  e^{-ikr}\sum_{n\geq0}\frac{e_{n}%
}{\left(  -\sqrt{-2ikr}\right)  ^{n}}. \label{scatu}%
\end{equation}

\subsubsection{The bound state}

By analytical continuation $k$ to whole complex plane, we can consider $k$ on
the positive imaginary axis, i.e.,
\begin{equation}
k=i\kappa,\text{ \ }\kappa>0.
\end{equation}
Notice that on the positive imaginary axis $\lambda=\alpha/\sqrt{2\kappa}$,
Eq. (\ref{scatu}) becomes%
\begin{equation}
u_{l}\left(  r\right)  =A_{l}K_{1}\left(  4l+2,0,0,8\lambda\right)  e^{-\kappa
r}\sum_{n\geq0}\frac{a_{n}}{\left(  2\kappa r\right)  ^{n/2}}+A_{l}%
K_{2}\left(  4l+2,0,0,8\lambda\right)  e^{\kappa r}\sum_{n\geq0}\frac{e_{n}%
}{\left(  2\kappa r\right)  ^{n/2}}.
\end{equation}
For bound states, only the first term contributes because the bound state
boundary condition requires $u_{l}\left(  \infty\right)  \rightarrow0$. Then
we have
\begin{equation}
K_{2}\left(  4l+2,0,0,8\lambda\right)  =0, \label{spectrum}%
\end{equation}
where \cite{ronveaux1995heun}%
\begin{equation}
K_{2}\left(  \alpha,\beta,\gamma,\delta\right)  =\frac{\Gamma\left(
1+\alpha\right)  }{\Gamma\left(  \left(  \alpha-\gamma\right)  /2\right)
\Gamma\left(  1+\left(  \alpha+\gamma\right)  /2\right)  }J_{1+\left(
\alpha+\gamma\right)  /2}\left(  \frac{1}{2}\left(  \alpha+\gamma\right)
,\beta,\frac{1}{2}\left(  3\alpha-\gamma\right)  ,\delta+\frac{1}{2}%
\beta\left(  \gamma-\alpha\right)  \right)  ,
\end{equation}
with%
\begin{equation}
J_{\lambda}\left(  \alpha,\beta,\gamma,\delta\right)  =\int_{0}^{\infty
}x^{\lambda-1}e^{-x^{2}-\beta x}N\left(  \alpha,\beta,\gamma,\delta,x\right)
dx.
\end{equation}

Eq. (\ref{spectrum}) is just an implicit expression of the bound-state eigenvalue.

The bound-state wave function reads%
\begin{equation}
u_{l}\left(  r\right)  =Ce^{-\kappa r}\sum_{n\geq0}\frac{a_{n}}{\left(
2\kappa r\right)  ^{n/2}},
\end{equation}
where $C$ is the normalization constant.

\subsubsection{The scattering state}

The singularity of the $S$-matrix on the positive imaginary axis corresponds
to the eigenvalues of bound states \cite{joachain1975quantum}, so the zero of
$K_{2}\left(  4l+2,0,0,8\lambda\right)  $ on the positive imaginary is just
the singularity of the $S$-matrix. Considering that the $S$-matrix is unitary,
i.e.,%
\begin{equation}
S_{l}=e^{2i\delta_{l}}, \label{SM}%
\end{equation}
we have%
\begin{equation}
S_{l}\left(  k\right)  =\frac{K_{2}^{\ast}\left(  4l+2,0,0,8\lambda\right)
}{K_{2}\left(  4l+2,0,0,8\lambda\right)  }=\frac{K_{2}\left(
4l+2,0,0,8i\lambda\right)  }{K_{2}\left(  4l+2,0,0,8\lambda\right)  },
\label{SM1}%
\end{equation}
where the relation $K_{2}^{\ast}\left(  4l+2,0,0,8\alpha/\sqrt{-2ik}\right)
=K_{2}^{\ast}\left(  4l+2,0,0,8i\alpha/\sqrt{-2ik}\right)  $ is used.

The scattering wave function can be constructed with the help of the
$S$-matrix. The scattering wave function can be expressed as a linear
combination of the radially ingoing wave $u_{in}$ and the radially outgoing
wave $u_{out}$, which are conjugate to each other, i.e.,
\cite{joachain1975quantum}%
\begin{equation}
u_{l}\left(  r\right)  =A_{l}\left[  \left(  -1\right)  ^{l+1}u_{in}\left(
r\right)  +S_{l}\left(  k\right)  u_{out}\left(  r\right)  \right]  .
\label{inout}%
\end{equation}
From Eq. (\ref{scatu}), we have%
\begin{align}
u_{in}  &  =e^{-ikr}\sum_{n\geq0}\frac{e_{n}}{\left(  -\sqrt{-2ikr}\right)
^{n}},\nonumber\\
u_{out}  &  =e^{ikr}\sum_{n\geq0}\frac{e_{n}^{\ast}}{\left(  -\sqrt
{2ikr}\right)  ^{n}}.
\end{align}
Then by Eq. (\ref{SM1}) we achieve the scattering wave function,%
\begin{equation}
u_{l}\left(  r\right)  =A_{l}\left[  \left(  -1\right)  ^{l+1}e^{-ikr}%
\sum_{n\geq0}\frac{e_{n}}{\left(  -\sqrt{-2ikr}\right)  ^{n}}+\frac
{K_{2}\left(  4l+2,0,0,8i\lambda\right)  }{K_{2}\left(  4l+2,0,0,8\lambda
\right)  }e^{ikr}\sum_{n\geq0}\frac{e_{n}^{\ast}}{\left(  -\sqrt{2ikr}\right)
^{n}}\right]  .
\end{equation}
Taking $r\rightarrow\infty$, we have
\begin{align}
&  u_{l}\left(  r\right)  \overset{r\rightarrow\infty}{\sim}A_{l}\left[
\left(  -1\right)  ^{l+1}e^{-ikr}+\frac{K_{2}\left(  4l+2,0,0,8i\lambda
\right)  }{K_{2}\left(  4l+2,0,0,8\lambda\right)  }e^{ikr}\right] \nonumber\\
&  =A_{l}e^{i\delta_{l}}\sin\left(  kr+\delta_{l}-\frac{l\pi}{2}\right)  .
\end{align}
By Eqs. (\ref{SM}) and (\ref{SM1}), we obtain the scattering phase shift%
\begin{equation}
\delta_{l}=-\arg K_{2}\left(  4l+2,0,0,8\lambda\right)  .
\end{equation}

\section{The exact solution of $U\left(  r\right)  =\eta r^{2/3}$ \label{V23}}

In this appendix, we provide an exact solution of the eigenproblem of the
$r^{2/3}$-potential by solving the radial equation directly. The $r^{2/3}%
$-potential has only bound states.

The radial equation of the $r^{2/3}$-potential
\begin{equation}
U\left(  r\right)  =\eta r^{2/3}%
\end{equation}
reads%
\begin{equation}
\frac{d^{2}u_{l}\left(  r\right)  }{dr^{2}}+\left[  E-\frac{l\left(
l+1\right)  }{r^{2}}-\eta r^{2/3}\right]  u_{l}\left(  r\right)  =0.
\label{req23}%
\end{equation}
By the variable substitution%
\begin{equation}
z=\sqrt{\frac{3}{2}}\eta^{1/4}r^{2/3} \label{zr23}%
\end{equation}
and introducing $f_{l}\left(  z\right)  $ by%
\begin{equation}
u_{l}\left(  z\right)  =A_{l}\exp\left(  -\frac{z^{2}}{2}+\frac{\sqrt{6}%
E}{4\eta^{3/4}}z\right)  z^{3\left(  l+1\right)  /2}f_{l}\left(  z\right)  ,
\end{equation}
we convert the radial equation (\ref{radwf32}) into an equation of
$f_{l}\left(  z\right)  $:%
\begin{equation}
zf_{l}^{\prime\prime}\left(  z\right)  +\left(  -2{z}^{2}+\frac{\sqrt{6}%
E}{2\eta^{3/4}}z+3l+\frac{5}{2}\right)  f_{l}^{\prime}\left(  z\right)
+\left\{  \left[  \left(  \frac{\sqrt{6}E}{4\eta^{3/4}}\right)  ^{2}%
-3l-\frac{7}{2}\right]  z+\frac{\sqrt{6}E}{4\eta^{3/4}}\left(  3l+\frac{5}%
{2}\right)  \right\}  f_{l}\left(  z\right)  =0, \label{eqf23}%
\end{equation}
where $A_{l}$ is a constant. This is a Biconfluent Heun equation
\cite{ronveaux1995heun}.

The choice of the boundary condition has been discussed in Ref.
\cite{li2016exact}.

\subsection{The regular solution}

The regular solution is a solution satisfying the boundary condition at $r=0$
\cite{li2016exact}. The regular solution at $r=0$\ should satisfy the boundary
condition $\lim_{r\rightarrow0}u_{l}\left(  r\right)  /r^{l+1}=1$. In this
section, we provide the regular solution of Eq. (\ref{eqf23}).

The Biconfluent Heun equation (\ref{eqf23}) has two linearly independent
solutions \cite{ronveaux1995heun}%
\begin{align}
y_{l}^{\left(  1\right)  }\left(  z\right)   &  =N\left(  3l+\frac{3}%
{2},-\frac{\sqrt{6}E}{2\eta^{3/4}},\frac{3E^{2}}{8\eta^{3/2}},0,z\right)
,\label{yl1z23}\\
y_{l}^{\left(  2\right)  }\left(  z\right)   &  =cN\left(  3l+\frac{3}%
{2},-\frac{\sqrt{6}E}{2\eta^{3/4}},\frac{3E^{2}}{8\eta^{3/2}},0,z\right)  \ln
z+\sum_{n\geq0}d_{n}z^{n-3l-3/2}, \label{yl2z23}%
\end{align}
where%
\begin{equation}
c=\frac{1}{3l+3/2}\left[  -d_{3l+1/2}\frac{\sqrt{6}E}{4\eta^{3/4}}\left(
3l+\frac{1}{2}\right)  -d_{3l-1/2}\left(  \frac{3E^{2}}{8\eta^{3/2}}+\frac
{1}{2}-3l\right)  \right]
\end{equation}
is a constant with the coefficient $d_{\nu}$ given by the following recurrence
relation,%
\begin{align}
&  d_{-1}=0,\text{ \ }d_{0}=1,\nonumber\\
&  \left(  \nu+2\right)  \left(  \nu+\frac{1}{2}-3l\right)  d_{\nu+2}%
+\frac{\sqrt{6}E}{4\eta^{3/4}}\left(  2\nu+\frac{3}{2}-3l\right)  d_{\nu
+1}+\left(  \frac{3E^{2}}{8\eta^{3/2}}-2v+3l-\frac{1}{2}\right)  d_{\nu}=0
\end{align}
and $N(\alpha,\beta,\gamma,\delta,z)$ is the biconfluent Heun function
\cite{ronveaux1995heun,slavyanov2000special,li2016exact}.

The biconfluent Heun function $N\left(  3l+\frac{3}{2},-\frac{\sqrt{6}E}%
{2\eta^{3/4}},\frac{3E^{2}}{8\eta^{3/2}},0,z\right)  $ has an expansion at
$z=0$ \cite{ronveaux1995heun}:%
\begin{equation}
N\left(  3l+\frac{3}{2},-\frac{\sqrt{6}E}{2\eta^{3/4}},\frac{3E^{2}}%
{8\eta^{3/2}},0,z\right)  =\sum_{n\geq0}\frac{A_{n}}{\left(  3l+5/2\right)
_{n}}\frac{z^{n}}{n!},
\end{equation}
where the expansion coefficients is determined by the recurrence relation,%
\begin{align}
A_{0}  &  =1,\text{ \ }A_{1}=-\frac{\sqrt{6}E}{4\eta^{3/4}}\left(  3l+\frac
{5}{2}\right)  ,\nonumber\\
A_{n+2}  &  =\left[  -\frac{\sqrt{6}E}{2\eta^{3/4}}\left(  n+1\right)
-\frac{\sqrt{6}E}{4\eta^{3/4}}\left(  3l+\frac{5}{2}\right)  \right]
A_{n+1}-\left(  n+1\right)  \left(  n+3l+\frac{5}{2}\right)  \left(
\frac{3E^{2}}{8\eta^{3/2}}-3l-\frac{7}{2}-2n\right)  A_{n},
\end{align}
and $\left(  a\right)  _{n}=\Gamma\left(  a+n\right)  /\Gamma\left(  a\right)
$ is Pochhammer's symbol.

Only $y_{l}^{\left(  1\right)  }\left(  z\right)  $ satisfies the boundary
condition of the regular solution at $r=0$, so the radial eigenfunction is%
\begin{align}
u_{l}\left(  z\right)   &  =A_{l}\exp\left(  -\frac{z^{2}}{2}+\frac{\sqrt{6}%
E}{4\eta^{3/4}}z\right)  z^{3\left(  l+1\right)  /2}y_{l}^{\left(  1\right)
}\left(  z\right) \nonumber\\
&  =A_{l}\exp\left(  -\frac{z^{2}}{2}+\frac{\sqrt{6}E}{4\eta^{3/4}}z\right)
z^{3\left(  l+1\right)  /2}N\left(  3l+\frac{3}{2},-\frac{\sqrt{6}E}%
{2\eta^{3/4}},\frac{3E^{2}}{8\eta^{3/2}},0,z\right)  .
\end{align}
By Eq. (\ref{zr23}), we obtain the regular solution,%
\begin{equation}
u_{l}\left(  r\right)  =A_{l}\exp\left(  -\frac{3}{4}\eta^{1/2}r^{4/3}%
+\frac{3E}{4\eta^{1/2}}r^{2/3}\right)  \left(  \frac{\sqrt{6}}{2}\eta
^{1/4}\right)  ^{3\left(  l+1\right)  /2}r^{l+1}N\left(  3l+\frac{3}{2}%
,-\frac{\sqrt{6}E}{2\eta^{3/4}},\frac{3E^{2}}{8\eta^{3/2}},0,\frac{\sqrt{6}%
}{2}\eta^{1/4}r^{2/3}\right)  . \label{regular23}%
\end{equation}

\subsection{The irregular solution}

The irregular solution is a solution satisfying the boundary condition at
$r\rightarrow\infty$ \cite{li2016exact}. The boundary conditions for bound
states and scattering states at $r\rightarrow\infty$ are different.

The Biconfluent Heun equation (\ref{eqf23}) has two linearly independent
irregular solutions \cite{ronveaux1995heun}:%
\begin{align}
B_{l}^{+}\left(  3l+\frac{3}{2},-\frac{\sqrt{6}E}{2\eta^{3/4}},\frac{3E^{2}%
}{8\eta^{3/2}},0,z\right)   &  =z^{\frac{1}{2}\left(  \frac{3E^{2}}%
{8\eta^{3/2}}-3l-\frac{7}{2}\right)  }\sum_{n\geq0}\frac{a_{n}}{z^{n}%
},\label{f123}\\
H_{l}^{+}\left(  3l+\frac{3}{2},-\frac{\sqrt{6}E}{2\eta^{3/4}},\frac{3E^{2}%
}{8\eta^{3/2}},0,z\right)   &  =z^{-\frac{1}{2}\left(  \frac{3E^{2}}%
{8\eta^{3/2}}+3l+\frac{7}{2}\right)  }\exp\left(  -\frac{\sqrt{6}E}%
{2\eta^{3/4}}z+z^{2}\right)  \sum_{n\geq0}\frac{e_{n}}{z^{n}} \label{f223}%
\end{align}
with the expansion coefficients given by the recurrence relation
\[
a_{0}=1,\text{ \ }a_{1}=-\frac{\sqrt{6}E}{8\eta^{3/4}}\left(  \frac{3E^{2}%
}{8\eta^{3/2}}-1\right)  ,
\]%
\begin{align}
&  2\left(  n+2\right)  a_{n+2}-\frac{\sqrt{6}E}{4\eta^{3/4}}\left(
3-\frac{3E^{2}}{8\eta^{3/2}}+2n\right)  a_{n+1}\nonumber\\
&  +\left\{  \frac{1}{4}\left[  \left(  \frac{3E^{2}}{8\eta^{3/2}}\right)
^{2}-\left(  3l+\frac{3}{2}\right)  ^{2}+4\right]  -\frac{3E^{2}}{8\eta^{3/2}%
}+n\left(  n+2-\frac{3E^{2}}{8\eta^{3/2}}\right)  \right\}  a_{n}=0
\end{align}
\qquad and%
\[
e_{0}=1,\text{ \ }e_{1}=\frac{\sqrt{6}E}{8\eta^{3/4}}\left(  \frac{3E^{2}%
}{8\eta^{3/2}}+1\right)  ,
\]%
\begin{align}
&  2\left(  n+2\right)  e_{n+2}-\frac{\sqrt{6}E}{4\eta^{3/4}}\left(
3+\frac{3E^{2}}{8\eta^{3/2}}+2n\right)  e_{n+1}\nonumber\\
&  -\left\{  \frac{1}{4}\left[  \left(  \frac{3E^{2}}{8\eta^{3/2}}\right)
^{2}-\left(  3l+\frac{3}{2}\right)  ^{2}+4\right]  +\frac{3E^{2}}{8\eta^{3/2}%
}+n\left(  n+2+\frac{3E^{2}}{8\eta^{3/2}}\right)  \right\}  e_{n}=0.
\end{align}

\subsection{Eigenfunctions and eigenvalues}

To construct the solution, we first express the regular solution
(\ref{regular23}) as a linear combination of the two irregular solutions
(\ref{f123}) and (\ref{f223}).

The regular solution (\ref{regular23}), with the relation
\cite{ronveaux1995heun,li2016exact}%

\begin{align}
N\left(  3l+\frac{3}{2},-\frac{\sqrt{6}E}{2\eta^{3/4}},\frac{3E^{2}}%
{8\eta^{3/2}},0,z\right)   &  =K_{1}\left(  3l+\frac{3}{2},-\frac{\sqrt{6}%
E}{2\eta^{3/4}},\frac{3E^{2}}{8\eta^{3/2}},0\right)  B_{l}^{+}\left(
3l+\frac{3}{2},-\frac{\sqrt{6}E}{2\eta^{3/4}},\frac{3E^{2}}{8\eta^{3/2}%
},0,z\right)  \nonumber\\
&  +K_{2}\left(  3l+\frac{3}{2},-\frac{\sqrt{6}E}{2\eta^{3/4}},\frac{3E^{2}%
}{8\eta^{3/2}},0\right)  H_{l}^{+}\left(  3l+\frac{3}{2},-\frac{\sqrt{6}%
E}{2\eta^{3/4}},\frac{3E^{2}}{8\eta^{3/2}},0,z\right)
\end{align}
and the expansions (\ref{f123}) and£¨\ref{f223}), becomes%
\begin{align}
u_{l}\left(  r\right)   &  =A_{l}K_{1}\left(  3l+\frac{3}{2},-\frac{\sqrt{6}%
E}{2\eta^{3/4}},\frac{3E^{2}}{8\eta^{3/2}},0\right)  \exp\left(  -\frac{3}%
{4}\eta^{1/2}r^{4/3}+\frac{3E}{4\eta^{1/2}}r^{2/3}\right)  \nonumber\\
&  \times r^{\frac{E^{2}}{8\eta^{3/2}}-\frac{1}{6}}\sum_{n\geq0}\frac{a_{n}%
}{\left(  \frac{\sqrt{6}}{2}\eta^{1/4}r^{2/3}\right)  ^{n}}\nonumber\\
&  +A_{l}K_{2}\left(  3l+\frac{3}{2},-\frac{\sqrt{6}E}{2\eta^{3/4}}%
,\frac{3E^{2}}{8\eta^{3/2}},0\right)  \exp\left(  \frac{3}{4}\eta^{1/2}%
r^{4/3}-\frac{3E}{4\eta^{1/2}}r^{2/3}\right)  \nonumber\\
&  \times r^{-\frac{E^{2}}{8\eta^{3/2}}-\frac{1}{6}}\sum_{n\geq0}\frac{e_{n}%
}{\left(  \frac{\sqrt{6}}{2}\eta^{1/4}r^{2/3}\right)  ^{n}},\label{uexpzf}%
\end{align}
where $K_{1}\left(  3l+\frac{3}{2},-\frac{\sqrt{6}E}{2\eta^{3/4}},\frac
{3E^{2}}{8\eta^{3/2}},0\right)  $ and $K_{2}\left(  3l+\frac{3}{2}%
,-\frac{\sqrt{6}E}{2\eta^{3/4}},\frac{3E^{2}}{8\eta^{3/2}},0\right)  $ are
combination coefficients\ and $z=\frac{\sqrt{6}}{2}\eta^{1/4}r^{2/3}$.

The boundary condition of bound states, $\left.  u\left(  r\right)
\right\vert _{r\rightarrow\infty}\rightarrow0$, requires that the coefficient
of the second term must vanish since this term diverges when $r\rightarrow
\infty$, i.e.,
\begin{equation}
K_{2}\left(  3l+\frac{3}{2},-\frac{\sqrt{6}E}{2\eta^{3/4}},\frac{3E^{2}}%
{8\eta^{3/2}},0\right)  =0, \label{eigenvalue23}%
\end{equation}
where%
\begin{align}
&  K_{2}\left(  3l+\frac{3}{2},-\frac{\sqrt{6}E}{2\eta^{3/4}},\frac{3E^{2}%
}{8\eta^{3/2}},0\right) \nonumber\\
&  =\frac{\Gamma\left(  3l+\frac{5}{2}\right)  }{\Gamma\left(  \frac{3}%
{2}l+\frac{3}{4}-\frac{3E^{2}}{16\eta^{3/2}}\right)  \Gamma\left(  \frac{3}%
{2}l+\frac{7}{4}+\frac{3E^{2}}{16\eta^{3/2}}\right)  }\nonumber\\
&  \times J_{\frac{3}{2}l+\frac{7}{4}+\frac{3E^{2}}{16\eta^{3/2}}}\left(
\frac{3}{2}l+\frac{3}{4}+\frac{3E^{2}}{16\eta^{3/2}},-\frac{\sqrt{6}E}%
{2\eta^{3/4}},\frac{9}{2}\left(  l+\frac{1}{2}\right)  -\frac{3E^{2}}%
{16\eta^{3/2}},\frac{\sqrt{6}E}{4\eta^{3/4}}\left(  \frac{3E^{2}}{8\eta^{3/2}%
}-3l-\frac{3}{2}\right)  \right)
\end{align}
with%
\begin{align}
&  J_{\frac{3}{2}l+\frac{7}{4}+\frac{3E^{2}}{16\eta^{3/2}}}\left(  \frac{3}%
{2}l+\frac{3}{4}+\frac{3E^{2}}{16\eta^{3/2}},-\frac{\sqrt{6}E}{2\eta^{3/4}%
},\frac{9}{2}\left(  l+\frac{1}{2}\right)  -\frac{3E^{2}}{16\eta^{3/2}}%
,\frac{\sqrt{6}E}{4\eta^{3/4}}\left(  \frac{3E^{2}}{8\eta^{3/2}}-3l-\frac
{3}{2}\right)  \right) \nonumber\\
&  =\int_{0}^{\infty}x^{\frac{3}{2}l+\frac{3}{4}+\frac{3E^{2}}{16\eta^{3/2}}%
}e^{-x^{2}+\frac{\sqrt{6}E}{2\eta^{3/4}}x}\nonumber\\
&  \times N\left(  \frac{3}{2}\left(  l+\frac{1}{2}\right)  +\frac{3E^{2}%
}{16\eta^{3/2}},-\frac{\sqrt{6}E}{2\eta^{3/4}},\frac{9}{2}\left(  l+\frac
{1}{2}\right)  -\frac{3E^{2}}{16\eta^{3/2}},\frac{\sqrt{6}E}{4\eta^{3/4}%
}\left[  \frac{3E^{2}}{8\eta^{3/2}}-3\left(  l+\frac{1}{2}\right)  \right]
,x\right)  .
\end{align}

Eq. (\ref{eigenvalue23}) is an implicit expression of the eigenvalue.

The eigenfunction, by Eqs. (\ref{uexpzf}) and (\ref{eigenvalue23}), reads%
\begin{equation}
u_{l}\left(  r\right)  =A_{l}K_{1}\left(  3l+\frac{3}{2},-\frac{\sqrt{6}%
E}{2\eta^{3/4}},\frac{3E^{2}}{8\eta^{3/2}},0\right)  \exp\left(  -\frac{3}%
{4}\eta^{1/2}r^{4/3}+\frac{3E}{4\eta^{1/2}}r^{2/3}\right)  r^{\frac{E^{2}%
}{8\eta^{3/2}}-\frac{1}{6}}\sum_{n\geq0}\frac{a_{n}}{\left(  \frac{\sqrt{6}%
}{2}\eta^{1/4}r^{2/3}\right)  ^{n}}.
\end{equation}

\section{The exact solution of $U\left(  r\right)  =\eta r^{6}$ \label{V6}}

In this appendix, we solve an exact solution of the $r^{6}$-potential through
solving the radial equation directly. The $r^{6}$-potential has only bound states.

The radial equation of the $r^{6}$-potential,
\begin{equation}
U\left(  r\right)  =\eta r^{6},
\end{equation}
reads%
\begin{equation}
\frac{d^{2}u_{l}\left(  r\right)  }{dr^{2}}+\left[  E-\frac{l\left(
l+1\right)  }{r^{2}}-\eta r^{6}\right]  u_{l}\left(  r\right)  =0.
\end{equation}
Introducing $f_{l}\left(  z\right)  $ by
\[
u_{l}\left(  z\right)  =A_{l}e^{-z^{2}/2}z^{\left(  l+1\right)  /2}%
f_{l}\left(  z\right)
\]
with%
\begin{equation}
z=\frac{1}{\sqrt{2}}\eta^{1/4}r^{2}, \label{zr6}%
\end{equation}
where $A_{l}$ is a constant, we arrive at an equation of $f_{l}\left(
z\right)  $
\begin{equation}
zf_{l}^{\prime\prime}\left(  z\right)  +\left(  -2{z}^{2}+l+\frac{3}%
{2}\right)  f_{l}^{\prime}\left(  z\right)  +\left[  -\left(  l+\frac{5}%
{2}\right)  z-\frac{E}{2\sqrt{2}\eta^{1/4}}\right]  f_{l}\left(  z\right)  =0.
\label{eqf6}%
\end{equation}
This is a Biconfluent Heun equation \cite{ronveaux1995heun}.

The choice of the boundary condition has been discussed in Ref.
\cite{li2016exact}.

\subsection{The regular solution}

First solve the regular solution which satisfies the boundary condition at
$r=0$ \cite{li2016exact}.

The Biconfluent Heun equation (\ref{eqf6}) has two linearly independent
solutions \cite{ronveaux1995heun}%

\begin{align}
y_{l}^{\left(  1\right)  }\left(  z\right)   &  =N\left(  l+\frac{1}%
{2},0,0,\frac{E}{\sqrt{2}\eta^{1/4}},z\right)  ,\\
y_{l}^{\left(  2\right)  }\left(  z\right)   &  =cN\left(  l+\frac{1}%
{2},0,0,\frac{E}{\sqrt{2}\eta^{1/4}},z\right)  \ln z+\sum_{n\geq0}%
d_{n}z^{n-l-1/2},
\end{align}
where%
\begin{equation}
c=\frac{1}{l+1/2}\left[  \frac{E}{2\sqrt{2}\eta^{1/4}}d_{l-1/2}-d_{l-3/2}%
\left(  \frac{3}{2}-l\right)  \right]
\end{equation}
is a constant and the expansion coefficient $d_{\nu}$ is determined by the
following recurrence relation,
\begin{align}
&  d_{-1}=0,\text{ \ }d_{0}=1,\nonumber\\
&  \left(  \nu+2\right)  \left(  \nu+\frac{3}{2}-l\right)  d_{\nu+2}-\frac
{E}{2\sqrt{2}\eta^{1/4}}d_{\nu+1}+\left(  -2v-\frac{3}{2}+l\right)  d_{\nu}=0,
\end{align}
where $N(\alpha,\beta,\gamma,\delta,z)$ is the biconfluent Heun function
\cite{ronveaux1995heun,slavyanov2000special,li2016exact}. The biconfluent Heun
function $N\left(  l+\frac{1}{2},0,0,\frac{E}{\sqrt{2}\eta^{1/4}},z\right)  $
has an expansion at $z=0$:
\begin{equation}
N\left(  l+\frac{1}{2},0,0,\frac{E}{\sqrt{2}\eta^{1/4}},z\right)  =\sum
_{n\geq0}\frac{A_{n}}{\left(  l+3/2\right)  _{n}}\frac{z^{n}}{n!},
\end{equation}
where the expansion coefficient is determined by the following recurrence
relation,%
\begin{align}
A_{0}  &  =1,\text{ \ }A_{1}=\frac{E}{2\sqrt{2}\eta^{1/4}},\nonumber\\
A_{n+2}  &  =\frac{E}{2\sqrt{2}\eta^{1/4}}A_{n+1}+\left(  n+1\right)  \left(
n+l+\frac{3}{2}\right)  \left(  l+\frac{5}{2}+2n\right)  A_{n},
\end{align}
and $\left(  a\right)  _{n}=\Gamma\left(  a+n\right)  /\Gamma\left(  a\right)
$ is Pochhammer's symbol.

Only $y_{l}^{\left(  1\right)  }\left(  z\right)  $ satisfies the boundary
condition of the regular solution at $r=0$, so the radial eigenfunction is%
\begin{align}
u_{l}\left(  r\right)   &  =A_{l}e^{-z^{2}/2}z^{\left(  l+1\right)  /2}%
y_{l}^{\left(  1\right)  }\left(  z\right) \nonumber\\
&  =A_{l}e^{-z^{2}/2}z^{\left(  l+1\right)  /2}N\left(  l+\frac{1}%
{2},0,0,\frac{E}{\sqrt{2}\eta^{1/4}},z\right)  .
\end{align}
By Eq. (\ref{zr6}), we obtain the regular solution,%
\begin{equation}
u_{l}\left(  r\right)  =A_{l}\left(  \frac{\eta^{1/2}}{2}\right)  ^{\left(
l+1\right)  /4}\exp\left(  -\frac{1}{4}\eta^{1/2}r^{4}\right)  r^{l+1}N\left(
l+\frac{1}{2},0,0,\frac{E}{\sqrt{2}\eta^{1/4}},\frac{1}{\sqrt{2}}\eta
^{1/4}r^{2}\right)  , \label{regular6}%
\end{equation}

\subsection{The irregular solution}

Now we solve the irregular solution which satisfies the boundary condition at
$r\rightarrow\infty$ \cite{li2016exact}.

The Biconfluent Heun equation (\ref{eqf6}) has two linearly independent
irregular solutions \cite{ronveaux1995heun}:%

\begin{align}
B_{l}^{+}\left(  l+\frac{1}{2},0,0,\frac{E}{\sqrt{2}\eta^{1/4}},z\right)   &
=z^{-\frac{1}{2}\left(  l+5/2\right)  }\sum_{n\geq0}\frac{a_{n}}{z^{n}%
},\label{f16}\\
H_{l}^{+}\left(  l+\frac{1}{2},0,0,\frac{E}{\sqrt{2}\eta^{1/4}},z\right)   &
=z^{-\frac{1}{2}\left(  l+5/2\right)  }e^{z^{2}}\sum_{n\geq0}\frac{e_{n}%
}{z^{n}} \label{f26}%
\end{align}
with the expansion coefficients given by the recurrence relation%
\[
a_{0}=1,\text{ \ }a_{1}=\frac{1}{2}\left(  q-\frac{1}{2}s\right)  ,
\]%
\begin{equation}
2\left(  n+2\right)  a_{n+2}-\frac{E}{2\sqrt{2}\eta^{1/4}}a_{n+1}+\left[
1-\left(  \frac{l+1/2}{2}\right)  ^{2}+n\left(  n+2\right)  \right]  a_{n}=0,
\end{equation}
and%
\[
e_{0}=1,\text{ \ }e_{1}=-\frac{1}{2}\left(  q+\frac{1}{2}s\right)  ,
\]%
\begin{equation}
2\left(  n+2\right)  e_{n+2}+\frac{E}{2\sqrt{2}\eta^{1/4}}e_{n+1}-\left[
1-\left(  \frac{l+1/2}{2}\right)  ^{2}+n\left(  n+2\right)  \right]  e_{n}=0.
\end{equation}

\subsection{Eigenfunctions and eigenvalues}

As above, first express the regular solution (\ref{regular6}) as a linear
combination of the two irregular solutions (\ref{f16}) and (\ref{f26}).

The regular solution (\ref{regular6}), with the relation
\cite{ronveaux1995heun}%
\begin{align}
N\left(  l+\frac{1}{2},0,0,\frac{E}{\sqrt{2}\eta^{1/4}},z\right)   &
=K_{1}\left(  l+\frac{1}{2},0,0,\frac{E}{\sqrt{2}\eta^{1/4}}\right)  B_{l}%
^{+}\left(  l+\frac{1}{2},0,0,\frac{E}{\sqrt{2}\eta^{1/4}},z\right)
\nonumber\\
&  +K_{2}\left(  l+\frac{1}{2},0,0,\frac{E}{\sqrt{2}\eta^{1/4}}\right)
H_{l}^{+}\left(  l+\frac{1}{2},0,0,\frac{E}{\sqrt{2}\eta^{1/4}},z\right)
\end{align}
and the expansions (\ref{f16}) and (\ref{f26}), becomes%
\begin{align}
u_{l}\left(  r\right)   &  =A_{l}K_{1}\left(  l+\frac{1}{2},0,0,\frac{E}%
{\sqrt{2}\eta^{1/4}}\right)  \left(  \frac{\sqrt{2}}{\eta^{1/4}}\right)
^{3/4}r^{-3/2}\exp\left(  -\frac{1}{4}\eta^{1/2}r^{4}\right)  \sum_{n\geq
0}\frac{a_{n}}{\left(  \frac{\eta^{1/4}}{\sqrt{2}}r^{2}\right)  ^{n}%
}\nonumber\\
&  +A_{l}K_{2}\left(  l+\frac{1}{2},0,0,\frac{E}{\sqrt{2}\eta^{1/4}}\right)
\left(  \frac{\sqrt{2}}{\eta^{1/4}}\right)  ^{3/4}r^{-3/2}\exp\left(  \frac
{1}{4}\eta^{1/2}r^{4}\right)  \sum_{n\geq0}\frac{e_{n}}{\left(  \frac
{\eta^{1/4}}{\sqrt{2}}r^{2}\right)  ^{n}}, \label{uexpzf6}%
\end{align}
where $K_{1}\left(  l+\frac{1}{2},0,0,\frac{E}{\sqrt{2}\eta^{1/4}}\right)  $
and $K_{2}\left(  l+\frac{1}{2},0,0,\frac{E}{\sqrt{2}\eta^{1/4}}\right)  $ are
combination coefficients\ and $z=\frac{1}{\sqrt{2}}\eta^{1/4}r^{2}$ is used.

Only the first term satisfies the boundary condition of bound states $\left.
u\left(  r\right)  \right\vert _{r\rightarrow\infty}\rightarrow0$, so the
coefficient of the second term must vanish, i.e.,%
\begin{equation}
K_{2}\left(  l+\frac{1}{2},0,0,\frac{E}{\sqrt{2}\eta^{1/4}}\right)  =0,
\label{eigenvalue6}%
\end{equation}
where
\begin{equation}
K_{2}\left(  l+\frac{1}{2},0,0,\frac{E}{\sqrt{2}\eta^{1/4}}\right)
=\frac{\Gamma\left(  l+\frac{3}{2}\right)  }{\Gamma\left(  \frac{l}{2}%
+\frac{1}{4}\right)  \Gamma\left(  \frac{l}{2}+\frac{5}{4}\right)  }%
J_{\frac{l}{2}+\frac{5}{4}}\left(  \frac{l}{2}+\frac{1}{4},0,\frac{3}%
{2}\left(  l+\frac{1}{2}\right)  ,\frac{E}{\sqrt{2}\eta^{1/4}}\right)
\end{equation}
with%
\begin{equation}
J_{\frac{l}{2}+\frac{5}{4}}\left(  \frac{l}{2}+\frac{1}{4},0,\frac{3}%
{2}\left(  l+\frac{1}{2}\right)  ,\frac{E}{\sqrt{2}\eta^{1/4}}\right)
=\int_{0}^{\infty}x^{\frac{l}{2}+\frac{1}{4}}e^{-x^{2}}N\left(  \frac{l}%
{2}+\frac{1}{4},0,\frac{3}{2}\left(  l+\frac{1}{2}\right)  ,\frac{E}{\sqrt
{2}\eta^{1/4}},x\right)  .
\end{equation}

The eigenvalue is the zero of $K_{2}\left(  l+\frac{1}{2},0,0,\frac{\sqrt{2}%
E}{2\eta^{1/4}}\right)  $ and Eq. (\ref{eigenvalue6}) is an implicit
expression of the eigenvalue.

The eigenfunction, by Eqs. (\ref{uexpzf}) and (\ref{eigenvalue23}), reads%
\begin{equation}
u_{l}\left(  r\right)  =A_{l}K_{1}\left(  l+\frac{1}{2},0,0,\frac{E}{\sqrt
{2}\eta^{1/4}}\right)  \left(  \frac{\sqrt{2}}{\eta^{1/4}}\right)
^{3/4}r^{-3/2}\exp\left(  -\frac{1}{4}\eta^{1/2}r^{4}\right)  \sum_{n\geq
0}\frac{a_{n}}{\left(  \frac{\eta^{1/4}}{\sqrt{2}}r^{2}\right)  ^{n}}.
\end{equation}

\section{The exact solution of $U\left(  r\right)  =\xi r^{2}+\frac{\mu}{r}$
\label{V2m1}}

In this appendix, we provide an exact solution of the eigenproblem of the
potential
\begin{equation}
U\left(  r\right)  =\xi r^{2}+\frac{\mu}{r} \label{Ur2rm1}%
\end{equation}
by solving the radial equation directly. This potential has only bound states.

The radial equation of the potential (\ref{Ur2rm1}) reads%
\begin{equation}
\frac{d^{2}u_{l}\left(  r\right)  }{dr^{2}}+\left[  E-\frac{l\left(
l+1\right)  }{r^{2}}-\xi r^{2}-\frac{\mu}{r}\right]  u_{l}\left(  r\right)
=0. \label{radialeq2o1}%
\end{equation}
Using the variable substitution%
\begin{equation}
z=i\xi^{1/4}r \label{zr2o1}%
\end{equation}
and introducing $f_{l}\left(  z\right)  $ by%
\begin{equation}
u_{l}\left(  z\right)  =A_{l}\exp\left(  -\frac{z^{2}}{2}\right)  z^{l+1}%
f_{l}\left(  z\right)
\end{equation}
with $A_{l}$ a constant, we convert the radial equation (\ref{radialeq2o1})
into an equation of $f_{l}\left(  z\right)  $:%
\begin{equation}
zf_{l}^{\prime\prime}\left(  z\right)  +\left(  -2{z}^{2}+2l+2\right)
f_{l}^{\prime}\left(  z\right)  +\left[  \left(  \frac{-E}{\sqrt{\xi}%
}-2l-3\right)  z+\frac{i\mu}{\xi^{1/4}}\right]  f_{l}\left(  z\right)  =0.
\label{eqf2o1}%
\end{equation}
This is a Biconfluent Heun equation \cite{ronveaux1995heun}.

The choice of the boundary condition has been discussed in Ref.
\cite{li2016exact}.

\subsection{The regular solution}

The regular solution is a solution satisfying the boundary condition at $r=0$
\cite{li2016exact}. The regular solution at $r=0$\ should satisfy the boundary
condition $\lim_{r\rightarrow0}u_{l}\left(  r\right)  /r^{l+1}=1$. In this
section, we provide the regular solution of Eq. (\ref{eqf2o1}).

The Biconfluent Heun equation (\ref{eqf2o1}) has two linearly independent
solutions \cite{ronveaux1995heun}%
\begin{align}
y_{l}^{\left(  1\right)  }\left(  z\right)   &  =N\left(  2l+1,0,\frac
{-E}{\sqrt{\xi}},\frac{-2i\mu}{\xi^{1/4}},z\right)  ,\label{yl1z2o1}\\
y_{l}^{\left(  2\right)  }\left(  z\right)   &  =cN\left(  2l+1,0,\frac
{-E}{\sqrt{\xi}},\frac{-2i\mu}{\xi^{1/4}},z\right)  \ln z+\sum_{n\geq0}%
d_{n}z^{n-2l-1}, \label{yl2z2o1}%
\end{align}
where%
\begin{equation}
c=\frac{1}{2l+1}\left[  \frac{-i\mu}{\xi^{1/4}}d_{2l}-\left(  \frac{-E}%
{\sqrt{\xi}}-2l-3\right)  d_{2l-1}\right]
\end{equation}
is a constant with the coefficient $d_{\nu}$ given by the following recurrence
relation,%
\begin{align}
&  d_{-1}=0,\text{ \ }d_{0}=1,\nonumber\\
&  \left(  \nu+2\right)  \left(  \nu+1-2l\right)  d_{\nu+2}+\frac{i\mu}%
{\xi^{1/4}}d_{\nu+1}+\left(  \frac{-E}{\sqrt{\xi}}-2v-1+2l\right)  d_{\nu}=0
\end{align}
and $N(\alpha,\beta,\gamma,\delta,z)$ is the biconfluent Heun function
\cite{ronveaux1995heun,slavyanov2000special,li2016exact}.

The biconfluent Heun function $N\left(  2l+1,0,\frac{-E}{\sqrt{\xi}}%
,\frac{-2i\mu}{\xi^{1/4}},z\right)  $ has an expansion at $z=0$
\cite{ronveaux1995heun}:%
\begin{equation}
N\left(  2l+1,0,\frac{-E}{\sqrt{\xi}},\frac{-2i\mu}{\xi^{1/4}},z\right)
=\sum_{n\geq0}\frac{A_{n}}{\left(  2l+2\right)  _{n}}\frac{z^{n}}{n!},
\end{equation}
where the expansion coefficients is determined by the recurrence relation,%
\begin{align}
A_{0}  &  =1,\text{ \ }A_{1}=-\frac{i\mu}{\xi^{1/4}},\nonumber\\
A_{n+2}  &  =-\frac{i\mu}{\xi^{1/4}}A_{n+1}-\left(  n+1\right)  \left(
n+2l+2\right)  \left(  \frac{-E}{\sqrt{\xi}}-3-2l-2n\right)  A_{n},
\end{align}
and $\left(  a\right)  _{n}=\Gamma\left(  a+n\right)  /\Gamma\left(  a\right)
$ is Pochhammer's symbol.

Only $y_{l}^{\left(  1\right)  }\left(  z\right)  $ satisfies the boundary
condition for the regular solution at $r=0$, so the radial eigenfunction reads%
\begin{align}
u_{l}\left(  z\right)   &  =A_{l}\exp\left(  -\frac{z^{2}}{2}\right)
z^{l+1}y_{l}^{\left(  1\right)  }\left(  z\right) \nonumber\\
&  =A_{l}\exp\left(  -\frac{z^{2}}{2}\right)  z^{l+1}N\left(  2l+1,0,\frac
{-E}{\sqrt{\xi}},\frac{-2i\mu}{\xi^{1/4}},z\right)  .
\end{align}
By Eq. (\ref{zr2o1}), we obtain the regular solution,%
\begin{equation}
u_{l}\left(  r\right)  =A_{l}\exp\left(  \frac{\xi^{1/2}r^{2}}{2}\right)
r^{l+1}N\left(  2l+1,0,\frac{-E}{\sqrt{\xi}},\frac{-2i\mu}{\xi^{1/4}}%
,i\xi^{1/4}r\right)  . \label{regular2o1}%
\end{equation}

\subsection{The irregular solution}

The irregular solution is a solution satisfying the boundary condition at
$r\rightarrow\infty$ \cite{li2016exact}.

The Biconfluent Heun equation (\ref{eqf2o1}) has two linearly independent
irregular solutions \cite{ronveaux1995heun}:%
\begin{align}
B_{l}^{+}\left(  2l+1,0,\frac{-E}{\sqrt{\xi}},\frac{-2i\mu}{\xi^{1/4}%
},z\right)   &  =e^{z^{2}}B_{l}^{+}\left(  2l+1,0,\frac{E}{\sqrt{\xi}}%
,\frac{2\mu}{\xi^{1/4}},-iz\right)  =e^{z^{2}}\left(  -iz\right)  ^{\frac
{1}{2}\left(  \frac{E}{\sqrt{\xi}}-2l-3\right)  }\sum_{n\geq0}\frac{a_{n}%
}{\left(  -iz\right)  ^{n}},\label{f12o1}\\
H_{l}^{+}\left(  2l+1,0,\frac{-E}{\sqrt{\xi}},\frac{-2i\mu}{\xi^{1/4}%
},z\right)   &  =e^{z^{2}}H_{l}^{+}\left(  2l+1,0,\frac{E}{\sqrt{\xi}}%
,\frac{2\mu}{\xi^{1/4}},-iz\right)  =\left(  -iz\right)  ^{-\frac{1}{2}\left(
\frac{E}{\sqrt{\xi}}+2l+3\right)  }\sum_{n\geq0}\frac{e_{n}}{\left(
-iz\right)  ^{n}} \label{f22o1}%
\end{align}
with the expansion coefficients given by the recurrence relation
\[
a_{0}=1,\text{ \ }a_{1}=\frac{\mu}{2\xi^{1/4}},
\]%
\begin{equation}
2\left(  n+2\right)  a_{n+2}-\frac{\mu}{\xi^{1/4}}a_{n+1}+\left[  \frac{E^{2}%
}{4\xi}-\frac{\left(  2l+1\right)  ^{2}}{4}+1-\frac{E}{\sqrt{\xi}}+n\left(
n+2-\frac{E}{\sqrt{\xi}}\right)  \right]  a_{n}=0
\end{equation}
and%
\[
e_{0}=1,\text{ \ }e_{1}=-\frac{\mu}{2\xi^{1/4}},
\]%
\begin{equation}
2\left(  n+2\right)  e_{n+2}+\frac{\mu}{\xi^{1/4}}e_{n+1}-\left[  \frac{E^{2}%
}{4\xi}-\frac{\left(  2l+1\right)  ^{2}}{4}+1+\frac{E}{\sqrt{\xi}}+n\left(
n+2+\frac{E}{\sqrt{\xi}}\right)  \right]  e_{n}=0.
\end{equation}

\subsection{Eigenfunctions and eigenvalues}

To construct the solution, we first express the regular solution
(\ref{regular2o1}) as a linear combination of the two irregular solutions
(\ref{f12o1}) and (\ref{f22o1}).

The regular solution (\ref{regular2o1}), with the relation
\cite{ronveaux1995heun,li2016exact}%

\begin{align}
N\left(  2l+1,0,\frac{-E}{\sqrt{\xi}},-\frac{2i\mu}{\xi^{1/4}},z\right)   &
=K_{1}\left(  2l+1,0,\frac{-E}{\sqrt{\xi}},-\frac{2i\mu}{\xi^{1/4}}\right)
B_{l}^{+}\left(  2l+1,0,\frac{-E}{\sqrt{\xi}},\frac{2i\mu}{\xi^{1/4}%
},z\right)  \nonumber\\
&  +K_{2}\left(  2l+1,0,\frac{-E}{\sqrt{\xi}},-\frac{2i\mu}{\xi^{1/4}}\right)
H_{l}^{+}\left(  2l+1,0,\frac{-E}{\sqrt{\xi}},\frac{2i\mu}{\xi^{1/4}%
},z\right)
\end{align}
and the expansions (\ref{f12o1}) and (\ref{f22o1}), becomes%
\begin{align}
u_{l}\left(  r\right)   &  =A_{l}K_{1}\left(  2l+1,0,\frac{-E}{\sqrt{\xi}%
},-\frac{2i\mu}{\xi^{1/4}}\right)  \exp\left(  -\frac{\xi^{1/2}r^{2}}%
{2}\right)  \exp\left(  \left(  \frac{E}{2\sqrt{\xi}}-\frac{1}{2}\right)
\ln\left(  \xi^{1/4}r\right)  \right)  \sum_{n\geq0}\frac{a_{n}}{\left(
\xi^{1/4}r\right)  ^{n}}\nonumber\\
&  +A_{l}K_{2}\left(  2l+1,0,\frac{-E}{\sqrt{\xi}},-\frac{2i\mu}{\xi^{1/4}%
}\right)  \exp\left(  \frac{\xi^{1/2}r^{2}}{2}\right)  \exp\left(  \left(
-\frac{E}{2\sqrt{\xi}}-\frac{1}{2}\right)  \ln\left(  \xi^{1/4}r\right)
\right)  \sum_{n\geq0}\frac{e_{n}}{\left(  \xi^{1/4}r\right)  ^{n}%
},\label{uexpzf2o1}%
\end{align}
where $K_{1}\left(  2l+1,0,\frac{-E}{\sqrt{\xi}},-\frac{2i\mu}{\xi^{1/4}%
}\right)  $ and $K_{2}\left(  2l+1,0,\frac{-E}{\sqrt{\xi}},-\frac{2i\mu}%
{\xi^{1/4}}\right)  $ are combination coefficients\ and $z=i\xi^{1/4}r$.

The boundary condition of bound states, $\left.  u\left(  r\right)
\right\vert _{r\rightarrow\infty}\rightarrow0$, requires that the coefficient
of the second term must vanish since this term diverges when $r\rightarrow
\infty$, i.e.,
\begin{equation}
K_{2}\left(  2l+1,0,\frac{-E}{\sqrt{\xi}},-\frac{2i\mu}{\xi^{1/4}}\right)  =0,
\label{eigenvalue2o1}%
\end{equation}
where%
\begin{align}
K_{2}\left(  2l+1,0,\frac{-E}{\sqrt{\xi}},-\frac{2i\mu}{\xi^{1/4}}\right)   &
=\frac{\Gamma\left(  2l+2\right)  }{\Gamma\left(  l+\frac{1}{2}+\frac
{E}{2\sqrt{\xi}}\right)  \Gamma\left(  l+\frac{3}{2}-\frac{E}{2\sqrt{\xi}%
}\right)  }\nonumber\\
&  \times J_{l+\frac{3}{2}-\frac{E}{2\sqrt{\xi}}}\left(  l+\frac{1}{2}%
-\frac{E}{2\sqrt{\xi}},0,3l+\frac{3}{2}+\frac{E}{2\sqrt{\xi}},-\frac{2i\mu
}{\xi^{1/4}}\right)
\end{align}
with%
\begin{align}
&  J_{l+\frac{3}{2}-\frac{E}{2\sqrt{\xi}}}\left(  l+\frac{1}{2}-\frac
{E}{2\sqrt{\xi}},0,3l+\frac{3}{2}+\frac{E}{2\sqrt{\xi}},-\frac{2i\mu}%
{\xi^{1/4}}\right) \nonumber\\
&  =\int_{0}^{\infty}x^{\frac{1}{2}\left(  2l+1-\frac{E}{\sqrt{\xi}}\right)
}e^{-x^{2}}N\left(  l+\frac{1}{2}-\frac{E}{2\sqrt{\xi}},0,3l+\frac{3}{2}%
+\frac{E}{2\sqrt{\xi}},-\frac{2i\mu}{\xi^{1/4}},x\right)  dx.
\end{align}

Eq. (\ref{eigenvalue2o1}) is an implicit expression of the eigenvalue.

The eigenfunction, by Eqs. (\ref{uexpzf2o1}) and (\ref{eigenvalue2o1}), reads%
\begin{equation}
u_{l}\left(  r\right)  =A_{l}K_{1}\left(  2l+1,0,\frac{-E}{\sqrt{\xi}}%
,-\frac{2i\mu}{\xi^{1/4}}\right)  \exp\left(  -\frac{\xi^{1/2}r^{2}}%
{2}\right)  \exp\left(  \left(  \frac{E}{2\sqrt{\xi}}-\frac{1}{2}\right)
\ln\left(  \xi^{1/4}r\right)  \right)  \sum_{n\geq0}\frac{a_{n}}{\left(
\xi^{1/4}r\right)  ^{n}}.
\end{equation}

\section{The exact solution of $U\left(  r\right)  =\frac{\xi}{r}+\frac{\mu
}{r^{3/2}}$ \label{Vm1m32}}

In this appendix, we provide an exact solution of the eigenproblem of the
potential
\begin{equation}
U\left(  r\right)  =\frac{\xi}{r}+\frac{\mu}{r^{3/2}}%
\end{equation}
by solving the radial equation directly. This potential has both bound states
and scattering states.

The radial equation reads%
\begin{equation}
\frac{d^{2}u_{l}\left(  r\right)  }{dr^{2}}+\left[  E-\frac{l\left(
l+1\right)  }{r^{2}}-\frac{\xi}{r}-\frac{\mu}{r^{3/2}}\right]  u_{l}\left(
r\right)  =0. \label{radialeqo132}%
\end{equation}
Using the variable substitution%
\begin{equation}
z=\left(  -E\right)  ^{1/4}\left(  -2r\right)  ^{1/2} \label{zro132}%
\end{equation}
and introducing $f_{l}\left(  z\right)  $ by%
\begin{equation}
u_{l}\left(  z\right)  =A_{l}\exp\left(  -\frac{z^{2}}{2}\right)
z^{2l+2}f_{l}\left(  z\right)
\end{equation}
with $A_{l}$ a constant, we convert the radial equation (\ref{radialeqo132})
into an equation of $f_{l}\left(  z\right)  $:%
\begin{equation}
zf_{l}^{\prime\prime}\left(  z\right)  +\left(  -2z^{2}+4l+3\right)
f_{l}^{\prime}\left(  z\right)  +\left[  \left(  \frac{2\xi}{\sqrt{-E}%
}-4l-4\right)  z-\frac{i2\sqrt{2}\mu}{\left(  -E\right)  ^{1/4}}\right]
f_{l}\left(  z\right)  =0. \label{eqfo132}%
\end{equation}
This is a Biconfluent Heun equation \cite{ronveaux1995heun}.

The choice of the boundary condition has been discussed in Ref.
\cite{li2016exact}.

\subsection{The regular solution}

The regular solution is a solution satisfying the boundary condition at $r=0$
\cite{li2016exact}. The regular solution at $r=0$\ should satisfy the boundary
condition $\lim_{r\rightarrow0}u_{l}\left(  r\right)  /r^{l+1}=1$ for both
bound states and scattering states. In this section, we provide the regular
solution of the radial equation, Eq. (\ref{eqfo132}).

The Biconfluent Heun equation (\ref{eqfo132}) has two linearly independent
solutions \cite{ronveaux1995heun}%
\begin{align}
y_{l}^{\left(  1\right)  }\left(  z\right)   &  =N\left(  4l+2,0,\frac{2\xi
}{\sqrt{-E}},\frac{-i4\sqrt{2}\mu}{\left(  -E\right)  ^{1/4}},z\right)  ,\\
y_{l}^{\left(  2\right)  }\left(  z\right)   &  =cN\left(  4l+2,0,\frac{2\xi
}{\sqrt{-E}},\frac{-i4\sqrt{2}\mu}{\left(  -E\right)  ^{1/4}},z\right)  \ln
z+\sum_{n\geq0}d_{n}z^{n-4l-2},
\end{align}
where%
\begin{equation}
c=\frac{1}{4l+2}\left[  \frac{-i2\sqrt{2}\mu}{\left(  -E\right)  ^{1/4}%
}d_{4l+1}-d_{4l}\left(  \frac{2\xi}{\sqrt{-E}}-4l\right)  \right]
\end{equation}
is a constant with the coefficient $d_{\nu}$ given by the following recurrence
relation,%
\begin{align}
&  d_{-1}=0,\text{ \ }d_{0}=1,\nonumber\\
&  \left(  v+2\right)  \left(  v-4l\right)  d_{v+2}+\frac{i2\sqrt{2}\mu
}{\left(  -E\right)  ^{1/4}}d_{v+1}+\left(  \frac{2\xi}{\sqrt{-E}%
}-2v+4l\right)  d_{v}=0
\end{align}
and $N(\alpha,\beta,\gamma,\delta,z)$ is the biconfluent Heun function
\cite{ronveaux1995heun,slavyanov2000special,li2016exact}.

The biconfluent Heun function $N\left(  4l+2,0,\frac{2\xi}{\sqrt{-E}}%
,\frac{-i4\sqrt{2}\mu}{\left(  -E\right)  ^{1/4}},z\right)  $ has an expansion
at $z=0$ \cite{ronveaux1995heun}:%
\begin{equation}
N\left(  4l+2,0,\frac{2\xi}{\sqrt{-E}},\frac{-i4\sqrt{2}\mu}{\left(
-E\right)  ^{1/4}},z\right)  =\sum_{n\geq0}\frac{A_{n}}{\left(  4l+3\right)
_{n}}\frac{z^{n}}{n!},
\end{equation}
where the expansion coefficients is determined by the recurrence relation,%
\begin{align}
A_{0}  &  =1,\text{ \ }A_{1}=-\frac{i2\sqrt{2}\mu}{\left(  -E\right)  ^{1/4}%
},\nonumber\\
A_{n+2}  &  =-\frac{i2\sqrt{2}\mu}{\left(  -E\right)  ^{1/4}}A_{n+1}-\left(
n+1\right)  \left(  n+4l+3\right)  \left(  \frac{2\xi}{\sqrt{-E}%
}-4l-4-2n\right)  A_{n},
\end{align}
and $\left(  a\right)  _{n}=\Gamma\left(  a+n\right)  /\Gamma\left(  a\right)
$ is Pochhammer's symbol.

Only $y_{l}^{\left(  1\right)  }\left(  z\right)  $ satisfies the boundary
condition for the regular solution at $r=0$, so the radial eigenfunction reads%
\begin{align}
u_{l}\left(  z\right)   &  =A_{l}\exp\left(  -\frac{z^{2}}{2}\right)
z^{2l+2}y_{l}^{\left(  1\right)  }\left(  z\right) \nonumber\\
&  =A_{l}\exp\left(  -\frac{z^{2}}{2}\right)  z^{2l+2}N\left(  4l+2,0,\frac
{2\xi}{\sqrt{-E}},\frac{-i4\sqrt{2}\mu}{\left(  -E\right)  ^{1/4}},z\right)  .
\end{align}
By Eq. (\ref{zro132}), we obtain the regular solution,%
\begin{equation}
u_{l}\left(  r\right)  =A_{l}\exp\left(  \left(  -E\right)  ^{1/2}r\right)
r^{l+1}N\left(  4l+2,0,\frac{2\xi}{\sqrt{-E}},\frac{-i4\sqrt{2}\mu}{\left(
-E\right)  ^{1/4}},\left(  -E\right)  ^{1/4}\left(  -2r\right)  ^{1/2}\right)
. \label{regularo132}%
\end{equation}

\subsection{The irregular solution}

The irregular solution is a solution satisfying the boundary condition at
$r\rightarrow\infty$ \cite{li2016exact}. The boundary conditions for bound
states and scattering states at $r\rightarrow\infty$ are different.

The Biconfluent Heun equation (\ref{eqfo132}) has two linearly independent
irregular solutions \cite{ronveaux1995heun}:%
\begin{align}
B_{l}^{+}\left(  4l+2,0,\frac{2\xi}{\sqrt{-E}},\frac{-i4\sqrt{2}\mu}{\left(
-E\right)  ^{1/4}},z\right)   &  =e^{z^{2}}B_{l}^{+}\left(  4l+2,0,-\frac
{2\xi}{\sqrt{-E}},\frac{4\sqrt{2}\mu}{\left(  -E\right)  ^{1/4}},-iz\right)
\nonumber\\
&  =e^{z^{2}}\left(  -iz\right)  ^{-\frac{\xi}{\sqrt{-E}}-2l-2}\sum_{n\geq
0}\frac{a_{n}}{\left(  -iz\right)  ^{n}}, \label{f1o132}%
\end{align}%
\begin{align}
H_{l}^{+}\left(  4l+2,0,\frac{2\xi}{\sqrt{-E}},\frac{-i4\sqrt{2}\mu}{\left(
-E\right)  ^{1/4}},z\right)   &  =e^{z^{2}}H_{l}^{+}\left(  4l+2,0,-\frac
{2\xi}{\sqrt{-E}},\frac{4\sqrt{2}\mu}{\left(  -E\right)  ^{1/4}},-iz\right)
\nonumber\\
&  =\left(  -iz\right)  ^{\frac{\xi}{\sqrt{-E}}-2l-2}\sum_{n\geq0}\frac{e_{n}%
}{\left(  -iz\right)  ^{n}} \label{f2o132}%
\end{align}
with the expansion coefficients given by the recurrence relation
\[
a_{0}=1,\text{ \ }a_{1}=\frac{\sqrt{2}\mu}{\left(  -E\right)  ^{1/4}},
\]%
\begin{equation}
2\left(  n+2\right)  a_{n+2}-\frac{2\sqrt{2}\mu}{\left(  -E\right)  ^{1/4}%
}a_{n+1}+\left[  -\frac{\xi^{2}}{E}-\left(  2l+1\right)  ^{2}+1+\frac{2\xi
}{\sqrt{-E}}+n\left(  n+2+\frac{2\xi}{\sqrt{-E}}\right)  \right]  a_{n}=0
\end{equation}
and%
\[
e_{0}=1,\text{ \ }e_{1}=-\frac{\sqrt{2}\mu}{\left(  -E\right)  ^{1/4}},
\]%
\begin{equation}
2\left(  n+2\right)  e_{n+2}+\frac{2\sqrt{2}\mu}{\left(  -E\right)  ^{1/4}%
}e_{n+1}-\left[  -\frac{\xi^{2}}{E}-\left(  2l+1\right)  ^{2}+1-\frac{2\xi
}{\sqrt{-E}}+n\left(  n+2-\frac{2\xi}{\sqrt{-E}}\right)  \right]  e_{n}=0.
\end{equation}

\subsection{Bound states and scattering states}

\subsubsection{The bound state}

To construct the solution, we first express the regular solution
(\ref{regularo132}) as a linear combination of the two irregular solutions
(\ref{f1o132}) and (\ref{f2o132}).

The regular solution (\ref{regularo132}), with the relation
\cite{ronveaux1995heun,li2016exact}%

\begin{align}
&  N\left(  4l+2,0,\frac{2\xi}{\sqrt{-E}},\frac{-i4\sqrt{2}\mu}{\left(
-E\right)  ^{1/4}},z\right) \nonumber\\
&  =K_{1}\left(  4l+2,0,\frac{2\xi}{\sqrt{-E}},\frac{-i4\sqrt{2}\mu}{\left(
-E\right)  ^{1/4}}\right)  B_{l}^{+}\left(  4l+2,0,\frac{2\xi}{\sqrt{-E}%
},\frac{-i4\sqrt{2}\mu}{\left(  -E\right)  ^{1/4}},z\right) \nonumber\\
&  +K_{2}\left(  4l+2,0,\frac{2\xi}{\sqrt{-E}},\frac{-i4\sqrt{2}\mu}{\left(
-E\right)  ^{1/4}}\right)  H_{l}^{+}\left(  4l+2,0,\frac{2\xi}{\sqrt{-E}%
},\frac{-i4\sqrt{2}\mu}{\left(  -E\right)  ^{1/4}},z\right)
\end{align}
and the expansions (\ref{f1o132}) and (\ref{f2o132}), become%
\begin{align}
u_{l}\left(  r\right)   &  =A_{l}K_{1}\left(  4l+2,0,\frac{2\xi}{\sqrt{-E}%
},\frac{-i4\sqrt{2}\mu}{\left(  -E\right)  ^{1/4}}\right)  \exp\left(
-\left(  -E\right)  ^{1/2}r-\frac{\xi}{2\sqrt{-E}}\ln\left(  2\left(
-E\right)  ^{1/2}r\right)  \right) \nonumber\\
&  \times\sum_{n\geq0}\frac{a_{n}}{\left[  \left(  -E\right)  ^{1/4}\left(
2r\right)  ^{1/2}\right]  ^{n}}\nonumber\\
&  +A_{l}K_{2}\left(  4l+2,0,\frac{2\xi}{\sqrt{-E}},\frac{-i4\sqrt{2}\mu
}{\left(  -E\right)  ^{1/4}}\right)  \exp\left(  \left(  -E\right)
^{1/2}r+\frac{\xi}{2\sqrt{-E}}\ln\left(  2\left(  -E\right)  ^{1/2}r\right)
\right) \nonumber\\
&  \times\sum_{n\geq0}\frac{e_{n}}{\left[  \left(  -E\right)  ^{1/4}\left(
2r\right)  ^{1/2}\right]  ^{n}}, \label{uexpzfo132}%
\end{align}
where $K_{1}\left(  4l+2,0,\frac{2\xi}{\sqrt{-E}},\frac{-i4\sqrt{2}\mu
}{\left(  -E\right)  ^{1/4}}\right)  $ and $K_{2}\left(  4l+2,0,\frac{2\xi
}{\sqrt{-E}},\frac{-i4\sqrt{2}\mu}{\left(  -E\right)  ^{1/4}}\right)  $ are
combination coefficients\ and $z=\left(  -E\right)  ^{1/4}\left(  -2r\right)
^{1/2}$.

The boundary condition of bound states, $\left.  u\left(  r\right)
\right\vert _{r\rightarrow\infty}\rightarrow0$, requires that the coefficient
of the second term must vanish since this term diverges when $r\rightarrow
\infty$, i.e.,
\begin{equation}
K_{2}\left(  4l+2,0,\frac{2\xi}{\sqrt{-E}},\frac{-i4\sqrt{2}\mu}{\left(
-E\right)  ^{1/4}}\right)  =0, \label{eigenvalueo132}%
\end{equation}
where%
\begin{align}
K_{2}\left(  4l+2,0,\frac{2\xi}{\sqrt{-E}},\frac{-i4\sqrt{2}\mu}{\left(
-E\right)  ^{1/4}}\right)   &  =\frac{\Gamma\left(  4l+3\right)  }%
{\Gamma\left(  2l+1-\frac{\xi}{\sqrt{-E}}\right)  \Gamma\left(  2l+2+\frac
{\xi}{\sqrt{-E}}\right)  }\nonumber\\
&  \times J_{2l+2+\xi/\sqrt{-E}}\left(  2l+1+\frac{\xi}{\sqrt{-E}%
},0,6l+3-\frac{\xi}{\sqrt{-E}},\frac{-i4\sqrt{2}\mu}{\left(  -E\right)
^{1/4}}\right)
\end{align}
with%
\begin{align}
&  J_{2l+2+\xi/\sqrt{-E}}\left(  2l+1+\frac{\xi}{\sqrt{-E}},0,6l+3-\frac{\xi
}{\sqrt{-E}},\frac{-i4\sqrt{2}\mu}{\left(  -E\right)  ^{1/4}}\right)
\nonumber\\
&  =\int_{0}^{\infty}x^{2l+1+\frac{\xi}{\sqrt{-E}}}e^{-x^{2}}N\left(
2l+1+\frac{\xi}{\sqrt{-E}},0,6l+3-\frac{\xi}{\sqrt{-E}},\frac{-i4\sqrt{2}\mu
}{\left(  -E\right)  ^{1/4}},x\right)  dx.
\end{align}

Eq. (\ref{eigenvalueo132}) is an implicit expression of the eigenvalue.

The eigenfunction, by Eqs. (\ref{uexpzfo132}) and (\ref{eigenvalueo132}),
reads%
\begin{align}
u_{l}\left(  r\right)   &  =A_{l}K_{1}\left(  4l+2,0,\frac{2\xi}{\sqrt{-E}%
},\frac{-i4\sqrt{2}\mu}{\left(  -E\right)  ^{1/4}}\right) \nonumber\\
&  \times\exp\left(  -\left(  -E\right)  ^{1/2}r-\frac{\xi}{2\sqrt{-E}}%
\ln\left(  2\left(  -E\right)  ^{1/2}r\right)  \right)  \sum_{n\geq0}%
\frac{a_{n}}{\left[  \left(  -E\right)  ^{1/4}\left(  2r\right)
^{1/2}\right]  ^{n}}.
\end{align}

\subsubsection{The scattering state}

For scattering states, $E>0$, introduce%
\begin{equation}
E=k^{2}.
\end{equation}

The singularity of the $S$-matrix on the positive imaginary axis corresponds
to the eigenvalues of bound states \cite{joachain1975quantum}, so the zero of
$K_{2}\left(  4l+2,0,i\frac{2\xi}{k},\frac{4\sqrt{-2i}\mu}{\sqrt{k}}\right)  $
on the positive imaginary is just the singularity of the $S$-matrix.
Considering that the $S$-matrix is unitary, i.e.,%
\begin{equation}
S_{l}=e^{2i\delta_{l}}, \label{SMo132}%
\end{equation}
we have%
\begin{equation}
S_{l}\left(  k\right)  =\frac{K_{2}^{\ast}\left(  4l+2,0,i\frac{2\xi}{k}%
,\frac{4\sqrt{-2i}\mu}{\sqrt{k}}\right)  }{K_{2}\left(  4l+2,0,i\frac{2\xi}%
{k},\frac{4\sqrt{-2i}\mu}{\sqrt{k}}\right)  }=\frac{K_{2}\left(
4l+2,0,-i\frac{2\xi}{k},\frac{4\sqrt{2i}\mu}{\sqrt{k}}\right)  }{K_{2}\left(
4l+2,0,i\frac{2\xi}{k},\frac{4\sqrt{-2i}\mu}{\sqrt{k}}\right)  }.
\label{SM1o132}%
\end{equation}

The scattering wave function can be constructed with the help of the
$S$-matrix. The scattering wave function can be expressed as a linear
combination of the radially ingoing wave $u_{in}\left(  r\right)  $ and the
radially outgoing wave $u_{out}\left(  r\right)  $, which are conjugate to
each other, i.e., \cite{joachain1975quantum}%
\begin{equation}
u_{l}\left(  r\right)  =A_{l}\left[  \left(  -1\right)  ^{l+1}u_{in}\left(
r\right)  +S_{l}\left(  k\right)  u_{out}\left(  r\right)  \right]  .
\label{inouto132}%
\end{equation}
From Eq. (\ref{uexpzfo132}), we have%
\begin{align}
u_{in}\left(  r\right)   &  =\exp\left(  -ikr+i\frac{\xi}{2k}\ln2kr\right)
\sum_{n\geq0}\frac{e_{n}}{\left(  -2ikr\right)  ^{n/2}},\nonumber\\
u_{out}\left(  r\right)   &  =\exp\left(  ikr-i\frac{\xi}{2k}\ln2kr\right)
\sum_{n\geq0}\frac{e_{n}^{\ast}}{\left(  2ikr\right)  ^{n/2}}.
\end{align}
Then by Eq. (\ref{SM1o132}), we obtain the scattering wave function,%
\begin{align}
u_{l}\left(  r\right)   &  =A_{l}\left[  \left(  -1\right)  ^{l+1}\exp\left(
-ikr+i\frac{\xi}{2k}\ln2kr\right)  \sum_{n\geq0}\frac{e_{n}}{\left(
-2ikr\right)  ^{n/2}}\right. \nonumber\\
&  +\left.  \frac{K_{2}\left(  4l+2,0,-i\frac{2\xi}{k},\frac{4\sqrt{2i}\mu
}{\sqrt{k}}\right)  }{K_{2}\left(  4l+2,0,i\frac{2\xi}{k},\frac{4\sqrt{-2i}%
\mu}{\sqrt{k}}\right)  }\exp\left(  ikr-i\frac{\xi}{2k}\ln2kr\right)
\sum_{n\geq0}\frac{e_{n}^{\ast}}{\left(  2ikr\right)  ^{n/2}}\right]  .
\end{align}
Taking $r\rightarrow\infty$, we have
\begin{align}
&  u_{l}\left(  r\right)  \overset{r\rightarrow\infty}{\sim}A_{l}\left[
\left(  -1\right)  ^{l+1}\exp\left(  -ikr+i\frac{\xi}{2k}\ln2kr\right)
+\frac{K_{2}\left(  4l+2,0,-i\frac{2\xi}{k},\frac{4\sqrt{2i}\mu}{\sqrt{k}%
}\right)  }{K_{2}\left(  4l+2,0,i\frac{2\xi}{k},\frac{4\sqrt{-2i}\mu}{\sqrt
{k}}\right)  }\exp\left(  ikr-i\frac{\xi}{2k}\ln2kr\right)  \right]
\nonumber\\
&  =A_{l}e^{i\delta_{l}}\sin\left(  kr-\frac{\xi}{2k}\ln2kr+\delta_{l}%
-\frac{l\pi}{2}\right)  .
\end{align}
By Eqs. (\ref{SMo132}) and (\ref{SM1o132}), we obtain the scattering phase
shift%
\begin{equation}
\delta_{l}=-\arg K_{2}\left(  4l+2,0,i\frac{2\xi}{k},\frac{4\sqrt{-2i}\mu
}{\sqrt{k}}\right)  .
\end{equation}

\section{The exact solution of $U\left(  r\right)  =\frac{\xi}{r^{2/3}}+\mu
r^{2/3}$\label{Vm2323}}

In this appendix, we provide an exact solution of the eigenproblem of the
potential
\begin{equation}
U\left(  r\right)  =\frac{\xi}{r^{2/3}}+\mu r^{2/3}%
\end{equation}
by solving the radial equation directly. This potential has only bound states.

The radial equation reads%
\begin{equation}
\frac{d^{2}u_{l}\left(  r\right)  }{dr^{2}}+\left[  E-\frac{l\left(
l+1\right)  }{r^{2}}-\frac{\xi}{r^{2/3}}-\mu r^{2/3}\right]  u_{l}\left(
r\right)  =0. \label{radialeq23o23}%
\end{equation}
Using the variable substitution%
\begin{equation}
z=i\frac{\sqrt{6}}{2}\mu^{1/4}r^{2/3} \label{zr23o23}%
\end{equation}
and introducing $f_{l}\left(  z\right)  $ by%
\begin{equation}
u_{l}\left(  z\right)  =A_{l}\exp\left(  -\frac{z^{2}}{2}-\frac{\beta}%
{2}z\right)  z^{3\left(  l+1\right)  /2}f_{l}\left(  z\right)
\end{equation}
with $A_{l}$ a constant, we convert the radial equation (\ref{radialeq23o23})
into an equation of $f_{l}\left(  z\right)  $:%
\begin{equation}
zf_{l}^{\prime\prime}\left(  z\right)  +\left(  3l+\frac{5}{2}+\frac{i\sqrt
{6}\varepsilon}{2\mu^{3/4}}z-2z^{2}\right)  f_{l}^{\prime}\left(  z\right)
+\left[  \left(  {\frac{3\xi}{2{\mu}^{1/2}}}-{\frac{3E^{2}}{8{\mu}^{3/2}}%
}-3l-\frac{7}{2}\right)  z+\frac{i\sqrt{6}E}{4{\mu}^{3/4}}\left(  3l+\frac
{5}{2}\right)  \right]  f_{l}\left(  z\right)  =0. \label{eqf23o23}%
\end{equation}
This is a Biconfluent Heun equation \cite{ronveaux1995heun}.

The choice of the boundary condition has been discussed in Ref.
\cite{li2016exact}.

\subsection{The regular solution}

The regular solution is a solution satisfying the boundary condition at $r=0$
\cite{li2016exact}. The regular solution at $r=0$\ should satisfy the boundary
condition $\lim_{r\rightarrow0}u_{l}\left(  r\right)  /r^{l+1}=1$. In this
section, we provide the regular solution of Eq. (\ref{eqf23o23}).

The Biconfluent Heun equation (\ref{eqf23o23}) has two linearly independent
solutions \cite{ronveaux1995heun}%
\begin{align}
y_{l}^{\left(  1\right)  }\left(  z\right)   &  =N\left(  3l+\frac{3}%
{2},{\frac{-i\sqrt{6}E}{2{\mu}^{3/4}}},{\frac{3\xi}{2{\mu}^{1/2}}}%
-{\frac{3E^{2}}{8{\mu}^{3/2}}},0,z\right)  ,\\
y_{l}^{\left(  2\right)  }\left(  z\right)   &  =cN\left(  3l+\frac{3}%
{2},{\frac{-i\sqrt{6}E}{2{\mu}^{3/4}}},{\frac{3\xi}{2{\mu}^{1/2}}}%
-{\frac{3E^{2}}{8{\mu}^{3/2}}},0,z\right)  \ln z+\sum_{n\geq0}d_{n}%
z^{n-3l-\frac{3}{2}},
\end{align}
where%
\begin{equation}
c=\frac{1}{\alpha}\left[  \frac{-i\sqrt{6}E}{4{\mu}^{3/4}}\left(  3l+\frac
{1}{2}\right)  d_{3l+1/2}-d_{3l-1/2}\left(  {\frac{3\xi}{2{\mu}^{1/2}}}%
-{\frac{3E^{2}}{8{\mu}^{3/2}}}+\frac{1}{2}-3l\right)  \right]
\end{equation}
is a constant with the coefficient $d_{\nu}$ given by the following recurrence
relation,%
\begin{align}
&  d_{-1}=0,\text{ \ }d_{0}=1,\nonumber\\
&  \left(  v+2\right)  \left(  v+\frac{1}{2}-3l\right)  d_{v+2}+{\frac
{i\sqrt{6}E}{4{\mu}^{3/4}}}\left(  2v+\frac{3}{2}-3l\right)  d_{v+1}+\left(
{\frac{3\xi}{2{\mu}^{1/2}}}-{\frac{3E^{2}}{8{\mu}^{3/2}}}-2v+3l-\frac{3}%
{2}\right)  d_{v}=0
\end{align}
and $N(\alpha,\beta,\gamma,\delta,z)$ is the biconfluent Heun function
\cite{ronveaux1995heun,slavyanov2000special,li2016exact}.

The biconfluent Heun function $N\left(  3l+\frac{3}{2},{\frac{-i\sqrt{6}%
E}{2{\mu}^{3/4}}},{\frac{3\xi}{2{\mu}^{1/2}}}-{\frac{3E^{2}}{8{\mu}^{3/2}}%
},0,z\right)  $ has an expansion at $z=0$ \cite{ronveaux1995heun}:%
\begin{equation}
N\left(  3l+\frac{3}{2},{\frac{-i\sqrt{6}E}{2{\mu}^{3/4}}},{\frac{3\xi}{2{\mu
}^{1/2}}}-{\frac{3E^{2}}{8{\mu}^{3/2}}},0,z\right)  =\sum_{n\geq0}\frac{A_{n}%
}{\left(  3l+5/2\right)  _{n}}\frac{z^{n}}{n!},
\end{equation}
where the expansion coefficients is determined by the recurrence relation,%
\begin{align}
A_{0}  &  =1,\text{ \ }A_{1}=\frac{-i\sqrt{6}E}{4{\mu}^{3/4}}\left(
3l+\frac{5}{2}\right)  ,\nonumber\\
A_{n+2}  &  =\left[  \frac{-i\sqrt{6}E}{2{\mu}^{3/4}}\left(  n+1\right)
+\frac{-i\sqrt{6}E}{4{\mu}^{3/4}}\left(  3l+\frac{5}{2}\right)  \right]
A_{n+1}\nonumber\\
&  -\left(  n+1\right)  \left(  n+3l+\frac{5}{2}\right)  \left(  {\frac{3\xi
}{2{\mu}^{1/2}}}-{\frac{3E^{2}}{8{\mu}^{3/2}}}-3l-\frac{7}{2}-2n\right)
A_{n},
\end{align}
and $\left(  a\right)  _{n}=\Gamma\left(  a+n\right)  /\Gamma\left(  a\right)
$ is Pochhammer's symbol.

Only $y_{l}^{\left(  1\right)  }\left(  z\right)  $ satisfies the boundary
condition for the regular solution at $r=0$, so the radial eigenfunction reads%
\begin{align}
u_{l}\left(  z\right)   &  =A_{l}\exp\left(  -\frac{z^{2}}{2}-\frac{\beta}%
{2}z\right)  z^{2\left(  l+1\right)  /3}y_{l}^{\left(  1\right)  }\left(
z\right) \nonumber\\
&  =A_{l}\exp\left(  -\frac{z^{2}}{2}-\frac{\beta}{2}z\right)  z^{2\left(
l+1\right)  /3}N\left(  3l+\frac{3}{2},{\frac{-i\sqrt{6}E}{2{\mu}^{3/4}}%
},{\frac{3\xi}{2{\mu}^{1/2}}}-{\frac{3E^{2}}{8{\mu}^{3/2}}},0,z\right)  .
\end{align}
By Eq. (\ref{zr23o23}), we obtain the regular solution,%
\begin{equation}
u_{l}\left(  r\right)  =A_{l}\exp\left(  \frac{3}{4}\mu^{1/2}r^{4/3}-\frac
{3E}{4}{\mu}^{-1/2}r^{2/3}\right)  r^{l+1}N\left(  3l+\frac{3}{2}%
,{\frac{-i\sqrt{6}E}{2{\mu}^{3/4}}},{\frac{3\xi}{2{\mu}^{1/2}}}-{\frac{3E^{2}%
}{8{\mu}^{3/2}}},0,i\frac{\sqrt{6}}{2}\mu^{1/4}{r}^{2/3}\right)  .
\label{regular23o23}%
\end{equation}

\subsection{The irregular solution}

The irregular solution is a solution satisfying the boundary condition at
$r\rightarrow\infty$ \cite{li2016exact}.

The Biconfluent Heun equation (\ref{eqf23o23}) has two linearly independent
irregular solutions \cite{ronveaux1995heun}:%
\begin{align}
&  B_{l}^{+}\left(  3l+\frac{3}{2},{\frac{-i\sqrt{6}E}{2{\mu}^{3/4}}}%
,{\frac{3\xi}{2{\mu}^{1/2}}}-{\frac{3E^{2}}{8{\mu}^{3/2}}},0,z\right)
\nonumber\\
&  =\exp\left(  {\frac{-i\sqrt{6}E}{2{\mu}^{3/4}}}z+z^{2}\right)  B_{l}%
^{+}\left(  3l+\frac{3}{2},-{\frac{\sqrt{6}E}{2{\mu}^{3/4}}},\frac{3E^{2}%
}{8{\mu}^{3/2}}-{\frac{3\xi}{2{\mu}^{1/2}}},0,-iz\right) \nonumber\\
&  =\exp\left(  {\frac{-i\sqrt{6}E}{2{\mu}^{3/4}}}z+z^{2}\right)  \left(
-iz\right)  ^{\frac{1}{2}\left(  \frac{3E^{2}}{8{\mu}^{3/2}}-{\frac{3\xi
}{2{\mu}^{1/2}}-}3l-\frac{7}{2}\right)  }\sum_{n\geq0}\frac{a_{n}}{\left(
-iz\right)  ^{n}}, \label{f123o23}%
\end{align}%
\begin{align}
&  H_{l}^{+}\left(  3l+\frac{3}{2},{\frac{-i\sqrt{6}E}{2{\mu}^{3/4}}}%
,{\frac{3\xi}{2{\mu}^{1/2}}}-{\frac{3E^{2}}{8{\mu}^{3/2}}},0,z\right)
\nonumber\\
&  =\exp\left(  {\frac{-i\sqrt{6}E}{2{\mu}^{3/4}}}z+z^{2}\right)  H_{l}%
^{+}\left(  3l+\frac{3}{2},-{\frac{\sqrt{6}E}{2{\mu}^{3/4}}},\frac{3E^{2}%
}{8{\mu}^{3/2}}-{\frac{3\xi}{2{\mu}^{1/2}}},0,-iz\right) \nonumber\\
&  =\left(  -iz\right)  ^{-\frac{1}{2}\left(  \frac{3E^{2}}{8{\mu}^{3/2}%
}-{\frac{3\xi}{2{\mu}^{1/2}}}+3l+\frac{7}{2}\right)  }\sum_{n\geq0}\frac
{e_{n}}{\left(  -iz\right)  ^{n}} \label{f223o23}%
\end{align}
with the expansion coefficients given by the recurrence relation
\[
a_{0}=1,\text{ \ }a_{1}=-{\frac{\sqrt{6}E}{8{\mu}^{3/4}}}\left(  \frac{3E^{2}%
}{8{\mu}^{3/2}}-{\frac{3\xi}{2{\mu}^{1/2}}}-1\right)  ,
\]%
\begin{align}
&  2\left(  n+2\right)  a_{n+2}-{\frac{\sqrt{6}E}{2{\mu}^{3/4}}}\left(
\frac{3}{2}-\frac{3E^{2}}{16{\mu}^{3/2}}+{\frac{3\xi}{4{\mu}^{1/2}}}+n\right)
a_{n+1}\nonumber\\
&  +\left[  \left(  \frac{3E^{2}}{16{\mu}^{3/2}}-{\frac{3\xi}{4{\mu}^{1/2}}%
}\right)  ^{2}-\frac{1}{4}\left(  3l+\frac{3}{2}\right)  ^{2}+1-\frac{3E^{2}%
}{8{\mu}^{3/2}}+{\frac{3\xi}{2{\mu}^{1/2}}}+n\left(  n+2-\frac{3E^{2}}{8{\mu
}^{3/2}}+{\frac{3\xi}{2{\mu}^{1/2}}}\right)  \right]  a_{n}=0
\end{align}
and%
\[
e_{0}=1,\text{ \ }e_{1}={\frac{\sqrt{6}E}{8{\mu}^{3/4}}}\left(  \frac{3E^{2}%
}{8{\mu}^{3/2}}-{\frac{3\xi}{2{\mu}^{1/2}}}+1\right)  ,
\]%
\begin{align}
&  2\left(  n+2\right)  e_{n+2}+\beta\left(  \frac{3}{2}+\frac{3E^{2}}{16{\mu
}^{3/2}}-\frac{3\xi}{4{\mu}^{1/2}}+n\right)  e_{n+1}\nonumber\\
&  -\left[  \left(  \frac{3E^{2}}{16{\mu}^{3/2}}-{\frac{3\xi}{4{\mu}^{1/2}}%
}\right)  ^{2}-\frac{1}{4}\left(  3l+\frac{3}{2}\right)  ^{2}+1+\frac{3E^{2}%
}{8{\mu}^{3/2}}-{\frac{3\xi}{2{\mu}^{1/2}}}+n\left(  n+2+\frac{3E^{2}}{8{\mu
}^{3/2}}-{\frac{3\xi}{2{\mu}^{1/2}}}\right)  \right]  e_{n}=0.
\end{align}

\subsection{Eigenfunctions and eigenvalues}

To construct the solution, we first express the regular solution
(\ref{regularo132}) as a linear combination of the two irregular solutions
(\ref{f123o23}) and (\ref{f223o23}).

The regular solution (\ref{regular23o23}), with the relation
\cite{ronveaux1995heun,li2016exact}%
\begin{align}
&  N\left(  3l+\frac{3}{2},{\frac{-i\sqrt{6}E}{2{\mu}^{3/4}}},{\frac{3\xi
}{2{\mu}^{1/2}}}-{\frac{3E^{2}}{8{\mu}^{3/2}}},0,z\right)  \nonumber\\
&  =K_{1}\left(  3l+\frac{3}{2},{\frac{-i\sqrt{6}E}{2{\mu}^{3/4}}},{\frac
{3\xi}{2{\mu}^{1/2}}}-{\frac{3E^{2}}{8{\mu}^{3/2}}},0\right)  B_{l}^{+}\left(
3l+\frac{3}{2},{\frac{-i\sqrt{6}E}{2{\mu}^{3/4}}},{\frac{3\xi}{2{\mu}^{1/2}}%
}-{\frac{3E^{2}}{8{\mu}^{3/2}}},0,z\right)  \nonumber\\
&  +K_{2}\left(  3l+\frac{3}{2},{\frac{-i\sqrt{6}E}{2{\mu}^{3/4}}},{\frac
{3\xi}{2{\mu}^{1/2}}}-{\frac{3E^{2}}{8{\mu}^{3/2}}},0\right)  H_{l}^{+}\left(
3l+\frac{3}{2},{\frac{-i\sqrt{6}E}{2{\mu}^{3/4}}},{\frac{3\xi}{2{\mu}^{1/2}}%
}-{\frac{3E^{2}}{8{\mu}^{3/2}}},0,z\right)
\end{align}
and the expansions (\ref{f123o23}) and (\ref{f223o23}), becomes%
\begin{align}
u_{l}\left(  r\right)   &  =A_{l}K_{1}\left(  3l+\frac{3}{2},{\frac{-i\sqrt
{6}E}{2{\mu}^{3/4}}},{\frac{3\xi}{2{\mu}^{1/2}}}-{\frac{3E^{2}}{8{\mu}^{3/2}}%
},0\right)  \nonumber\\
&  \times\exp\left(  -\frac{3\mu^{1/2}}{4}r^{4/3}+\frac{3E}{4{\mu}^{1/2}%
}r^{2/3}\right)  r^{\frac{E^{2}}{8{\mu}^{3/2}}-{\frac{\xi}{2{\mu}^{1/2}}%
}-\frac{1}{6}}\sum_{n\geq0}\frac{a_{n}}{\left(  \frac{\sqrt{6}}{2}\mu
^{1/4}r^{2/3}\right)  ^{n}}\nonumber\\
&  +A_{l}K_{2}\left(  3l+\frac{3}{2},{\frac{-i\sqrt{6}E}{2{\mu}^{3/4}}}%
,{\frac{3\xi}{2{\mu}^{1/2}}}-{\frac{3E^{2}}{8{\mu}^{3/2}}},0\right)
\nonumber\\
&  \times\exp\left(  \frac{3\mu^{1/2}}{4}r^{4/3}-\frac{3E}{4{\mu}^{3/4}}%
\mu^{1/4}r^{2/3}\right)  r^{-\left(  \frac{E^{2}}{8{\mu}^{3/2}}-{\frac{\xi
}{2{\mu}^{1/2}}}+\frac{1}{6}\right)  }\sum_{n\geq0}\frac{e_{n}}{\left(
\frac{\sqrt{6}}{2}\mu^{1/4}r^{2/3}\right)  ^{n}}\label{uexpzf23o23}%
\end{align}
where $K_{1}\left(  3l+\frac{3}{2},{\frac{-i\sqrt{6}E}{2{\mu}^{3/4}}}%
,{\frac{3\xi}{2{\mu}^{1/2}}}-{\frac{3E^{2}}{8{\mu}^{3/2}}},0\right)  $ and
$K_{2}\left(  3l+\frac{3}{2},{\frac{-i\sqrt{6}E}{2{\mu}^{3/4}}},{\frac{3\xi
}{2{\mu}^{1/2}}}-{\frac{3E^{2}}{8{\mu}^{3/2}}},0\right)  $ are combination
coefficients\ and $z=i\frac{\sqrt{6}}{2}\mu^{1/4}r^{2/3}$.

The boundary condition of bound states, $\left.  u\left(  r\right)
\right\vert _{r\rightarrow\infty}\rightarrow0$, requires that the coefficient
of the second term must vanish since this term diverges when $r\rightarrow
\infty$, i.e.,
\begin{equation}
K_{2}\left(  3l+\frac{3}{2},{\frac{-i\sqrt{6}E}{2{\mu}^{3/4}}},{\frac{3\xi
}{2{\mu}^{1/2}}}-{\frac{3E^{2}}{8{\mu}^{3/2}}},0\right)  =0,
\label{eigenvalue23023}%
\end{equation}
where%
\begin{align}
&  K_{2}\left(  3l+\frac{3}{2},{\frac{-i\sqrt{6}E}{2{\mu}^{3/4}}},{\frac{3\xi
}{2{\mu}^{1/2}}}-{\frac{3E^{2}}{8{\mu}^{3/2}}},0\right) \nonumber\\
&  =\frac{\Gamma\left(  3l+\frac{5}{2}\right)  }{\Gamma\left(  \frac{3l}%
{2}+\frac{3}{4}-{\frac{3\xi}{4{\mu}^{1/2}}}+{\frac{3E^{2}}{16{\mu}^{3/2}}%
}\right)  \Gamma\left(  \frac{3l}{2}+\frac{7}{4}+{\frac{3\xi}{4{\mu}^{1/2}}%
}-{\frac{3E^{2}}{16{\mu}^{3/2}}}\right)  }\nonumber\\
&  \times J_{\frac{3l}{2}+\frac{7}{4}+{\frac{3\xi}{4{\mu}^{1/2}}}%
-{\frac{3E^{2}}{16{\mu}^{3/2}}}}\left(  \frac{3l}{2}+\frac{3}{4}+{\frac{3\xi
}{4{\mu}^{1/2}}}-{\frac{3E^{2}}{16{\mu}^{3/2}}},{\frac{-i\sqrt{6}E}{2{\mu
}^{3/4}}},\frac{9l}{2}+\frac{9}{4}\right. \nonumber\\
&  -\left.  {\frac{3\xi}{4{\mu}^{1/2}}}+{\frac{3E^{2}}{16{\mu}^{3/2}}}%
,{\frac{-i\sqrt{6}E}{4{\mu}^{3/4}}}\left(  {\frac{3\xi}{2{\mu}^{1/2}}}%
-{\frac{3E^{2}}{8{\mu}^{3/2}}}-3l-\frac{3}{2}\right)  \right)
\end{align}
with%
\begin{align}
&  J_{\frac{3l}{2}+\frac{7}{4}+{\frac{3\xi}{4{\mu}^{1/2}}}-{\frac{3E^{2}%
}{16{\mu}^{3/2}}}}\left(  \frac{3l}{2}+\frac{3}{4}+{\frac{3\xi}{4{\mu}^{1/2}}%
}-{\frac{3E^{2}}{16{\mu}^{3/2}}},{\frac{-i\sqrt{6}E}{2{\mu}^{3/4}}},\frac
{9l}{2}+\frac{9}{4}\right. \nonumber\\
&  -\left.  {\frac{3\xi}{4{\mu}^{1/2}}}+{\frac{3E^{2}}{16{\mu}^{3/2}}}%
,{\frac{-i\sqrt{6}E}{4{\mu}^{3/4}}}\left(  {\frac{3\xi}{2{\mu}^{1/2}}}%
-{\frac{3E^{2}}{8{\mu}^{3/2}}}-3l-\frac{3}{2}\right)  \right) \nonumber\\
&  =\int_{0}^{\infty}x^{\frac{3l}{2}+\frac{3}{4}+{\frac{3\xi}{4{\mu}^{1/2}}%
}-{\frac{3E^{2}}{16{\mu}^{3/2}}}}e^{-x^{2}}N\left(  \frac{3l}{2}+\frac{3}%
{4}+{\frac{3\xi}{4{\mu}^{1/2}}}-{\frac{3E^{2}}{16{\mu}^{3/2}}},{\frac
{-i\sqrt{6}E}{2{\mu}^{3/4}}},\frac{9l}{2}+\frac{9}{4}\right. \nonumber\\
&  -\left.  {\frac{3\xi}{4{\mu}^{1/2}}}+{\frac{3E^{2}}{16{\mu}^{3/2}}}%
,{\frac{-i\sqrt{6}E}{4{\mu}^{3/4}}}\left(  {\frac{3\xi}{2{\mu}^{1/2}}}%
-{\frac{3E^{2}}{8{\mu}^{3/2}}}-3l-\frac{3}{2}\right)  ,x\right)  dx.
\end{align}

Eq. (\ref{eigenvalue23023}) is an implicit expression of the eigenvalue.

The eigenfunction, by Eqs. (\ref{uexpzf23o23}) and (\ref{eigenvalue23023}),
reads%
\begin{align}
u_{l}\left(  r\right)   &  =A_{l}K_{1}\left(  3l+\frac{3}{2},{\frac{-i\sqrt
{6}E}{2{\mu}^{3/4}}},{\frac{3\xi}{2{\mu}^{1/2}}}-{\frac{3E^{2}}{8{\mu}^{3/2}}%
},0\right)  \nonumber\\
&  \times\exp\left(  -\frac{3\mu^{1/2}}{4}r^{4/3}+\frac{3E}{4{\mu}^{1/2}%
}r^{2/3}\right)  r^{\frac{E^{2}}{8{\mu}^{3/2}}-{\frac{\xi}{2{\mu}^{1/2}}%
}-\frac{1}{6}}\sum_{n\geq0}\frac{a_{n}}{\left(  \frac{\sqrt{6}}{2}\mu
^{1/4}r^{2/3}\right)  ^{n}}.
\end{align}

\section{The exact solution of $U\left(  r\right)  =\frac{\xi}{\sqrt{r}}%
+\frac{\mu}{r^{3/2}}$\label{Vm12m32}}

In this appendix, we provide an exact solution of the eigenproblem of the
potential
\begin{equation}
U\left(  r\right)  =\frac{\xi}{\sqrt{r}}+\frac{\mu}{r^{3/2}}%
\end{equation}
by solving the radial equation directly. This potential has both bound states
and scattering states.

The radial equation reads%
\begin{equation}
\frac{d^{2}}{dr^{2}}u_{l}\left(  r\right)  +\left[  E-\frac{l\left(
l+1\right)  }{r^{2}}-\frac{\xi}{\sqrt{r}}-\frac{\mu}{r^{3/2}}\right]
u_{l}\left(  r\right)  =0. \label{radialeqo1232}%
\end{equation}
Using the variable substitution%
\begin{equation}
z=i\sqrt{2r}\left(  -E\right)  ^{1/4}%
\end{equation}
and introducing $f_{l}\left(  z\right)  $ by%
\begin{equation}
u_{l}\left(  z\right)  =A_{l}\exp\left(  -\frac{z^{2}}{2}-\frac{\beta}%
{2}z\right)  z^{2\left(  l+1\right)  }f_{l}\left(  z\right)
\end{equation}
with $A_{l}$ a constant, we convert the radial equation (\ref{radialeqo1232})
into an equation of $f_{l}\left(  z\right)  $:%
\begin{align}
&  zf_{l}^{\prime\prime}\left(  z\right)  +\left(  4l+3-\frac{i\sqrt{2}\xi
}{\left(  -E\right)  ^{3/4}}z-2z^{2}\right)  f_{l}^{\prime}\left(  z\right)
\nonumber\\
&  +\left[  \left(  -{\frac{{\xi}^{2}}{2\left(  -E\right)  ^{3/2}}%
}-4l-4\right)  z-\frac{1}{2}\left(  \frac{-i4\sqrt{2}\mu}{\left(  -E\right)
^{1/4}}+\frac{i\sqrt{2}\xi}{\left(  -E\right)  ^{3/4}}\left(  4l+3\right)
\right)  \right]  f_{l}\left(  z\right)  =0. \label{eqfo1232}%
\end{align}
This is a Biconfluent Heun equation \cite{ronveaux1995heun}.

The choice of the boundary condition has been discussed in Ref.
\cite{li2016exact}.

\subsection{The regular solution}

The regular solution is a solution satisfying the boundary condition at $r=0$
\cite{li2016exact}. The regular solution at $r=0$\ should satisfy the boundary
condition $\lim_{r\rightarrow0}u_{l}\left(  r\right)  /r^{l+1}=1$ for both
bound states and scattering states. In this section, we provide the regular
solution of Eq. (\ref{eqfo1232}).

The Biconfluent Heun equation (\ref{eqfo1232}) has two linearly independent
solutions \cite{ronveaux1995heun}%
\begin{align}
y_{l}^{\left(  1\right)  }\left(  z\right)   &  =N\left(  4l+2,{\frac
{i\sqrt{2}\xi}{\left(  -E\right)  ^{3/4}}},-{\frac{{\xi}^{2}}{2\left(
-E\right)  ^{3/2}}},{\frac{-i4\sqrt{2}\mu}{\left(  -E\right)  ^{1/4}}%
},z\right)  ,\\
y_{l}^{\left(  2\right)  }\left(  z\right)   &  =cN\left(  4l+2,{\frac
{i\sqrt{2}\xi}{\left(  -E\right)  ^{3/4}}},-{\frac{{\xi}^{2}}{2\left(
-E\right)  ^{3/2}}},{\frac{-i4\sqrt{2}\mu}{\left(  -E\right)  ^{1/4}}%
},z\right)  \ln z+\sum_{n\geq0}d_{n}z^{n-4l-2},
\end{align}
where%
\begin{equation}
c=\frac{1}{4l+2}\left[  d_{4l+1}\frac{1}{2}\left(  \frac{-i4\sqrt{2}\mu
}{\left(  -E\right)  ^{1/4}}+\frac{i\sqrt{2}\xi}{\left(  -E\right)  ^{3/4}%
}\left(  4l+1\right)  \right)  -d_{4l}\left(  -{\frac{{\xi}^{2}}{2\left(
-E\right)  ^{3/2}}}-4l\right)  \right]
\end{equation}
is a constant with the coefficient $d_{\nu}$ given by the following recurrence
relation,%
\begin{align}
&  d_{-1}=0,\text{ \ }d_{0}=1,\nonumber\\
&  \left(  v+2\right)  \left(  v-4l\right)  d_{v+2}-\frac{1}{2}\left[
\frac{-i4\sqrt{2}\mu}{\left(  -E\right)  ^{1/4}}+\frac{i\sqrt{2}\xi}{\left(
-E\right)  ^{3/4}}\left(  2v+1-4l\right)  \right]  d_{v+1}+\left[
-{\frac{{\xi}^{2}}{2\left(  -E\right)  ^{3/2}}}-2v+4l\right]  d_{v}=0
\end{align}
and $N(\alpha,\beta,\gamma,\delta,z)$ is the biconfluent Heun function
\cite{ronveaux1995heun,slavyanov2000special,li2016exact}.

The biconfluent Heun function $N\left(  4l+2,{\frac{i\sqrt{2}\xi}{\left(
-E\right)  ^{3/4}}},-{\frac{{\xi}^{2}}{2\left(  -E\right)  ^{3/2}}}%
,{\frac{-i4\sqrt{2}\mu}{\left(  -E\right)  ^{1/4}}},z\right)  $ has an
expansion at $z=0$ \cite{ronveaux1995heun}:%
\begin{equation}
N\left(  4l+2,{\frac{i\sqrt{2}\xi}{\left(  -E\right)  ^{3/4}}},-{\frac{{\xi
}^{2}}{2\left(  -E\right)  ^{3/2}}},{\frac{-i4\sqrt{2}\mu}{\left(  -E\right)
^{1/4}}},z\right)  =\sum_{n\geq0}\frac{A_{n}}{\left(  4l+3\right)  _{n}}%
\frac{z^{n}}{n!},
\end{equation}
where the expansion coefficients is determined by the recurrence relation,%
\begin{align}
A_{0}  &  =1,\text{ \ }A_{1}=\frac{1}{2}\left[  {\frac{-i4\sqrt{2}\mu}{\left(
-E\right)  ^{1/4}}}+\frac{i\sqrt{2}\xi}{\left(  -E\right)  ^{3/4}}\left(
4l+3\right)  \right]  ,\nonumber\\
A_{n+2}  &  =\left\{  \frac{i\sqrt{2}\xi}{\left(  -E\right)  ^{3/4}}\left(
n+1\right)  +\frac{1}{2}\left[  \frac{-i4\sqrt{2}\mu}{\left(  -E\right)
^{1/4}}+{\frac{i\sqrt{2}\xi}{\left(  -E\right)  ^{3/4}}}\left(  4l+3\right)
\right]  \right\}  A_{n+1}\nonumber\\
&  -\left(  n+1\right)  \left(  n+4l+3\right)  \left[  -{\frac{{\xi}^{2}%
}{2\left(  -E\right)  ^{3/2}}}-4-4l-2n\right]  A_{n},
\end{align}
and $\left(  a\right)  _{n}=\Gamma\left(  a+n\right)  /\Gamma\left(  a\right)
$ is Pochhammer's symbol.

Only $y_{l}^{\left(  1\right)  }\left(  z\right)  $ satisfies the boundary
condition for the regular solution at $r=0$, so the radial eigenfunction reads%
\begin{align}
u_{l}\left(  z\right)   &  =A_{l}\left(  -\frac{z^{2}}{2}-\frac{\beta}%
{2}z\right)  z^{2\left(  l+1\right)  }y_{l}^{\left(  1\right)  }\left(
z\right) \nonumber\\
&  =A_{l}\left(  -\frac{z^{2}}{2}-\frac{\beta}{2}z\right)  z^{2\left(
l+1\right)  }N\left(  4l+2,{\frac{i\sqrt{2}\xi}{\left(  -E\right)  ^{3/4}}%
},-{\frac{{\xi}^{2}}{2\left(  -E\right)  ^{3/2}}},{\frac{-i4\sqrt{2}\mu
}{\left(  -E\right)  ^{1/4}}},z\right)  .
\end{align}
By Eq. (\ref{zr2o1}), we obtain the regular solution,%
\begin{equation}
u_{l}\left(  r\right)  =A_{l}\left[  \left(  -E\right)  ^{1/2}r+\frac{\xi
}{\left(  -E\right)  ^{1/2}}\sqrt{r}\right]  r^{l+1}N\left(  4l+2,{\frac
{i\sqrt{2}\xi}{\left(  -E\right)  ^{3/4}}},-{\frac{{\xi}^{2}}{2\left(
-E\right)  ^{3/2}}},{\frac{-i4\sqrt{2}\mu}{\left(  -E\right)  ^{1/4}}}%
,i\sqrt{2r}\left(  -E\right)  ^{1/4}\right)  . \label{regularo1232}%
\end{equation}

\subsection{Irregular solution}

The irregular solution is a solution satisfying the boundary condition at
$r\rightarrow\infty$ \cite{li2016exact}. The boundary conditions for bound
states and scattering states at $r\rightarrow\infty$ are different.

The Biconfluent Heun equation (\ref{eqfo1232}) has two linearly independent
irregular solutions \cite{ronveaux1995heun}:%
\begin{align}
&  B_{l}^{+}\left(  4l+2,{\frac{i\sqrt{2}\xi}{\left(  -E\right)  ^{3/4}}%
},-{\frac{{\xi}^{2}}{2\left(  -E\right)  ^{3/2}}},{\frac{-i4\sqrt{2}\mu
}{\left(  -E\right)  ^{1/4}}},z\right) \nonumber\\
&  =\exp\left(  \frac{i\sqrt{2}\xi}{\left(  -E\right)  ^{3/4}}z+z^{2}\right)
B_{l}^{+}\left(  4l+2,{\frac{\sqrt{2}\xi}{\left(  -E\right)  ^{3/4}}}%
,{\frac{{\xi}^{2}}{2\left(  -E\right)  ^{3/2}}},{\frac{4\sqrt{2}\mu}{\left(
-E\right)  ^{1/4}}},-iz\right) \nonumber\\
&  =\exp\left(  \frac{i\sqrt{2}\xi}{\left(  -E\right)  ^{3/4}}z+z^{2}\right)
\left(  -iz\right)  ^{\frac{1}{2}\left(  {\frac{{\xi}^{2}}{2\left(  -E\right)
^{3/2}}}-4l-4\right)  }\sum_{n\geq0}\frac{a_{n}}{\left(  -iz\right)  ^{n}},
\label{f1o1232}%
\end{align}%
\begin{align}
&  H_{l}^{+}\left(  4l+2,{\frac{i\sqrt{2}\xi}{\left(  -E\right)  ^{3/4}}%
},-{\frac{{\xi}^{2}}{2\left(  -E\right)  ^{3/2}}},{\frac{-i4\sqrt{2}\mu
}{\left(  -E\right)  ^{1/4}}},z\right) \nonumber\\
&  =\exp\left(  \frac{i\sqrt{2}\xi}{\left(  -E\right)  ^{3/4}}z+z^{2}\right)
H_{l}^{+}\left(  4l+2,{\frac{\sqrt{2}\xi}{\left(  -E\right)  ^{3/4}}}%
,{\frac{{\xi}^{2}}{2\left(  -E\right)  ^{3/2}}},{\frac{4\sqrt{2}\mu}{\left(
-E\right)  ^{1/4}}},-iz\right) \nonumber\\
&  =\left(  -iz\right)  ^{-\frac{1}{2}\left(  \frac{{\xi}^{2}}{2\left(
-E\right)  ^{3/2}}+4l+4\right)  }\sum_{n\geq0}\frac{e_{n}}{\left(  -iz\right)
^{n}} \label{f2o1232}%
\end{align}
with the expansion coefficients given by the recurrence relation
\[
a_{0}=1,\text{ \ }a_{1}=\frac{1}{4}\left[  {\frac{4\sqrt{2}\mu}{\left(
-E\right)  ^{1/4}}}+\frac{\sqrt{2}\xi}{\left(  -E\right)  ^{3/4}}\left(
\frac{{\xi}^{2}}{2\left(  -E\right)  ^{3/2}}-1\right)  \right]  ,
\]%
\begin{align}
&  2\left(  n+2\right)  a_{n+2}+\left[  \frac{\sqrt{2}\xi}{\left(  -E\right)
^{3/4}}\left(  \frac{3}{2}-\frac{{\xi}^{2}}{4\left(  -E\right)  ^{3/2}%
}+n\right)  -{\frac{2\sqrt{2}\mu}{\left(  -E\right)  ^{1/4}}}\right]
a_{n+1}\nonumber\\
&  +\left[  \frac{{\xi}^{4}}{16\left(  -E\right)  ^{3}}-\frac{\left(
4l+2\right)  ^{2}}{4}+1-\frac{{\xi}^{2}}{2\left(  -E\right)  ^{3/2}}+n\left(
n+2-\frac{{\xi}^{2}}{2\left(  -E\right)  ^{3/2}}\right)  \right]  a_{n}=0
\end{align}
and%
\[
e_{0}=1,\text{ \ }e_{1}=-\frac{1}{4}\left[  {\frac{4\sqrt{2}\mu}{\left(
-E\right)  ^{1/4}}}+\frac{\sqrt{2}\xi}{\left(  -E\right)  ^{3/4}}\left(
\frac{{\xi}^{2}}{2\left(  -E\right)  ^{3/2}}+1\right)  \right]  ,
\]%
\begin{align}
&  2\left(  n+2\right)  e_{n+2}+\left[  \beta\left(  \frac{3}{2}+\frac{{\xi
}^{2}}{4\left(  -E\right)  ^{3/2}}+n\right)  +\frac{\delta}{2}\right]
e_{n+1}\nonumber\\
&  -\left[  \frac{{\xi}^{4}}{16\left(  -E\right)  ^{3}}-\frac{\left(
4l+2\right)  ^{2}}{4}+1+\frac{{\xi}^{2}}{2\left(  -E\right)  ^{3/2}}+n\left(
n+2+\frac{{\xi}^{2}}{2\left(  -E\right)  ^{3/2}}\right)  \right]  e_{n}=0.
\end{align}

\subsection{Bound states and scattering states}

\subsubsection{The bound state}

To construct the solution, we first express the regular solution
(\ref{regularo1232}) as a linear combination of the two irregular solutions
(\ref{f1o1232}) and (\ref{f2o1232}).

The regular solution (\ref{regularo1232}), with the relation
\cite{ronveaux1995heun,li2016exact}%

\begin{align}
&  N\left(  4l+2,{\frac{i\sqrt{2}\xi}{\left(  -E\right)  ^{3/4}}},-{\frac
{{\xi}^{2}}{2\left(  -E\right)  ^{3/2}}},{\frac{-i4\sqrt{2}\mu}{\left(
-E\right)  ^{1/4}}},z\right) \nonumber\\
&  =K_{1}\left(  4l+2,{\frac{i\sqrt{2}\xi}{\left(  -E\right)  ^{3/4}}}%
,-{\frac{{\xi}^{2}}{2\left(  -E\right)  ^{3/2}}},{\frac{-i4\sqrt{2}\mu
}{\left(  -E\right)  ^{1/4}}}\right)  B_{l}^{+}\left(  4l+2,{\frac{i\sqrt
{2}\xi}{\left(  -E\right)  ^{3/4}}},-{\frac{{\xi}^{2}}{2\left(  -E\right)
^{3/2}}},{\frac{-i4\sqrt{2}\mu}{\left(  -E\right)  ^{1/4}}},z\right)
\nonumber\\
&  +K_{2}\left(  4l+2,{\frac{i\sqrt{2}\xi}{\left(  -E\right)  ^{3/4}}}%
,-{\frac{{\xi}^{2}}{2\left(  -E\right)  ^{3/2}}},{\frac{-i4\sqrt{2}\mu
}{\left(  -E\right)  ^{1/4}}}\right)  H_{l}^{+}\left(  4l+2,{\frac{i\sqrt
{2}\xi}{\left(  -E\right)  ^{3/4}}},-{\frac{{\xi}^{2}}{2\left(  -E\right)
^{3/2}}},{\frac{-i4\sqrt{2}\mu}{\left(  -E\right)  ^{1/4}}},z\right)
\end{align}
and the expansions (\ref{f1o1232}) and (\ref{f2o1232}), become%
\begin{align}
u_{l}\left(  r\right)   &  =A_{l}K_{1}\left(  4l+2,{\frac{i\sqrt{2}\xi
}{\left(  -E\right)  ^{3/4}}},-{\frac{{\xi}^{2}}{2\left(  -E\right)  ^{3/2}}%
},{\frac{-i4\sqrt{2}\mu}{\left(  -E\right)  ^{1/4}}}\right) \nonumber\\
&  \times\exp\left(  -\left(  -E\right)  ^{1/2}r-\frac{\xi}{\left(  -E\right)
^{1/2}}\sqrt{r}\right)  \left[  i\left(  -E\right)  ^{1/4}\sqrt{2r}\right]
^{{\xi}^{2}/\left[  4\left(  -E\right)  ^{3/2}\right]  }\sum_{n\geq0}%
\frac{a_{n}}{\left[  \sqrt{2r}\left(  -E\right)  ^{1/4}\right]  ^{n}%
}\nonumber\\
&  +A_{l}K_{2}\left(  4l+2,{\frac{i\sqrt{2}\xi}{\left(  -E\right)  ^{3/4}}%
},-{\frac{{\xi}^{2}}{2\left(  -E\right)  ^{3/2}}},{\frac{-i4\sqrt{2}\mu
}{\left(  -E\right)  ^{1/4}}}\right) \nonumber\\
&  \times\exp\left(  \left(  -E\right)  ^{1/2}r+\frac{\xi}{\left(  -E\right)
^{1/2}}\sqrt{r}\right)  \left[  i\left(  -E\right)  ^{1/4}\sqrt{2r}\right]
^{-{\xi}^{2}/\left[  4\left(  -E\right)  ^{3/2}\right]  }\sum_{n\geq0}%
\frac{e_{n}}{\left[  \sqrt{2r}\left(  -E\right)  ^{1/4}\right]  ^{n}},
\label{uexpzfo1232}%
\end{align}
where $K_{1}\left(  4l+2,{\frac{i\sqrt{2}\xi}{\left(  -E\right)  ^{3/4}}%
},-{\frac{{\xi}^{2}}{2\left(  -E\right)  ^{3/2}}},{\frac{-i4\sqrt{2}\mu
}{\left(  -E\right)  ^{1/4}}}\right)  $ and $K_{2}\left(  4l+2,{\frac
{i\sqrt{2}\xi}{\left(  -E\right)  ^{3/4}}},-{\frac{{\xi}^{2}}{2\left(
-E\right)  ^{3/2}}},{\frac{-i4\sqrt{2}\mu}{\left(  -E\right)  ^{1/4}}}\right)
$ are combination coefficients\ and $z=i\sqrt{2r}\left(  -E\right)  ^{1/4}$.

The boundary condition of bound states, $\left.  u\left(  r\right)
\right\vert _{r\rightarrow\infty}\rightarrow0$, requires that the coefficient
of the second term must vanish since this term diverges when $r\rightarrow
\infty$, i.e.,
\begin{equation}
K_{2}\left(  4l+2,{\frac{i\sqrt{2}\xi}{\left(  -E\right)  ^{3/4}}}%
,-{\frac{{\xi}^{2}}{2\left(  -E\right)  ^{3/2}}},{\frac{-i4\sqrt{2}\mu
}{\left(  -E\right)  ^{1/4}}}\right)  =0, \label{eigenvalueo1232}%
\end{equation}
where%
\begin{align}
&  K_{2}\left(  4l+2,{\frac{i\sqrt{2}\xi}{\left(  -E\right)  ^{3/4}}}%
,-{\frac{{\xi}^{2}}{2\left(  -E\right)  ^{3/2}}},{\frac{-i4\sqrt{2}\mu
}{\left(  -E\right)  ^{1/4}}}\right) \nonumber\\
&  =\frac{\Gamma\left(  4l+3\right)  }{\Gamma\left(  2l+1+{\frac{{\xi}^{2}%
}{4\left(  -E\right)  ^{3/2}}}\right)  \Gamma\left(  2l+2-{\frac{{\xi}^{2}%
}{4\left(  -E\right)  ^{3/2}}}\right)  }\nonumber\\
\times J_{2l+2-{\frac{{\xi}^{2}}{4\left(  -E\right)  ^{3/2}}}}  &  \left(
2l+1-{\frac{{\xi}^{2}}{4\left(  -E\right)  ^{3/2}}},\frac{i\sqrt{2}\xi
}{\left(  -E\right)  ^{3/4}},6l+3+{\frac{{\xi}^{2}}{4\left(  -E\right)
^{3/2}}},\frac{-i4\sqrt{2}\mu}{\left(  -E\right)  ^{1/4}}\right. \nonumber\\
&  \left.  -\frac{i\sqrt{2}\xi}{2\left(  -E\right)  ^{3/4}}\left(  {\frac
{{\xi}^{2}}{2\left(  -E\right)  ^{3/2}}}+4l+2\right)  \right)
\end{align}
with%
\begin{align}
&  J_{2l+2-{\frac{{\xi}^{2}}{4\left(  -E\right)  ^{3/2}}}}\left(
2l+1-{\frac{{\xi}^{2}}{4\left(  -E\right)  ^{3/2}}}, \frac{i\sqrt{2}\xi
}{\left(  -E\right)  ^{3/4}},6l+3+{\frac{{\xi}^{2}}{4\left(  -E\right)
^{3/2}}},\frac{-i4\sqrt{2}\mu}{\left(  -E\right)  ^{1/4}}\right. \nonumber\\
&  \left.  -\frac{i\sqrt{2}\xi}{2\left(  -E\right)  ^{3/4}}\left(  {\frac
{{\xi}^{2}}{2\left(  -E\right)  ^{3/2}}}+4l+2\right)  \right) \\
&  =\int_{0}^{\infty}dxx^{2l+1-{\frac{{\xi}^{2}}{4\left(  -E\right)  ^{3/2}}}%
}e^{-x^{2}}\nonumber\\
&  \times N\left(  2l+1-{\frac{{\xi}^{2}}{4\left(  -E\right)  ^{3/2}}}%
,\frac{i\sqrt{2}\xi}{\left(  -E\right)  ^{3/4}},6l+3+{\frac{{\xi}^{2}%
}{4\left(  -E\right)  ^{3/2}}},\frac{-i4\sqrt{2}\mu}{\left(  -E\right)
^{1/4}}-\frac{i\sqrt{2}\xi}{2\left(  -E\right)  ^{3/4}}\left(  {\frac{{\xi
}^{2}}{2\left(  -E\right)  ^{3/2}}}+4l+2\right)  ,x\right)  .
\end{align}

Eq. (\ref{eigenvalueo1232}) is an implicit expression of the eigenvalue.

The eigenfunction, by Eqs. (\ref{uexpzfo1232}) and (\ref{eigenvalueo1232}),
reads%
\begin{align}
u_{l}\left(  r\right)   &  =A_{l}K_{1}\left(  4l+2,{\frac{i\sqrt{2}\xi
}{\left(  -E\right)  ^{3/4}}},-{\frac{{\xi}^{2}}{2\left(  -E\right)  ^{3/2}}%
},{\frac{-i4\sqrt{2}\mu}{\left(  -E\right)  ^{1/4}}}\right)  \exp\left(
-\left(  -E\right)  ^{1/2}r-\frac{\xi}{\left(  -E\right)  ^{1/2}}\sqrt
{r}\right) \nonumber\\
&  \times\left[  i\left(  -E\right)  ^{1/4}\sqrt{2r}\right]  ^{{\frac{{\xi
}^{2}}{4\left(  -E\right)  ^{3/2}}}}\sum_{n\geq0}\frac{a_{n}}{\left[
\sqrt{2r}\left(  -E\right)  ^{1/4}\right]  ^{n}}.
\end{align}

\subsubsection{The scattering state}

For scattering states, $E>0$, introduce%
\begin{equation}
E=k^{2}.
\end{equation}

The singularity of the $S$-matrix on the positive imaginary axis corresponds
to the eigenvalues of bound states \cite{joachain1975quantum}, so the zero of
$K_{2}\left(  4l+2,-{\frac{\left(  1+i\right)  \xi}{k^{3/2}}},{\frac{i{\xi
}^{2}}{2k^{3}}},{\frac{4\sqrt{-2i}\mu}{\sqrt{k}}}\right)  $ on the positive
imaginary is just the singularity of the $S$-matrix. Considering that the
$S$-matrix is unitary, i.e.,%
\begin{equation}
S_{l}=e^{2i\delta_{l}}, \label{SMo1232}%
\end{equation}
we have%
\begin{equation}
S_{l}\left(  k\right)  =\frac{K_{2}^{\ast}\left(  4l+2,-{\frac{\left(
1+i\right)  \xi}{k^{3/2}}},{\frac{i{\xi}^{2}}{2k^{3}}},{\frac{4\sqrt{-2i}\mu
}{\sqrt{k}}}\right)  }{K_{2}\left(  4l+2,-{\frac{\left(  1+i\right)  \xi
}{k^{3/2}}},{\frac{i{\xi}^{2}}{2k^{3}}},{\frac{4\sqrt{-2i}\mu}{\sqrt{k}}%
}\right)  }=\frac{K_{2}\left(  4l+2,-{\frac{\left(  1-i\right)  \xi}{k^{3/2}}%
},-{\frac{i{\xi}^{2}}{2k^{3}}},{\frac{4\sqrt{2i}\mu}{\sqrt{k}}}\right)
}{K_{2}\left(  4l+2,-{\frac{\left(  1+i\right)  \xi}{k^{3/2}}},{\frac{i{\xi
}^{2}}{2k^{3}}},{\frac{4\sqrt{-2i}\mu}{\sqrt{k}}}\right)  }. \label{SM1o1232}%
\end{equation}

The scattering wave function can be constructed with the help of the
$S$-matrix. The scattering wave function can be expressed as a linear
combination of the radially ingoing wave $u_{in}\left(  r\right)  $ and the
radially outgoing wave $u_{out}\left(  r\right)  $, which are conjugate to
each other, i.e., \cite{joachain1975quantum}%
\begin{equation}
u_{l}\left(  r\right)  =A_{l}\left[  \left(  -1\right)  ^{l+1}u_{in}\left(
r\right)  +S_{l}\left(  k\right)  u_{out}\left(  r\right)  \right]  .
\label{inouto1232}%
\end{equation}
From Eq. (\ref{uexpzfo1232}), we have%
\begin{align}
u_{in}\left(  r\right)   &  =\exp\left(  -ikr+i\frac{\xi}{k}\sqrt{r}\right)
\left(  2ikr\right)  ^{-i{\xi}^{2}/\left(  8k^{3}\right)  }\sum_{n\geq0}%
\frac{e_{n}}{\left(  -2ikr\right)  ^{n/2}},\nonumber\\
u_{out}\left(  r\right)   &  =\exp\left(  ikr-i\frac{\xi}{k}\sqrt{r}\right)
\left(  -2ikr\right)  ^{i{\xi}^{2}/\left(  8k^{3}\right)  }\sum_{n\geq0}%
\frac{e_{n}^{\ast}}{\left(  2ikr\right)  ^{n/2}}.
\end{align}
Then by Eq. (\ref{SM1o1232}), we obtain the scattering wave function,%
\begin{align}
u_{l}\left(  r\right)   &  =A_{l}\left[  \left(  -1\right)  ^{l+1}\exp\left(
-ikr+i\frac{\xi}{k}\sqrt{r}\right)  \left(  2ikr\right)  ^{-i{\xi}^{2}/\left(
8k^{3}\right)  }\sum_{n\geq0}\frac{e_{n}}{\left(  -2ikr\right)  ^{n/2}}\right.
\nonumber\\
&  +\left.  \frac{K_{2}\left(  4l+2,-{\frac{\left(  1-i\right)  \xi}{k^{3/2}}%
},-{\frac{i{\xi}^{2}}{2k^{3}}},{\frac{4\sqrt{2i}\mu}{\sqrt{k}}}\right)
}{K_{2}\left(  4l+2,-{\frac{\left(  1+i\right)  \xi}{k^{3/2}}},{\frac{i{\xi
}^{2}}{2k^{3}}},{\frac{4\sqrt{-2i}\mu}{\sqrt{k}}}\right)  }\exp\left(
ikr-i\frac{\xi}{k}\sqrt{r}\right)  \left(  -2ikr\right)  ^{i{\xi}^{2}/\left(
8k^{3}\right)  }\sum_{n\geq0}\frac{e_{n}^{\ast}}{\left(  2ikr\right)  ^{n/2}%
}\right]  .
\end{align}
Taking $r\rightarrow\infty$, we have
\begin{align}
&  u_{l}\left(  r\right)  \overset{r\rightarrow\infty}{\sim}A_{l}\left[
\left(  -1\right)  ^{l+1}\exp\left(  -ikr+i\frac{\xi}{k}\sqrt{r}\right)
\left(  2kr\right)  ^{-i{\xi}^{2}/\left(  8k^{3}\right)  }\right. \nonumber\\
&  +\left.  \frac{K_{2}\left(  4l+2,-{\frac{\left(  1-i\right)  \xi}{k^{3/2}}%
},-{\frac{i{\xi}^{2}}{2k^{3}}},{\frac{4\sqrt{2i}\mu}{\sqrt{k}}}\right)
}{K_{2}\left(  4l+2,-{\frac{\left(  1+i\right)  \xi}{k^{3/2}}},{\frac{i{\xi
}^{2}}{2k^{3}}},{\frac{4\sqrt{-2i}\mu}{\sqrt{k}}}\right)  }\exp\left(
ikr-i\frac{\xi}{k}\sqrt{r}\right)  \left(  2kr\right)  ^{i{\xi}^{2}/\left(
8k^{3}\right)  }\right] \nonumber\\
&  =A_{l}e^{i\delta_{l}}\sin\left(  kr-\frac{\xi}{k}\sqrt{r}+\frac{{\xi}^{2}%
}{8k^{3}}\ln2kr+\delta_{l}-\frac{l\pi}{2}\right)  .
\end{align}
By Eqs. (\ref{SMo1232}) and (\ref{SM1o1232}), we obtain the scattering phase
shift%
\begin{equation}
\delta_{l}=-\arg K_{2}\left(  4l+2,-{\frac{\left(  1+i\right)  \xi}{k^{3/2}}%
},{\frac{i{\xi}^{2}}{2k^{3}}},{\frac{4\sqrt{-2i}\mu}{\sqrt{k}}}\right)  .
\end{equation}

\section{The exact solution of $U\left(  r\right)  =\xi r^{2/3}+\frac{\mu
}{r^{4/3}}$\label{V23m43}}

In this appendix, we provide an exact solution of the eigenproblem of the
potential
\begin{equation}
U\left(  r\right)  =\xi r^{2/3}+\frac{\mu}{r^{4/3}}%
\end{equation}
by solving the radial equation directly. This potential has only bound states.

The radial equation reads%
\begin{equation}
\frac{d^{2}u_{l}\left(  r\right)  }{dr^{2}}+\left[  E-\frac{l\left(
l+1\right)  }{r^{2}}-\xi r^{2/3}-\frac{\mu}{r^{4/3}}\right]  u_{l}\left(
r\right)  =0. \label{radialeq23o43}%
\end{equation}
Using the variable substitution%
\begin{equation}
z=i\frac{\sqrt{6}}{2}\xi^{1/4}{r}^{2/3} \label{zr23o43}%
\end{equation}
and introducing $f_{l}\left(  z\right)  $ by%
\begin{equation}
u_{l}\left(  z\right)  =A_{l}\exp\left(  -\frac{z^{2}}{2}-\frac{\beta}%
{2}z\right)  z^{3\left(  l+1\right)  /2}f_{l}\left(  z\right)
\end{equation}
with $A_{l}$ a constant, we convert the radial equation (\ref{radialeq23o43})
into an equation of $f_{l}\left(  z\right)  $:%
\begin{align}
&  zf_{l}^{\prime\prime}\left(  z\right)  +\left(  3l+\frac{5}{2}%
+{\frac{i\sqrt{6}E}{2{\xi}^{3/4}}}z-2z^{2}\right)  f_{l}^{\prime}\left(
z\right) \nonumber\\
&  +\left\{  \left(  -{\frac{3E^{2}}{8{\xi}^{3/2}}}-3l-\frac{7}{2}\right)
z-\frac{1}{2}\left[  -{\frac{i3\sqrt{6}\mu}{2\xi^{1/4}}}-{\frac{i\sqrt{6}%
E}{2{\xi}^{3/4}}}\left(  3l+\frac{5}{2}\right)  \right]  \right\}
f_{l}\left(  z\right)  =0. \label{eqf23o43}%
\end{align}
This is a Biconfluent Heun equation \cite{ronveaux1995heun}.

The choice of the boundary condition has been discussed in Ref.
\cite{li2016exact}.

\subsection{The regular solution}

The regular solution is a solution satisfying the boundary condition at $r=0$
\cite{li2016exact}. The regular solution at $r=0$\ should satisfy the boundary
condition $\lim_{r\rightarrow0}u_{l}\left(  r\right)  /r^{l+1}=1$. In this
section, we provide the regular solution of Eq. (\ref{eqf23o43}).

The Biconfluent Heun equation (\ref{eqf23o43}) has two linearly independent
solutions \cite{ronveaux1995heun}%
\begin{align}
y_{l}^{\left(  1\right)  }\left(  z\right)   &  =N\left(  3l+\frac{3}%
{2},-{\frac{i\sqrt{6}E}{2{\xi}^{3/4}}},-{\frac{3E^{2}}{8{\xi}^{3/2}}}%
,-{\frac{i3\sqrt{6}\mu}{2\xi^{1/4}}},z\right)  ,\\
y_{l}^{\left(  2\right)  }\left(  z\right)   &  =cN\left(  3l+\frac{3}%
{2},-{\frac{i\sqrt{6}E}{2{\xi}^{3/4}}},-{\frac{3E^{2}}{8{\xi}^{3/2}}}%
,-{\frac{i3\sqrt{6}\mu}{2\xi^{1/4}}},z\right)  \ln z+\sum_{n\geq0}%
d_{n}z^{n-3l-3/2},
\end{align}
where%
\begin{equation}
c=\frac{1}{3l+3/2}\left\{  \frac{1}{2}\left[  -{\frac{i3\sqrt{6}\mu}%
{2\xi^{1/4}}}-{\frac{i\sqrt{6}E}{2{\xi}^{3/4}}}\left(  3l+\frac{1}{2}\right)
\right]  d_{3l+1/2}-\left(  -{\frac{3E^{2}}{8{\xi}^{3/2}}}+\frac{1}%
{2}-3l\right)  d_{3l-1/2}\right\}
\end{equation}
is a constant with the coefficient $d_{\nu}$ given by the following recurrence
relation,%
\begin{align}
&  d_{-1}=0,\text{ \ }d_{0}=1,\nonumber\\
&  \left(  v+2\right)  \left(  v+\frac{1}{2}-3l\right)  d_{v+2}-\frac{1}%
{2}\left[  -{\frac{i3\sqrt{6}\mu}{2\xi^{1/4}}}-{\frac{i\sqrt{6}E}{2{\xi}%
^{3/4}}}\left(  2v+\frac{3}{2}-3l\right)  \right]  d_{v+1}+\left(
-{\frac{3E^{2}}{8{\xi}^{3/2}}}-2v-\frac{1}{2}+3l\right)  d_{v}=0
\end{align}
and $N(\alpha,\beta,\gamma,\delta,z)$ is the biconfluent Heun function
\cite{ronveaux1995heun,slavyanov2000special,li2016exact}.

The biconfluent Heun function $N\left(  3l+\frac{3}{2},-{\frac{i\sqrt{6}%
E}{2{\xi}^{3/4}}},-{\frac{3E^{2}}{8{\xi}^{3/2}}},-{\frac{i3\sqrt{6}\mu}%
{2\xi^{1/4}}},z\right)  $ has an expansion at $z=0$ \cite{ronveaux1995heun}:%
\begin{equation}
N\left(  3l+\frac{3}{2},-{\frac{i\sqrt{6}E}{2{\xi}^{3/4}}},-{\frac{3E^{2}%
}{8{\xi}^{3/2}}},-{\frac{i3\sqrt{6}\mu}{2\xi^{1/4}}},z\right)  =\sum_{n\geq
0}\frac{A_{n}}{\left(  3l+5/2\right)  _{n}}\frac{z^{n}}{n!},
\end{equation}
where the expansion coefficients is determined by the recurrence relation,%
\begin{align}
A_{0}  &  =1,\text{ \ }A_{1}=\frac{1}{2}\left[  -{\frac{i3\sqrt{6}\mu}%
{2\xi^{1/4}}}-{\frac{i\sqrt{6}E}{2{\xi}^{3/4}}}\left(  3l+\frac{5}{2}\right)
\right]  ,\nonumber\\
A_{n+2}  &  =\left\{  -\left(  n+1\right)  {\frac{i\sqrt{6}E}{2{\xi}^{3/4}}%
}+\frac{1}{2}\left[  -{\frac{i3\sqrt{6}\mu}{2\xi^{1/4}}}-{\frac{i\sqrt{6}%
E}{2{\xi}^{3/4}}}\left(  3l+\frac{5}{2}\right)  \right]  \right\}
A_{n+1}\nonumber\\
&  -\left(  n+1\right)  \left(  n+3l+\frac{5}{2}\right)  \left(
-{\frac{3E^{2}}{8{\xi}^{3/2}}}-3l-\frac{7}{2}-2n\right)  A_{n},
\end{align}
and $\left(  a\right)  _{n}=\Gamma\left(  a+n\right)  /\Gamma\left(  a\right)
$ is Pochhammer's symbol.

Only $y_{l}^{\left(  1\right)  }\left(  z\right)  $ satisfies the boundary
condition for the regular solution at $r=0$, so the radial eigenfunction reads%
\begin{align}
u_{l}\left(  z\right)   &  =A_{l}\exp\left(  -\frac{z^{2}}{2}-\frac{\beta}%
{2}z\right)  z^{3\left(  l+1\right)  /2}y_{l}^{\left(  1\right)  }\left(
z\right) \nonumber\\
&  =A_{l}\exp\left(  -\frac{z^{2}}{2}-\frac{\beta}{2}z\right)  z^{3\left(
l+1\right)  /2}N\left(  3l+\frac{3}{2},-{\frac{i\sqrt{6}E}{2{\xi}^{3/4}}%
},-{\frac{3E^{2}}{8{\xi}^{3/2}}},-{\frac{i3\sqrt{6}\mu}{2\xi^{1/4}}},z\right)
.
\end{align}
By Eq. (\ref{zr23o43}), we obtain the regular solution,%
\begin{equation}
u_{l}\left(  r\right)  =A_{l}\exp\left(  \frac{3}{4}\xi^{1/2}{r}^{4/3}%
-{\frac{3E}{4{\xi}^{1/2}}}r^{2/3}\right)  r^{l+1}N\left(  3l+\frac{3}%
{2},-{\frac{i\sqrt{6}E}{2{\xi}^{3/4}}},-{\frac{3E^{2}}{8{\xi}^{3/2}}}%
,-{\frac{i3\sqrt{6}\mu}{2\xi^{1/4}}},i\frac{\sqrt{6}}{2}\xi^{1/4}{r}%
^{2/3}\right)  . \label{regular23o43}%
\end{equation}

\subsection{The irregular solution}

The irregular solution is a solution satisfying the boundary condition at
$r\rightarrow\infty$ \cite{li2016exact}. The boundary conditions for bound
states and scattering states at $r\rightarrow\infty$ are different.

The Biconfluent Heun equation (\ref{eqf23o43}) has two linearly independent
irregular solutions \cite{ronveaux1995heun}:%
\begin{align}
&  B_{l}^{+}\left(  3l+\frac{3}{2},-{\frac{i\sqrt{6}E}{2{\xi}^{3/4}}}%
,-{\frac{3E^{2}}{8{\xi}^{3/2}}},-{\frac{i3\sqrt{6}\mu}{2\xi^{1/4}}},z\right)
\nonumber\\
&  =\exp\left(  -{\frac{i\sqrt{6}E}{2{\xi}^{3/4}}}z+z^{2}\right)  B_{l}%
^{+}\left(  3l+\frac{3}{2},-{\frac{\sqrt{6}E}{2{\xi}^{3/4}}},{\frac{3E^{2}%
}{8{\xi}^{3/2}}},{\frac{3\sqrt{6}\mu}{2\xi^{1/4}}},-iz\right) \nonumber\\
&  =\exp\left(  -{\frac{i\sqrt{6}E}{2{\xi}^{3/4}}}z+z^{2}\right)  \left(
-iz\right)  ^{\frac{1}{2}\left(  \frac{3E^{2}}{8{\xi}^{3/2}}-3l-\frac{7}%
{2}\right)  }\sum_{n\geq0}\frac{a_{n}}{\left(  -iz\right)  ^{n}},
\label{f123o43}%
\end{align}%
\begin{align}
&  H_{l}^{+}\left(  3l+\frac{3}{2},-{\frac{i\sqrt{6}E}{2{\xi}^{3/4}}}%
,-{\frac{3E^{2}}{8{\xi}^{3/2}}},-{\frac{i3\sqrt{6}\mu}{2\xi^{1/4}}},z\right)
\nonumber\\
&  =\exp\left(  -{\frac{i\sqrt{6}E}{2{\xi}^{3/4}}}z+z^{2}\right)  H_{l}%
^{+}\left(  3l+\frac{3}{2},-{\frac{\sqrt{6}E}{2{\xi}^{3/4}}},{\frac{3E^{2}%
}{8{\xi}^{3/2}}},{\frac{3\sqrt{6}\mu}{2\xi^{1/4}}},-iz\right) \nonumber\\
&  =\left(  -iz\right)  ^{-\frac{1}{2}\left(  \frac{3E^{2}}{8{\xi}^{3/2}%
}+3l+\frac{7}{2}\right)  }\sum_{n\geq0}\frac{e_{n}}{\left(  -iz\right)  ^{n}}
\label{f223o43}%
\end{align}
with the expansion coefficients given by the recurrence relation
\[
a_{0}=1,\text{ \ }a_{1}=\frac{1}{4}\left[  \frac{3\sqrt{6}\mu}{2\xi^{1/4}%
}-{\frac{\sqrt{6}E}{2{\xi}^{3/4}}}\left(  \frac{3E^{2}}{8{\xi}^{3/2}%
}-1\right)  \right]  ,
\]%
\begin{align}
&  2\left(  n+2\right)  a_{n+2}+\left[  \beta\left(  \frac{3}{2}-\frac{3E^{2}%
}{16{\xi}^{3/2}}+n\right)  -\frac{3\sqrt{6}\mu}{4\xi^{1/4}}\right]
a_{n+1}\nonumber\\
&  +\left[  \frac{1}{4}\left(  \frac{3E^{2}}{8{\xi}^{3/2}}\right)  ^{2}%
-\frac{1}{4}\left(  3l+\frac{3}{2}\right)  ^{2}+1-{\frac{3E^{2}}{8{\xi}^{3/2}%
}}+n\left(  n+2-{\frac{3E^{2}}{8{\xi}^{3/2}}}\right)  \right]  a_{n}=0
\end{align}
and%
\[
e_{0}=1,\text{ \ }e_{1}=-\frac{1}{4}\left[  \frac{3\sqrt{6}\mu}{2\xi^{1/4}%
}-{\frac{\sqrt{6}E}{2{\xi}^{3/4}}}\left(  {\frac{3E^{2}}{8{\xi}^{3/2}}%
}+1\right)  \right]  ,
\]%
\begin{align}
&  2\left(  n+2\right)  e_{n+2}+\left[  \beta\left(  \frac{3}{2}+\frac{3E^{2}%
}{16{\xi}^{3/2}}+n\right)  +\frac{3\sqrt{6}\mu}{4\xi^{1/4}}\right]
e_{n+1}\nonumber\\
&  -\left[  \frac{1}{4}\left(  \frac{3E^{2}}{8{\xi}^{3/2}}\right)  ^{2}%
-\frac{1}{4}\left(  3l+\frac{3}{2}\right)  ^{2}+1+\frac{3E^{2}}{8{\xi}^{3/2}%
}+n\left(  n+2+\frac{3E^{2}}{8{\xi}^{3/2}}\right)  \right]  e_{n}=0.
\end{align}

\subsection{Eigenfunctions and eigenvalues}

To construct the solution, we first express the regular solution
(\ref{regular23o43}) as a linear combination of the two irregular solutions
(\ref{f123o43}) and (\ref{f223o43}).

The regular solution (\ref{regular23o43}), with the relation
\cite{ronveaux1995heun,li2016exact}%

\begin{align}
&  N\left(  3l+\frac{3}{2},-{\frac{i\sqrt{6}E}{2{\xi}^{3/4}}},-{\frac{3E^{2}%
}{8{\xi}^{3/2}}},-{\frac{i3\sqrt{6}\mu}{2\xi^{1/4}}},z\right)  \nonumber\\
&  =K_{1}\left(  3l+\frac{3}{2},-{\frac{i\sqrt{6}E}{2{\xi}^{3/4}}}%
,-{\frac{3E^{2}}{8{\xi}^{3/2}}},-{\frac{i3\sqrt{6}\mu}{2\xi^{1/4}}}\right)
B_{l}^{+}\left(  3l+\frac{3}{2},-{\frac{i\sqrt{6}E}{2{\xi}^{3/4}}}%
,-{\frac{3E^{2}}{8{\xi}^{3/2}}},-{\frac{i3\sqrt{6}\mu}{2\xi^{1/4}}},z\right)
\nonumber\\
&  +K_{2}\left(  3l+\frac{3}{2},-{\frac{i\sqrt{6}E}{2{\xi}^{3/4}}}%
,-{\frac{3E^{2}}{8{\xi}^{3/2}}},-{\frac{i3\sqrt{6}\mu}{2\xi^{1/4}}}\right)
H_{l}^{+}\left(  3l+\frac{3}{2},-{\frac{i\sqrt{6}E}{2{\xi}^{3/4}}}%
,-{\frac{3E^{2}}{8{\xi}^{3/2}}},-{\frac{i3\sqrt{6}\mu}{2\xi^{1/4}}},z\right)
\end{align}
and the expansions (\ref{f123o43}) and (\ref{f223o43}), become%
\begin{align}
u_{l}\left(  r\right)   &  =A_{l}K_{1}\left(  3l+\frac{3}{2},-{\frac{i\sqrt
{6}E}{2{\xi}^{3/4}}},-{\frac{3E^{2}}{8{\xi}^{3/2}}},-{\frac{i3\sqrt{6}\mu
}{2\xi^{1/4}}}\right)  \nonumber\\
&  \times\exp\left(  -\frac{3\xi^{1/2}}{4}{r}^{4/3}+{\frac{3E}{4{\xi}^{1/2}}%
}r^{2/3}\right)  {r}^{\frac{E^{2}}{8{\xi}^{3/2}}-\frac{1}{6}}\sum_{n\geq
0}\frac{a_{n}}{\left(  \frac{\sqrt{6}}{2}\xi^{1/4}{r}^{2/3}\right)  ^{n}%
}\nonumber\\
&  +A_{l}K_{2}\left(  3l+\frac{3}{2},-{\frac{i\sqrt{6}E}{2{\xi}^{3/4}}%
},-{\frac{3E^{2}}{8{\xi}^{3/2}}},-{\frac{i3\sqrt{6}\mu}{2\xi^{1/4}}}\right)
\nonumber\\
&  \times\exp\left(  \frac{3\xi^{1/2}}{4}{r}^{4/3}-{\frac{3E}{4{\xi}^{1/2}}%
}r^{2/3}\right)  r^{-\frac{E^{2}}{8{\xi}^{3/2}}-\frac{1}{6}}\sum_{n\geq0}%
\frac{e_{n}}{\left(  \frac{\sqrt{6}}{2}\xi^{1/4}{r}^{2/3}\right)  ^{n}%
},\label{uexpzf23o43}%
\end{align}
where $K_{1}\left(  3l+\frac{3}{2},-{\frac{i\sqrt{6}E}{2{\xi}^{3/4}}}%
,-{\frac{3E^{2}}{8{\xi}^{3/2}}},-{\frac{i3\sqrt{6}\mu}{2\xi^{1/4}}}\right)  $
and $K_{2}\left(  3l+\frac{3}{2},-{\frac{i\sqrt{6}E}{2{\xi}^{3/4}}}%
,-{\frac{3E^{2}}{8{\xi}^{3/2}}},-{\frac{i3\sqrt{6}\mu}{2\xi^{1/4}}}\right)  $
are combination coefficients\ and $z=i\frac{\sqrt{6}}{2}\xi^{1/4}{r}^{2/3}$.

The boundary condition of bound states, $\left.  u\left(  r\right)
\right\vert _{r\rightarrow\infty}\rightarrow0$, requires that the coefficient
of the second term must vanish since this term diverges when $r\rightarrow
\infty$, i.e.,
\begin{equation}
K_{2}\left(  3l+\frac{3}{2},-{\frac{i\sqrt{6}E}{2{\xi}^{3/4}}},-{\frac{3E^{2}%
}{8{\xi}^{3/2}}},-{\frac{i3\sqrt{6}\mu}{2\xi^{1/4}}}\right)  =0,
\label{eigenvalue23o43}%
\end{equation}
where%
\begin{align}
&  K_{2}\left(  3l+\frac{3}{2},-{\frac{i\sqrt{6}E}{2{\xi}^{3/4}}}%
,-{\frac{3E^{2}}{8{\xi}^{3/2}}},-{\frac{i3\sqrt{6}\mu}{2\xi^{1/4}}}\right)
\nonumber\\
&  =\frac{\Gamma\left(  3l+\frac{5}{2}\right)  }{\Gamma\left(  \frac{3}%
{2}l+\frac{3}{4}+{\frac{3E^{2}}{16{\xi}^{3/2}}}\right)  \Gamma\left(  \frac
{3}{2}l+\frac{7}{4}-{\frac{3E^{2}}{16{\xi}^{3/2}}}\right)  }\nonumber\\
&  \times J_{\frac{3}{2}l+\frac{7}{4}-{\frac{3E^{2}}{16{\xi}^{3/2}}}}\left(
\frac{3}{2}l+\frac{3}{4}-{\frac{3E^{2}}{16{\xi}^{3/2}}},-{\frac{i\sqrt{6}%
E}{2{\xi}^{3/4}}},\frac{9}{2}l+\frac{9}{4}+{\frac{3E^{2}}{16{\xi}^{3/2}}%
},-{\frac{i3\sqrt{6}\mu}{2\xi^{1/4}}}-{\frac{i\sqrt{6}E}{4{\xi}^{3/4}}}\left(
-{\frac{3E^{2}}{8{\xi}^{3/2}}}-3l-\frac{3}{2}\right)  \right)
\end{align}
with%
\begin{align}
&  J_{\frac{3}{2}l+\frac{7}{4}-{\frac{3E^{2}}{16{\xi}^{3/2}}}}\left(  \frac
{3}{2}l+\frac{3}{4}-{\frac{3E^{2}}{16{\xi}^{3/2}}},-{\frac{i\sqrt{6}E}{2{\xi
}^{3/4}}},\frac{9}{2}l+\frac{9}{4}+{\frac{3E^{2}}{16{\xi}^{3/2}}}%
,-{\frac{i3\sqrt{6}\mu}{2\xi^{1/4}}}-{\frac{i\sqrt{6}E}{4{\xi}^{3/4}}}\left(
-{\frac{3E^{2}}{8{\xi}^{3/2}}}-3l-\frac{3}{2}\right)  \right) \nonumber\\
&  =\int_{0}^{\infty}dxx^{\frac{3}{2}l+\frac{3}{4}-{\frac{3E^{2}}{16{\xi
}^{3/2}}}}e^{-x^{2}}\\
&  \times N\left(  \frac{3}{2}l+\frac{3}{4}-{\frac{3E^{2}}{16{\xi}^{3/2}}%
},-{\frac{i\sqrt{6}E}{2{\xi}^{3/4}}},\frac{9}{2}l+\frac{9}{4}+{\frac{3E^{2}%
}{16{\xi}^{3/2}}},-{\frac{i3\sqrt{6}\mu}{2\xi^{1/4}}}-{\frac{i\sqrt{6}E}%
{4{\xi}^{3/4}}}\left(  -{\frac{3E^{2}}{8{\xi}^{3/2}}}-3l-\frac{3}{2}\right)
,x\right)  .
\end{align}

Eq. (\ref{eigenvalue23o43}) is an implicit expression of the eigenvalue.

The eigenfunction, by Eqs. (\ref{uexpzf23o43}) and (\ref{eigenvalue23o43}),
reads%
\begin{align}
u_{l}\left(  r\right)   &  =A_{l}K_{1}\left(  3l+\frac{3}{2},-{\frac{i\sqrt
{6}E}{2{\xi}^{3/4}}},-{\frac{3E^{2}}{8{\xi}^{3/2}}},-{\frac{i3\sqrt{6}\mu
}{2\xi^{1/4}}}\right)  \nonumber\\
&  \times\exp\left(  -\frac{3\xi^{1/2}}{4}{r}^{4/3}+{\frac{3E}{4{\xi}^{1/2}}%
}r^{2/3}\right)  {r}^{\frac{E^{2}}{8{\xi}^{3/2}}-\frac{1}{6}}\sum_{n\geq
0}\frac{a_{n}}{\left(  \frac{\sqrt{6}}{2}\xi^{1/4}{r}^{2/3}\right)  ^{n}}.
\end{align}

\section{The exact solution of $U\left(  r\right)  =\xi r^{6}+\mu r^{4}$
\label{V64}}

In this appendix, we provide an exact solution of the eigenproblem of the
potential $U\left(  r\right)  =\xi r^{6}+\mu r^{4}$ by solving the radial
equation directly. This potential has only bound states.

The radial equation of the potential $U\left(  r\right)  =\xi r^{6}+\mu r^{4}$
reads%
\begin{equation}
\frac{d^{2}u_{l}\left(  r\right)  }{dr^{2}}+\left[  E-\frac{l\left(
l+1\right)  }{r^{2}}-\xi r^{6}-\mu r^{4}\right]  u_{l}\left(  r\right)  =0.
\label{radialeq64}%
\end{equation}
Using the variable substitution%
\begin{equation}
z=i\frac{\sqrt{2}}{2}\xi^{1/4}r^{2} \label{zr64}%
\end{equation}
and introducing $f_{l}\left(  z\right)  $ by%
\begin{equation}
u_{l}\left(  z\right)  =A_{l}\exp\left(  -\frac{z^{2}}{2}-\frac{\beta}%
{2}z\right)  z^{\left(  l+1\right)  /2}f_{l}\left(  z\right)
\end{equation}
with $A_{l}$ a constant, we convert the radial equation (\ref{radialeq23o43})
into an equation of $f_{l}\left(  z\right)  $:%
\begin{align}
&  zf_{l}^{\prime\prime}\left(  z\right)  +\left(  l+\frac{3}{2}-\frac
{i\sqrt{2}\mu}{2\xi^{3/4}}z-2z^{2}\right)  f_{l}^{\prime}\left(  z\right)
\nonumber\\
&  +\left\{  \left(  -{\frac{\mu^{2}}{8\xi^{3/2}}}-l-\frac{5}{2}\right)
z-\frac{1}{2}\left[  \frac{i\sqrt{2}E}{2\xi^{1/4}}+\left(  l+\frac{3}%
{2}\right)  {\frac{i\sqrt{2}\mu}{2\xi^{3/4}}}\right]  \right\}  f_{l}\left(
z\right)  =0. \label{eqf64}%
\end{align}
This is a Biconfluent Heun equation \cite{ronveaux1995heun}.

The choice of the boundary condition has been discussed in Ref.
\cite{li2016exact}.

\subsection{The regular solution}

The regular solution is a solution satisfying the boundary condition at $r=0$
\cite{li2016exact}. The regular solution at $r=0$\ should satisfy the boundary
condition $\lim_{r\rightarrow0}u_{l}\left(  r\right)  /r^{l+1}=1$. In this
section, we provide the regular solution of Eq. (\ref{eqf64}).

The Biconfluent Heun equation (\ref{eqf64}) has two linearly independent
solutions \cite{ronveaux1995heun}%
\begin{align}
y_{l}^{\left(  1\right)  }\left(  z\right)   &  =N\left(  l+\frac{1}{2}%
,{\frac{i\sqrt{2}\mu}{2\xi^{3/4}}},-{\frac{\mu^{2}}{8\xi^{3/2}}},{\frac
{i\sqrt{2}E}{2\xi^{1/4}}},z\right)  ,\\
y_{l}^{\left(  2\right)  }\left(  z\right)   &  =cN\left(  l+\frac{1}%
{2},{\frac{i\sqrt{2}\mu}{2\xi^{3/4}}},-{\frac{\mu^{2}}{8\xi^{3/2}}}%
,{\frac{i\sqrt{2}E}{2\xi^{1/4}}},z\right)  \ln z+\sum_{n\geq0}d_{n}%
z^{n-l-1/2},
\end{align}
where%
\begin{equation}
c=\frac{1}{l+1/2}\left\{  d_{l-1/2}\left[  \frac{i\sqrt{2}E}{4\xi^{1/4}%
}+{\frac{i\sqrt{2}\mu}{4\xi^{3/4}}}\left(  l-\frac{1}{2}\right)  \right]
-d_{l-3/2}\left(  -{\frac{\mu^{2}}{8\xi^{3/2}}}+\frac{3}{2}-l\right)
\right\}
\end{equation}
is a constant with the coefficient $d_{\nu}$ given by the following recurrence
relation,%
\begin{align}
&  d_{-1}=0,\text{ \ }d_{0}=1,\nonumber\\
&  \left(  v+2\right)  \left(  v+\frac{3}{2}-l\right)  d_{v+2}-\frac{1}%
{2}\left(  {\frac{i\sqrt{2}E}{2\xi^{1/4}}}+{\frac{i\sqrt{2}\mu}{2\xi^{3/4}}%
}\left(  2v+\frac{5}{2}-l\right)  \right)  d_{v+1}+\left(  -{\frac{\mu^{2}%
}{8\xi^{3/2}}}-2v-\frac{3}{2}+l\right)  d_{v}=0
\end{align}
and $N(\alpha,\beta,\gamma,\delta,z)$ is the biconfluent Heun function
\cite{ronveaux1995heun,slavyanov2000special,li2016exact}.

The biconfluent Heun function $N\left(  l+\frac{1}{2},{\frac{i\sqrt{2}\mu
}{2\xi^{3/4}}},-{\frac{\mu^{2}}{8\xi^{3/2}}},{\frac{i\sqrt{2}E}{2\xi^{1/4}}%
},z\right)  $ has an expansion at $z=0$ \cite{ronveaux1995heun}:%
\begin{equation}
N\left(  l+\frac{1}{2},{\frac{i\sqrt{2}\mu}{2\xi^{3/4}}},-{\frac{\mu^{2}}%
{8\xi^{3/2}}},{\frac{i\sqrt{2}E}{2\xi^{1/4}}},z\right)  =\sum_{n\geq0}%
\frac{A_{n}}{\left(  l+\frac{3}{2}\right)  _{n}}\frac{z^{n}}{n!},
\end{equation}
where the expansion coefficients is determined by the recurrence relation,%
\begin{align}
A_{0}  &  =1,\text{ \ }A_{1}=\frac{i\sqrt{2}E}{4\xi^{1/4}}+{\frac{i\sqrt{2}%
\mu}{4\xi^{3/4}}}\left(  l+\frac{3}{2}\right)  ,\nonumber\\
A_{n+2}  &  =\left[  \frac{i\sqrt{2}\mu}{2\xi^{3/4}}\left(  n+1\right)
+\frac{i\sqrt{2}E}{4\xi^{1/4}}+{\frac{i\sqrt{2}\mu}{4\xi^{3/4}}}\left(
l+\frac{3}{2}\right)  \right]  A_{n+1}\nonumber\\
&  -\left(  n+1\right)  \left(  n+l+\frac{3}{2}\right)  \left(  -{\frac
{\mu^{2}}{8\xi^{3/2}}}-l-\frac{5}{2}-2n\right)  A_{n},
\end{align}
and $\left(  a\right)  _{n}=\Gamma\left(  a+n\right)  /\Gamma\left(  a\right)
$ is Pochhammer's symbol.

Only $y_{l}^{\left(  1\right)  }\left(  z\right)  $ satisfies the boundary
condition for the regular solution at $r=0$, so the radial eigenfunction reads%
\begin{align}
u_{l}\left(  z\right)   &  =A_{l}\exp\left(  -\frac{z^{2}}{2}-\frac{\beta}%
{2}z\right)  z^{\left(  l+1\right)  /2}y_{l}^{\left(  1\right)  }\left(
z\right)  \nonumber\\
&  =A_{l}\exp\left(  -\frac{z^{2}}{2}-\frac{\beta}{2}z\right)  z^{\left(
l+1\right)  /2}N\left(  l+\frac{1}{2},{\frac{i\sqrt{2}\mu}{2\xi^{3/4}}%
},-{\frac{\mu^{2}}{8\xi^{3/2}}},{\frac{i\sqrt{2}E}{2\xi^{1/4}}},z\right)  .
\end{align}
By Eq. (\ref{zr64}), we obtain the regular solution,%
\begin{equation}
u_{l}\left(  r\right)  =A_{l}\exp\left(  \frac{\xi^{1/2}}{4}r^{4}+\frac{\mu
}{4\xi^{1/2}}r^{2}\right)  r^{l+1}N\left(  l+\frac{1}{2},{\frac{i\sqrt{2}\mu
}{2\xi^{3/4}}},-{\frac{\mu^{2}}{8\xi^{3/2}}},{\frac{i\sqrt{2}E}{2\xi^{1/4}}%
},i\frac{\sqrt{2}}{2}\xi^{1/4}r^{2}\right)  .\label{regular64}%
\end{equation}

\subsection{The irregular solution}

The irregular solution is a solution satisfying the boundary condition at
$r\rightarrow\infty$ \cite{li2016exact}.

The Biconfluent Heun equation (\ref{eqf64}) has two linearly independent
irregular solutions \cite{ronveaux1995heun}:%
\begin{align}
&  B_{l}^{+}\left(  l+\frac{1}{2},{\frac{i\sqrt{2}\mu}{2\xi^{3/4}}}%
,-{\frac{\mu^{2}}{8\xi^{3/2}}},{\frac{i\sqrt{2}E}{2\xi^{1/4}}},z\right)
\nonumber\\
&  =\exp\left(  \frac{i\sqrt{2}\mu}{2\xi^{3/4}}z+z^{2}\right)  B_{l}%
^{+}\left(  l+\frac{1}{2},{\frac{\sqrt{2}\mu}{2\xi^{3/4}}},{\frac{\mu^{2}%
}{8\xi^{3/2}}},-{\frac{\sqrt{2}E}{2\xi^{1/4}}},-iz\right) \nonumber\\
&  =\exp\left(  \frac{i\sqrt{2}\mu}{2\xi^{3/4}}z+z^{2}\right)  \left(
-iz\right)  ^{\frac{1}{2}\left(  \frac{\mu^{2}}{8\xi^{3/2}}-l-\frac{5}%
{2}\right)  }\sum_{n\geq0}\frac{a_{n}}{\left(  -iz\right)  ^{n}}, \label{f164}%
\end{align}%
\begin{align}
&  H_{l}^{+}\left(  l+\frac{1}{2},{\frac{i\sqrt{2}\mu}{2\xi^{3/4}}}%
,-{\frac{\mu^{2}}{8\xi^{3/2}}},{\frac{i\sqrt{2}E}{2\xi^{1/4}}},z\right)
\nonumber\\
&  =\exp\left(  \frac{i\sqrt{2}\mu}{2\xi^{3/4}}z+z^{2}\right)  H_{l}%
^{+}\left(  l+\frac{1}{2},{\frac{\sqrt{2}\mu}{2\xi^{3/4}}},{\frac{\mu^{2}%
}{8\xi^{3/2}}},-{\frac{\sqrt{2}E}{2\xi^{1/4}}},-iz\right) \nonumber\\
&  =\left(  -iz\right)  ^{-\frac{1}{2}\left(  \frac{\mu^{2}}{8\xi^{3/2}%
}+l+\frac{5}{2}\right)  }\sum_{n\geq0}\frac{e_{n}}{\left(  -iz\right)  ^{n}}
\label{f264}%
\end{align}
with the expansion coefficients given by the recurrence relation
\[
a_{0}=1,\text{ \ }a_{1}=-{\frac{\sqrt{2}E}{8\xi^{1/4}}}+\frac{\sqrt{2}\mu
}{8\xi^{3/4}}\left(  \frac{\mu^{2}}{8\xi^{3/2}}-1\right)  ,
\]%
\begin{align}
&  2\left(  n+2\right)  a_{n+2}+\left[  \beta\left(  \frac{3}{2}-\frac{\mu
^{2}}{16\xi^{3/2}}+n\right)  +{\frac{\sqrt{2}E}{4\xi^{1/4}}}\right]
a_{n+1}\nonumber\\
&  +\left[  \frac{1}{4}\left(  \frac{\mu^{2}}{8\xi^{3/2}}\right)  ^{2}%
-\frac{1}{4}\left(  l+\frac{1}{2}\right)  ^{2}+1-\frac{\mu^{2}}{8\xi^{3/2}%
}+n\left(  n+2-\frac{\mu^{2}}{8\xi^{3/2}}\right)  \right]  a_{n}=0
\end{align}
and%
\[
e_{0}=1,\text{ \ }e_{1}={\frac{\sqrt{2}E}{8\xi^{1/4}}}-\frac{\sqrt{2}\mu}%
{8\xi^{3/4}}\left(  {\frac{\mu^{2}}{8\xi^{3/2}}+1}\right)  ,
\]%
\begin{align}
&  2\left(  n+2\right)  e_{n+2}+\left[  \frac{\sqrt{2}\mu}{2\xi^{3/4}}\left(
\frac{3}{2}+\frac{\mu^{2}}{16\xi^{3/2}}+n\right)  -{\frac{\sqrt{2}E}%
{4\xi^{1/4}}}\right]  e_{n+1}\nonumber\\
&  -\left[  \frac{1}{4}\left(  \frac{\mu^{2}}{8\xi^{3/2}}\right)  ^{2}%
-\frac{1}{4}\left(  l+\frac{1}{2}\right)  ^{2}+1+\frac{\mu^{2}}{8\xi^{3/2}%
}+n\left(  n+2+\frac{\mu^{2}}{8\xi^{3/2}}\right)  \right]  e_{n}=0.
\end{align}

\subsection{Eigenfunctions and eigenvalues}

To construct the solution, we first express the regular solution
(\ref{regular23o43}) as a linear combination of the two irregular solutions
(\ref{f164}) and (\ref{f264}).

The regular solution (\ref{regular64}), with the relation
\cite{ronveaux1995heun,li2016exact}%

\begin{align}
&  N\left(  l+\frac{1}{2},{\frac{\sqrt{2}\mu}{2\xi^{3/4}}},{\frac{\mu^{2}%
}{8\xi^{3/2}}},-{\frac{\sqrt{2}E}{2\xi^{1/4}}},z\right)  \nonumber\\
&  =K_{1}\left(  l+\frac{1}{2},{\frac{\sqrt{2}\mu}{2\xi^{3/4}}},{\frac{\mu
^{2}}{8\xi^{3/2}}},-{\frac{\sqrt{2}E}{2\xi^{1/4}}}\right)  B_{l}^{+}\left(
l+\frac{1}{2},{\frac{\sqrt{2}\mu}{2\xi^{3/4}}},{\frac{\mu^{2}}{8\xi^{3/2}}%
},-{\frac{\sqrt{2}E}{2\xi^{1/4}}},z\right)  \nonumber\\
&  +K_{2}\left(  l+\frac{1}{2},{\frac{\sqrt{2}\mu}{2\xi^{3/4}}},{\frac{\mu
^{2}}{8\xi^{3/2}}},-{\frac{\sqrt{2}E}{2\xi^{1/4}}}\right)  H_{l}^{+}\left(
l+\frac{1}{2},{\frac{\sqrt{2}\mu}{2\xi^{3/4}}},{\frac{\mu^{2}}{8\xi^{3/2}}%
},-{\frac{\sqrt{2}E}{2\xi^{1/4}},}z\right)
\end{align}
and the expansions (\ref{f164}) and (\ref{f264}), becomes%
\begin{align}
u_{l}\left(  r\right)   &  =A_{l}K_{1}\left(  l+\frac{1}{2},{\frac{i\sqrt
{2}\mu}{2\xi^{3/4}}},-{\frac{\mu^{2}}{8\xi^{3/2}}},{\frac{i\sqrt{2}E}%
{2\xi^{1/4}}}\right)  \exp\left(  -\frac{\xi^{1/4}}{4}r^{4}-\frac{\mu}%
{4\xi^{1/2}}r^{2}\right)  r^{\frac{\mu^{2}}{8\xi^{3/2}}-\frac{3}{2}}%
\sum_{n\geq0}\frac{a_{n}}{\left(  \frac{\sqrt{2}}{2}\xi^{1/4}r^{2}\right)
^{n}},\nonumber\\
&  +A_{l}K_{2}\left(  l+\frac{1}{2},{\frac{i\sqrt{2}\mu}{2\xi^{3/4}}}%
,-{\frac{\mu^{2}}{8\xi^{3/2}}},{\frac{i\sqrt{2}E}{2\xi^{1/4}}}\right)
\exp\left(  \frac{\xi^{1/4}}{4}r^{4}+\frac{\mu}{4\xi^{1/2}}r^{2}\right)
r^{-\frac{\mu^{2}}{8\xi^{3/2}}-\frac{3}{2}}\sum_{n\geq0}\frac{e_{n}}{\left(
\frac{\sqrt{2}}{2}\xi^{1/4}r^{2}\right)  ^{n}},\label{uexpzf64}%
\end{align}
where $K_{1}\left(  l+\frac{1}{2},{\frac{i\sqrt{2}\mu}{2\xi^{3/4}}}%
,-{\frac{\mu^{2}}{8\xi^{3/2}}},{\frac{i\sqrt{2}E}{2\xi^{1/4}}}\right)  $ and
$K_{2}\left(  l+\frac{1}{2},{\frac{i\sqrt{2}\mu}{2\xi^{3/4}}},-{\frac{\mu^{2}%
}{8\xi^{3/2}}},{\frac{i\sqrt{2}E}{2\xi^{1/4}}}\right)  $ are combination
coefficients\ and $z=i\frac{\sqrt{2}}{2}\xi^{1/4}r^{2}$.

The boundary condition of bound states, $\left.  u\left(  r\right)
\right\vert _{r\rightarrow\infty}\rightarrow0$, requires that the coefficient
of the second term must vanish since this term diverges when $r\rightarrow
\infty$, i.e.,
\begin{equation}
K_{2}\left(  l+\frac{1}{2},{\frac{i\sqrt{2}\mu}{2\xi^{3/4}}},-{\frac{\mu^{2}%
}{8\xi^{3/2}}},{\frac{i\sqrt{2}E}{2\xi^{1/4}}}\right)  =0,
\label{eigenvalue64}%
\end{equation}
where%
\begin{align}
&  K_{2}\left(  l+\frac{1}{2},{\frac{i\sqrt{2}\mu}{2\xi^{3/4}}},-{\frac
{\mu^{2}}{8\xi^{3/2}}},{\frac{i\sqrt{2}E}{2\xi^{1/4}}}\right) \nonumber\\
&  =\frac{\Gamma\left(  l+\frac{3}{2}\right)  }{\Gamma\left(  \frac{l}%
{2}+\frac{1}{4}+{\frac{\mu^{2}}{16\xi^{3/2}}}\right)  \Gamma\left(  \frac
{l}{2}+\frac{5}{4}-{\frac{\mu^{2}}{16\xi^{3/2}}}\right)  }\nonumber\\
&  \times J_{\frac{l}{2}+\frac{5}{4}-{\frac{\mu^{2}}{16\xi^{3/2}}}}\left(
\frac{l}{2}+\frac{1}{4}-{\frac{\mu^{2}}{16\xi^{3/2}}},{\frac{i\sqrt{2}\mu
}{2\xi^{3/4}}},\frac{3}{2}\left(  l+\frac{1}{2}\right)  +{\frac{\mu^{2}}%
{16\xi^{3/2}}},\frac{i\sqrt{2}E}{2\xi^{1/4}}+\frac{i\sqrt{2}\mu}{4\xi^{3/4}%
}\left(  -{\frac{\mu^{2}}{8\xi^{3/2}}}-l-\frac{1}{2}\right)  \right)
\end{align}
with%
\begin{align}
&  J_{\frac{l}{2}+\frac{5}{4}-{\frac{\mu^{2}}{16\xi^{3/2}}}}\left(  \frac
{l}{2}+\frac{1}{4}-{\frac{\mu^{2}}{16\xi^{3/2}}},{\frac{i\sqrt{2}\mu}%
{2\xi^{3/4}}},\frac{3}{2}\left(  l+\frac{1}{2}\right)  +{\frac{\mu^{2}}%
{16\xi^{3/2}}},\frac{i\sqrt{2}E}{2\xi^{1/4}}+\frac{i\sqrt{2}\mu}{4\xi^{3/4}%
}\left(  -{\frac{\mu^{2}}{8\xi^{3/2}}}-l-\frac{1}{2}\right)  \right)
\nonumber\\
&  =\int_{0}^{\infty}dxx^{\frac{l}{2}+\frac{1}{4}-{\frac{\mu^{2}}{16\xi^{3/2}%
}}}e^{-x^{2}}\nonumber\\
&  \times N\left(  \frac{l}{2}+\frac{1}{4}-{\frac{\mu^{2}}{16\xi^{3/2}}%
},{\frac{i\sqrt{2}\mu}{2\xi^{3/4}}},\frac{3}{2}\left(  l+\frac{1}{2}\right)
+{\frac{\mu^{2}}{16\xi^{3/2}}},\frac{i\sqrt{2}E}{2\xi^{1/4}}+\frac{i\sqrt
{2}\mu}{4\xi^{3/4}}\left(  -{\frac{\mu^{2}}{8\xi^{3/2}}}-l-\frac{1}{2}\right)
,x\right)  .
\end{align}

Eq. (\ref{eigenvalue64}) is an implicit expression of the eigenvalue.

The eigenfunction, by Eqs. (\ref{uexpzf64}) and (\ref{eigenvalue64}), reads%
\begin{equation}
u_{l}\left(  r\right)  =A_{l}K_{1}\left(  l+\frac{1}{2},{\frac{i\sqrt{2}\mu
}{2\xi^{3/4}}},-{\frac{\mu^{2}}{8\xi^{3/2}}},{\frac{i\sqrt{2}E}{2\xi^{1/4}}%
}\right)  \exp\left(  -\frac{\xi^{1/4}}{4}r^{4}-\frac{\mu}{4\xi^{1/2}}%
r^{2}\right)  r^{\frac{\mu^{2}}{8\xi^{3/2}}-\frac{3}{2}}\sum_{n\geq0}%
\frac{a_{n}}{\left(  \frac{\sqrt{2}}{2}\xi^{1/4}r^{2}\right)  ^{n}}.
\end{equation}

\section{The exact solution of $U\left(  r\right)  =\xi r^{6}+\mu r^{2}$
\label{V62}}

In this appendix, we provide an exact solution of the eigenproblem of the
potential%
\begin{equation}
U\left(  r\right)  =\xi r^{6}+\mu r^{2}%
\end{equation}
by solving the radial equation directly. This potential has only bound states.

The radial equation reads%
\begin{equation}
\frac{d^{2}u_{l}\left(  r\right)  }{dr^{2}}+\left[  E-\frac{l\left(
l+1\right)  }{r^{2}}-\xi r^{6}-\mu r^{2}\right]  u_{l}\left(  r\right)  =0.
\label{radialeq62}%
\end{equation}
Using the variable substitution%
\begin{equation}
z=\frac{i\sqrt{2}}{2}\xi^{1/4}{r}^{2} \label{zr62}%
\end{equation}
and introducing $f_{l}\left(  z\right)  $ by%
\begin{equation}
u_{l}\left(  z\right)  =A_{l}\exp\left(  -\frac{z^{2}}{2}\right)  z^{\left(
l+1\right)  /2}f_{l}\left(  z\right)
\end{equation}
with $A_{l}$ a constant, we convert the radial equation (\ref{radialeq62})
into an equation of $f_{l}\left(  z\right)  $:%
\begin{equation}
zf_{l}^{\prime\prime}\left(  z\right)  +\left(  l+\frac{3}{2}-2z^{2}\right)
f_{l}^{\prime}\left(  z\right)  +\left[  \left(  \frac{\mu}{2\sqrt{\xi}%
}-l-\frac{5}{2}\right)  z-{\frac{iE\sqrt{2}}{4\xi^{1/4}}}\right]  f_{l}\left(
z\right)  =0. \label{eqf62}%
\end{equation}
This is a Biconfluent Heun equation \cite{ronveaux1995heun}.

The choice of the boundary condition has been discussed in Ref.
\cite{li2016exact}.

\subsection{The regular solution}

The regular solution is a solution satisfying the boundary condition at $r=0$
\cite{li2016exact}. The regular solution at $r=0$\ should satisfy the boundary
condition $\lim_{r\rightarrow0}u_{l}\left(  r\right)  /r^{l+1}=1$. In this
section, we provide the regular solution of Eq. (\ref{eqf62}).

The Biconfluent Heun equation (\ref{eqf62}) has two linearly independent
solutions \cite{ronveaux1995heun}%
\begin{align}
y_{l}^{\left(  1\right)  }\left(  z\right)   &  =N\left(  l+\frac{1}%
{2},0,{\frac{\mu}{2\sqrt{\xi}}},{\frac{iE\sqrt{2}}{2\xi^{1/4}}},z\right)  ,\\
y_{l}^{\left(  2\right)  }\left(  z\right)   &  =cN\left(  l+\frac{1}%
{2},0,{\frac{\mu}{2\sqrt{\xi}}},{\frac{iE\sqrt{2}}{2\xi^{1/4}}},z\right)  \ln
z+\sum_{n\geq0}d_{n}z^{n-l-1/2},
\end{align}
where%
\begin{equation}
c=\frac{1}{l+1/2}\left[  \frac{iE\sqrt{2}}{4\xi^{1/4}}d_{l-1/2}-d_{l-3/2}%
\left(  \frac{\mu}{2\sqrt{\xi}}+\frac{3}{2}-l\right)  \right]
\end{equation}
is a constant with the coefficient $d_{\nu}$ given by the following recurrence
relation,%
\begin{align}
&  d_{-1}=0,\text{ \ }d_{0}=1,\nonumber\\
&  \left(  v+2\right)  \left(  v+\frac{3}{2}-l\right)  d_{v+2}-\frac
{iE\sqrt{2}}{4\xi^{1/4}}d_{v+1}+\left(  \frac{\mu}{2\sqrt{\xi}}-2v+l-\frac
{3}{2}\right)  d_{v}=0
\end{align}
and $N(\alpha,\beta,\gamma,\delta,z)$ is the biconfluent Heun function
\cite{ronveaux1995heun,slavyanov2000special,li2016exact}.

The biconfluent Heun function $N\left(  l+\frac{1}{2},0,{\frac{\mu}{2\sqrt
{\xi}}},{\frac{iE\sqrt{2}}{2\xi^{1/4}}},z\right)  $ has an expansion at $z=0$
\cite{ronveaux1995heun}:%
\begin{equation}
N\left(  l+\frac{1}{2},0,{\frac{\mu}{2\sqrt{\xi}}},{\frac{iE\sqrt{2}}%
{2\xi^{1/4}}},z\right)  =\sum_{n\geq0}\frac{A_{n}}{\left(  l+\frac{3}%
{2}\right)  _{n}}\frac{z^{n}}{n!},
\end{equation}
where the expansion coefficients is determined by the recurrence relation,%
\begin{align}
A_{0}  &  =1,\text{ \ }A_{1}={\frac{iE\sqrt{2}}{4\xi^{1/4}}},\nonumber\\
A_{n+2}  &  =\frac{iE\sqrt{2}}{4\xi^{1/4}}A_{n+1}-\left(  n+1\right)  \left(
n+l+\frac{3}{2}\right)  \left(  \frac{\mu}{2\sqrt{\xi}}-l-\frac{5}%
{2}-2n\right)  A_{n},
\end{align}
and $\left(  a\right)  _{n}=\Gamma\left(  a+n\right)  /\Gamma\left(  a\right)
$ is Pochhammer's symbol.

Only $y_{l}^{\left(  1\right)  }\left(  z\right)  $ satisfies the boundary
condition for the regular solution at $r=0$, so the radial eigenfunction reads%
\begin{align}
u_{l}\left(  z\right)   &  =A_{l}\exp\left(  -\frac{z^{2}}{2}\right)
z^{\left(  l+1\right)  /2}y_{l}^{\left(  1\right)  }\left(  z\right)
\nonumber\\
&  =A_{l}\exp\left(  -\frac{z^{2}}{2}\right)  z^{\left(  l+1\right)
/2}N\left(  l+\frac{1}{2},0,{\frac{\mu}{2\sqrt{\xi}}},{\frac{iE\sqrt{2}}%
{2\xi^{1/4}}},z\right)  .
\end{align}
By Eq. (\ref{zr62}), we obtain the regular solution,%
\begin{equation}
u_{l}\left(  r\right)  =A_{l}\exp\left(  \frac{\xi^{1/2}}{4}{r}^{4}\right)
r^{l+1}N\left(  l+\frac{1}{2},0,{\frac{\mu}{2\sqrt{\xi}}},{\frac{iE\sqrt{2}%
}{2\xi^{1/4}}},\frac{i\sqrt{2}}{2}\xi^{1/4}{r}^{2}\right)  . \label{regular62}%
\end{equation}

\subsection{The irregular solution}

The irregular solution is a solution satisfying the boundary condition at
$r\rightarrow\infty$ \cite{li2016exact}.

The Biconfluent Heun equation (\ref{eqf62}) has two linearly independent
irregular solutions \cite{ronveaux1995heun}:%
\begin{align}
B_{l}^{+}\left(  l+\frac{1}{2},0,{\frac{\mu}{2\sqrt{\xi}}},{\frac{iE\sqrt{2}%
}{2\xi^{1/4}}},z\right)   &  =e^{z^{2}}B_{l}^{+}\left(  l+\frac{1}%
{2},0,-{\frac{\mu}{2\sqrt{\xi}}},-{\frac{E\sqrt{2}}{2\xi^{1/4}}},-iz\right)
\nonumber\\
&  =e^{z^{2}}\left(  -iz\right)  ^{\frac{1}{2}\left(  -{\frac{\mu}{2\sqrt{\xi
}}}-l-\frac{5}{2}\right)  }\sum_{n\geq0}\frac{a_{n}}{\left(  -iz\right)  ^{n}%
}, \label{f162}%
\end{align}%
\begin{align}
H_{l}^{+}\left(  4l+2,0,\frac{2\xi}{\sqrt{-E}},\frac{-i4\sqrt{2}\mu}{\left(
-E\right)  ^{1/4}},z\right)   &  =e^{z^{2}}H_{l}^{+}\left(  l+\frac{1}%
{2},0,-{\frac{\mu}{2\sqrt{\xi}}},-{\frac{E\sqrt{2}}{2\xi^{1/4}}},-iz\right)
\nonumber\\
&  =\left(  -iz\right)  ^{-\frac{1}{2}\left(  -\frac{\mu}{2\sqrt{\xi}}%
+l+\frac{5}{2}\right)  }\sum_{n\geq0}\frac{e_{n}}{\left(  -iz\right)  ^{n}}
\label{f262}%
\end{align}
with the expansion coefficients given by the recurrence relation
\[
a_{0}=1,\text{ \ }a_{1}=-{\frac{E\sqrt{2}}{8\xi^{1/4}}},
\]%
\begin{equation}
2\left(  n+2\right)  a_{n+2}+{\frac{E\sqrt{2}}{4\xi^{1/4}}}a_{n+1}+\left[
{\frac{\mu^{2}}{16\xi}}-\frac{1}{4}\left(  l+\frac{1}{2}\right)  ^{2}%
+1+{\frac{\mu}{2\sqrt{\xi}}}+n\left(  n+2+\frac{\mu}{2\sqrt{\xi}}\right)
\right]  a_{n}=0
\end{equation}
and%
\[
e_{0}=1,\text{ \ }e_{1}={\frac{E\sqrt{2}}{8\xi^{1/4}},}%
\]%
\begin{equation}
2\left(  n+2\right)  e_{n+2}-{\frac{E\sqrt{2}}{4\xi^{1/4}}}e_{n+1}-\left[
\frac{\mu^{2}}{16\xi}-\frac{1}{4}\left(  l+\frac{1}{2}\right)  ^{2}%
+1-{\frac{\mu}{2\sqrt{\xi}}}+n\left(  n+2-{\frac{\mu}{2\sqrt{\xi}}}\right)
\right]  e_{n}=0.
\end{equation}

\subsection{Eigenfunctions and eigenvalues}

To construct the solution, we first express the regular solution
(\ref{regularo132}) as a linear combination of the two irregular solutions
(\ref{f162}) and (\ref{f262}).

The regular solution (\ref{regular62}), with the relation
\cite{ronveaux1995heun,li2016exact}%

\begin{align}
N\left(  l+\frac{1}{2},0,{\frac{\mu}{2\sqrt{\xi}}},{\frac{iE\sqrt{2}}%
{2\xi^{1/4}}},z\right)   &  =K_{1}\left(  l+\frac{1}{2},0,{\frac{\mu}%
{2\sqrt{\xi}}},{\frac{iE\sqrt{2}}{2\xi^{1/4}}}\right)  B_{l}^{+}\left(
l+\frac{1}{2},0,{\frac{\mu}{2\sqrt{\xi}}},{\frac{iE\sqrt{2}}{2\xi^{1/4}}%
},z\right)  \nonumber\\
&  +K_{2}\left(  l+\frac{1}{2},0,{\frac{\mu}{2\sqrt{\xi}}},{\frac{iE\sqrt{2}%
}{2\xi^{1/4}}}\right)  H_{l}^{+}\left(  l+\frac{1}{2},0,{\frac{\mu}{2\sqrt
{\xi}}},{\frac{iE\sqrt{2}}{2\xi^{1/4}}},z\right)
\end{align}
and the expansions (\ref{f162}) and (\ref{f262}), becomes%
\begin{align}
u_{l}\left(  r\right)   &  =A_{l}K_{1}\left(  l+\frac{1}{2},0,{\frac{\mu
}{2\sqrt{\xi}}},{\frac{iE\sqrt{2}}{2\xi^{1/4}}}\right)  \exp\left(  -\frac
{\xi^{1/2}}{4}{r}^{4}\right)  r^{-{\frac{\mu}{2\sqrt{\xi}}}-\frac{3}{2}}%
\sum_{n\geq0}\frac{a_{n}}{\left(  \frac{\sqrt{2}}{2}\xi^{1/4}{r}^{2}\right)
^{n}}\nonumber\\
&  +A_{l}K_{2}\left(  l+\frac{1}{2},0,{\frac{\mu}{2\sqrt{\xi}}},{\frac
{iE\sqrt{2}}{2\xi^{1/4}}}\right)  \exp\left(  \frac{\xi^{1/2}}{4}{r}%
^{4}\right)  r^{\frac{\mu}{2\sqrt{\xi}}-\frac{3}{2}}\sum_{n\geq0}\frac{e_{n}%
}{\left(  \frac{\sqrt{2}}{2}\xi^{1/4}{r}^{2}\right)  ^{n}},\label{uexpzf62}%
\end{align}
where $K_{1}\left(  l+\frac{1}{2},0,{\frac{\mu}{2\sqrt{\xi}}},{\frac
{iE\sqrt{2}}{2\xi^{1/4}}}\right)  $ and $K_{2}\left(  l+\frac{1}{2}%
,0,{\frac{\mu}{2\sqrt{\xi}}},{\frac{iE\sqrt{2}}{2\xi^{1/4}}}\right)  $ are
combination coefficients\ and $z=\frac{i\sqrt{2}}{2}\xi^{1/4}{r}^{2}$.

The boundary condition of bound states, $\left.  u\left(  r\right)
\right\vert _{r\rightarrow\infty}\rightarrow0$, requires that the coefficient
of the second term must vanish since this term diverges when $r\rightarrow
\infty$, i.e.,
\begin{equation}
K_{2}\left(  l+\frac{1}{2},0,{\frac{\mu}{2\sqrt{\xi}}},{\frac{iE\sqrt{2}}%
{2\xi^{1/4}}}\right)  =0, \label{eigenvalue62}%
\end{equation}
where%
\begin{align}
K_{2}\left(  l+\frac{1}{2},0,{\frac{\mu}{2\sqrt{\xi}}},{\frac{iE\sqrt{2}}%
{2\xi^{1/4}}}\right)   &  =\frac{\Gamma\left(  l+\frac{3}{2}\right)  }%
{\Gamma\left(  \frac{l}{2}+\frac{1}{4}-\frac{\mu}{4\sqrt{\xi}}\right)
\Gamma\left(  \frac{l}{2}+\frac{5}{4}+\frac{\mu}{4\sqrt{\xi}}\right)
}\nonumber\\
&  \times J_{\frac{l}{2}+\frac{5}{4}+\frac{\mu}{4\sqrt{\xi}}}\left(  \frac
{l}{2}+\frac{1}{4}+\frac{\mu}{4\sqrt{\xi}},0,\frac{3}{2}\left(  l+\frac{1}%
{2}\right)  -{\frac{\mu}{4\sqrt{\xi}}},\frac{iE\sqrt{2}}{2\xi^{1/4}}\right)
\end{align}
with%
\begin{align}
&  J_{\frac{l}{2}+\frac{5}{4}+\frac{\mu}{4\sqrt{\xi}}}\left(  \frac{l}%
{2}+\frac{1}{4}+\frac{\mu}{4\sqrt{\xi}},0,\frac{3}{2}\left(  l+\frac{1}%
{2}\right)  -{\frac{\mu}{4\sqrt{\xi}}},\frac{iE\sqrt{2}}{2\xi^{1/4}}\right)
\nonumber\\
&  =\int_{0}^{\infty}x^{\frac{l}{2}+\frac{1}{4}+\frac{\mu}{4\sqrt{\xi}}%
}e^{-x^{2}}N\left(  \frac{l}{2}+\frac{1}{4}+\frac{\mu}{4\sqrt{\xi}},0,\frac
{3}{2}\left(  l+\frac{1}{2}\right)  -{\frac{\mu}{4\sqrt{\xi}}},\frac
{iE\sqrt{2}}{2\xi^{1/4}},x\right)  dx.
\end{align}

Eq. (\ref{eigenvalue62}) is an implicit expression of the eigenvalue.

The eigenfunction, by Eqs. (\ref{uexpzf62}) and (\ref{eigenvalue62}), reads%
\begin{equation}
u_{l}\left(  r\right)  =A_{l}K_{1}\left(  l+\frac{1}{2},0,{\frac{\mu}%
{2\sqrt{\xi}}},{\frac{iE\sqrt{2}}{2\xi^{1/4}}}\right)  \exp\left(  -\frac
{\xi^{1/2}}{4}{r}^{4}\right)  r^{-{\frac{\mu}{2\sqrt{\xi}}}-\frac{3}{2}}%
\sum_{n\geq0}\frac{a_{n}}{\left(  \frac{\sqrt{2}}{2}\xi^{1/4}{r}^{2}\right)
^{n}}.
\end{equation}

\section{The exact solution of $U\left(  r\right)  =\xi r^{2}+\mu
r$\label{V21}}

In this appendix, we provide an exact solution of the eigenproblem of the
potential
\begin{equation}
U\left(  r\right)  =\xi r^{2}+\mu r
\end{equation}
by solving the radial equation directly. This potential has only bound states.

The radial equation reads%
\begin{equation}
\frac{d^{2}u_{l}\left(  r\right)  }{dr^{2}}+\left[  E-\frac{l\left(
l+1\right)  }{r^{2}}-\xi r^{2}-\mu r\right]  u_{l}\left(  r\right)  =0.
\label{radialeq21}%
\end{equation}
Using the variable substitution%
\begin{equation}
z=i{\xi}^{1/4}r \label{zr21}%
\end{equation}
and introducing $f_{l}\left(  z\right)  $ by%
\begin{equation}
u_{l}\left(  z\right)  =A_{l}\exp\left(  -\frac{z^{2}}{2}-\frac{\beta}%
{2}z\right)  z^{l+1}f_{l}\left(  z\right)
\end{equation}
with $A_{l}$ a constant, we convert the radial equation (\ref{radialeq21})
into an equation of $f_{l}\left(  z\right)  $:%
\begin{equation}
zf_{l}^{\prime\prime}\left(  z\right)  +\left(  2l+2-{\frac{i\mu}{{\xi}^{3/4}%
}}z-2z^{2}\right)  f_{l}^{\prime}\left(  z\right)  +\left[  \left(  -{\frac
{E}{{\xi}^{1/2}}}-{\frac{{\mu}^{2}}{4{\xi}^{3/2}}}-2l-3\right)  z-\frac{i\mu
}{{\xi}^{3/4}}\left(  l+1\right)  \right]  f_{l}\left(  z\right)  =0.
\label{eqf21}%
\end{equation}
This is a Biconfluent Heun equation \cite{ronveaux1995heun}.

The choice of the boundary condition has been discussed in Ref.
\cite{li2016exact}.

\subsection{The regular solution}

The regular solution is a solution satisfying the boundary condition at $r=0$
\cite{li2016exact}. The regular solution at $r=0$\ should satisfy the boundary
condition $\lim_{r\rightarrow0}u_{l}\left(  r\right)  /r^{l+1}=1$. In this
section, we provide the regular solution of Eq. (\ref{eqf21}).

The Biconfluent Heun equation (\ref{eqf21}) has two linearly independent
solutions \cite{ronveaux1995heun}%
\begin{align}
y_{l}^{\left(  1\right)  }\left(  z\right)   &  =N\left(  2l+1,{\frac{i\mu
}{{\xi}^{3/4}}},-{\frac{E}{{\xi}^{1/2}}}-{\frac{{\mu}^{2}}{4{\xi}^{3/2}}%
,0},z\right)  ,\\
y_{l}^{\left(  2\right)  }\left(  z\right)   &  =cN\left(  2l+1,{\frac{i\mu
}{{\xi}^{3/4}}},-{\frac{E}{{\xi}^{1/2}}}-{\frac{{\mu}^{2}}{4{\xi}^{3/2}}%
},0,z\right)  \ln z+\sum_{n\geq0}d_{n}z^{n-2l-1},
\end{align}
where%
\begin{equation}
c=\frac{1}{2l+1}\left[  \frac{il\mu}{{\xi}^{3/4}}d_{2l}-d_{2l-1}\left(
-{\frac{E}{{\xi}^{1/2}}}-{\frac{{\mu}^{2}}{4{\xi}^{3/2}}}+1-2l\right)
\right]
\end{equation}
is a constant with the coefficient $d_{\nu}$ given by the following recurrence
relation,%
\begin{align}
&  d_{-1}=0,\text{ \ }d_{0}=1,\nonumber\\
&  \left(  v+2\right)  \left(  v+1-2l\right)  d_{v+2}-{\frac{i\mu}{{\xi}%
^{3/4}}}\left(  v+1-l\right)  d_{v+1}+\left(  -{\frac{E}{{\xi}^{1/2}}}%
-{\frac{{\mu}^{2}}{4{\xi}^{3/2}}}-2v-1+2l\right)  d_{v}=0
\end{align}
and $N(\alpha,\beta,\gamma,\delta,z)$ is the biconfluent Heun function
\cite{ronveaux1995heun,slavyanov2000special,li2016exact}.

The biconfluent Heun function $N\left(  2l+1,{\frac{i\mu}{{\xi}^{3/4}}%
},-{\frac{E}{{\xi}^{1/2}}}-{\frac{{\mu}^{2}}{4{\xi}^{3/2}},0},z\right)  $ has
an expansion at $z=0$ \cite{ronveaux1995heun}:%
\begin{equation}
N\left(  2l+1,{\frac{i\mu}{{\xi}^{3/4}}},-{\frac{E}{{\xi}^{1/2}}}-{\frac{{\mu
}^{2}}{4{\xi}^{3/2}},0},z\right)  =\sum_{n\geq0}\frac{A_{n}}{\left(
2l+2\right)  _{n}}\frac{z^{n}}{n!},
\end{equation}
where the expansion coefficients is determined by the recurrence relation,%
\begin{align}
A_{0}  &  =1,\text{ \ }A_{1}=\frac{i\mu}{{\xi}^{3/4}}\left(  l+1\right)
,\nonumber\\
A_{n+2}  &  =\frac{i\mu}{{\xi}^{3/4}}\left(  n+l+2\right)  A_{n+1}-\left(
n+1\right)  \left(  n+2l+2\right)  \left(  -{\frac{E}{{\xi}^{1/2}}}%
-{\frac{{\mu}^{2}}{4{\xi}^{3/2}}}-3-2l-2n\right)  A_{n},
\end{align}
and $\left(  a\right)  _{n}=\Gamma\left(  a+n\right)  /\Gamma\left(  a\right)
$ is Pochhammer's symbol.

Only $y_{l}^{\left(  1\right)  }\left(  z\right)  $ satisfies the boundary
condition for the regular solution at $r=0$, so the radial eigenfunction reads%
\begin{align}
u_{l}\left(  z\right)   &  =A_{l}\exp\left(  -\frac{z^{2}}{2}-\frac{\beta}%
{2}z\right)  z^{l+1}y_{l}^{\left(  1\right)  }\left(  z\right)  \nonumber\\
&  =A_{l}\exp\left(  -\frac{z^{2}}{2}-\frac{\beta}{2}z\right)  z^{l+1}N\left(
2l+1,{\frac{i\mu}{{\xi}^{3/4}}},-{\frac{E}{{\xi}^{1/2}}}-{\frac{{\mu}^{2}%
}{4{\xi}^{3/2}},0},z\right)  .
\end{align}
By Eq. (\ref{zr21}), we obtain the regular solution,%
\begin{equation}
u_{l}\left(  r\right)  =A_{l}\exp\left(  \frac{{\xi}^{1/2}}{2}r^{2}+\frac{\mu
}{2{\xi}^{1/2}}r\right)  r^{l+1}N\left(  2l+1,{\frac{i\mu}{{\xi}^{3/4}}%
},-{\frac{E}{{\xi}^{1/2}}}-{\frac{{\mu}^{2}}{4{\xi}^{3/2}}},0,i{\xi}%
^{1/4}r\right)  .\label{regular21}%
\end{equation}

\subsection{The irregular solution}

The irregular solution is a solution satisfying the boundary condition at
$r\rightarrow\infty$ \cite{li2016exact}.

The Biconfluent Heun equation (\ref{eqf21}) has two linearly independent
irregular solutions \cite{ronveaux1995heun}:%
\begin{align}
B_{l}^{+}\left(  2l+1,{\frac{i\mu}{{\xi}^{3/4}}},-{\frac{E}{{\xi}^{1/2}}%
}-{\frac{{\mu}^{2}}{4{\xi}^{3/2}}},0,z\right)   &  =\exp\left(  \frac{i\mu
}{{\xi}^{3/4}}z+z^{2}\right)  B_{l}^{+}\left(  2l+1,{\frac{\mu}{{\xi}^{3/4}}%
},{\frac{E}{{\xi}^{1/2}}}+{\frac{{\mu}^{2}}{4{\xi}^{3/2}}},0,-iz\right)
\nonumber\\
&  =\exp\left(  \frac{i\mu}{{\xi}^{3/4}}z+z^{2}\right)  \left(  -iz\right)
^{\frac{1}{2}\left(  {\frac{E}{{\xi}^{1/2}}}+{\frac{{\mu}^{2}}{4{\xi}^{3/2}}%
-}2l-3\right)  }\sum_{n\geq0}\frac{a_{n}}{\left(  -iz\right)  ^{n}},
\label{f121}%
\end{align}%
\begin{align}
H_{l}^{+}\left(  2l+1,{\frac{i\mu}{{\xi}^{3/4}}},-{\frac{E}{{\xi}^{1/2}}%
}-{\frac{{\mu}^{2}}{4{\xi}^{3/2}}},0,z\right)   &  =\exp\left(  \frac{i\mu
}{{\xi}^{3/4}}z+z^{2}\right)  H_{l}^{+}\left(  2l+1,{\frac{\mu}{{\xi}^{3/4}}%
},{\frac{E}{{\xi}^{1/2}}}+{\frac{{\mu}^{2}}{4{\xi}^{3/2}}},0,-iz\right)
\nonumber\\
&  =\left(  -iz\right)  ^{-\frac{1}{2}\left(  {\frac{E}{{\xi}^{1/2}}}%
+{\frac{{\mu}^{2}}{4{\xi}^{3/2}}}+2l+3\right)  }\sum_{n\geq0}\frac{e_{n}%
}{\left(  -iz\right)  ^{n}} \label{f221}%
\end{align}
with the expansion coefficients given by the recurrence relation
\[
a_{0}=1,\text{ \ }a_{1}=\frac{\mu}{4{\xi}^{3/4}}\left(  {\frac{E}{{\xi}^{1/2}%
}}+{\frac{{\mu}^{2}}{4{\xi}^{3/2}}}-1\right)  ,
\]%
\begin{align}
&  2\left(  n+2\right)  a_{n+2}+\frac{\mu}{{\xi}^{3/4}}\left(  \frac{3}%
{2}-{\frac{E}{2{\xi}^{1/2}}}-{\frac{{\mu}^{2}}{8{\xi}^{3/2}}}+n\right)
a_{n+1}\nonumber\\
&  +\left[  \frac{1}{4}\left(  {\frac{E}{{\xi}^{1/2}}}+{\frac{{\mu}^{2}}%
{4{\xi}^{3/2}}}\right)  ^{2}-\frac{1}{4}\left(  2l+1\right)  ^{2}+1-{\frac
{E}{{\xi}^{1/2}}}-{\frac{{\mu}^{2}}{4{\xi}^{3/2}}}+n\left(  n+2-{\frac{E}%
{{\xi}^{1/2}}}-{\frac{{\mu}^{2}}{4{\xi}^{3/2}}}\right)  \right]  a_{n}=0
\end{align}
and%
\[
e_{0}=1,\text{ \ }e_{1}=-\frac{\mu}{4{\xi}^{3/4}}\left(  \gamma+1\right)  ,
\]%
\begin{align}
&  2\left(  n+2\right)  e_{n+2}+\frac{\mu}{{\xi}^{3/4}}\left(  \frac{3}%
{2}+{\frac{E}{2{\xi}^{1/2}}}+{\frac{{\mu}^{2}}{8{\xi}^{3/2}}}+n\right)
e_{n+1}\nonumber\\
&  -\left[  \frac{1}{4}\left(  {\frac{E}{{\xi}^{1/2}}}+{\frac{{\mu}^{2}}%
{4{\xi}^{3/2}}}\right)  ^{2}-\frac{1}{4}\left(  2l+1\right)  ^{2}+1+{\frac
{E}{{\xi}^{1/2}}}+{\frac{{\mu}^{2}}{4{\xi}^{3/2}}}+n\left(  n+2+{\frac{E}%
{{\xi}^{1/2}}}+{\frac{{\mu}^{2}}{4{\xi}^{3/2}}}\right)  \right]  e_{n}=0.
\end{align}

\subsection{Eigenfunctions and eigenvalues}

To construct the solution, we first express the regular solution
(\ref{regular21}) as a linear combination of the two irregular solutions
(\ref{f121}) and (\ref{f221}).

The regular solution (\ref{regular21}), with the relation
\cite{ronveaux1995heun,li2016exact}%

\begin{align}
&  N\left(  2l+1,{\frac{i\mu}{{\xi}^{3/4}}},-{\frac{E}{{\xi}^{1/2}}}%
-{\frac{{\mu}^{2}}{4{\xi}^{3/2}}},0,z\right)  \nonumber\\
&  =K_{1}\left(  2l+1,{\frac{i\mu}{{\xi}^{3/4}}},-{\frac{E}{{\xi}^{1/2}}%
}-{\frac{{\mu}^{2}}{4{\xi}^{3/2}}},0\right)  B_{l}^{+}\left(  2l+1,{\frac
{i\mu}{{\xi}^{3/4}}},-{\frac{E}{{\xi}^{1/2}}}-{\frac{{\mu}^{2}}{4{\xi}^{3/2}}%
},0,z\right)  \nonumber\\
&  +K_{2}\left(  2l+1,{\frac{i\mu}{{\xi}^{3/4}}},-{\frac{E}{{\xi}^{1/2}}%
}-{\frac{{\mu}^{2}}{4{\xi}^{3/2}}},0\right)  H_{l}^{+}\left(  2l+1,{\frac
{i\mu}{{\xi}^{3/4}}},-{\frac{E}{{\xi}^{1/2}}}-{\frac{{\mu}^{2}}{4{\xi}^{3/2}}%
},0,z\right)
\end{align}
and the expansions (\ref{f121}) and (\ref{f221}), becomes%
\begin{align}
u_{l}\left(  r\right)   &  =A_{l}K_{1}\left(  2l+1,{\frac{i\mu}{{\xi}^{3/4}}%
},-{\frac{E}{{\xi}^{1/2}}}-{\frac{{\mu}^{2}}{4{\xi}^{3/2}}},0\right)
\exp\left(  -\frac{{\xi}^{1/2}}{2}r^{2}-\frac{\mu}{2{\xi}^{1/2}}r\right)
r^{{\frac{E}{2{\xi}^{1/2}}}+{\frac{{\mu}^{2}}{8{\xi}^{3/2}}}-\frac{1}{2}}%
\sum_{n\geq0}\frac{a_{n}}{\left(  {\xi}^{1/4}r\right)  ^{n}}\nonumber\\
&  +A_{l}K_{2}\left(  2l+1,{\frac{i\mu}{{\xi}^{3/4}}},-{\frac{E}{{\xi}^{1/2}}%
}-{\frac{{\mu}^{2}}{4{\xi}^{3/2}}},0\right)  \exp\left(  \frac{{\xi}^{1/2}}%
{2}r^{2}+\frac{\mu}{2{\xi}^{1/2}}r\right)  r^{-{\frac{E}{2{\xi}^{1/2}}}%
-{\frac{{\mu}^{2}}{8{\xi}^{3/2}}}-\frac{1}{2}}\sum_{n\geq0}\frac{e_{n}%
}{\left(  {\xi}^{1/4}r\right)  ^{n}},\label{uexpzf21}%
\end{align}
where $K_{1}\left(  2l+1,{\frac{i\mu}{{\xi}^{3/4}}},-{\frac{E}{{\xi}^{1/2}}%
}-{\frac{{\mu}^{2}}{4{\xi}^{3/2}}},0\right)  $ and $K_{2}\left(
2l+1,{\frac{i\mu}{{\xi}^{3/4}}},-{\frac{E}{{\xi}^{1/2}}}-{\frac{{\mu}^{2}%
}{4{\xi}^{3/2}}},0\right)  $ are combination coefficients\ and $z=i{\xi}%
^{1/4}r$.

The boundary condition of bound states, $\left.  u\left(  r\right)
\right\vert _{r\rightarrow\infty}\rightarrow0$, requires that the coefficient
of the second term must vanish since this term diverges when $r\rightarrow
\infty$, i.e.,
\begin{equation}
K_{2}\left(  2l+1,{\frac{i\mu}{{\xi}^{3/4}}},-{\frac{E}{{\xi}^{1/2}}}%
-{\frac{{\mu}^{2}}{4{\xi}^{3/2}}},0\right)  =0, \label{eigenvalue21}%
\end{equation}
where%
\begin{align}
&  K_{2}\left(  2l+1,{\frac{i\mu}{{\xi}^{3/4}}},-{\frac{E}{{\xi}^{1/2}}%
}-{\frac{{\mu}^{2}}{4{\xi}^{3/2}}},0\right) \nonumber\\
&  =\frac{\Gamma\left(  2l+2\right)  }{\Gamma\left(  l+\frac{1}{2}+{\frac
{E}{2{\xi}^{1/2}}}+{\frac{{\mu}^{2}}{8{\xi}^{3/2}}}\right)  \Gamma\left(
l+\frac{3}{2}-{\frac{E}{2{\xi}^{1/2}}}-{\frac{{\mu}^{2}}{8{\xi}^{3/2}}%
}\right)  }\nonumber\\
&  \times J_{l+\frac{3}{2}-{\frac{E}{2{\xi}^{1/2}}}-{\frac{{\mu}^{2}}{8{\xi
}^{3/2}}}}\left(  l+\frac{1}{2}-{\frac{E}{2{\xi}^{1/2}}}-{\frac{{\mu}^{2}%
}{8{\xi}^{3/2}}},\frac{i\mu}{{\xi}^{3/4}},\frac{3}{2}\left(  2l+1\right)
+{\frac{E}{2{\xi}^{1/2}}}\right. \nonumber\\
&  +\left.  {\frac{{\mu}^{2}}{8{\xi}^{3/2}}},{\frac{i\mu}{2{\xi}^{3/4}}%
}\left(  -{\frac{E}{{\xi}^{1/2}}}-{\frac{{\mu}^{2}}{4{\xi}^{3/2}}%
}-2l-1\right)  \right)
\end{align}
with%
\begin{align}
&  J_{l+\frac{3}{2}-{\frac{E}{2{\xi}^{1/2}}}-{\frac{{\mu}^{2}}{8{\xi}^{3/2}}}%
}\left(  l+\frac{1}{2}-{\frac{E}{2{\xi}^{1/2}}}-{\frac{{\mu}^{2}}{8{\xi}%
^{3/2}}},\frac{i\mu}{{\xi}^{3/4}},\frac{3}{2}\left(  2l+1\right)  +{\frac
{E}{2{\xi}^{1/2}}}+{\frac{{\mu}^{2}}{8{\xi}^{3/2}}},{\frac{i\mu}{2{\xi}^{3/4}%
}}\left(  -{\frac{E}{{\xi}^{1/2}}}-{\frac{{\mu}^{2}}{4{\xi}^{3/2}}%
}-2l-1\right)  \right) \nonumber\\
&  =\int_{0}^{\infty}dxx^{l+\frac{1}{2}-{\frac{E}{2{\xi}^{1/2}}}-{\frac{{\mu
}^{2}}{8{\xi}^{3/2}}}}e^{-x^{2}}\\
&  \times N\left(  l+\frac{1}{2}-{\frac{E}{2{\xi}^{1/2}}}-{\frac{{\mu}^{2}%
}{8{\xi}^{3/2}}},\frac{i\mu}{{\xi}^{3/4}},\frac{3}{2}\left(  2l+1\right)
+{\frac{E}{2{\xi}^{1/2}}}+{\frac{{\mu}^{2}}{8{\xi}^{3/2}}},{\frac{i\mu}{2{\xi
}^{3/4}}}\left(  -{\frac{E}{{\xi}^{1/2}}}-{\frac{{\mu}^{2}}{4{\xi}^{3/2}}%
}-2l-1\right)  ,x\right)  .
\end{align}

Eq. (\ref{eigenvalue21}) is an implicit expression of the eigenvalue.

The eigenfunction, by Eqs. (\ref{uexpzf21}) and (\ref{eigenvalue21}), reads%
\begin{equation}
u_{l}\left(  r\right)  =A_{l}K_{1}\left(  2l+1,{\frac{i\mu}{{\xi}^{3/4}}%
},-{\frac{E}{{\xi}^{1/2}}}-{\frac{{\mu}^{2}}{4{\xi}^{3/2}}},0\right)
\exp\left(  -\frac{{\xi}^{1/2}}{2}r^{2}-\frac{\mu}{2{\xi}^{1/2}}r\right)
r^{{\frac{E}{2{\xi}^{1/2}}}+{\frac{{\mu}^{2}}{8{\xi}^{3/2}}}-\frac{1}{2}}%
\sum_{n\geq0}\frac{a_{n}}{\left(  {\xi}^{1/4}r\right)  ^{n}}.
\end{equation}

\section{The exact solution of $U\left(  r\right)  =\frac{\xi}{r}+\frac{\mu
}{\sqrt{r}}$\label{Vm1m12}}

In this appendix, we provide an exact solution of the eigenproblem of the
potential
\begin{equation}
U\left(  r\right)  =\frac{\xi}{r}+\frac{\mu}{\sqrt{r}}%
\end{equation}
by solving the radial equation directly. This potential has both bound states
and scattering states.

The radial equation of the potential $U\left(  r\right)  =\frac{\xi}{r}%
+\frac{\mu}{\sqrt{r}}$ reads%
\begin{equation}
\frac{d^{2}}{dr^{2}}u_{l}\left(  r\right)  +\left[  E-\frac{l\left(
l+1\right)  }{r^{2}}-\frac{\xi}{r}-\frac{\mu}{\sqrt{r}}\right]  u_{l}\left(
r\right)  =0. \label{radialeqo112}%
\end{equation}
Using the variable substitution%
\begin{equation}
z=i\left(  -E\right)  ^{1/4}\sqrt{2r} \label{zro112}%
\end{equation}
and introducing $f_{l}\left(  z\right)  $ by%
\begin{equation}
u_{l}\left(  z\right)  =A_{l}\exp\left(  -\frac{z^{2}}{2}-\frac{\beta}%
{2}z\right)  z^{2\left(  l+1\right)  }f_{l}\left(  z\right)
\end{equation}
with $A_{l}$ a constant, we convert the radial equation (\ref{radialeqo112})
into an equation of $f_{l}\left(  z\right)  $:%
\begin{align}
zf_{l}^{\prime\prime}  &  \left(  z\right)  +\left(  4l+3-\frac{i\sqrt{2}\mu
}{\left(  -E\right)  ^{3/4}}z-2z^{2}\right)  f_{l}^{\prime}\left(  z\right)
\nonumber\\
&  +\left[  \left(  {\frac{2\xi}{\left(  -E\right)  ^{1/2}}}-{\frac{{\mu}^{2}%
}{2\left(  -E\right)  ^{3/2}}}-4l-4\right)  z-\frac{i\sqrt{2}\mu}{2\left(
-E\right)  ^{3/4}}\left(  4l+3\right)  \right]  f_{l}\left(  z\right)  =0.
\label{eqfo112}%
\end{align}
This is a Biconfluent Heun equation \cite{ronveaux1995heun}.

The choice of the boundary condition has been discussed in Ref.
\cite{li2016exact}.

\subsection{The regular solution}

The regular solution is a solution satisfying the boundary condition at $r=0$
\cite{li2016exact}. The regular solution at $r=0$\ should satisfy the boundary
condition $\lim_{r\rightarrow0}u_{l}\left(  r\right)  /r^{l+1}=1$ for both
bound states and scattering states. In this section, we provide the regular
solution of Eq. (\ref{eqf21}).

The Biconfluent Heun equation (\ref{eqf21}) has two linearly independent
solutions \cite{ronveaux1995heun}%
\begin{align}
y_{l}^{\left(  1\right)  }\left(  z\right)   &  =N\left(  4l+2,{\frac
{i\sqrt{2}\mu}{\left(  -E\right)  ^{3/4}}},{\frac{2\xi}{\left(  -E\right)
^{1/2}}}-{\frac{{\mu}^{2}}{2\left(  -E\right)  ^{3/2}}},0,z\right)  ,\\
y_{l}^{\left(  2\right)  }\left(  z\right)   &  =cN\left(  4l+2,{\frac
{i\sqrt{2}\mu}{\left(  -E\right)  ^{3/4}}},{\frac{2\xi}{\left(  -E\right)
^{1/2}}}-{\frac{{\mu}^{2}}{2\left(  -E\right)  ^{3/2}}},0,z\right)  \ln
z+\sum_{n\geq0}d_{n}z^{n-4l-2},
\end{align}
where%
\begin{equation}
c=\frac{1}{4l+2}\left\{  {\frac{i\sqrt{2}\mu}{2\left(  -E\right)  ^{3/4}}%
}\left(  4l+1\right)  d_{4l+1}-d_{4l}\left[  {\frac{2\xi}{\left(  -E\right)
^{1/2}}}-{\frac{{\mu}^{2}}{2\left(  -E\right)  ^{3/2}}}-4l\right]  \right\}
\end{equation}
is a constant with the coefficient $d_{\nu}$ given by the following recurrence
relation,%
\begin{align}
&  d_{-1}=0,\text{ \ }d_{0}=1,\nonumber\\
&  \left(  v+2\right)  \left(  v-4l\right)  d_{v+2}-{\frac{i\sqrt{2}\mu
}{2\left(  -E\right)  ^{3/4}}}\left(  2v+1-4l\right)  d_{v+1}+\left[
{\frac{2\xi}{\left(  -E\right)  ^{1/2}}}-{\frac{{\mu}^{2}}{2\left(  -E\right)
^{3/2}}}-2v+4l\right]  d_{v}=0
\end{align}
and $N(\alpha,\beta,\gamma,\delta,z)$ is the biconfluent Heun function
\cite{ronveaux1995heun,slavyanov2000special,li2016exact}.

The biconfluent Heun function $N\left(  4l+2,{\frac{i\sqrt{2}\mu}{\left(
-E\right)  ^{3/4}}},{\frac{2\xi}{\left(  -E\right)  ^{1/2}}}-{\frac{{\mu}^{2}%
}{2\left(  -E\right)  ^{3/2}}},0,z\right)  $ has an expansion at $z=0$
\cite{ronveaux1995heun}:%
\begin{equation}
N\left(  4l+2,{\frac{i\sqrt{2}\mu}{\left(  -E\right)  ^{3/4}}},{\frac{2\xi
}{\left(  -E\right)  ^{1/2}}}-{\frac{{\mu}^{2}}{2\left(  -E\right)  ^{3/2}}%
},0,z\right)  =\sum_{n\geq0}\frac{A_{n}}{\left(  4l+3\right)  _{n}}\frac
{z^{n}}{n!},
\end{equation}
where the expansion coefficients is determined by the recurrence relation,%
\begin{align}
A_{0}  &  =1,\text{ \ }A_{1}={\frac{i\sqrt{2}\mu}{2\left(  -E\right)  ^{3/4}}%
}\left(  4l+3\right)  ,\nonumber\\
A_{n+2}  &  =\frac{i\sqrt{2}\mu}{2\left(  -E\right)  ^{3/4}}\left(
2n+4l+5\right)  A_{n+1}-\left(  n+1\right)  \left(  n+4l+3\right)  \left[
{\frac{2\xi}{\left(  -E\right)  ^{1/2}}}-{\frac{{\mu}^{2}}{2\left(  -E\right)
^{3/2}}}-4l-4-2n\right]  A_{n},
\end{align}
and $\left(  a\right)  _{n}=\Gamma\left(  a+n\right)  /\Gamma\left(  a\right)
$ is Pochhammer's symbol.

Only $y_{l}^{\left(  1\right)  }\left(  z\right)  $ satisfies the boundary
condition for the regular solution at $r=0$, so the radial eigenfunction reads%
\begin{align}
u_{l}\left(  z\right)   &  =A_{l}\exp\left(  -\frac{z^{2}}{2}-\frac{\beta}%
{2}z\right)  z^{2\left(  l+1\right)  }y_{l}^{\left(  1\right)  }\left(
z\right) \nonumber\\
&  =A_{l}\exp\left(  -\frac{z^{2}}{2}-\frac{\beta}{2}z\right)  z^{2\left(
l+1\right)  }N\left(  4l+2,{\frac{i\sqrt{2}\mu}{\left(  -E\right)  ^{3/4}}%
},{\frac{2\xi}{\left(  -E\right)  ^{1/2}}}-{\frac{{\mu}^{2}}{2\left(
-E\right)  ^{3/2}}},0,z\right)  .
\end{align}
By Eq. (\ref{zr21}), we obtain the regular solution%
\begin{align}
u_{l}\left(  r\right)   &  =A_{l}\exp\left(  \left(  -E\right)  ^{1/2}%
r+\frac{\mu}{\left(  -E\right)  ^{1/2}}\sqrt{r}\right)  r^{l+1}\nonumber\\
&  \times N\left(  4l+2,{\frac{i\sqrt{2}\mu}{\left(  -E\right)  ^{3/4}}%
},{\frac{2\xi}{\left(  -E\right)  ^{1/2}}}-{\frac{{\mu}^{2}}{2\left(
-E\right)  ^{3/2}}},0,i\left(  -E\right)  ^{1/4}\sqrt{2r}\right)  .
\label{regularo112}%
\end{align}

\subsection{The irregular solution}

The irregular solution is a solution satisfying the boundary condition at
$r\rightarrow\infty$ \cite{li2016exact}. The boundary conditions for bound
states and scattering states at $r\rightarrow\infty$ are different.

The Biconfluent Heun equation (\ref{eqfo112}) has two linearly independent
irregular solutions \cite{ronveaux1995heun}:%
\begin{align}
&  B_{l}^{+}\left(  4l+2,{\frac{i\sqrt{2}\mu}{\left(  -E\right)  ^{3/4}}%
},{\frac{2\xi}{\left(  -E\right)  ^{1/2}}}-{\frac{{\mu}^{2}}{2\left(
-E\right)  ^{3/2}}},0,z\right) \nonumber\\
&  =\exp\left(  \frac{i\sqrt{2}\mu}{\left(  -E\right)  ^{3/4}}z+z^{2}\right)
B_{l}^{+}\left(  4l+2,{\frac{\sqrt{2}\mu}{\left(  -E\right)  ^{3/4}}}%
,{\frac{{\mu}^{2}}{2\left(  -E\right)  ^{3/2}}}-{\frac{2\xi}{\left(
-E\right)  ^{1/2}}},0,-iz\right) \nonumber\\
&  =\exp\left(  \frac{i\sqrt{2}\mu}{\left(  -E\right)  ^{3/4}}z+z^{2}\right)
\left(  -iz\right)  ^{\frac{1}{2}\left(  {\frac{{\mu}^{2}}{2\left(  -E\right)
^{3/2}}}-{\frac{2\xi}{\left(  -E\right)  ^{1/2}}}-4l-4\right)  }\sum_{n\geq
0}\frac{a_{n}}{\left(  -iz\right)  ^{n}}, \label{f1o112}%
\end{align}%
\begin{align}
&  H_{l}^{+}\left(  4l+2,{\frac{i\sqrt{2}\mu}{\left(  -E\right)  ^{3/4}}%
},{\frac{2\xi}{\left(  -E\right)  ^{1/2}}}-{\frac{{\mu}^{2}}{2\left(
-E\right)  ^{3/2}}},0,z\right) \nonumber\\
&  =\exp\left(  \frac{i\sqrt{2}\mu}{\left(  -E\right)  ^{3/4}}z+z^{2}\right)
H_{l}^{+}\left(  4l+2,{\frac{\sqrt{2}\mu}{\left(  -E\right)  ^{3/4}}}%
,{\frac{{\mu}^{2}}{2\left(  -E\right)  ^{3/2}}}-{\frac{2\xi}{\left(
-E\right)  ^{1/2}}},0,-iz\right) \nonumber\\
&  =\left(  -iz\right)  ^{-\frac{1}{2}\left(  {\frac{{\mu}^{2}}{2\left(
-E\right)  ^{3/2}}}-{\frac{2\xi}{\left(  -E\right)  ^{1/2}}}+4l+4\right)
}\sum_{n\geq0}\frac{e_{n}}{\left(  -iz\right)  ^{n}} \label{f2o112}%
\end{align}
with the expansion coefficients given by the recurrence relation
\[
a_{0}=1,\text{ \ }a_{1}=\frac{\sqrt{2}\mu}{4\left(  -E\right)  ^{3/4}}\left(
\gamma-1\right)  ,
\]%
\begin{align}
2\left(  n+2\right)  a_{n+2}+  &  \frac{\sqrt{2}\mu}{\left(  -E\right)
^{3/4}}\left[  \frac{3}{2}-{\frac{{\mu}^{2}}{4\left(  -E\right)  ^{3/2}}%
}+{\frac{2\xi}{2\left(  -E\right)  ^{1/2}}}+n\right]  a_{n+1}\nonumber\\
&  +\left\{  \frac{1}{4}\left[  {\frac{{\mu}^{2}}{2\left(  -E\right)  ^{3/2}}%
}-{\frac{2\xi}{\left(  -E\right)  ^{1/2}}}\right]  ^{2}-\left(  2l+1\right)
^{2}+1-{\frac{{\mu}^{2}}{2\left(  -E\right)  ^{3/2}}}\right. \nonumber\\
&  +\left.  {\frac{2\xi}{\left(  -E\right)  ^{1/2}}}+n\left[  n+2-{\frac{{\mu
}^{2}}{2\left(  -E\right)  ^{3/2}}}+{\frac{2\xi}{\left(  -E\right)  ^{1/2}}%
}\right]  \right\}  a_{n}=0
\end{align}
and%
\[
e_{0}=1,\ e_{1}=-\frac{\sqrt{2}\mu}{4\left(  -E\right)  ^{3/4}}\left(
\gamma+1\right)  ,
\]%
\begin{align}
2\left(  n+2\right)  e_{n+2}+  &  \frac{\sqrt{2}\mu}{\left(  -E\right)
^{3/4}}\left[  \frac{3}{2}+{\frac{{\mu}^{2}}{4\left(  -E\right)  ^{3/2}}%
}-{\frac{\xi}{\left(  -E\right)  ^{1/2}}}+n\right]  e_{n+1}\nonumber\\
&  -\left\{  \frac{1}{4}\left[  {\frac{{\mu}^{2}}{2\left(  -E\right)  ^{3/2}}%
}-{\frac{2\xi}{\left(  -E\right)  ^{1/2}}}\right]  ^{2}-\left(  2l+1\right)
^{2}+1+{\frac{{\mu}^{2}}{2\left(  -E\right)  ^{3/2}}}\right. \nonumber\\
&  -\left.  {\frac{2\xi}{\left(  -E\right)  ^{1/2}}}+n\left[  n+2+{\frac{{\mu
}^{2}}{2\left(  -E\right)  ^{3/2}}}-{\frac{2\xi}{\left(  -E\right)  ^{1/2}}%
}\right]  \right\}  e_{n}=0.
\end{align}

\subsection{Bound states and scattering states}

\subsubsection{The bound state}

To construct the solution, we first express the regular solution
(\ref{regularo132}) as a linear combination of the two irregular solutions
(\ref{f1o112}) and (\ref{f2o112}).

The regular solution (\ref{regularo112}), with the relation
\cite{ronveaux1995heun,li2016exact}%

\begin{align}
&  N\left(  4l+2,{\frac{i\sqrt{2}\mu}{\left(  -E\right)  ^{3/4}}},{\frac{2\xi
}{\left(  -E\right)  ^{1/2}}}-{\frac{{\mu}^{2}}{2\left(  -E\right)  ^{3/2}}%
},0,z\right)  \nonumber\\
&  =K_{1}\left(  4l+2,{\frac{i\sqrt{2}\mu}{\left(  -E\right)  ^{3/4}}}%
,{\frac{2\xi}{\left(  -E\right)  ^{1/2}}}-{\frac{{\mu}^{2}}{2\left(
-E\right)  ^{3/2}}},0\right)  B_{l}^{+}\left(  4l+2,{\frac{i\sqrt{2}\mu
}{\left(  -E\right)  ^{3/4}}},{\frac{2\xi}{\left(  -E\right)  ^{1/2}}}%
-{\frac{{\mu}^{2}}{2\left(  -E\right)  ^{3/2}}},0,z\right)  \nonumber\\
&  +K_{2}\left(  4l+2,{\frac{i\sqrt{2}\mu}{\left(  -E\right)  ^{3/4}}}%
,{\frac{2\xi}{\left(  -E\right)  ^{1/2}}}-{\frac{{\mu}^{2}}{2\left(
-E\right)  ^{3/2}}},0\right)  H_{l}^{+}\left(  4l+2,{\frac{i\sqrt{2}\mu
}{\left(  -E\right)  ^{3/4}}},{\frac{2\xi}{\left(  -E\right)  ^{1/2}}}%
-{\frac{{\mu}^{2}}{2\left(  -E\right)  ^{3/2}}},0,z\right)
\end{align}
and the expansions (\ref{f1o112}) and (\ref{f2o112}), becomes%
\begin{align}
u_{l}\left(  r\right)   &  =A_{l}K_{1}\left(  4l+2,{\frac{i\sqrt{2}\mu
}{\left(  -E\right)  ^{3/4}}},{\frac{2\xi}{\left(  -E\right)  ^{1/2}}}%
-{\frac{{\mu}^{2}}{2\left(  -E\right)  ^{3/2}}},0\right)  \nonumber\\
&  \times\exp\left(  -\left(  -E\right)  ^{1/2}r-\frac{\mu\sqrt{r}}{\left(
-E\right)  ^{1/2}}\right)  r^{{\frac{{\mu}^{2}}{8\left(  -E\right)  ^{3/2}}%
}-{\frac{\xi}{2\left(  -E\right)  ^{1/2}}}}\sum_{n\geq0}\frac{a_{n}}{\left[
\left(  -E\right)  ^{1/4}\sqrt{2r}\right]  ^{n}}\nonumber\\
&  +A_{l}K_{2}\left(  4l+2,{\frac{i\sqrt{2}\mu}{\left(  -E\right)  ^{3/4}}%
},{\frac{2\xi}{\left(  -E\right)  ^{1/2}}}-{\frac{{\mu}^{2}}{2\left(
-E\right)  ^{3/2}}},0\right)  \nonumber\\
&  \times\exp\left(  \left(  -E\right)  ^{1/2}r+\frac{\mu\sqrt{r}}{\left(
-E\right)  ^{1/2}}\right)  r^{-{\frac{{\mu}^{2}}{8\left(  -E\right)  ^{3/2}}%
}+{\frac{\xi}{2\left(  -E\right)  ^{1/2}}}}\sum_{n\geq0}\frac{e_{n}}{\left[
\left(  -E\right)  ^{1/4}\sqrt{2r}\right]  ^{n}},\label{uexpzfo112}%
\end{align}
where $K_{1}\left(  4l+2,{\frac{i\sqrt{2}\mu}{\left(  -E\right)  ^{3/4}}%
},{\frac{2\xi}{\left(  -E\right)  ^{1/2}}}-{\frac{{\mu}^{2}}{2\left(
-E\right)  ^{3/2}}},0\right)  $ and $K_{2}\left(  4l+2,{\frac{i\sqrt{2}\mu
}{\left(  -E\right)  ^{3/4}}},{\frac{2\xi}{\left(  -E\right)  ^{1/2}}}%
-{\frac{{\mu}^{2}}{2\left(  -E\right)  ^{3/2}}},0\right)  $ are combination
coefficients\ and $z=i\left(  -E\right)  ^{1/4}\sqrt{2r}$.

The boundary condition of bound states, $\left.  u\left(  r\right)
\right\vert _{r\rightarrow\infty}\rightarrow0$, requires that the coefficient
of the second term must vanish since this term diverges when $r\rightarrow
\infty$, i.e.,
\begin{equation}
K_{2}\left(  4l+2,{\frac{i\sqrt{2}\mu}{\left(  -E\right)  ^{3/4}}},{\frac
{2\xi}{\left(  -E\right)  ^{1/2}}}-{\frac{{\mu}^{2}}{2\left(  -E\right)
^{3/2}}},0\right)  =0, \label{eigenvalueo112}%
\end{equation}
where%
\begin{align}
&  K_{2}\left(  4l+2,{\frac{i\sqrt{2}\mu}{\left(  -E\right)  ^{3/4}}}%
,{\frac{2\xi}{\left(  -E\right)  ^{1/2}}}-{\frac{{\mu}^{2}}{2\left(
-E\right)  ^{3/2}}},0\right) \nonumber\\
&  =\frac{\Gamma\left(  4l+3\right)  }{\Gamma\left(  2l+1-{\frac{\xi}{\left(
-E\right)  ^{1/2}}}+{\frac{{\mu}^{2}}{4\left(  -E\right)  ^{3/2}}}\right)
\Gamma\left(  2l+2+{\frac{\xi}{\left(  -E\right)  ^{1/2}}}-{\frac{{\mu}^{2}%
}{4\left(  -E\right)  ^{3/2}}}\right)  }\nonumber\\
&  \times J_{2l+2+{\frac{\xi}{\left(  -E\right)  ^{1/2}}}-{\frac{{\mu}^{2}%
}{4\left(  -E\right)  ^{3/2}}}}\left(  2l+1+{\frac{\xi}{\left(  -E\right)
^{1/2}}}-{\frac{{\mu}^{2}}{4\left(  -E\right)  ^{3/2}}},\frac{i\sqrt{2}\mu
}{\left(  -E\right)  ^{3/4}},6l+3-{\frac{\xi}{\left(  -E\right)  ^{1/2}}%
}\right. \nonumber\\
&  +\left.  {\frac{{\mu}^{2}}{4\left(  -E\right)  ^{3/2}}},\frac{i\sqrt{2}\mu
}{2\left(  -E\right)  ^{3/4}}\left[  {\frac{2\xi}{\left(  -E\right)  ^{1/2}}%
}-{\frac{{\mu}^{2}}{2\left(  -E\right)  ^{3/2}}}-4l-2\right]  \right)
\end{align}
with%
\begin{align}
&  J_{2l+2+{\frac{\xi}{\left(  -E\right)  ^{1/2}}}-{\frac{{\mu}^{2}}{4\left(
-E\right)  ^{3/2}}}}\left(  2l+1+{\frac{\xi}{\left(  -E\right)  ^{1/2}}%
}-{\frac{{\mu}^{2}}{4\left(  -E\right)  ^{3/2}}},\frac{i\sqrt{2}\mu}{\left(
-E\right)  ^{3/4}},6l+3-{\frac{\xi}{\left(  -E\right)  ^{1/2}}}\right.
\nonumber\\
&  +\left.  {\frac{{\mu}^{2}}{4\left(  -E\right)  ^{3/2}}},\frac{i\sqrt{2}\mu
}{2\left(  -E\right)  ^{3/4}}\left[  {\frac{2\xi}{\left(  -E\right)  ^{1/2}}%
}-{\frac{{\mu}^{2}}{2\left(  -E\right)  ^{3/2}}}-4l-2\right]  \right)
\nonumber\\
&  =\int_{0}^{\infty}dxx^{2l+1+{\frac{\xi}{\left(  -E\right)  ^{1/2}}}%
-{\frac{{\mu}^{2}}{4\left(  -E\right)  ^{3/2}}}}e^{-x^{2}}N\left(
2l+1+{\frac{\xi}{\left(  -E\right)  ^{1/2}}}-{\frac{{\mu}^{2}}{4\left(
-E\right)  ^{3/2}}},\frac{i\sqrt{2}\mu}{\left(  -E\right)  ^{3/4}}%
,6l+3-{\frac{\xi}{\left(  -E\right)  ^{1/2}}}\right. \nonumber\\
&  +\left.  {\frac{{\mu}^{2}}{4\left(  -E\right)  ^{3/2}}},\frac{i\sqrt{2}\mu
}{2\left(  -E\right)  ^{3/4}}\left[  {\frac{2\xi}{\left(  -E\right)  ^{1/2}}%
}-{\frac{{\mu}^{2}}{2\left(  -E\right)  ^{3/2}}}-4l-2\right]  ,x\right)  .
\end{align}

Eq. (\ref{eigenvalueo112}) is an implicit expression of the eigenvalue.

The eigenfunction, by Eqs. (\ref{uexpzfo112}) and (\ref{eigenvalueo112}),
reads%
\begin{align}
u_{l}\left(  r\right)   &  =A_{l}K_{1}\left(  4l+2,{\frac{i\sqrt{2}\mu
}{\left(  -E\right)  ^{3/4}}},{\frac{2\xi}{\left(  -E\right)  ^{1/2}}}%
-{\frac{{\mu}^{2}}{2\left(  -E\right)  ^{3/2}}},0\right)  \nonumber\\
&  \times\exp\left(  -\left(  -E\right)  ^{1/2}r-\frac{\mu\sqrt{r}}{\left(
-E\right)  ^{1/2}}\right)  r^{{\frac{{\mu}^{2}}{8\left(  -E\right)  ^{3/2}}%
}-{\frac{\xi}{2\left(  -E\right)  ^{1/2}}}}\sum_{n\geq0}\frac{a_{n}}{\left[
\left(  -E\right)  ^{1/4}\sqrt{2r}\right]  ^{n}}.
\end{align}

\subsubsection{The scattering state}

For scattering states, $E>0$, introduce%
\begin{equation}
E=k^{2}.
\end{equation}
The singularity of the $S$-matrix on the positive imaginary axis corresponds
to the eigenvalues of bound states \cite{joachain1975quantum}, so the zero of
$K_{2}\left(  4l+2,-{\frac{\left(  1+i\right)  \mu}{k^{3/2}}},{\frac{2i\xi}%
{k}}-i{\frac{{\mu}^{2}}{2k^{3}}},0\right)  $ on the positive imaginary is just
the singularity of the $S$-matrix. Considering that the $S$-matrix is unitary,
i.e.,%
\begin{equation}
S_{l}=e^{2i\delta_{l}}, \label{SMo112}%
\end{equation}
we have%
\begin{equation}
S_{l}\left(  k\right)  =\frac{K_{2}^{\ast}\left(  4l+2,-{\frac{\left(
1+i\right)  \mu}{k^{3/2}}},{\frac{2i\xi}{k}}-i{\frac{{\mu}^{2}}{2k^{3}}%
},0\right)  }{K_{2}\left(  4l+2,-{\frac{\left(  1+i\right)  \mu}{k^{3/2}}%
},{\frac{2i\xi}{k}}-i{\frac{{\mu}^{2}}{2k^{3}}},0\right)  }=\frac{K_{2}\left(
4l+2,-{\frac{\left(  1-i\right)  \mu}{k^{3/2}}},-{\frac{2i\xi}{k}}%
+i{\frac{{\mu}^{2}}{2k^{3}}},0\right)  }{K_{2}\left(  4l+2,-{\frac{\left(
1+i\right)  \mu}{k^{3/2}}},{\frac{2i\xi}{k}}-i{\frac{{\mu}^{2}}{2k^{3}}%
},0\right)  }. \label{SM1o112}%
\end{equation}

The scattering wave function can be constructed with the help of the
$S$-matrix. The scattering wave function can be expressed as a linear
combination of the radially ingoing wave $u_{in}\left(  r\right)  $ and the
radially outgoing wave $u_{out}\left(  r\right)  $, which are conjugate to
each other, i.e., \cite{joachain1975quantum}%
\begin{equation}
u_{l}\left(  r\right)  =A_{l}\left[  \left(  -1\right)  ^{l+1}u_{in}\left(
r\right)  +S_{l}\left(  k\right)  u_{out}\left(  r\right)  \right]
.\label{inouto112}%
\end{equation}
From Eq. (\ref{uexpzfo112}), we have%
\begin{align}
u_{in}\left(  r\right)   &  =\exp\left(  -ikr+i\frac{\mu\sqrt{r}}{k}\right)
r^{i{\frac{{\mu}^{2}}{8k^{3}}}+i{\frac{\xi}{2k}}}\sum_{n\geq0}\frac{e_{n}%
}{\left(  -2ikr\right)  ^{n/2}},\nonumber\\
u_{out}\left(  r\right)   &  =\exp\left(  ikr-i\frac{\mu\sqrt{r}}{k}\right)
r^{-i{\frac{{\mu}^{2}}{8k^{3}}}-i{\frac{\xi}{2k}}}\sum_{n\geq0}\frac
{e_{n}^{\ast}}{\left(  2ikr\right)  ^{n/2}}.
\end{align}
Then by Eq. (\ref{SM1o112}) we obtain the scattering wave function,%
\begin{align}
u_{l}\left(  r\right)   &  =A_{l}\left[  \left(  -1\right)  ^{l+1}\exp\left(
-ikr+i\frac{\mu\sqrt{r}}{k}\right)  r^{i{\frac{{\mu}^{2}}{8k^{3}}}+i{\frac
{\xi}{2k}}}\sum_{n\geq0}\frac{e_{n}}{\left(  -2ikr\right)  ^{n/2}}\right.
\nonumber\\
&  +\left.  \frac{K_{2}\left(  4l+2,-{\frac{\left(  1-i\right)  \mu}{k^{3/2}}%
},-{\frac{2i\xi}{k}}+i{\frac{{\mu}^{2}}{2k^{3}}},0\right)  }{K_{2}\left(
4l+2,-{\frac{\left(  1+i\right)  \mu}{k^{3/2}}},{\frac{2i\xi}{k}}-i{\frac
{{\mu}^{2}}{2k^{3}}},0\right)  }\exp\left(  ikr-i\frac{\mu\sqrt{r}}{k}\right)
r^{-i{\frac{{\mu}^{2}}{8k^{3}}}-i{\frac{\xi}{2k}}}\sum_{n\geq0}\frac
{e_{n}^{\ast}}{\left(  2ikr\right)  ^{n/2}}\right]  .
\end{align}
Taking $r\rightarrow\infty$, we have
\begin{align}
&  u_{l}\left(  r\right)  \overset{r\rightarrow\infty}{\sim}\nonumber\\
&  A_{l}\left[  \left(  -1\right)  ^{l+1}\exp\left(  -ikr+\frac{i\mu\sqrt{r}%
}{k}\right)  r^{{\frac{i{\mu}^{2}}{8k^{3}}}+{\frac{i\xi}{2k}}}+\frac
{K_{2}\left(  4l+2,{\frac{-\left(  1-i\right)  \mu}{k^{3/2}}},{\frac{-2i\xi
}{k}}+{\frac{i{\mu}^{2}}{2k^{3}}},0\right)  }{K_{2}\left(  4l+2,{\frac
{-\left(  1+i\right)  \mu}{k^{3/2}}},{\frac{2i\xi}{k}}-{\frac{i{\mu}^{2}%
}{2k^{3}}},0\right)  }\exp\left(  ikr-\frac{i\mu\sqrt{r}}{k}\right)
r^{-{\frac{i{\mu}^{2}}{8k^{3}}}-{\frac{i\xi}{2k}}}\right]  \nonumber\\
&  =A_{l}e^{i\delta_{l}}\sin\left(  kr-\frac{\mu\sqrt{r}}{k}-\left(
{\frac{{\mu}^{2}}{8k^{3}}}+{\frac{\xi}{2k}}\right)  \ln2kr+\delta_{l}%
-\frac{l\pi}{2}\right)  .
\end{align}
By Eqs. (\ref{SMo112}) and (\ref{SM1o112}), we obtain the scattering phase
shift%
\begin{equation}
\delta_{l}=-\arg K_{2}\left(  4l+2,-{\frac{\left(  1+i\right)  \mu}{k^{3/2}}%
},{\frac{2i\xi}{k}}-i{\frac{{\mu}^{2}}{2k^{3}}},0\right)  .
\end{equation}

\section{The exact solution of $U\left(  r\right)  =\xi r^{2}+\frac{\mu}%
{r}+\kappa r$\label{V2m11}}

In this appendix, we provide an exact solution of the eigenproblem of the
potential
\begin{equation}
U\left(  r\right)  =\xi r^{2}+\frac{\mu}{r}+\kappa r
\end{equation}
by solving the radial equation directly. This potential has only bound states.

The radial equation reads%
\begin{equation}
\frac{d^{2}}{dr^{2}}u_{l}\left(  r\right)  +\left[  E-\frac{l\left(
l+1\right)  }{r^{2}}-\xi r^{2}-\frac{\mu}{r}-\kappa r\right]  u_{l}\left(
r\right)  =0. \label{radialeq2o1o1}%
\end{equation}
Using the variable substitution%
\begin{equation}
z=i\xi^{1/4}r \label{zr2o1o1}%
\end{equation}
and introducing $f_{l}\left(  z\right)  $ by%
\begin{equation}
u_{l}\left(  z\right)  =A_{l}\exp\left(  -\frac{z^{2}}{2}-\frac{\beta}%
{2}z\right)  z^{l+1}f_{l}\left(  z\right)
\end{equation}
with $A_{l}$ a constant, we convert the radial equation (\ref{radialeq2o1o1})
into an equation of $f_{l}\left(  z\right)  $:%
\begin{equation}
zf_{l}^{\prime\prime}\left(  z\right)  +\left(  2l+2-{\frac{i\kappa}{{\xi
}^{3/4}}}z-2z^{2}\right)  f_{l}^{\prime}\left(  z\right)  +\left[  \left(
-{\frac{E}{{\xi}^{1/2}}}-{\frac{{\kappa}^{2}}{4{\xi}^{3/2}}}-2l-3\right)
z+\frac{i\mu}{\xi^{1/4}}-\left(  l+1\right)  \frac{i\kappa}{{\xi}^{3/4}%
}\right]  f_{l}\left(  z\right)  =0. \label{eqf2o1o1}%
\end{equation}
This is a Biconfluent Heun equation \cite{ronveaux1995heun}.

The choice of the boundary condition has been discussed in Ref.
\cite{li2016exact}.

\subsection{The regular solution}

The regular solution is a solution satisfying the boundary condition at $r=0$
\cite{li2016exact}. The regular solution at $r=0$\ should satisfy the boundary
condition $\lim_{r\rightarrow0}u_{l}\left(  r\right)  /r^{l+1}=1$. In this
section, we provide the regular solution of Eq. (\ref{eqf2o1o1}).

The Biconfluent Heun equation (\ref{eqf2o1o1}) has two linearly independent
solutions \cite{ronveaux1995heun}%
\begin{align}
y_{l}^{\left(  1\right)  }\left(  z\right)   &  =N\left(  2l+1,{\frac{i\kappa
}{{\xi}^{3/4}}},-{\frac{E}{{\xi}^{1/2}}}-{\frac{{\kappa}^{2}}{4{\xi}^{3/2}}%
},{\frac{-i2\mu}{\xi^{1/4}}},z\right)  ,\\
y_{l}^{\left(  2\right)  }\left(  z\right)   &  =cN\left(  2l+1,{\frac
{i\kappa}{{\xi}^{3/4}}},-{\frac{E}{{\xi}^{1/2}}}-{\frac{{\kappa}^{2}}{4{\xi
}^{3/2}}},{\frac{-i2\mu}{\xi^{1/4}}},z\right)  \ln z+\sum_{n\geq0}%
d_{n}z^{n-2l-1},
\end{align}
where%
\[
c=\frac{1}{2l+1}\left[  d_{2l}\left(  -\frac{i\mu}{\xi^{1/4}}+\frac{il\kappa
}{{\xi}^{3/4}}\right)  -d_{2l-1}\left(  -{\frac{E}{{\xi}^{1/2}}}%
-{\frac{{\kappa}^{2}}{4{\xi}^{3/2}}}+1-2l\right)  \right]
\]
is a constant with the coefficient $d_{\nu}$ given by the following recurrence
relation,%
\begin{align}
&  d_{-1}=0,\text{ \ }d_{0}=1,\nonumber\\
&  \left(  v+2\right)  \left(  v+1-2l\right)  d_{v+2}-\left(  \frac{-i\mu}%
{\xi^{1/4}}+\frac{i\kappa}{{\xi}^{3/4}}\left(  v+1-l\right)  \right)
d_{v+1}+\left(  -{\frac{E}{{\xi}^{1/2}}}-{\frac{{\kappa}^{2}}{4{\xi}^{3/2}}%
}-2v-1+2l\right)  d_{v}=0
\end{align}
and $N(\alpha,\beta,\gamma,\delta,z)$ is the biconfluent Heun function
\cite{ronveaux1995heun,slavyanov2000special,li2016exact}.

The biconfluent Heun function $N\left(  2l+1,{\frac{i\kappa}{{\xi}^{3/4}}%
},-{\frac{E}{{\xi}^{1/2}}}-{\frac{{\kappa}^{2}}{4{\xi}^{3/2}}},{\frac{-i2\mu
}{\xi^{1/4}}},z\right)  $ has an expansion at $z=0$ \cite{ronveaux1995heun}:%
\begin{equation}
N\left(  2l+1,{\frac{i\kappa}{{\xi}^{3/4}}},-{\frac{E}{{\xi}^{1/2}}}%
-{\frac{{\kappa}^{2}}{4{\xi}^{3/2}}},{\frac{-i2\mu}{\xi^{1/4}}},z\right)
=\sum_{n\geq0}\frac{A_{n}}{\left(  2l+2\right)  _{n}}\frac{z^{n}}{n!},
\end{equation}
where the expansion coefficients is determined by the recurrence relation,%
\begin{align}
A_{0}  &  =1,\text{ \ }A_{1}=\frac{-i\mu}{\xi^{1/4}}+\frac{i\kappa}{{\xi
}^{3/4}}\left(  l+1\right)  ,\nonumber\\
A_{n+2}  &  =\left[  \left(  n+l+2\right)  \frac{i\kappa}{{\xi}^{3/4}}%
+\frac{-i\mu}{\xi^{1/4}}\right]  A_{n+1}-\left(  n+1\right)  \left(
n+2l+2\right)  \left(  -{\frac{E}{{\xi}^{1/2}}}-{\frac{{\kappa}^{2}}{4{\xi
}^{3/2}}}-2l-3-2n\right)  A_{n},
\end{align}
and $\left(  a\right)  _{n}=\Gamma\left(  a+n\right)  /\Gamma\left(  a\right)
$ is Pochhammer's symbol.

Only $y_{l}^{\left(  1\right)  }\left(  z\right)  $ satisfies the boundary
condition for the regular solution at $r=0$, so the radial eigenfunction reads%
\begin{align}
u_{l}\left(  z\right)   &  =A_{l}\exp\left(  -\frac{z^{2}}{2}-\frac{\beta}%
{2}z\right)  z^{l+1}y_{l}^{\left(  1\right)  }\left(  z\right) \nonumber\\
&  =A_{l}\exp\left(  -\frac{z^{2}}{2}-\frac{\beta}{2}z\right)  z^{l+1}N\left(
2l+1,{\frac{i\kappa}{{\xi}^{3/4}}},-{\frac{E}{{\xi}^{1/2}}}-{\frac{{\kappa
}^{2}}{4{\xi}^{3/2}}},{\frac{-i2\mu}{\xi^{1/4}}},z\right)  .
\end{align}
By Eq. (\ref{zr2o1o1}), we obtain the regular solution,%

\begin{equation}
u_{l}\left(  r\right)  =A_{l}\exp\left(  \frac{\xi^{1/2}}{2}r^{2}+\frac
{\kappa}{2{\xi}^{1/2}}r\right)  r^{l+1}N\left(  2l+1,{\frac{i\kappa}{{\xi
}^{3/4}}},-{\frac{E}{{\xi}^{1/2}}}-{\frac{{\kappa}^{2}}{4{\xi}^{3/2}}}%
,{\frac{-i2\mu}{\xi^{1/4}}},i\xi^{1/4}r\right)  . \label{regular2o1o1}%
\end{equation}

\subsection{The irregular solution}

The irregular solution is a solution satisfying the boundary condition at
$r\rightarrow\infty$ \cite{li2016exact}.

The Biconfluent Heun equation (\ref{eqf2o1o1}) has two linearly independent
irregular solutions \cite{ronveaux1995heun}:%
\begin{align}
&  B_{l}^{+}\left(  2l+1,{\frac{i\kappa}{{\xi}^{3/4}}},-{\frac{E}{{\xi}^{1/2}%
}}-{\frac{{\kappa}^{2}}{4{\xi}^{3/2}}},{\frac{-i2\mu}{\xi^{1/4}},}z\right)
\nonumber\\
&  =\exp\left(  \frac{i\kappa}{{\xi}^{3/4}}z+z^{2}\right)  B_{l}^{+}\left(
2l+1,{\frac{\kappa}{{\xi}^{3/4}}},{\frac{E}{{\xi}^{1/2}}}+{\frac{{\kappa}^{2}%
}{4{\xi}^{3/2}}},{\frac{2\mu}{\xi^{1/4}},-i}z\right) \nonumber\\
&  =\exp\left(  \frac{i\kappa}{{\xi}^{3/4}}z+z^{2}\right)  \left(  -iz\right)
^{\frac{1}{2}\left(  {\frac{E}{{\xi}^{1/2}}}+{\frac{{\kappa}^{2}}{4{\xi}%
^{3/2}}}-2l-3\right)  }\sum_{n\geq0}\frac{a_{n}}{\left(  -iz\right)  ^{n}},
\label{f12o1o1}%
\end{align}%
\begin{align}
&  H_{l}^{+}\left(  2l+1,{\frac{i\kappa}{{\xi}^{3/4}}},-{\frac{E}{{\xi}^{1/2}%
}}-{\frac{{\kappa}^{2}}{4{\xi}^{3/2}}},{\frac{-i2\mu}{\xi^{1/4}},}z\right)
\nonumber\\
&  =\exp\left(  \frac{i\kappa}{{\xi}^{3/4}}z+z^{2}\right)  H_{l}^{+}\left(
2l+1,{\frac{\kappa}{{\xi}^{3/4}}},{\frac{E}{{\xi}^{1/2}}}+{\frac{{\kappa}^{2}%
}{4{\xi}^{3/2}}},{\frac{2\mu}{\xi^{1/4}},-i}z\right) \nonumber\\
&  =\left(  -iz\right)  ^{-\frac{1}{2}\left(  {\frac{E}{{\xi}^{1/2}}}%
+{\frac{{\kappa}^{2}}{4{\xi}^{3/2}}}+2l+3\right)  }\sum_{n\geq0}\frac{e_{n}%
}{\left(  -iz\right)  ^{n}} \label{f22o1o1}%
\end{align}
with the expansion coefficients given by the recurrence relation
\[
a_{0}=1,\text{ \ }a_{1}={\frac{\mu}{2\xi^{1/4}}}+\frac{\kappa}{4{\xi}^{3/4}%
}\left(  {\frac{E}{{\xi}^{1/2}}}+{\frac{{\kappa}^{2}}{4{\xi}^{3/2}}}-1\right)
,
\]%
\begin{align}
&  2\left(  n+2\right)  a_{n+2}+\left[  \frac{\kappa}{{\xi}^{3/4}}\left(
\frac{3}{2}-{\frac{E}{2{\xi}^{1/2}}}-{\frac{{\kappa}^{2}}{8{\xi}^{3/2}}%
}+n\right)  -\frac{\mu}{\xi^{1/4}}\right]  a_{n+1}\nonumber\\
&  +\left[  \left(  {\frac{E}{2{\xi}^{1/2}}}+{\frac{{\kappa}^{2}}{8{\xi}%
^{3/2}}}-1\right)  ^{2}-\frac{\left(  2l+1\right)  ^{2}}{4}+n\left(
n+2-{\frac{E}{{\xi}^{1/2}}}-{\frac{{\kappa}^{2}}{4{\xi}^{3/2}}}\right)
\right]  a_{n}=0
\end{align}
and%
\[
e_{0}=1,\text{ \ }e_{1}=-\frac{\mu}{2\xi^{1/4}}-\frac{\kappa}{4{\xi}^{3/4}%
}\left(  {\frac{E}{{\xi}^{1/2}}}+{\frac{{\kappa}^{2}}{4{\xi}^{3/2}}}+1\right)
,
\]%
\begin{align}
&  2\left(  n+2\right)  e_{n+2}+\left[  \frac{\kappa}{{\xi}^{3/4}}\left(
\frac{3}{2}+{\frac{E}{2{\xi}^{1/2}}}+{\frac{{\kappa}^{2}}{8{\xi}^{3/2}}%
}+n\right)  +\frac{\mu}{\xi^{1/4}}\right]  e_{n+1}\nonumber\\
&  -\left[  \left(  {\frac{E}{2{\xi}^{1/2}}}+{\frac{{\kappa}^{2}}{8{\xi}%
^{3/2}}}+1\right)  ^{2}-\frac{\left(  2l+1\right)  ^{2}}{4}+n\left(
n+2+{\frac{E}{{\xi}^{1/2}}}+{\frac{{\kappa}^{2}}{4{\xi}^{3/2}}}\right)
\right]  e_{n}=0.
\end{align}

\subsection{Eigenfunctions and eigenvalues}

To construct the solution, we first express the regular solution
(\ref{regular2o1o1}) as a linear combination of the two irregular solutions
(\ref{f12o1o1}) and (\ref{f22o1o1}).

The regular solution (\ref{regular2o1o1}), with the relation
\cite{ronveaux1995heun,li2016exact}%

\begin{align}
&  N\left(  2l+1,{\frac{i\kappa}{{\xi}^{3/4}}},-{\frac{E}{{\xi}^{1/2}}}%
-{\frac{{\kappa}^{2}}{4{\xi}^{3/2}}},{\frac{-i2\mu}{\xi^{1/4}}},z\right)
\nonumber\\
&  =K_{1}\left(  2l+1,{\frac{i\kappa}{{\xi}^{3/4}}},-{\frac{E}{{\xi}^{1/2}}%
}-{\frac{{\kappa}^{2}}{4{\xi}^{3/2}}},{\frac{-i2\mu}{\xi^{1/4}}}\right)
B_{l}^{+}\left(  2l+1,{\frac{i\kappa}{{\xi}^{3/4}}},-{\frac{E}{{\xi}^{1/2}}%
}-{\frac{{\kappa}^{2}}{4{\xi}^{3/2}}},{\frac{-i2\mu}{\xi^{1/4}}},z\right)
\nonumber\\
&  +K_{2}\left(  2l+1,{\frac{i\kappa}{{\xi}^{3/4}}},-{\frac{E}{{\xi}^{1/2}}%
}-{\frac{{\kappa}^{2}}{4{\xi}^{3/2}}},{\frac{-i2\mu}{\xi^{1/4}}}\right)
H_{l}^{+}\left(  2l+1,{\frac{i\kappa}{{\xi}^{3/4}}},-{\frac{E}{{\xi}^{1/2}}%
}-{\frac{{\kappa}^{2}}{4{\xi}^{3/2}}},{\frac{-i2\mu}{\xi^{1/4}}},z\right)
\end{align}
and the expansions (\ref{f12o1o1}) and (\ref{f22o1o1}), becomes%
\begin{align}
u_{l}\left(  r\right)   &  =A_{l}K_{1}\left(  2l+1,{\frac{i\kappa}{{\xi}%
^{3/4}}},-{\frac{E}{{\xi}^{1/2}}}-{\frac{{\kappa}^{2}}{4{\xi}^{3/2}}}%
,{\frac{-i2\mu}{\xi^{1/4}}}\right) \nonumber\\
&  \times\exp\left(  -\frac{\xi^{1/2}r^{2}}{2}-\frac{\kappa}{2{\xi}^{1/2}%
}r\right)  r^{{\frac{E}{2{\xi}^{1/2}}}+{\frac{{\kappa}^{2}}{8{\xi}^{3/2}}%
}-\frac{1}{2}}\sum_{n\geq0}\frac{a_{n}}{\left(  \xi^{1/4}r\right)  ^{n}%
}\nonumber\\
&  +A_{l}K_{2}\left(  2l+1,{\frac{i\kappa}{{\xi}^{3/4}}},-{\frac{E}{{\xi
}^{1/2}}}-{\frac{{\kappa}^{2}}{4{\xi}^{3/2}}},{\frac{-i2\mu}{\xi^{1/4}}%
}\right) \nonumber\\
&  \times\exp\left(  \frac{\xi^{1/2}r^{2}}{2}+\frac{\kappa}{2{\xi}^{1/2}%
}r\right)  r^{-{\frac{E}{2{\xi}^{1/2}}}-{\frac{{\kappa}^{2}}{8{\xi}^{3/2}}%
}-\frac{1}{2}}\sum_{n\geq0}\frac{e_{n}}{\left(  \xi^{1/4}r\right)  ^{n}},
\label{uexpzf2o1o1}%
\end{align}
where $K_{1}\left(  2l+1,{\frac{i\kappa}{{\xi}^{3/4}}},-{\frac{E}{{\xi}^{1/2}%
}}-{\frac{{\kappa}^{2}}{4{\xi}^{3/2}}},{\frac{-i2\mu}{\xi^{1/4}}}\right)  $
and $K_{2}\left(  2l+1,{\frac{i\kappa}{{\xi}^{3/4}}},-{\frac{E}{{\xi}^{1/2}}%
}-{\frac{{\kappa}^{2}}{4{\xi}^{3/2}}},{\frac{-i2\mu}{\xi^{1/4}}}\right)  $ are
combination coefficients\ and $z=i\xi^{1/4}r$.

The boundary condition of bound states, $\left.  u\left(  r\right)
\right\vert _{r\rightarrow\infty}\rightarrow0$, requires that the coefficient
of the second term must vanish since this term diverges when $r\rightarrow
\infty$, i.e.,
\begin{equation}
K_{2}\left(  2l+1,{\frac{i\kappa}{{\xi}^{3/4}}},-{\frac{E}{{\xi}^{1/2}}%
}-{\frac{{\kappa}^{2}}{4{\xi}^{3/2}}},{\frac{-i2\mu}{\xi^{1/4}}}\right)  =0,
\label{eigenvalue2o1o1}%
\end{equation}
where%
\begin{align}
&  K_{2}\left(  2l+1,{\frac{i\kappa}{{\xi}^{3/4}}},-{\frac{E}{{\xi}^{1/2}}%
}-{\frac{{\kappa}^{2}}{4{\xi}^{3/2}}},{\frac{-i2\mu}{\xi^{1/4}}}\right)
\nonumber\\
&  =\frac{\Gamma\left(  2l+2\right)  }{\Gamma\left(  l+\frac{1}{2}+{\frac
{E}{2{\xi}^{1/2}}}+{\frac{{\kappa}^{2}}{8{\xi}^{3/2}}}\right)  \Gamma\left(
l+\frac{3}{2}-{\frac{E}{2{\xi}^{1/2}}}-{\frac{{\kappa}^{2}}{8{\xi}^{3/2}}%
}\right)  }\nonumber\\
&  \times J_{l+\frac{3}{2}-{\frac{E}{2{\xi}^{1/2}}}-{\frac{{\kappa}^{2}}%
{8{\xi}^{3/2}}}}\left(  l+\frac{1}{2}-{\frac{E}{2{\xi}^{1/2}}}-{\frac{{\kappa
}^{2}}{8{\xi}^{3/2}}},\frac{i\kappa}{{\xi}^{3/4}},\frac{3}{2}\left(
2l+1\right)  +{\frac{E}{2{\xi}^{1/2}}}\right. \nonumber\\
&  +\left.  {\frac{{\kappa}^{2}}{8{\xi}^{3/2}}},\frac{-i2\mu}{\xi^{1/4}}%
+\frac{i\kappa}{2{\xi}^{3/4}}\left(  -{\frac{E}{{\xi}^{1/2}}}-{\frac{{\kappa
}^{2}}{4{\xi}^{3/2}}}-2l-1\right)  \right)
\end{align}
with%
\begin{align}
&  J_{l+\frac{3}{2}-{\frac{E}{2{\xi}^{1/2}}}-{\frac{{\kappa}^{2}}{8{\xi}%
^{3/2}}}}\left(  l+\frac{1}{2}-{\frac{E}{2{\xi}^{1/2}}}-{\frac{{\kappa}^{2}%
}{8{\xi}^{3/2}}},\frac{i\kappa}{{\xi}^{3/4}},\frac{3}{2}\left(  2l+1\right)
+{\frac{E}{2{\xi}^{1/2}}}\right. \nonumber\\
&  +\left.  {\frac{{\kappa}^{2}}{8{\xi}^{3/2}}},\frac{-i2\mu}{\xi^{1/4}}%
+\frac{i\kappa}{2{\xi}^{3/4}}\left(  -{\frac{E}{{\xi}^{1/2}}}-{\frac{{\kappa
}^{2}}{4{\xi}^{3/2}}}-2l-1\right)  \right) \nonumber\\
&  =\int_{0}^{\infty}dxx^{l+\frac{1}{2}-{\frac{E}{2{\xi}^{1/2}}}%
-{\frac{{\kappa}^{2}}{8{\xi}^{3/2}}}}e^{-x^{2}}N\left(  l+\frac{1}{2}%
-{\frac{E}{2{\xi}^{1/2}}}-{\frac{{\kappa}^{2}}{8{\xi}^{3/2}}},\frac{i\kappa
}{{\xi}^{3/4}},\frac{3}{2}\left(  2l+1\right)  +{\frac{E}{2{\xi}^{1/2}}%
}\right. \nonumber\\
&  +\left.  {\frac{{\kappa}^{2}}{8{\xi}^{3/2}}},\frac{-i2\mu}{\xi^{1/4}}%
+\frac{i\kappa}{2{\xi}^{3/4}}\left(  -{\frac{E}{{\xi}^{1/2}}}-{\frac{{\kappa
}^{2}}{4{\xi}^{3/2}}}-2l-1\right)  ,x\right)  .
\end{align}

Eq. (\ref{eigenvalue2o1o1}) is an implicit expression of the eigenvalue.

The eigenfunction, by Eqs. (\ref{uexpzf2o1o1}) and (\ref{eigenvalue2o1o1}),
reads%
\begin{equation}
u_{l}\left(  r\right)  =A_{l}K_{1}\left(  2l+1,{\frac{i\kappa}{{\xi}^{3/4}}%
},-{\frac{E}{{\xi}^{1/2}}}-{\frac{{\kappa}^{2}}{4{\xi}^{3/2}}},{\frac{-i2\mu
}{\xi^{1/4}}}\right)  \exp\left(  -\frac{\xi^{1/2}r^{2}}{2}-\frac{\kappa
}{2{\xi}^{1/2}}r\right)  r^{{\frac{E}{2{\xi}^{1/2}}}+{\frac{{\kappa}^{2}%
}{8{\xi}^{3/2}}}-\frac{1}{2}}\sum_{n\geq0}\frac{a_{n}}{\left(  \xi
^{1/4}r\right)  ^{n}}.
\end{equation}

\section{The exact solution of $U\left(  r\right)  =\frac{\xi}{r}+\frac{\mu
}{r^{3/2}}+\frac{\kappa}{r^{1/2}}$\label{Vm1m32m12}}

In this appendix, we provide an exact solution of the eigenproblem of the
potential
\begin{equation}
U\left(  r\right)  =\frac{\xi}{r}+\frac{\mu}{r^{3/2}}+\frac{\kappa}{r^{1/2}}%
\end{equation}
by solving the radial equation directly. This potential has both bound states
and scattering states.

The radial equation reads%
\begin{equation}
\frac{d^{2}}{dr^{2}}u_{l}\left(  r\right)  +\left[  E-\frac{l\left(
l+1\right)  }{r^{2}}-\frac{\xi}{r}-\frac{\mu}{r^{3/2}}-\frac{\kappa}{r^{1/2}%
}\right]  u_{l}\left(  r\right)  =0. \label{radialeqo13212}%
\end{equation}
Using the variable substitution%
\begin{equation}
z=i\left(  -E\right)  ^{1/4}\sqrt{2r} \label{zro13212}%
\end{equation}
and introducing $f_{l}\left(  z\right)  $ by%
\begin{equation}
u_{l}\left(  z\right)  =A_{l}\exp\left(  -\frac{z^{2}}{2}-\frac{\beta}%
{2}z\right)  z^{2\left(  l+1\right)  }f_{l}\left(  z\right)
\end{equation}
with $A_{l}$ a constant, we convert the radial equation (\ref{radialeqo13212})
into an equation of $f_{l}\left(  z\right)  $:%
\begin{align}
zf_{l}^{\prime\prime}  &  \left(  z\right)  +\left(  4l+3-\frac{i\sqrt
{2}\kappa}{\left(  -E\right)  ^{3/4}}z-2z^{2}\right)  f_{l}^{\prime}\left(
z\right) \nonumber\\
&  +\left\{  \left[  {\frac{2\xi}{\left(  -E\right)  ^{1/2}}}-{\frac{{\kappa
}^{2}}{2\left(  -E\right)  ^{3/2}}}-4l-4\right]  z-\left[  \frac{-i2\sqrt
{2}\mu}{\left(  -E\right)  ^{1/4}}+\left(  4l+3\right)  \frac{i\sqrt{2}\kappa
}{2\left(  -E\right)  ^{3/4}}\right]  \right\}  f_{l}\left(  z\right)  =0.
\label{eqfo13212}%
\end{align}
This is a Biconfluent Heun equation \cite{ronveaux1995heun}.

The choice of the boundary condition has been discussed in Ref.
\cite{li2016exact}.

\subsection{The regular solution}

The regular solution is a solution satisfying the boundary condition at $r=0$
\cite{li2016exact}. The regular solution at $r=0$\ should satisfy the boundary
condition $\lim_{r\rightarrow0}u_{l}\left(  r\right)  /r^{l+1}=1$ for both
bound states and scattering states. In this section, we provide the regular
solution of Eq. (\pageref{eqfo13212}).

The Biconfluent Heun equation (\ref{eqfo13212}) has two linearly independent
solutions \cite{ronveaux1995heun}%
\begin{align}
y_{l}^{\left(  1\right)  }\left(  z\right)   &  =N\left(  4l+2,{\frac
{i\sqrt{2}\kappa}{\left(  -E\right)  ^{3/4}}},{\frac{2\xi}{\left(  -E\right)
^{1/2}}}-{\frac{{\kappa}^{2}}{2\left(  -E\right)  ^{3/2}}},{\frac{-i4\sqrt
{2}\mu}{\left(  -E\right)  ^{1/4}}},z\right)  ,\\
y_{l}^{\left(  2\right)  }\left(  z\right)   &  =cN\left(  4l+2,{\frac
{i\sqrt{2}\kappa}{\left(  -E\right)  ^{3/4}}},{\frac{2\xi}{\left(  -E\right)
^{1/2}}}-{\frac{{\kappa}^{2}}{2\left(  -E\right)  ^{3/2}}},{\frac{-i4\sqrt
{2}\mu}{\left(  -E\right)  ^{1/4}}},z\right)  \ln z+\sum_{n\geq0}%
d_{n}z^{n-4l-2},
\end{align}
where%
\begin{equation}
c=\frac{1}{4l+2}\left\{  d_{4l+1}\left[  \frac{-i2\sqrt{2}\mu}{\left(
-E\right)  ^{1/4}}+\frac{i\sqrt{2}\kappa}{2\left(  -E\right)  ^{3/4}}\left(
4l+1\right)  \right]  -d_{4l}\left[  {\frac{2\xi}{\left(  -E\right)  ^{1/2}}%
}-{\frac{{\kappa}^{2}}{2\left(  -E\right)  ^{3/2}}}-4l\right]  \right\}
\end{equation}
is a constant with the coefficient $d_{\nu}$ given by the following recurrence
relation,%
\begin{align}
&  d_{-1}=0,\text{ \ }d_{0}=1,\nonumber\\
&  \left(  v+2\right)  \left(  v-4l\right)  d_{v+2}-\left[  \frac{-i2\sqrt
{2}\mu}{\left(  -E\right)  ^{1/4}}+\frac{i\sqrt{2}\kappa}{2\left(  -E\right)
^{3/4}}\left(  2v+1-4l\right)  \right]  d_{v+1}\nonumber\\
+  &  \left[  {\frac{2\xi}{\left(  -E\right)  ^{1/2}}}-{\frac{{\kappa}^{2}%
}{2\left(  -E\right)  ^{3/2}}}-2v+4l\right]  d_{v}=0
\end{align}
and $N(\alpha,\beta,\gamma,\delta,z)$ is the biconfluent Heun function
\cite{ronveaux1995heun,slavyanov2000special,li2016exact}.

The biconfluent Heun function $N\left(  4l+2,{\frac{i\sqrt{2}\kappa}{\left(
-E\right)  ^{3/4}}},{\frac{2\xi}{\left(  -E\right)  ^{1/2}}}-{\frac{{\kappa
}^{2}}{2\left(  -E\right)  ^{3/2}}},{\frac{-i4\sqrt{2}\mu}{\left(  -E\right)
^{1/4}}},z\right)  $ has an expansion at $z=0$ \cite{ronveaux1995heun}:%
\begin{equation}
N\left(  4l+2,{\frac{i\sqrt{2}\kappa}{\left(  -E\right)  ^{3/4}}},{\frac{2\xi
}{\left(  -E\right)  ^{1/2}}}-{\frac{{\kappa}^{2}}{2\left(  -E\right)  ^{3/2}%
}},{\frac{-i4\sqrt{2}\mu}{\left(  -E\right)  ^{1/4}}},z\right)  =\sum_{n\geq
0}\frac{A_{n}}{\left(  4l+3\right)  _{n}}\frac{z^{n}}{n!},
\end{equation}
where the expansion coefficients is determined by the recurrence relation,%
\begin{align}
A_{0}  &  =1,\text{ \ }A_{1}=\frac{-i2\sqrt{2}\mu}{\left(  -E\right)  ^{1/4}%
}+\frac{i\sqrt{2}\kappa}{2\left(  -E\right)  ^{3/4}}\left(  4l+3\right)
,\nonumber\\
A_{n+2}  &  =\left[  \left(  2n+4l+5\right)  \frac{i\sqrt{2}\kappa}{2\left(
-E\right)  ^{3/4}}+\frac{-i2\sqrt{2}\mu}{\left(  -E\right)  ^{1/4}}\right]
A_{n+1}\nonumber\\
&  -\left(  n+1\right)  \left(  n+4l+3\right)  \left[  {\frac{2\xi}{\left(
-E\right)  ^{1/2}}}-{\frac{{\kappa}^{2}}{2\left(  -E\right)  ^{3/2}}%
}-4l-4-2n\right]  A_{n},
\end{align}
and $\left(  a\right)  _{n}=\Gamma\left(  a+n\right)  /\Gamma\left(  a\right)
$ is Pochhammer's symbol.

Only $y_{l}^{\left(  1\right)  }\left(  z\right)  $ satisfies the boundary
condition for the regular solution at $r=0$, so the radial eigenfunction reads%
\begin{align}
u_{l}\left(  z\right)   &  =A_{l}\exp\left(  -\frac{z^{2}}{2}-\frac{\beta}%
{2}z\right)  z^{2\left(  l+1\right)  }y_{l}^{\left(  1\right)  }\left(
z\right) \nonumber\\
&  =A_{l}\exp\left(  -\frac{z^{2}}{2}-\frac{\beta}{2}z\right)  z^{2\left(
l+1\right)  }N\left(  4l+2,{\frac{i\sqrt{2}\kappa}{\left(  -E\right)  ^{3/4}}%
},{\frac{2\xi}{\left(  -E\right)  ^{1/2}}}-{\frac{{\kappa}^{2}}{2\left(
-E\right)  ^{3/2}}},{\frac{-i4\sqrt{2}\mu}{\left(  -E\right)  ^{1/4}}%
},z\right)  .
\end{align}
By Eq. (\ref{zro13212}), we obtain the regular solution,%

\begin{align}
u_{l}\left(  r\right)   &  =A_{l}\exp\left(  \left(  -E\right)  ^{1/2}%
r+\frac{\kappa\sqrt{r}}{\left(  -E\right)  ^{1/2}}\right)  r^{l+1}\nonumber\\
&  \times N\left(  4l+2,{\frac{i\sqrt{2}\kappa}{\left(  -E\right)  ^{3/4}}%
},{\frac{2\xi}{\left(  -E\right)  ^{1/2}}}-{\frac{{\kappa}^{2}}{2\left(
-E\right)  ^{3/2}}},{\frac{-i4\sqrt{2}\mu}{\left(  -E\right)  ^{1/4}}%
},i\left(  -E\right)  ^{1/4}\sqrt{2r}\right)  . \label{regularo13212}%
\end{align}

\subsection{The irregular solution}

The irregular solution is a solution satisfying the boundary condition at
$r\rightarrow\infty$ \cite{li2016exact}. The boundary conditions for bound
states and scattering states at $r\rightarrow\infty$ are different.

The Biconfluent Heun equation (\ref{eqfo13212}) has two linearly independent
irregular solutions \cite{ronveaux1995heun}:%
\begin{align}
&  B_{l}^{+}\left(  4l+2,{\frac{i\sqrt{2}\kappa}{\left(  -E\right)  ^{3/4}}%
},{\frac{2\xi}{\left(  -E\right)  ^{1/2}}}-{\frac{{\kappa}^{2}}{2\left(
-E\right)  ^{3/2}}},{\frac{-i4\sqrt{2}\mu}{\left(  -E\right)  ^{1/4}}%
,}z\right) \nonumber\\
&  =\exp\left(  \frac{i\sqrt{2}\kappa}{\left(  -E\right)  ^{3/4}}%
z+z^{2}\right)  B_{l}^{+}\left(  4l+2,{\frac{\sqrt{2}\kappa}{\left(
-E\right)  ^{3/4}}},{\frac{{\kappa}^{2}}{2\left(  -E\right)  ^{3/2}}%
-\frac{2\xi}{\left(  -E\right)  ^{1/2}}},{\frac{4\sqrt{2}\mu}{\left(
-E\right)  ^{1/4}},-i}z\right) \nonumber\\
&  =\exp\left(  \frac{i\sqrt{2}\kappa}{\left(  -E\right)  ^{3/4}}%
z+z^{2}\right)  \left(  -iz\right)  ^{\frac{1}{2}\left(  {\frac{{\kappa}^{2}%
}{2\left(  -E\right)  ^{3/2}}-\frac{2\xi}{\left(  -E\right)  ^{1/2}}%
}-4l-4\right)  }\sum_{n\geq0}\frac{a_{n}}{\left(  -iz\right)  ^{n}},
\label{f1o13212}%
\end{align}%
\begin{align}
&  H_{l}^{+}\left(  4l+2,{\frac{i\sqrt{2}\kappa}{\left(  -E\right)  ^{3/4}}%
},{\frac{2\xi}{\left(  -E\right)  ^{1/2}}}-{\frac{{\kappa}^{2}}{2\left(
-E\right)  ^{3/2}}},{\frac{-i4\sqrt{2}\mu}{\left(  -E\right)  ^{1/4}}%
,}z\right) \nonumber\\
&  =\exp\left(  \frac{i\sqrt{2}\kappa}{\left(  -E\right)  ^{3/4}}%
z+z^{2}\right)  H_{l}^{+}\left(  4l+2,{\frac{\sqrt{2}\kappa}{\left(
-E\right)  ^{3/4}}},{\frac{{\kappa}^{2}}{2\left(  -E\right)  ^{3/2}}%
-\frac{2\xi}{\left(  -E\right)  ^{1/2}}},{\frac{4\sqrt{2}\mu}{\left(
-E\right)  ^{1/4}},-i}z\right) \nonumber\\
&  =\left(  -iz\right)  ^{-\frac{1}{2}\left(  {\frac{{\kappa}^{2}}{2\left(
-E\right)  ^{3/2}}-\frac{2\xi}{\left(  -E\right)  ^{1/2}}}+4l+4\right)  }%
\sum_{n\geq0}\frac{e_{n}}{\left(  -iz\right)  ^{n}} \label{f2o13212}%
\end{align}
with the expansion coefficients given by the recurrence relation
\[
a_{0}=1,\text{ \ }a_{1}=\frac{\sqrt{2}\mu}{\left(  -E\right)  ^{1/4}}%
+\frac{\sqrt{2}\kappa}{4\left(  -E\right)  ^{3/4}}\left[  {\frac{{\kappa}^{2}%
}{2\left(  -E\right)  ^{3/2}}-\frac{2\xi}{\left(  -E\right)  ^{1/2}}%
}-1\right]  ,
\]%
\begin{align}
&  2\left(  n+2\right)  a_{n+2}+\left\{  \frac{\sqrt{2}\kappa}{\left(
-E\right)  ^{3/4}}\left[  \frac{3}{2}-{\frac{{\kappa}^{2}}{4\left(  -E\right)
^{3/2}}+\frac{\xi}{\left(  -E\right)  ^{1/2}}}+n\right]  -\frac{2\sqrt{2}\mu
}{\left(  -E\right)  ^{1/4}}\right\}  a_{n+1}\nonumber\\
&  +\left\{  \left[  {\frac{{\kappa}^{2}}{4\left(  -E\right)  ^{3/2}}%
-\frac{\xi}{\left(  -E\right)  ^{1/2}}}-1\right]  ^{2}-\frac{1}{4}\left(
4l+2\right)  ^{2}+n\left[  n+2-{\frac{{\kappa}^{2}}{2\left(  -E\right)
^{3/2}}+\frac{2\xi}{\left(  -E\right)  ^{1/2}}}\right]  \right\}  a_{n}=0
\end{align}
and%
\[
e_{0}=1,\text{ \ }e_{1}=-\frac{\sqrt{2}\mu}{\left(  -E\right)  ^{1/4}}%
-\frac{\sqrt{2}\kappa}{4\left(  -E\right)  ^{3/4}}\left(  {\frac{{\kappa}^{2}%
}{2\left(  -E\right)  ^{3/2}}-\frac{2\xi}{\left(  -E\right)  ^{1/2}}%
}+1\right)
\]%
\begin{align}
&  2\left(  n+2\right)  e_{n+2}+\left\{  \frac{\sqrt{2}\kappa}{\left(
-E\right)  ^{3/4}}\left[  \frac{3}{2}+{\frac{{\kappa}^{2}}{4\left(  -E\right)
^{3/2}}-\frac{\xi}{\left(  -E\right)  ^{1/2}}}+n\right]  +\frac{2\sqrt{2}\mu
}{\left(  -E\right)  ^{1/4}}\right\}  e_{n+1}\nonumber\\
&  -\left\{  \left[  {\frac{{\kappa}^{2}}{4\left(  -E\right)  ^{3/2}}%
-\frac{\xi}{\left(  -E\right)  ^{1/2}}}+1\right]  ^{2}-\frac{1}{4}\left(
4l+2\right)  ^{2}+n\left[  n+2+{\frac{{\kappa}^{2}}{2\left(  -E\right)
^{3/2}}-\frac{2\xi}{\left(  -E\right)  ^{1/2}}}\right]  \right\}  e_{n}=0.
\end{align}

\subsection{Bound states and scattering states}

\subsubsection{The bound state}

To construct the solution, we first express the regular solution
(\ref{regular2o1o1}) as a linear combination of the two irregular solutions
(\ref{f1o13212}) and (\ref{f2o13212}).

The regular solution (\ref{regularo13212}), with the relation
\cite{ronveaux1995heun,li2016exact}%

\begin{align}
&  N\left(  4l+2,{\frac{i\sqrt{2}\kappa}{\left(  -E\right)  ^{3/4}}}%
,{\frac{2\xi}{\left(  -E\right)  ^{1/2}}}-{\frac{{\kappa}^{2}}{2\left(
-E\right)  ^{3/2}}},{\frac{-i4\sqrt{2}\mu}{\left(  -E\right)  ^{1/4}}%
},z\right) \nonumber\\
&  =K_{1}\left(  4l+2,{\frac{i\sqrt{2}\kappa}{\left(  -E\right)  ^{3/4}}%
},{\frac{2\xi}{\left(  -E\right)  ^{1/2}}}-{\frac{{\kappa}^{2}}{2\left(
-E\right)  ^{3/2}}},{\frac{-i4\sqrt{2}\mu}{\left(  -E\right)  ^{1/4}}}\right)
\nonumber\\
&  \times B_{l}^{+}\left(  4l+2,{\frac{i\sqrt{2}\kappa}{\left(  -E\right)
^{3/4}}},{\frac{2\xi}{\left(  -E\right)  ^{1/2}}}-{\frac{{\kappa}^{2}%
}{2\left(  -E\right)  ^{3/2}}},{\frac{-i4\sqrt{2}\mu}{\left(  -E\right)
^{1/4}}},z\right) \nonumber\\
&  +K_{2}\left(  4l+2,{\frac{i\sqrt{2}\kappa}{\left(  -E\right)  ^{3/4}}%
},{\frac{2\xi}{\left(  -E\right)  ^{1/2}}}-{\frac{{\kappa}^{2}}{2\left(
-E\right)  ^{3/2}}},{\frac{-i4\sqrt{2}\mu}{\left(  -E\right)  ^{1/4}}}\right)
\nonumber\\
&  \times H_{l}^{+}\left(  4l+2,{\frac{i\sqrt{2}\kappa}{\left(  -E\right)
^{3/4}}},{\frac{2\xi}{\left(  -E\right)  ^{1/2}}}-{\frac{{\kappa}^{2}%
}{2\left(  -E\right)  ^{3/2}}},{\frac{-i4\sqrt{2}\mu}{\left(  -E\right)
^{1/4}}},z\right)
\end{align}

and the expansions (\ref{f1o13212}) and (\ref{f2o13212}), become%
\begin{align}
u_{l}\left(  r\right)   &  =A_{l}K_{1}\left(  4l+2,{\frac{i\sqrt{2}\kappa
}{\left(  -E\right)  ^{3/4}}},{\frac{2\xi}{\left(  -E\right)  ^{1/2}}}%
-{\frac{{\kappa}^{2}}{2\left(  -E\right)  ^{3/2}}},{\frac{-i4\sqrt{2}\mu
}{\left(  -E\right)  ^{1/4}}}\right)  \nonumber\\
&  \times\exp\left(  -\left(  -E\right)  ^{1/2}r-\frac{\kappa\sqrt{r}}{\left(
-E\right)  ^{1/2}}\right)  r^{{\frac{{\kappa}^{2}}{8\left(  -E\right)  ^{3/2}%
}-\frac{\xi}{2\left(  -E\right)  ^{1/2}}}}\sum_{n\geq0}\frac{a_{n}}{\left[
\left(  -E\right)  ^{1/4}\sqrt{2r}\right]  ^{n}}\nonumber\\
&  +A_{l}K_{2}\left(  4l+2,{\frac{i\sqrt{2}\kappa}{\left(  -E\right)  ^{3/4}}%
},{\frac{2\xi}{\left(  -E\right)  ^{1/2}}}-{\frac{{\kappa}^{2}}{2\left(
-E\right)  ^{3/2}}},{\frac{-i4\sqrt{2}\mu}{\left(  -E\right)  ^{1/4}}}\right)
\nonumber\\
&  \times\exp\left(  \left(  -E\right)  ^{1/2}r+\frac{\kappa\sqrt{r}}{\left(
-E\right)  ^{1/2}}\right)  r^{-{\frac{{\kappa}^{2}}{8\left(  -E\right)
^{3/2}}+\frac{\xi}{2\left(  -E\right)  ^{1/2}}}}\sum_{n\geq0}\frac{e_{n}%
}{\left[  \left(  -E\right)  ^{1/4}\sqrt{2r}\right]  ^{n}}%
,\label{uexpzfo13212}%
\end{align}
where $K_{1}\left(  4l+2,{\frac{i\sqrt{2}\kappa}{\left(  -E\right)  ^{3/4}}%
},{\frac{2\xi}{\left(  -E\right)  ^{1/2}}}-{\frac{{\kappa}^{2}}{2\left(
-E\right)  ^{3/2}}},{\frac{-i4\sqrt{2}\mu}{\left(  -E\right)  ^{1/4}}}\right)
$ and $K_{2}\left(  4l+2,{\frac{i\sqrt{2}\kappa}{\left(  -E\right)  ^{3/4}}%
},{\frac{2\xi}{\left(  -E\right)  ^{1/2}}}-{\frac{{\kappa}^{2}}{2\left(
-E\right)  ^{3/2}}},{\frac{-i4\sqrt{2}\mu}{\left(  -E\right)  ^{1/4}}}\right)
$ are combination coefficients\ and $z=i\left(  -E\right)  ^{1/4}\sqrt{2r}$.

The boundary condition of bound states, $\left.  u\left(  r\right)
\right\vert _{r\rightarrow\infty}\rightarrow0$, requires that the coefficient
of the second term must vanish since this term diverges when $r\rightarrow
\infty$, i.e.,
\begin{equation}
K_{2}\left(  4l+2,{\frac{i\sqrt{2}\kappa}{\left(  -E\right)  ^{3/4}}}%
,{\frac{2\xi}{\left(  -E\right)  ^{1/2}}}-{\frac{{\kappa}^{2}}{2\left(
-E\right)  ^{3/2}}},{\frac{-i4\sqrt{2}\mu}{\left(  -E\right)  ^{1/4}}}\right)
=0, \label{eigenvalueo13212}%
\end{equation}
where%
\begin{align}
&  K_{2}\left(  4l+2,{\frac{i\sqrt{2}\kappa}{\left(  -E\right)  ^{3/4}}%
},{\frac{2\xi}{\left(  -E\right)  ^{1/2}}}-{\frac{{\kappa}^{2}}{2\left(
-E\right)  ^{3/2}}},{\frac{-i4\sqrt{2}\mu}{\left(  -E\right)  ^{1/4}}}\right)
\nonumber\\
&  =\frac{\Gamma\left(  4l+3\right)  }{\Gamma\left(  2l+1-{\frac{\xi}{\left(
-E\right)  ^{1/2}}}+{\frac{{\kappa}^{2}}{4\left(  -E\right)  ^{3/2}}}\right)
\Gamma\left(  2l+2+{\frac{\xi}{\left(  -E\right)  ^{1/2}}}-{\frac{{\kappa}%
^{2}}{4\left(  -E\right)  ^{3/2}}}\right)  }\nonumber\\
&  \times J_{2l+2+{\frac{\xi}{\left(  -E\right)  ^{1/2}}}-{\frac{{\kappa}^{2}%
}{4\left(  -E\right)  ^{3/2}}}}\left(  2l+1+{\frac{\xi}{\left(  -E\right)
^{1/2}}}-{\frac{{\kappa}^{2}}{4\left(  -E\right)  ^{3/2}}},\frac{i\sqrt
{2}\kappa}{\left(  -E\right)  ^{3/4}},6l+3-{\frac{\xi}{\left(  -E\right)
^{1/2}}}\right. \nonumber\\
&  +\left.  {\frac{{\kappa}^{2}}{4\left(  -E\right)  ^{3/2}}},\frac
{-i4\sqrt{2}\mu}{\left(  -E\right)  ^{1/4}}+\frac{i\sqrt{2}\kappa}{2\left(
-E\right)  ^{3/4}}\left[  {\frac{2\xi}{\left(  -E\right)  ^{1/2}}}%
-{\frac{{\kappa}^{2}}{2\left(  -E\right)  ^{3/2}}}-4l-2\right]  \right)
\end{align}
with%
\begin{align}
&  J_{2l+2+{\frac{\xi}{\left(  -E\right)  ^{1/2}}}-{\frac{{\kappa}^{2}%
}{4\left(  -E\right)  ^{3/2}}}}\left(  2l+1+{\frac{\xi}{\left(  -E\right)
^{1/2}}}-{\frac{{\kappa}^{2}}{4\left(  -E\right)  ^{3/2}}},\frac{i\sqrt
{2}\kappa}{\left(  -E\right)  ^{3/4}},6l+3-{\frac{\xi}{\left(  -E\right)
^{1/2}}}\right. \nonumber\\
&  +\left.  {\frac{{\kappa}^{2}}{4\left(  -E\right)  ^{3/2}}},\frac
{-i4\sqrt{2}\mu}{\left(  -E\right)  ^{1/4}}+\frac{i\sqrt{2}\kappa}{2\left(
-E\right)  ^{3/4}}\left[  {\frac{2\xi}{\left(  -E\right)  ^{1/2}}}%
-{\frac{{\kappa}^{2}}{2\left(  -E\right)  ^{3/2}}}-4l-2\right]  \right)
\nonumber\\
&  =\int_{0}^{\infty}dxx^{2l+1+{\frac{\xi}{\left(  -E\right)  ^{1/2}}}%
-{\frac{{\kappa}^{2}}{4\left(  -E\right)  ^{3/2}}}}e^{-x^{2}}N\left(
2l+1+{\frac{\xi}{\left(  -E\right)  ^{1/2}}}-{\frac{{\kappa}^{2}}{4\left(
-E\right)  ^{3/2}}},\frac{i\sqrt{2}\kappa}{\left(  -E\right)  ^{3/4}%
},6l+3-{\frac{\xi}{\left(  -E\right)  ^{1/2}}}\right. \nonumber\\
&  +\left.  {\frac{{\kappa}^{2}}{4\left(  -E\right)  ^{3/2}}},\frac
{-i4\sqrt{2}\mu}{\left(  -E\right)  ^{1/4}}+\frac{i\sqrt{2}\kappa}{2\left(
-E\right)  ^{3/4}}\left[  {\frac{2\xi}{\left(  -E\right)  ^{1/2}}}%
-{\frac{{\kappa}^{2}}{2\left(  -E\right)  ^{3/2}}}-4l-2\right]  ,x\right)  .
\end{align}

Eq. (\ref{eigenvalueo13212}) is an implicit expression of the eigenvalue.

The eigenfunction, by Eqs. (\ref{uexpzfo13212}) and (\ref{eigenvalueo13212}),
reads%
\begin{align}
u_{l}\left(  r\right)   &  =A_{l}K_{1}\left(  4l+2,{\frac{i\sqrt{2}\kappa
}{\left(  -E\right)  ^{3/4}}},{\frac{2\xi}{\left(  -E\right)  ^{1/2}}}%
-{\frac{{\kappa}^{2}}{2\left(  -E\right)  ^{3/2}}},{\frac{-i4\sqrt{2}\mu
}{\left(  -E\right)  ^{1/4}}}\right)  \nonumber\\
&  \times\exp\left(  -\left(  -E\right)  ^{1/2}r-\frac{\kappa\sqrt{r}}{\left(
-E\right)  ^{1/2}}\right)  r^{{\frac{{\kappa}^{2}}{8\left(  -E\right)  ^{3/2}%
}-\frac{\xi}{2\left(  -E\right)  ^{1/2}}}}\sum_{n\geq0}\frac{a_{n}}{\left[
\left(  -E\right)  ^{1/4}\sqrt{2r}\right]  ^{n}}.
\end{align}

\subsubsection{The scattering state}

For scattering states, $E>0$, introduce%
\begin{equation}
E=k^{2}.
\end{equation}

The singularity of the $S$-matrix on the positive imaginary axis corresponds
to the eigenvalues of bound states \cite{joachain1975quantum}, so the zero of
$K_{2}\left(  4l+2,-{\frac{\left(  1+i\right)  \kappa}{k^{3/2}}},{\frac{2i\xi
}{k}}+i{\frac{{\kappa}^{2}}{2k^{3}}},{\frac{4\left(  1-i\right)  \mu}{k^{1/2}%
}}\right)  $ on the positive imaginary is just the singularity of the
$S$-matrix. Considering that the $S$-matrix is unitary, i.e.,%
\begin{equation}
S_{l}=e^{2i\delta_{l}}, \label{SMo13212}%
\end{equation}
we have%
\begin{equation}
S_{l}\left(  k\right)  =\frac{K_{2}^{\ast}\left(  4l+2,-{\frac{\left(
1+i\right)  \kappa}{k^{3/2}}},{\frac{2i\xi}{k}}+i{\frac{{\kappa}^{2}}{2k^{3}}%
},{\frac{4\left(  1-i\right)  \mu}{k^{1/2}}}\right)  }{\left(  4l+2,-{\frac
{\left(  1+i\right)  \kappa}{k^{3/2}}},{\frac{2i\xi}{k}}+i{\frac{{\kappa}^{2}%
}{2k^{3}}},{\frac{4\left(  1-i\right)  \mu}{k^{1/2}}}\right)  }=\frac
{K_{2}\left(  4l+2,-{\frac{\left(  1-i\right)  \kappa}{k^{3/2}}},-{\frac
{2i\xi}{k}}-i{\frac{{\kappa}^{2}}{2k^{3}}},{\frac{4\left(  1+i\right)  \mu
}{k^{1/2}}}\right)  }{K_{2}\left(  4l+2,-{\frac{\left(  1+i\right)  \kappa
}{k^{3/2}}},{\frac{2i\xi}{k}}+i{\frac{{\kappa}^{2}}{2k^{3}}},{\frac{4\left(
1-i\right)  \mu}{k^{1/2}}}\right)  }. \label{SM1o13212}%
\end{equation}

The scattering wave function can be constructed with the help of the
$S$-matrix. The scattering wave function can be expressed as a linear
combination of the radially ingoing wave $u_{in}\left(  r\right)  $ and the
radially outgoing wave $u_{out}\left(  r\right)  $, which are conjugate to
each other, i.e., \cite{joachain1975quantum}%
\begin{equation}
u_{l}\left(  r\right)  =A_{l}\left[  \left(  -1\right)  ^{l+1}u_{in}\left(
r\right)  +S_{l}\left(  k\right)  u_{out}\left(  r\right)  \right]
.\label{inouto13212}%
\end{equation}
From Eq. (\ref{uexpzfo13212}), we have%
\begin{align}
u_{in}\left(  r\right)   &  =\exp\left(  -ikr+i\frac{\kappa\sqrt{r}}%
{k}\right)  r^{i{\frac{{\kappa}^{2}}{8k^{3}}+i\frac{\xi}{2k}}}\sum_{n\geq
0}\frac{e_{n}}{\left(  -2ikr\right)  ^{n/2}},\nonumber\\
u_{out}\left(  r\right)   &  =\exp\left(  ikr-i\frac{\kappa\sqrt{r}}%
{k}\right)  r^{-i{\frac{{\kappa}^{2}}{8k^{3}}-i\frac{\xi}{2k}}}\sum_{n\geq
0}\frac{e_{n}^{\ast}}{\left(  2ikr\right)  ^{n/2}}.
\end{align}
Then by Eq. (\ref{SM1o13212}) we obtain the scattering wave function,%
\begin{align}
u_{l}\left(  r\right)   &  =A_{l}\left[  \left(  -1\right)  ^{l+1}\exp\left(
-ikr+i\frac{\kappa\sqrt{r}}{k}\right)  r^{i{\frac{{\kappa}^{2}}{8k^{3}}%
+i\frac{\xi}{2k}}}\sum_{n\geq0}\frac{e_{n}}{\left(  -2ikr\right)  ^{n/2}%
}\right.  \nonumber\\
&  +\left.  \frac{K_{2}\left(  4l+2,-{\frac{\left(  1-i\right)  \kappa
}{k^{3/2}}},-{\frac{2i\xi}{k}}-i{\frac{{\kappa}^{2}}{2k^{3}}},{\frac{4\left(
1+i\right)  \mu}{k^{1/2}}}\right)  }{K_{2}\left(  4l+2,-{\frac{\left(
1+i\right)  \kappa}{k^{3/2}}},{\frac{2i\xi}{k}}+i{\frac{{\kappa}^{2}}{2k^{3}}%
},{\frac{4\left(  1-i\right)  \mu}{k^{1/2}}}\right)  }\exp\left(
ikr-i\frac{\kappa\sqrt{r}}{k}\right)  r^{-i{\frac{{\kappa}^{2}}{8k^{3}}%
-i\frac{\xi}{2k}}}\sum_{n\geq0}\frac{e_{n}^{\ast}}{\left(  2ikr\right)
^{n/2}}\right]  .
\end{align}
Taking $r\rightarrow\infty$, we have
\begin{align}
&  u_{l}\left(  r\right)  \overset{r\rightarrow\infty}{\sim}A_{l}\left[
\left(  -1\right)  ^{l+1}\exp\left(  -ikr+i\frac{\kappa\sqrt{r}}{k}\right)
r^{i{\frac{{\kappa}^{2}}{8k^{3}}+i\frac{\xi}{2k}}}\right.  \nonumber\\
&  \left.  +\frac{K_{2}\left(  4l+2,-{\frac{\left(  1-i\right)  \kappa
}{k^{3/2}}},-{\frac{2i\xi}{k}}-i{\frac{{\kappa}^{2}}{2k^{3}}},{\frac{4\left(
1+i\right)  \mu}{k^{1/2}}}\right)  }{K_{2}\left(  4l+2,-{\frac{\left(
1+i\right)  \kappa}{k^{3/2}}},{\frac{2i\xi}{k}}+i{\frac{{\kappa}^{2}}{2k^{3}}%
},{\frac{4\left(  1-i\right)  \mu}{k^{1/2}}}\right)  }\exp\left(
ikr-i\frac{\kappa\sqrt{r}}{k}\right)  r^{-i{\frac{{\kappa}^{2}}{8k^{3}}%
-i\frac{\xi}{2k}}}\right]  \nonumber\\
&  =A_{l}e^{i\delta_{l}}\sin\left(  kr-\frac{\kappa}{k}\sqrt{r}-\left(
{\frac{{\kappa}^{2}}{8k^{3}}+\frac{\xi}{2k}}\right)  \ln2kr+\delta_{l}%
-\frac{l\pi}{2}\right)  .
\end{align}
By Eqs. (\ref{SMo13212}) and (\ref{SM1o13212}), we obtain the scattering phase
shift%
\begin{equation}
\delta_{l}=-\arg K_{2}\left(  4l+2,-{\frac{\left(  1+i\right)  \kappa}%
{k^{3/2}}},{\frac{2i\xi}{k}}+i{\frac{{\kappa}^{2}}{2k^{3}}},{\frac{4\left(
1-i\right)  \mu}{k^{1/2}}}\right)  .
\end{equation}

\section{The exact solution of $U\left(  r\right)  =\xi r^{2}+\mu r^{6}+\kappa
r^{4}$ \label{V264}}

In this appendix, we provide an exact solution of the eigenproblem of the
potential
\begin{equation}
U\left(  r\right)  =\xi r^{2}+\mu r^{6}+\kappa r^{4}%
\end{equation}
by solving the radial equation directly. This potential has only bound states.

The radial equation reads%
\begin{equation}
\frac{d^{2}u_{l}\left(  r\right)  }{dr^{2}}+\left[  E-\frac{l\left(
l+1\right)  }{r^{2}}-\xi r^{2}-\mu r^{6}-\kappa r^{4}\right]  u_{l}\left(
r\right)  =0. \label{radialeq264}%
\end{equation}
Using the variable substitution%
\begin{equation}
z=i\frac{\sqrt{2}}{2}\mu^{1/4}{r}^{2} \label{zr264}%
\end{equation}
and introducing $f_{l}\left(  z\right)  $ by%
\begin{equation}
u_{l}\left(  z\right)  =A_{l}\exp\left(  -\frac{z^{2}}{2}-\frac{\beta}%
{2}z\right)  z^{\left(  l+1\right)  /2}f_{l}\left(  z\right)
\end{equation}
with $A_{l}$ a constant, we convert the radial equation (\ref{radialeqo13212})
into an equation of $f_{l}\left(  z\right)  $:%
\begin{align}
&  zf_{l}^{\prime\prime}\left(  z\right)  +\left(  l+\frac{3}{2}-{\frac
{i\sqrt{2}\kappa}{2{\mu}^{3/4}}}z-2z^{2}\right)  f_{l}^{\prime}\left(
z\right) \nonumber\\
&  +\left[  \left(  {\frac{\xi}{2{\mu}^{1/2}}}-{\frac{{\kappa}^{2}}{8{\mu
}^{3/2}}}-l-\frac{5}{2}\right)  z-\frac{i\sqrt{2}E}{4\mu^{1/4}}-\frac
{i\sqrt{2}\kappa}{4{\mu}^{3/4}}\left(  l+\frac{3}{2}\right)  \right]
f_{l}\left(  z\right)  =0. \label{eqf264}%
\end{align}
This is a Biconfluent Heun equation \cite{ronveaux1995heun}.

The choice of the boundary condition has been discussed in Ref.
\cite{li2016exact}.

\subsection{The regular solution}

The regular solution is a solution satisfying the boundary condition at $r=0$
\cite{li2016exact}. The regular solution at $r=0$\ should satisfy the boundary
condition $\lim_{r\rightarrow0}u_{l}\left(  r\right)  /r^{l+1}=1$. In this
section, we provide the regular solution of Eq. (\pageref{eqf264}).

The Biconfluent Heun equation (\ref{eqf264}) has two linearly independent
solutions \cite{ronveaux1995heun}%
\begin{align}
y_{l}^{\left(  1\right)  }\left(  z\right)   &  =N\left(  l+\frac{1}{2}%
,{\frac{i\sqrt{2}\kappa}{2{\mu}^{3/4}}},{\frac{\xi}{2{\mu}^{1/2}}}%
-{\frac{{\kappa}^{2}}{8{\mu}^{3/2}}},{\frac{i\sqrt{2}E}{2\mu^{1/4}}},z\right)
,\\
y_{l}^{\left(  2\right)  }\left(  z\right)   &  =cN\left(  l+\frac{1}%
{2},{\frac{i\sqrt{2}\kappa}{2{\mu}^{3/4}}},{\frac{\xi}{2{\mu}^{1/2}}}%
-{\frac{{\kappa}^{2}}{8{\mu}^{3/2}}},{\frac{i\sqrt{2}E}{2\mu^{1/4}}},z\right)
\ln z+\sum_{n\geq0}d_{n}z^{n-l-\frac{1}{2}},
\end{align}
where%
\begin{equation}
c=\frac{1}{l+\frac{1}{2}}\left\{  d_{l-1/2}\left[  \frac{i\sqrt{2}E}%
{4\mu^{1/4}}+\frac{i\sqrt{2}\kappa}{4{\mu}^{3/4}}\left(  l-\frac{1}{2}\right)
\right]  -d_{l-3/2}\left(  {\frac{\xi}{2{\mu}^{1/2}}}-{\frac{{\kappa}^{2}%
}{8{\mu}^{3/2}}}+\frac{3}{2}-l\right)  \right\}
\end{equation}
is a constant with the coefficient $d_{\nu}$ given by the following recurrence
relation,%
\begin{align}
&  d_{-1}=0,\text{ \ }d_{0}=1,\nonumber\\
&  \left(  v+2\right)  \left(  v+\frac{3}{2}-l\right)  d_{v+2}-\left[
\frac{i\sqrt{2}E}{4\mu^{1/4}}+\frac{i\sqrt{2}\kappa}{4{\mu}^{3/4}}\left(
2v+\frac{5}{2}-l\right)  \right]  d_{v+1}\nonumber\\
&  +\left(  {\frac{\xi}{2{\mu}^{1/2}}}-{\frac{{\kappa}^{2}}{8{\mu}^{3/2}}%
}-2v-\frac{3}{2}+l\right)  d_{v}=0
\end{align}
and $N(\alpha,\beta,\gamma,\delta,z)$ is the biconfluent Heun function
\cite{ronveaux1995heun,slavyanov2000special,li2016exact}.

The biconfluent Heun function $N\left(  l+\frac{1}{2},{\frac{i\sqrt{2}\kappa
}{2{\mu}^{3/4}}},{\frac{\xi}{2{\mu}^{1/2}}}-{\frac{{\kappa}^{2}}{8{\mu}^{3/2}%
}},{\frac{i\sqrt{2}E}{2\mu^{1/4}}},z\right)  $ has an expansion at $z=0$
\cite{ronveaux1995heun}:%
\begin{equation}
N\left(  l+\frac{1}{2},{\frac{i\sqrt{2}\kappa}{2{\mu}^{3/4}}},{\frac{\xi
}{2{\mu}^{1/2}}}-{\frac{{\kappa}^{2}}{8{\mu}^{3/2}}},{\frac{i\sqrt{2}E}%
{2\mu^{1/4}}},z\right)  =\sum_{n\geq0}\frac{A_{n}}{\left(  l+\frac{3}%
{2}\right)  _{n}}\frac{z^{n}}{n!},
\end{equation}
where the expansion coefficients is determined by the recurrence relation,%
\begin{align}
A_{0}  &  =1,\text{ \ }A_{1}=\frac{i\sqrt{2}E}{4\mu^{1/4}}+\frac{i\sqrt
{2}\kappa}{4{\mu}^{3/4}}\left(  l+\frac{3}{2}\right)  ,\nonumber\\
A_{n+2}  &  =\left[  \left(  2n+2+l+\frac{3}{2}\right)  \frac{i\sqrt{2}\kappa
}{4{\mu}^{3/4}}+\frac{i\sqrt{2}E}{4\mu^{1/4}}\right]  A_{n+1}\nonumber\\
&  -\left(  n+1\right)  \left(  n+l+\frac{3}{2}\right)  \left(  {\frac{\xi
}{2{\mu}^{1/2}}}-{\frac{{\kappa}^{2}}{8{\mu}^{3/2}}}-l-\frac{5}{2}-2n\right)
A_{n},
\end{align}
and $\left(  a\right)  _{n}=\Gamma\left(  a+n\right)  /\Gamma\left(  a\right)
$ is Pochhammer's symbol.

Only $y_{l}^{\left(  1\right)  }\left(  z\right)  $ satisfies the boundary
condition for the regular solution at $r=0$, so the radial eigenfunction reads%
\begin{align}
u_{l}\left(  z\right)   &  =A_{l}\exp\left(  -\frac{z^{2}}{2}-\frac{\beta}%
{2}z\right)  z^{\left(  l+1\right)  /2}y_{l}^{\left(  1\right)  }\left(
z\right) \nonumber\\
&  =A_{l}\exp\left(  -\frac{z^{2}}{2}-\frac{\beta}{2}z\right)  z^{\left(
l+1\right)  /2}N\left(  l+\frac{1}{2},{\frac{i\sqrt{2}\kappa}{2{\mu}^{3/4}}%
},{\frac{\xi}{2{\mu}^{1/2}}}-{\frac{{\kappa}^{2}}{8{\mu}^{3/2}}},{\frac
{i\sqrt{2}E}{2\mu^{1/4}}},z\right)  .
\end{align}
By Eq. (\ref{zr264}), we obtain the regular solution,%

\begin{equation}
u_{l}\left(  r\right)  =A_{l}\exp\left(  \frac{\mu^{1/2}}{4}{r}^{4}%
+\frac{\kappa}{4{\mu}^{1/2}}{r}^{2}\right)  r^{l+1}N\left(  l+\frac{1}%
{2},{\frac{i\sqrt{2}\kappa}{2{\mu}^{3/4}}},{\frac{\xi}{2{\mu}^{1/2}}}%
-{\frac{{\kappa}^{2}}{8{\mu}^{3/2}}},{\frac{i\sqrt{2}E}{2\mu^{1/4}}}%
,i\frac{\sqrt{2}}{2}\mu^{1/4}{r}^{2}\right)  . \label{regular264}%
\end{equation}

\subsection{The irregular solution}

The irregular solution is a solution satisfying the boundary condition at
$r\rightarrow\infty$ \cite{li2016exact}.

The Biconfluent Heun equation (\ref{eqf264}) has two linearly independent
irregular solutions \cite{ronveaux1995heun}:%
\begin{align}
&  B_{l}^{+}\left(  l+\frac{1}{2},{\frac{i\sqrt{2}\kappa}{2{\mu}^{3/4}}%
},{\frac{\xi}{2{\mu}^{1/2}}}-{\frac{{\kappa}^{2}}{8{\mu}^{3/2}}},{\frac
{i\sqrt{2}E}{2\mu^{1/4}},}z\right) \nonumber\\
&  =\exp\left(  \frac{i\sqrt{2}\kappa}{2{\mu}^{3/4}}z+z^{2}\right)  B_{l}%
^{+}\left(  l+\frac{1}{2},{\frac{\sqrt{2}\kappa}{2{\mu}^{3/4}}},\frac{{\kappa
}^{2}}{8{\mu}^{3/2}}-{\frac{\xi}{2{\mu}^{1/2}}},-{\frac{\sqrt{2}E}{2\mu^{1/4}%
},-i}z\right) \nonumber\\
&  =\exp\left(  \frac{i\sqrt{2}\kappa}{2{\mu}^{3/4}}z+z^{2}\right)  \left(
-iz\right)  ^{\frac{1}{2}\left(  \frac{{\kappa}^{2}}{8{\mu}^{3/2}}-{\frac{\xi
}{2{\mu}^{1/2}}}-l-\frac{5}{2}\right)  }\sum_{n\geq0}\frac{a_{n}}{\left(
-iz\right)  ^{n}}, \label{f1264}%
\end{align}%
\begin{align}
&  H_{l}^{+}\left(  l+\frac{1}{2},{\frac{i\sqrt{2}\kappa}{2{\mu}^{3/4}}%
},{\frac{\xi}{2{\mu}^{1/2}}}-{\frac{{\kappa}^{2}}{8{\mu}^{3/2}}},{\frac
{i\sqrt{2}E}{2\mu^{1/4}},}z\right) \nonumber\\
&  =\exp\left(  \frac{i\sqrt{2}\kappa}{2{\mu}^{3/4}}z+z^{2}\right)  H_{l}%
^{+}\left(  l+\frac{1}{2},{\frac{\sqrt{2}\kappa}{2{\mu}^{3/4}}},\frac{{\kappa
}^{2}}{8{\mu}^{3/2}}-{\frac{\xi}{2{\mu}^{1/2}}},-{\frac{\sqrt{2}E}{2\mu^{1/4}%
},-i}z\right) \nonumber\\
&  =\left(  -iz\right)  ^{-\frac{1}{2}\left(  \frac{{\kappa}^{2}}{8{\mu}%
^{3/2}}-\frac{\xi}{2{\mu}^{1/2}}+l+\frac{5}{2}\right)  }\sum_{n\geq0}%
\frac{e_{n}}{\left(  -iz\right)  ^{n}} \label{f2264}%
\end{align}
with the expansion coefficients given by the recurrence relation
\[
a_{0}=1,\text{ \ }a_{1}=-{\frac{\sqrt{2}E}{8\mu^{1/4}}}+\frac{\sqrt{2}\kappa
}{8{\mu}^{3/4}}\left(  \frac{{\kappa}^{2}}{8{\mu}^{3/2}}-{\frac{\xi}{2{\mu
}^{1/2}}}-1\right)  ,
\]%
\begin{align}
&  2\left(  n+2\right)  a_{n+2}+\left[  \frac{\sqrt{2}\kappa}{2{\mu}^{3/4}%
}\left(  \frac{3}{2}-\frac{{\kappa}^{2}}{16{\mu}^{3/2}}+{\frac{\xi}{4{\mu
}^{1/2}}}+n\right)  +{\frac{\sqrt{2}E}{4\mu^{1/4}}}\right]  a_{n+1}\nonumber\\
&  +\left[  \left(  \frac{{\kappa}^{2}}{16{\mu}^{3/2}}-{\frac{\xi}{4{\mu
}^{1/2}}}-1\right)  ^{2}-\frac{1}{4}\left(  l+\frac{1}{2}\right)
^{2}+n\left(  n+2-\frac{{\kappa}^{2}}{8{\mu}^{3/2}}+{\frac{\xi}{2{\mu}^{1/2}}%
}\right)  \right]  a_{n}=0
\end{align}
and%
\[
e_{0}=1,\text{ \ }e_{1}={\frac{\sqrt{2}E}{8\mu^{1/4}}}-\frac{\sqrt{2}\kappa
}{8{\mu}^{3/4}}\left(  \frac{{\kappa}^{2}}{8{\mu}^{3/2}}-{\frac{\xi}{2{\mu
}^{1/2}}}+1\right)  ,
\]%
\begin{align}
&  2\left(  n+2\right)  e_{n+2}+\left[  \frac{\sqrt{2}\kappa}{2{\mu}^{3/4}%
}\left(  \frac{3}{2}+\frac{{\kappa}^{2}}{16{\mu}^{3/2}}-{\frac{\xi}{4{\mu
}^{1/2}}}+n\right)  -{\frac{\sqrt{2}E}{4\mu^{1/4}}}\right]  e_{n+1}\nonumber\\
&  -\left[  \left(  \frac{{\kappa}^{2}}{16{\mu}^{3/2}}-{\frac{\xi}{4{\mu
}^{1/2}}}+1\right)  ^{2}-\frac{1}{4}\left(  l+\frac{1}{2}\right)
^{2}+n\left(  n+2+\frac{{\kappa}^{2}}{8{\mu}^{3/2}}-{\frac{\xi}{2{\mu}^{1/2}}%
}\right)  \right]  e_{n}=0.
\end{align}

\subsection{Eigenfunctions and eigenvalues}

To construct the solution, we first express the regular solution
(\ref{regular264}) as a linear combination of the two irregular solutions
(\ref{f1264}) and (\ref{f2264}).

The regular solution (\ref{regular264}), with the relation
\cite{ronveaux1995heun,li2016exact}%

\begin{align}
&  N\left(  l+\frac{1}{2},{\frac{i\sqrt{2}\kappa}{2{\mu}^{3/4}}},{\frac{\xi
}{2{\mu}^{1/2}}}-{\frac{{\kappa}^{2}}{8{\mu}^{3/2}}},{\frac{i\sqrt{2}E}%
{2\mu^{1/4}}},z\right)  \nonumber\\
&  =K_{1}\left(  l+\frac{1}{2},{\frac{i\sqrt{2}\kappa}{2{\mu}^{3/4}}}%
,{\frac{\xi}{2{\mu}^{1/2}}}-{\frac{{\kappa}^{2}}{8{\mu}^{3/2}}},{\frac
{i\sqrt{2}E}{2\mu^{1/4}}}\right)  B_{l}^{+}\left(  l+\frac{1}{2},{\frac
{i\sqrt{2}\kappa}{2{\mu}^{3/4}}},{\frac{\xi}{2{\mu}^{1/2}}}-{\frac{{\kappa
}^{2}}{8{\mu}^{3/2}}},{\frac{i\sqrt{2}E}{2\mu^{1/4}}},z\right)  \nonumber\\
&  +K_{2}\left(  l+\frac{1}{2},{\frac{i\sqrt{2}\kappa}{2{\mu}^{3/4}}}%
,{\frac{\xi}{2{\mu}^{1/2}}}-{\frac{{\kappa}^{2}}{8{\mu}^{3/2}}},{\frac
{i\sqrt{2}E}{2\mu^{1/4}}}\right)  H_{l}^{+}\left(  l+\frac{1}{2},{\frac
{i\sqrt{2}\kappa}{2{\mu}^{3/4}}},{\frac{\xi}{2{\mu}^{1/2}}}-{\frac{{\kappa
}^{2}}{8{\mu}^{3/2}}},{\frac{i\sqrt{2}E}{2\mu^{1/4}}},z\right)
\end{align}
and the expansions (\ref{f1264}) and (\ref{f2264}), becomes%
\begin{align}
&  u_{l}\left(  r\right)  \nonumber\\
&  =A_{l}K_{1}\left(  l+\frac{1}{2},{\frac{i\sqrt{2}\kappa}{2{\mu}^{3/4}}%
},{\frac{\xi}{2{\mu}^{1/2}}}-{\frac{{\kappa}^{2}}{8{\mu}^{3/2}}},{\frac
{i\sqrt{2}E}{2\mu^{1/4}}}\right)  \nonumber\\
&  \times\exp\left(  -\frac{\mu^{1/2}}{4}{r}^{4}-\frac{\kappa}{4{\mu}^{1/2}%
}{r}^{2}\right)  r^{\frac{{\kappa}^{2}}{8{\mu}^{3/2}}-{\frac{\xi}{2{\mu}%
^{1/2}}}-\frac{3}{2}}\sum_{n\geq0}\frac{a_{n}}{\left(  \frac{\sqrt{2}}{2}%
\mu^{1/4}{r}^{2}\right)  ^{n}}\nonumber\\
&  +A_{l}K_{2}\left(  l+\frac{1}{2},{\frac{i\sqrt{2}\kappa}{2{\mu}^{3/4}}%
},{\frac{\xi}{2{\mu}^{1/2}}}-{\frac{{\kappa}^{2}}{8{\mu}^{3/2}}},{\frac
{i\sqrt{2}E}{2\mu^{1/4}}}\right)  \nonumber\\
&  \times\exp\left(  \frac{\mu^{1/2}}{4}{r}^{4}+\frac{\kappa}{4{\mu}^{1/2}}%
{r}^{2}\right)  r^{-\frac{{\kappa}^{2}}{8{\mu}^{3/2}}+\frac{\xi}{2{\mu}^{1/2}%
}-\frac{3}{2}}\sum_{n\geq0}\frac{e_{n}}{\left(  \frac{\sqrt{2}}{2}\mu^{1/4}%
{r}^{2}\right)  ^{n}},\label{uexpzf264}%
\end{align}
where $K_{1}\left(  l+\frac{1}{2},{\frac{i\sqrt{2}\kappa}{2{\mu}^{3/4}}%
},{\frac{\xi}{2{\mu}^{1/2}}}-{\frac{{\kappa}^{2}}{8{\mu}^{3/2}}},{\frac
{i\sqrt{2}E}{2\mu^{1/4}}}\right)  $ and $K_{2}\left(  l+\frac{1}{2}%
,{\frac{i\sqrt{2}\kappa}{2{\mu}^{3/4}}},{\frac{\xi}{2{\mu}^{1/2}}}%
-{\frac{{\kappa}^{2}}{8{\mu}^{3/2}}},{\frac{i\sqrt{2}E}{2\mu^{1/4}}}\right)  $
are combination coefficients\ and $z=i\frac{\sqrt{2}}{2}\mu^{1/4}{r}^{2}$.

The boundary condition of bound states, $\left.  u\left(  r\right)
\right\vert _{r\rightarrow\infty}\rightarrow0$, requires that the coefficient
of the second term must vanish since this term diverges when $r\rightarrow
\infty$, i.e.,
\begin{equation}
K_{2}\left(  l+\frac{1}{2},{\frac{i\sqrt{2}\kappa}{2{\mu}^{3/4}}},{\frac{\xi
}{2{\mu}^{1/2}}}-{\frac{{\kappa}^{2}}{8{\mu}^{3/2}}},{\frac{i\sqrt{2}E}%
{2\mu^{1/4}}}\right)  =0, \label{eigenvalue264}%
\end{equation}
where%
\begin{align}
&  K_{2}\left(  l+\frac{1}{2},{\frac{i\sqrt{2}\kappa}{2{\mu}^{3/4}}}%
,{\frac{\xi}{2{\mu}^{1/2}}}-{\frac{{\kappa}^{2}}{8{\mu}^{3/2}}},{\frac
{i\sqrt{2}E}{2\mu^{1/4}}}\right) \nonumber\\
&  =\frac{\Gamma\left(  l+\frac{3}{2}\right)  }{\Gamma\left(  \frac{l}%
{2}+\frac{1}{4}-{\frac{\xi}{4{\mu}^{1/2}}}+{\frac{{\kappa}^{2}}{16{\mu}^{3/2}%
}}\right)  \Gamma\left(  \frac{l}{2}+\frac{5}{4}+{\frac{\xi}{4{\mu}^{1/2}}%
}-{\frac{{\kappa}^{2}}{16{\mu}^{3/2}}}\right)  }\nonumber\\
&  \times J_{\frac{l}{2}+\frac{5}{4}+{\frac{\xi}{4{\mu}^{1/2}}}-{\frac
{{\kappa}^{2}}{16{\mu}^{3/2}}}}\left(  \frac{l}{2}+\frac{1}{4}+{\frac{\xi
}{4{\mu}^{1/2}}}-{\frac{{\kappa}^{2}}{16{\mu}^{3/2}}},\frac{i\sqrt{2}\kappa
}{2{\mu}^{3/4}},\frac{3}{2}\left(  l+\frac{1}{2}\right)  \right. \nonumber\\
&  -\left.  {\frac{\xi}{4{\mu}^{1/2}}}+{\frac{{\kappa}^{2}}{16{\mu}^{3/2}}%
},\frac{i\sqrt{2}E}{2\mu^{1/4}}+\frac{i\sqrt{2}\kappa}{4{\mu}^{3/4}}\left(
{\frac{\xi}{2{\mu}^{1/2}}}-{\frac{{\kappa}^{2}}{8{\mu}^{3/2}}}-l-\frac{1}%
{2}\right)  \right)
\end{align}
with%
\begin{align}
&  J_{\frac{l}{2}+\frac{5}{4}+{\frac{\xi}{4{\mu}^{1/2}}}-{\frac{{\kappa}^{2}%
}{16{\mu}^{3/2}}}}\left(  \frac{l}{2}+\frac{1}{4}+{\frac{\xi}{4{\mu}^{1/2}}%
}-{\frac{{\kappa}^{2}}{16{\mu}^{3/2}}},\frac{i\sqrt{2}\kappa}{2{\mu}^{3/4}%
},\frac{3}{2}\left(  l+\frac{1}{2}\right)  -{\frac{\xi}{4{\mu}^{1/2}}}\right.
\nonumber\\
&  +\left.  {\frac{{\kappa}^{2}}{16{\mu}^{3/2}}},\frac{i\sqrt{2}E}{2\mu^{1/4}%
}+\frac{i\sqrt{2}\kappa}{4{\mu}^{3/4}}\left(  {\frac{\xi}{2{\mu}^{1/2}}%
}-{\frac{{\kappa}^{2}}{8{\mu}^{3/2}}}-l-\frac{1}{2}\right)  \right)
\nonumber\\
&  =\int_{0}^{\infty}dxx^{\frac{l}{2}+\frac{1}{4}+{\frac{\xi}{4{\mu}^{1/2}}%
}-{\frac{{\kappa}^{2}}{16{\mu}^{3/2}}}}e^{-x^{2}}N\left(  \frac{l}{2}+\frac
{1}{4}+{\frac{\xi}{4{\mu}^{1/2}}}-{\frac{{\kappa}^{2}}{16{\mu}^{3/2}}}%
,\frac{i\sqrt{2}\kappa}{2{\mu}^{3/4}},\frac{3}{2}\left(  l+\frac{1}{2}\right)
-{\frac{\xi}{4{\mu}^{1/2}}}\right. \nonumber\\
&  +\left.  {\frac{{\kappa}^{2}}{16{\mu}^{3/2}}},\frac{i\sqrt{2}E}{2\mu^{1/4}%
}+\frac{i\sqrt{2}\kappa}{4{\mu}^{3/4}}\left(  {\frac{\xi}{2{\mu}^{1/2}}%
}-{\frac{{\kappa}^{2}}{8{\mu}^{3/2}}}-l-\frac{1}{2}\right)  ,x\right)  .
\end{align}

Eq. (\ref{eigenvalue264}) is an implicit expression of the eigenvalue.

The eigenfunction, by Eqs. (\ref{uexpzf264}) and (\ref{eigenvalue264}), reads%
\begin{align}
u_{l}\left(  r\right)   &  =A_{l}K_{1}\left(  l+\frac{1}{2},{\frac{i\sqrt
{2}\kappa}{2{\mu}^{3/4}}},{\frac{\xi}{2{\mu}^{1/2}}}-{\frac{{\kappa}^{2}%
}{8{\mu}^{3/2}}},{\frac{i\sqrt{2}E}{2\mu^{1/4}}}\right)  \nonumber\\
&  \times\exp\left(  -\frac{\mu^{1/2}}{4}{r}^{4}-\frac{\kappa}{4{\mu}^{1/2}%
}{r}^{2}\right)  r^{\frac{{\kappa}^{2}}{8{\mu}^{3/2}}-{\frac{\xi}{2{\mu}%
^{1/2}}}-\frac{3}{2}}\sum_{n\geq0}\frac{a_{n}}{\left(  \frac{\sqrt{2}}{2}%
\mu^{1/4}{r}^{2}\right)  ^{n}}.
\end{align}

\section{The exact solution of $U\left(  r\right)  =\frac{\xi}{r^{2/3}}+\mu
r^{2/3}+\frac{\kappa}{r^{4/3}}$\label{Vm2323m43}}

In this appendix, we provide an exact solution of the eigenproblem of the
potential%
\begin{equation}
U\left(  r\right)  =\frac{\xi}{r^{2/3}}+\mu r^{2/3}+\frac{\kappa}{r^{4/3}}%
\end{equation}
by solving the radial equation directly. This potential has only bound states.

The radial equation of the potential $U\left(  r\right)  =\frac{\xi}{r^{2/3}%
}+\mu r^{2/3}+\frac{\kappa}{r^{4/3}}$, reads%
\begin{equation}
\frac{d^{2}}{dr^{2}}u_{l}\left(  r\right)  +\left[  E-\frac{l\left(
l+1\right)  }{r^{2}}-\frac{\xi}{r^{2/3}}-\mu r^{2/3}-\frac{\kappa}{r^{4/3}%
}\right]  u_{l}\left(  r\right)  =0. \label{radialeq23o23o43}%
\end{equation}
Using the variable substitution%
\begin{equation}
z=i\frac{\sqrt{6}}{2}\mu^{1/4}{r}^{2/3} \label{zr23o23o43}%
\end{equation}
and introducing $f_{l}\left(  z\right)  $ by%
\begin{equation}
u_{l}\left(  z\right)  =A_{l}\exp\left(  -\frac{z^{2}}{2}-\frac{\beta}%
{2}z\right)  z^{3\left(  l+1\right)  /2}f_{l}\left(  z\right)
\end{equation}
with $A_{l}$ a constant, we convert the radial equation
(\ref{radialeq23o23o43}) into an equation of $f_{l}\left(  z\right)  $:%
\begin{align}
&  zf_{l}^{\prime\prime}\left(  z\right)  +\left(  3l+\frac{5}{2}%
+{\frac{iE\sqrt{6}}{2{\mu}^{3/4}}}z-2z^{2}\right)  f_{l}^{\prime}\left(
z\right) \nonumber\\
&  +\left\{  \left(  {\frac{3\xi}{2{\mu}^{1/2}}}-{\frac{3E^{2}}{8{\mu}^{3/2}}%
}-3l-\frac{7}{2}\right)  z-\left[  -{\frac{i3\sqrt{6}\kappa}{4\mu^{1/4}}%
}-{\frac{iE\sqrt{6}}{4{\mu}^{3/4}}}\left(  3l+\frac{5}{2}\right)  \right]
\right\}  f_{l}\left(  z\right)  =0. \label{eqf23o23o43}%
\end{align}
This is a Biconfluent Heun equation \cite{ronveaux1995heun}.

The choice of the boundary condition has been discussed in Ref.
\cite{li2016exact}.

\subsection{The regular solution}

The regular solution is a solution satisfying the boundary condition at $r=0$
\cite{li2016exact}. The regular solution at $r=0$\ should satisfy the boundary
condition $\lim_{r\rightarrow0}u_{l}\left(  r\right)  /r^{l+1}=1$. In this
section, we provide the regular solution of Eq. (\pageref{eqf23o23o43}).

The Biconfluent Heun equation (\ref{eqf23o23o43}) has two linearly independent
solutions \cite{ronveaux1995heun}%
\begin{align}
y_{l}^{\left(  1\right)  }\left(  z\right)   &  =N\left(  3l+\frac{3}%
{2},-{\frac{iE\sqrt{6}}{2{\mu}^{3/4}}},{\frac{3\xi}{2{\mu}^{1/2}}}%
-{\frac{3E^{2}}{8{\mu}^{3/2}}},-{\frac{i3\sqrt{6}\kappa}{2\mu^{1/4}}%
},z\right)  ,\\
y_{l}^{\left(  2\right)  }\left(  z\right)   &  =cN\left(  3l+\frac{3}%
{2},-{\frac{iE\sqrt{6}}{2{\mu}^{3/4}}},{\frac{3\xi}{2{\mu}^{1/2}}}%
-{\frac{3E^{2}}{8{\mu}^{3/2}}},-{\frac{i3\sqrt{6}\kappa}{2\mu^{1/4}}%
},z\right)  \ln z+\sum_{n\geq0}d_{n}z^{n-3l-3/2},
\end{align}
where%
\begin{equation}
c=\frac{1}{3\left(  l+1/2\right)  }\left\{  d_{3l+1/2}\left[  -{\frac
{i3\sqrt{6}\kappa}{4\mu^{1/4}}}-{\frac{iE\sqrt{6}}{4{\mu}^{3/4}}}\left(
3l+\frac{1}{2}\right)  \right]  -d_{3l-1/2}\left(  {\frac{3\xi}{2{\mu}^{1/2}}%
}-{\frac{3E^{2}}{8{\mu}^{3/2}}}+\frac{1}{2}-3l\right)  \right\}
\end{equation}
is a constant with the coefficient $d_{\nu}$ given by the following recurrence
relation,%
\begin{align}
&  d_{-1}=0,\text{ \ }d_{0}=1,\nonumber\\
&  \left(  v+2\right)  \left(  v+\frac{1}{2}-3l\right)  d_{v+2}\nonumber\\
&  +\left[  {\frac{i3\sqrt{6}\kappa}{4\mu^{1/4}}}+{\frac{iE\sqrt{6}}{4{\mu
}^{3/4}}}\left(  2v+\frac{3}{2}-3l\right)  \right]  d_{v+1}+\left(
{\frac{3\xi}{2{\mu}^{1/2}}}-{\frac{3E^{2}}{8{\mu}^{3/2}}}-2v-\frac{1}%
{2}+3l\right)  d_{v}=0
\end{align}
and $N(\alpha,\beta,\gamma,\delta,z)$ is the biconfluent Heun function
\cite{ronveaux1995heun,slavyanov2000special,li2016exact}.

The biconfluent Heun function $N\left(  3l+\frac{3}{2},-{\frac{iE\sqrt{6}%
}{2{\mu}^{3/4}}},{\frac{3\xi}{2{\mu}^{1/2}}}-{\frac{3E^{2}}{8{\mu}^{3/2}}%
},-{\frac{i3\sqrt{6}\kappa}{2\mu^{1/4}}},z\right)  $ has an expansion at $z=0$
\cite{ronveaux1995heun}:%
\begin{equation}
N\left(  3l+\frac{3}{2},-{\frac{iE\sqrt{6}}{2{\mu}^{3/4}}},{\frac{3\xi}{2{\mu
}^{1/2}}}-{\frac{3E^{2}}{8{\mu}^{3/2}}},-{\frac{i3\sqrt{6}\kappa}{2\mu^{1/4}}%
},z\right)  =\sum_{n\geq0}\frac{A_{n}}{\left(  3l+\frac{5}{2}\right)  _{n}%
}\frac{z^{n}}{n!},
\end{equation}
where the expansion coefficients is determined by the recurrence relation,%
\begin{align}
A_{0}  &  =1,\text{ \ }A_{1}=-{\frac{i3\sqrt{6}\kappa}{4\mu^{1/4}}}%
-{\frac{iE\sqrt{6}}{4{\mu}^{3/4}}}\left(  3l+\frac{5}{2}\right)  ,\nonumber\\
A_{n+2}  &  =\left[  -{\frac{iE\sqrt{6}}{4{\mu}^{3/4}}}\left(  2n+3l+\frac
{9}{2}\right)  -{\frac{i3\sqrt{6}\kappa}{4\mu^{1/4}}}\right]  A_{n+1}%
\nonumber\\
&  -\left(  n+1\right)  \left(  n+3l+\frac{5}{2}\right)  \left(  {\frac{3\xi
}{2{\mu}^{1/2}}}-{\frac{3E^{2}}{8{\mu}^{3/2}}}-3l-\frac{7}{2}-2n\right)
A_{n},
\end{align}
and $\left(  a\right)  _{n}=\Gamma\left(  a+n\right)  /\Gamma\left(  a\right)
$ is Pochhammer's symbol.

Only $y_{l}^{\left(  1\right)  }\left(  z\right)  $ satisfies the boundary
condition for the regular solution at $r=0$, so the radial eigenfunction reads%
\begin{align}
u_{l}\left(  z\right)   &  =A_{l}\exp\left(  -\frac{z^{2}}{2}-\frac{\beta}%
{2}z\right)  z^{3\left(  l+1\right)  /2}y_{l}^{\left(  1\right)  }\left(
z\right) \nonumber\\
&  =A_{l}\exp\left(  -\frac{z^{2}}{2}-\frac{\beta}{2}z\right)  z^{3\left(
l+1\right)  /2}N\left(  3l+\frac{3}{2},-{\frac{iE\sqrt{6}}{2{\mu}^{3/4}}%
},{\frac{3\xi}{2{\mu}^{1/2}}}-{\frac{3E^{2}}{8{\mu}^{3/2}}},-{\frac{i3\sqrt
{6}\kappa}{2\mu^{1/4}}},z\right)  .
\end{align}
By Eq. (\ref{zr23o23o43}), we obtain the regular solution,%

\begin{equation}
u_{l}\left(  r\right)  =A_{l}\exp\left(  \frac{3\mu^{1/4}}{4}{r}^{4/3}%
-{\frac{3E}{4{\mu}^{1/2}}}r^{2/3}\right)  r^{l+1}N\left(  3l+\frac{3}%
{2},-{\frac{iE\sqrt{6}}{2{\mu}^{3/4}}},{\frac{3\xi}{2{\mu}^{1/2}}}%
-{\frac{3E^{2}}{8{\mu}^{3/2}}},-{\frac{i3\sqrt{6}\kappa}{2\mu^{1/4}}}%
,i\frac{\sqrt{6}}{2}\mu^{1/4}{r}^{2/3}\right)  .\label{regular23o23o43}%
\end{equation}

\subsection{The irregular solution}

The irregular solution is a solution satisfying the boundary condition at
$r\rightarrow\infty$ \cite{li2016exact}.

The Biconfluent Heun equation (\ref{eqf23o23o43}) has two linearly independent
irregular solutions \cite{ronveaux1995heun}:%
\begin{align}
&  B_{l}^{+}\left(  3l+\frac{3}{2},-{\frac{iE\sqrt{6}}{2{\mu}^{3/4}}}%
,{\frac{3\xi}{2{\mu}^{1/2}}}-{\frac{3E^{2}}{8{\mu}^{3/2}}},-{\frac{i3\sqrt
{6}\kappa}{2\mu^{1/4}},}z\right) \nonumber\\
&  =\exp\left(  -{\frac{iE\sqrt{6}}{2{\mu}^{3/4}}}z+z^{2}\right)  B_{l}%
^{+}\left(  3l+\frac{3}{2},-{\frac{E\sqrt{6}}{2{\mu}^{3/4}}},{\frac{3E^{2}%
}{8{\mu}^{3/2}}-\frac{3\xi}{2{\mu}^{1/2}}},{\frac{3\sqrt{6}\kappa}{2\mu^{1/4}%
},-i}z\right) \nonumber\\
&  =\exp\left(  -{\frac{iE\sqrt{6}}{2{\mu}^{3/4}}}z+z^{2}\right)  \left(
-iz\right)  ^{\frac{1}{2}\left(  {\frac{3E^{2}}{8{\mu}^{3/2}}-\frac{3\xi
}{2{\mu}^{1/2}}}-3l-\frac{7}{2}\right)  }\sum_{n\geq0}\frac{a_{n}}{\left(
-iz\right)  ^{n}}, \label{f123o23o43}%
\end{align}%
\begin{align}
&  H_{l}^{+}\left(  3l+\frac{3}{2},-{\frac{iE\sqrt{6}}{2{\mu}^{3/4}}}%
,{\frac{3\xi}{2{\mu}^{1/2}}}-{\frac{3E^{2}}{8{\mu}^{3/2}}},-{\frac{i3\sqrt
{6}\kappa}{2\mu^{1/4}},}z\right) \nonumber\\
&  =\exp\left(  \frac{i\sqrt{2}\kappa}{2{\mu}^{3/4}}z+z^{2}\right)  H_{l}%
^{+}\left(  3l+\frac{3}{2},-{\frac{E\sqrt{6}}{2{\mu}^{3/4}}},{\frac{3E^{2}%
}{8{\mu}^{3/2}}-\frac{3\xi}{2{\mu}^{1/2}}},{\frac{3\sqrt{6}\kappa}{2\mu^{1/4}%
},-i}z\right) \nonumber\\
&  =\left(  -iz\right)  ^{-\frac{1}{2}\left(  {\frac{3E^{2}}{8{\mu}^{3/2}%
}-\frac{3\xi}{2{\mu}^{1/2}}}+3l+\frac{7}{2}\right)  }\sum_{n\geq0}\frac{e_{n}%
}{\left(  -iz\right)  ^{n}} \label{f223o23o43}%
\end{align}
with the expansion coefficients given by the recurrence relation
\[
a_{0}=1,\text{ \ }a_{1}=\frac{3\sqrt{6}\kappa}{8\mu^{1/4}}-{\frac{E\sqrt{6}%
}{8{\mu}^{3/4}}}\left(  {\frac{3E^{2}}{8{\mu}^{3/2}}-\frac{3\xi}{2{\mu}^{1/2}%
}}-1\right)  ,
\]%
\begin{align}
&  2\left(  n+2\right)  a_{n+2}+\left[  -{\frac{E\sqrt{6}}{2{\mu}^{3/4}}%
}\left(  \frac{3}{2}-{\frac{3E^{2}}{16{\mu}^{3/2}}+\frac{3\xi}{4{\mu}^{1/2}}%
}+n\right)  -\frac{3\sqrt{6}\kappa}{4\mu^{1/4}}\right]  a_{n+1}\nonumber\\
&  +\left[  \left(  {\frac{3E^{2}}{16{\mu}^{3/2}}-\frac{3\xi}{4{\mu}^{1/2}}%
}-1\right)  ^{2}-\frac{1}{4}\left(  3l+\frac{3}{2}\right)  ^{2}+n\left(
n+2-{\frac{3E^{2}}{8{\mu}^{3/2}}+\frac{3\xi}{2{\mu}^{1/2}}}\right)  \right]
a_{n}=0
\end{align}
and%
\[
e_{0}=1,\text{ \ }e_{1}=-\frac{3\sqrt{6}\kappa}{8\mu^{1/4}}+{\frac{E\sqrt{6}%
}{8{\mu}^{3/4}}}\left(  {\frac{3E^{2}}{8{\mu}^{3/2}}-\frac{3\xi}{2{\mu}^{1/2}%
}}+1\right)  ,
\]%
\begin{align}
&  2\left(  n+2\right)  e_{n+2}+\left[  \beta\left(  \frac{3}{2}+{\frac
{3E^{2}}{16{\mu}^{3/2}}-\frac{3\xi}{4{\mu}^{1/2}}}+n\right)  +{\frac{3\sqrt
{6}\kappa}{4\mu^{1/4}}}\right]  e_{n+1}\nonumber\\
&  -\left[  \left(  {\frac{3E^{2}}{16{\mu}^{3/2}}-\frac{3\xi}{4{\mu}^{1/2}}%
}+1\right)  ^{2}-\frac{1}{4}\left(  3l+\frac{3}{2}\right)  ^{2}+n\left(
n+2+{\frac{3E^{2}}{8{\mu}^{3/2}}-\frac{3\xi}{2{\mu}^{1/2}}}\right)  \right]
e_{n}=0.
\end{align}

\subsection{Eigenfunctions and eigenvalues}

To construct the solution, we first express the regular solution
(\ref{regular23o23o43}) as a linear combination of the two irregular solutions
(\ref{f123o23o43}) and (\ref{f223o23o43}).

The regular solution (\ref{regular23o23o43}), with the relation
\cite{ronveaux1995heun,li2016exact}%

\begin{align}
&  N\left(  3l+\frac{3}{2},-{\frac{iE\sqrt{6}}{2{\mu}^{3/4}}},{\frac{3\xi
}{2{\mu}^{1/2}}}-{\frac{3E^{2}}{8{\mu}^{3/2}}},-{\frac{i3\sqrt{6}\kappa}%
{2\mu^{1/4}}},z\right)  \nonumber\\
&  =K_{1}\left(  3l+\frac{3}{2},-{\frac{iE\sqrt{6}}{2{\mu}^{3/4}}},{\frac
{3\xi}{2{\mu}^{1/2}}}-{\frac{3E^{2}}{8{\mu}^{3/2}}},-{\frac{i3\sqrt{6}\kappa
}{2\mu^{1/4}}}\right)  B_{l}^{+}\left(  3l+\frac{3}{2},-{\frac{iE\sqrt{6}%
}{2{\mu}^{3/4}}},{\frac{3\xi}{2{\mu}^{1/2}}}-{\frac{3E^{2}}{8{\mu}^{3/2}}%
},-{\frac{i3\sqrt{6}\kappa}{2\mu^{1/4}}},z\right)  \nonumber\\
&  +K_{2}\left(  3l+\frac{3}{2},-{\frac{iE\sqrt{6}}{2{\mu}^{3/4}}},{\frac
{3\xi}{2{\mu}^{1/2}}}-{\frac{3E^{2}}{8{\mu}^{3/2}}},-{\frac{i3\sqrt{6}\kappa
}{2\mu^{1/4}}}\right)  H_{l}^{+}\left(  3l+\frac{3}{2},-{\frac{iE\sqrt{6}%
}{2{\mu}^{3/4}}},{\frac{3\xi}{2{\mu}^{1/2}}}-{\frac{3E^{2}}{8{\mu}^{3/2}}%
},-{\frac{i3\sqrt{6}\kappa}{2\mu^{1/4}}},z\right)
\end{align}
and the expansions (\ref{f123o23o43}) and (\ref{f223o23o43}), becomes%
\begin{align}
u_{l}\left(  r\right)   &  =A_{l}K_{1}\left(  3l+\frac{3}{2},-{\frac
{iE\sqrt{6}}{2{\mu}^{3/4}}},{\frac{3\xi}{2{\mu}^{1/2}}}-{\frac{3E^{2}}{8{\mu
}^{3/2}}},-{\frac{i3\sqrt{6}\kappa}{2\mu^{1/4}}}\right)  \nonumber\\
&  \times\exp\left(  -\frac{3\mu^{1/2}}{4}{r}^{4/3}+{\frac{3E}{4{\mu}^{1/2}}%
}r^{2/3}\right)  r^{{\frac{E^{2}}{8{\mu}^{3/2}}-\frac{\xi}{2{\mu}^{1/2}}%
}-\frac{1}{6}}\sum_{n\geq0}\frac{a_{n}}{\left(  \frac{\sqrt{6}}{2}\mu^{1/4}%
{r}^{2/3}\right)  ^{n}}\nonumber\\
&  +A_{l}K_{2}\left(  3l+\frac{3}{2},-{\frac{iE\sqrt{6}}{2{\mu}^{3/4}}}%
,{\frac{3\xi}{2{\mu}^{1/2}}}-{\frac{3E^{2}}{8{\mu}^{3/2}}},-{\frac{i3\sqrt
{6}\kappa}{2\mu^{1/4}}}\right)  \nonumber\\
&  \times\exp\left(  \frac{3\mu^{1/2}}{4}{r}^{4/3}-{\frac{3E}{4{\mu}^{1/2}}%
}r^{2/3}\right)  r^{-\frac{E^{2}}{8{\mu}^{3/2}}{+}\frac{\xi}{2{\mu}^{1/2}%
}-\frac{1}{6}}\sum_{n\geq0}\frac{e_{n}}{\left(  \frac{\sqrt{6}}{2}\mu^{1/4}%
{r}^{2/3}\right)  ^{n}},\label{uexpzf23o23o43}%
\end{align}
where $K_{1}\left(  3l+\frac{3}{2},-{\frac{iE\sqrt{6}}{2{\mu}^{3/4}}}%
,{\frac{3\xi}{2{\mu}^{1/2}}}-{\frac{3E^{2}}{8{\mu}^{3/2}}},-{\frac{i3\sqrt
{6}\kappa}{2\mu^{1/4}}}\right)  $ and $K_{2}\left(  3l+\frac{3}{2}%
,-{\frac{iE\sqrt{6}}{2{\mu}^{3/4}}},{\frac{3\xi}{2{\mu}^{1/2}}}-{\frac{3E^{2}%
}{8{\mu}^{3/2}}},-{\frac{i3\sqrt{6}\kappa}{2\mu^{1/4}}}\right)  $ are
combination coefficients\ and $z=i\frac{\sqrt{6}}{2}\mu^{1/4}{r}^{2/3}$.

The boundary condition of bound states, $\left.  u\left(  r\right)
\right\vert _{r\rightarrow\infty}\rightarrow0$, requires that the coefficient
of the second term must vanish since this term diverges when $r\rightarrow
\infty$, i.e.,
\begin{equation}
K_{2}\left(  3l+\frac{3}{2},-{\frac{iE\sqrt{6}}{2{\mu}^{3/4}}},{\frac{3\xi
}{2{\mu}^{1/2}}}-{\frac{3E^{2}}{8{\mu}^{3/2}}},-{\frac{i3\sqrt{6}\kappa}%
{2\mu^{1/4}}}\right)  =0, \label{eigenvalue23o23o43}%
\end{equation}
where%
\begin{align}
&  K_{2}\left(  3l+\frac{3}{2},-{\frac{iE\sqrt{6}}{2{\mu}^{3/4}}},{\frac{3\xi
}{2{\mu}^{1/2}}}-{\frac{3E^{2}}{8{\mu}^{3/2}}},-{\frac{i3\sqrt{6}\kappa}%
{2\mu^{1/4}}}\right) \nonumber\\
&  =\frac{\Gamma\left(  3l+\frac{5}{2}\right)  }{\Gamma\left(  \frac{3}%
{2}l+\frac{3}{4}-{\frac{3\xi}{4{\mu}^{1/2}}}+{\frac{3E^{2}}{16{\mu}^{3/2}}%
}\right)  \Gamma\left(  \frac{3}{2}l+\frac{7}{4}+{\frac{3\xi}{4{\mu}^{1/2}}%
}-{\frac{3E^{2}}{16{\mu}^{3/2}}}\right)  }\nonumber\\
&  \times J_{\frac{3}{2}l+\frac{7}{4}+{\frac{3\xi}{4{\mu}^{1/2}}}%
-{\frac{3E^{2}}{16{\mu}^{3/2}}}}\left(  \frac{3}{2}l+\frac{3}{4}+{\frac{3\xi
}{4{\mu}^{1/2}}}-{\frac{3E^{2}}{16{\mu}^{3/2}}},-{\frac{iE\sqrt{6}}{2{\mu
}^{3/4}}},\frac{3}{2}\left(  3l+\frac{3}{2}\right)  \right. \nonumber\\
&  -\left.  {\frac{3\xi}{4{\mu}^{1/2}}}+{\frac{3E^{2}}{16{\mu}^{3/2}}}%
,-{\frac{i3\sqrt{6}\kappa}{2\mu^{1/4}}}-{\frac{iE\sqrt{6}}{4{\mu}^{3/4}}%
}\left(  {\frac{3\xi}{2{\mu}^{1/2}}}-{\frac{3E^{2}}{8{\mu}^{3/2}}}-3l-\frac
{3}{2}\right)  \right)
\end{align}
with%
\begin{align}
&  J_{\frac{3}{2}l+\frac{7}{4}+{\frac{3\xi}{4{\mu}^{1/2}}}-{\frac{3E^{2}%
}{16{\mu}^{3/2}}}}\left(  \frac{3}{2}l+\frac{3}{4}+{\frac{3\xi}{4{\mu}^{1/2}}%
}-{\frac{3E^{2}}{16{\mu}^{3/2}}},-{\frac{iE\sqrt{6}}{2{\mu}^{3/4}}},\frac
{3}{2}\left(  3l+\frac{3}{2}\right)  \right. \nonumber\\
&  -\left.  {\frac{3\xi}{4{\mu}^{1/2}}}+{\frac{3E^{2}}{16{\mu}^{3/2}}}%
,-{\frac{i3\sqrt{6}\kappa}{2\mu^{1/4}}}-{\frac{iE\sqrt{6}}{4{\mu}^{3/4}}%
}\left(  {\frac{3\xi}{2{\mu}^{1/2}}}-{\frac{3E^{2}}{8{\mu}^{3/2}}}-3l-\frac
{3}{2}\right)  \right) \nonumber\\
&  =\int_{0}^{\infty}dxx^{\frac{3}{2}l+\frac{3}{4}+{\frac{3\xi}{4{\mu}^{1/2}}%
}-{\frac{3E^{2}}{16{\mu}^{3/2}}}}e^{-x^{2}}N\left(  \frac{3}{2}l+\frac{3}%
{4}+{\frac{3\xi}{4{\mu}^{1/2}}}-{\frac{3E^{2}}{16{\mu}^{3/2}}},-{\frac
{iE\sqrt{6}}{2{\mu}^{3/4}}},\frac{3}{2}\left(  3l+\frac{3}{2}\right)  \right.
\nonumber\\
&  -\left.  {\frac{3\xi}{4{\mu}^{1/2}}}+{\frac{3E^{2}}{16{\mu}^{3/2}}}%
,-{\frac{i3\sqrt{6}\kappa}{2\mu^{1/4}}}-{\frac{iE\sqrt{6}}{4{\mu}^{3/4}}%
}\left(  {\frac{3\xi}{2{\mu}^{1/2}}}-{\frac{3E^{2}}{8{\mu}^{3/2}}}-3l-\frac
{3}{2}\right)  ,x\right)  .
\end{align}

Eq. (\ref{eigenvalue23o23o43}) is an implicit expression of the eigenvalue.

The eigenfunction, by Eqs. (\ref{uexpzf23o23o43}) and
(\ref{eigenvalue23o23o43}), reads%
\begin{align}
u_{l}\left(  r\right)   &  =A_{l}K_{1}\left(  3l+\frac{3}{2},-{\frac
{iE\sqrt{6}}{2{\mu}^{3/4}}},{\frac{3\xi}{2{\mu}^{1/2}}}-{\frac{3E^{2}}{8{\mu
}^{3/2}}},-{\frac{i3\sqrt{6}\kappa}{2\mu^{1/4}}}\right)  \nonumber\\
&  \times\exp\left(  -\frac{3\mu^{1/2}}{4}{r}^{4/3}+{\frac{3E}{4{\mu}^{1/2}}%
}r^{2/3}\right)  r^{{\frac{E^{2}}{8{\mu}^{3/2}}-\frac{\xi}{2{\mu}^{1/2}}%
}-\frac{1}{6}}\sum_{n\geq0}\frac{a_{n}}{\left(  \frac{\sqrt{6}}{2}\mu^{1/4}%
{r}^{2/3}\right)  ^{n}}.
\end{align}

\section{The solution of the harmonic-oscillator potential $U\left(  r\right)
=\xi r^{2}$ in terms of the Heun biconfluent function \label{Harmonic}}

In this appendix, for consistency with the solutions of other potentials, we
solve the harmonic-oscillator potential in terms of the Heun biconfluent function.

The radial equation of the harmonic-oscillator potential $U\left(  r\right)
=\xi r^{2}$,
\begin{equation}
\frac{d^{2}}{dr^{2}}u_{l}\left(  r\right)  +\left[  E-\frac{l\left(
l+1\right)  }{r^{2}}-\xi r^{2}\right]  u_{l}\left(  r\right)  =0.
\end{equation}
Introducing $f_{l}\left(  z\right)  $ by%
\begin{equation}
u_{l}\left(  z\right)  =A_{l}e^{-\frac{z^{2}}{2}}z^{l+1}f_{l}\left(  z\right)
\label{hocuy}%
\end{equation}
with%
\begin{equation}
z=\xi^{1/4}r, \label{zrham}%
\end{equation}
where $A_{l}$ is a constant, we then arrive at an equation of $f_{l}\left(
z\right)  $%
\begin{equation}
f_{l}^{\prime\prime}\left(  z\right)  +\frac{-2{z}^{2}+2\left(  l+1\right)
}{z}f_{l}^{\prime}\left(  z\right)  +\left(  \frac{E}{\xi^{1/2}}-2l-3\right)
f_{l}\left(  z\right)  =0. \label{eqfhoc}%
\end{equation}
This is a Biconfluent Heun equation \cite{ronveaux1995heun}.

The choice of the boundary condition has been discussed in Ref.
\cite{li2016exact}.

\subsection{The regular solution}

The regular solution is a solution satisfying the boundary condition at $r=0$
\cite{li2016exact}. The regular solution at $r=0$\ should satisfy the boundary
condition $\lim_{r\rightarrow0}u_{l}\left(  r\right)  /r^{l+1}=1$ for both
bound states and scattering states. In this section, we provide the regular
solution of Eq. (\ref{eqfhoc}).

The Biconfluent Heun equation (\ref{eqfhoc}) has two linearly independent
solutions \cite{ronveaux1995heun}%
\begin{align}
y_{l}^{\left(  1\right)  }\left(  z\right)   &  =N\left(  2l+1,0,\frac{E}%
{\xi^{1/2}},0,z\right)  ,\\
y_{l}^{\left(  2\right)  }\left(  z\right)   &  =cN\left(  2l+1,0,\frac{E}%
{\xi^{1/2}},0,z\right)  \ln z+\sum_{n\geq0}d_{n}z^{n-2l-1},
\end{align}
where%
\begin{equation}
c=\frac{1}{2l+1}\left\{  -d_{2l-1}\left(  \frac{E}{\xi^{1/2}}+1-2l\right)
\right\}
\end{equation}
is a constant with the coefficient $d_{\nu}$ given by the following recurrence
relation,%
\begin{align}
&  d_{-1}=0,\text{ \ }d_{0}=1,\\
&  \left(  \nu+2\right)  \left(  \nu+1-2l\right)  d_{\nu+2}+\left(  \frac
{E}{\xi^{1/2}}-2\left(  \nu+1\right)  +2l+1\right)  d_{\nu}=0.
\end{align}
and $N(\alpha,\beta,\gamma,\delta,z)$ is the biconfluent Heun function
\cite{ronveaux1995heun,slavyanov2000special,li2016exact}.

The biconfluent Heun function $N\left(  2l+1,0,\frac{E}{\xi^{1/2}},0,z\right)
$ has an expansion at $z=0$ \cite{ronveaux1995heun}:%
\begin{equation}
N\left(  2l+1,0,\frac{E}{\xi^{1/2}},0,z\right)  =\sum_{n\geq0}\frac{A_{n}%
}{\left(  2l+2\right)  _{n}}\frac{z^{n}}{n!},
\end{equation}
where the expansion coefficients is determined by the recurrence relation,%
\begin{align}
A_{0}  &  =1,\text{ \ }A_{1}=0,\\
A_{n+2}  &  =-\left(  n+1\right)  \left(  n+2l+2\right)  \left(  \frac{E}%
{\xi^{1/2}}-2l-3-2n\right)  A_{n}%
\end{align}
and $\left(  a\right)  _{n}=\Gamma\left(  a+n\right)  /\Gamma\left(  a\right)
$ is Pochhammer's symbol. In the case of harmonic-oscillator potential, the
recurrence relation among the adjacent three terms reduces to a recurrence
relation between the adjacent two terms.

Only $y_{l}^{\left(  1\right)  }\left(  z\right)  $ satisfies the boundary
condition for the regular solution at $r=0$, so the radial eigenfunction reads%
\begin{align}
u_{l}\left(  z\right)   &  =A_{l}e^{-z^{2}/2}z^{l+1}y_{l}^{\left(  1\right)
}\left(  z\right) \nonumber\\
&  =A_{l}e^{-\frac{z^{2}}{2}}z^{l+1}N\left(  2l+1,0,\frac{E}{\xi^{1/2}%
},0,z\right)  .
\end{align}
By Eq. (\ref{zrham}), we obtain the regular solution,%
\begin{equation}
u_{l}\left(  r\right)  =A_{l}e^{-\sqrt{\xi}r^{2}/2}\xi^{\left(  l+1\right)
/4}r^{l+1}N\left(  2l+1,0,\frac{E}{\xi^{1/2}},0,\xi^{1/4}r\right)  .
\label{regularham}%
\end{equation}

In the case of harmonic-oscillator potential, by the relation between the Heun
function and the hypergeometric function
\cite{ronveaux1995heun,slavyanov2000special},%
\begin{equation}
N\left(  \alpha,0,\gamma,0,z\right)  =\text{ }_{1}F_{1}\left(  \frac{1}%
{2}+\frac{\alpha}{4}-\frac{\gamma}{4},{1+}\frac{\alpha}{2},{{z}^{2}}\right)
{,} \label{Heun1F1ham}%
\end{equation}
the regular solution (\ref{regularham}) reduces to
\begin{equation}
u_{l}\left(  z\right)  =A_{l}e^{-\sqrt{\xi}r^{2}/2}\xi^{\left(  l+1\right)
/4}r^{l+1}\text{ }_{1}F_{1}\left(  \frac{l}{2}+\frac{3}{4}-\frac{E}{4\sqrt
{\xi}},\frac{3}{2}+l,\sqrt{\xi}r^{2}\right)  .
\end{equation}

\subsection{The irregular solution}

The irregular solution is a solution satisfying the boundary condition at
$r\rightarrow\infty$ \cite{li2016exact}. The boundary conditions for bound
states and scattering states at $r\rightarrow\infty$ are different.

The Biconfluent Heun equation (\ref{eqfhoc}) has two linearly independent
irregular solutions \cite{ronveaux1995heun}:%

\begin{align}
B_{l}^{+}\left(  2l+1,0,\frac{E}{\xi^{1/2}},0,z\right)   &  =z^{\frac{1}%
{2}\left[  \frac{E}{\xi^{1/2}}-2-\left(  2l+1\right)  \right]  }\sum_{n\geq
0}\frac{a_{n}}{z^{n}},\label{hoc1}\\
H_{l}^{+}\left(  2l+1,0,\frac{E}{\xi^{1/2}},0,z\right)   &  =z^{-\frac{1}%
{2}\left[  \frac{E}{\xi^{1/2}}+2+\left(  2l+1\right)  \right]  }e^{z^{2}}%
\sum_{n\geq0}\frac{e_{n}}{z^{n}}, \label{hoc2}%
\end{align}
with the expansion coefficients given by the recurrence relation%
\[
a_{0}=1,\text{ \ }a_{1}=0,
\]%
\begin{equation}
2\left(  n+2\right)  a_{n+2}+\left[  \frac{1}{4}\left(  \frac{E^{2}}{\xi
}-\left(  2l+1\right)  ^{2}+4\right)  -\frac{E}{\xi^{1/2}}+n\left(
n+2-\frac{E}{\xi^{1/2}}\right)  \right]  a_{n}=0,
\end{equation}
and%
\[
e_{0}=1,\text{ \ }e_{1}=0,
\]%
\begin{equation}
2\left(  n+2\right)  e_{n+2}-\left[  \frac{1}{4}\left(  \frac{E^{2}}{\xi
}-\left(  2l+1\right)  ^{2}+4\right)  +\frac{E}{\xi^{1/2}}+n\left(
n+2+\frac{E}{\xi^{1/2}}\right)  \right]  e_{n}=0,
\end{equation}
where the recurrence relation among the adjacent three terms reduces to a
recurrence relation between the adjacent two terms.

\subsection{Eigenfunctions and eigenvalues}

To construct the solution, we first express the regular solution
(\ref{regularham}) as a linear combination of the two irregular solutions
(\ref{hoc1}) and (\ref{hoc2}).

The regular solution (\ref{regularham}) with the relation
\cite{ronveaux1995heun,li2016exact}%

\begin{align}
N\left(  2l+1,0,\frac{E}{\xi^{1/2}},0,z\right)   &  =K_{1}\left(
2l+1,0,\frac{E}{\xi^{1/2}},0\right)  B_{l}^{+}\left(  2l+1,0,\frac{E}%
{\xi^{1/2}},0,z\right) \nonumber\\
&  +K_{2}\left(  2l+1,0,\frac{E}{\xi^{1/2}},0\right)  H_{l}^{+}\left(
2l+1,0,\frac{E}{\xi^{1/2}},0,z\right)
\end{align}
and the expansions (\ref{hoc1}) and (\ref{hoc2}) become%
\begin{align}
u_{l}\left(  r\right)   &  =A_{l}e^{-\frac{\sqrt{\xi}}{2}r^{2}}\xi^{\left(
l+1\right)  /4}r^{l+1}N\left(  2l+1,0,\frac{E}{\xi^{1/2}},0,\xi^{1/4}r\right)
\nonumber\\
&  =A_{l}e^{-\frac{\sqrt{\xi}}{2}r^{2}}K_{1}\left(  2l+1,0,\frac{E}{\xi^{1/2}%
},0\right)  \left(  \xi^{1/4}r\right)  ^{\frac{E}{2\xi^{1/2}}-1/2}\sum
_{n\geq0}\frac{a_{n}}{\left(  \xi^{1/4}r\right)  ^{n/2}}\nonumber\\
&  +A_{l}e^{\frac{\sqrt{\xi}}{2}r^{2}}K_{2}\left(  2l+1,0,\frac{E}{\xi^{1/2}%
},0\right)  \left(  \xi^{1/4}r\right)  ^{-\frac{E}{2\xi^{1/2}}-1/2}\sum
_{n\geq0}\frac{e_{n}}{\left(  \xi^{1/4}r\right)  ^{n/2}}, \label{uexpham}%
\end{align}
where $K_{1}\left(  2l+1,0,\frac{E}{\xi^{1/2}},0\right)  $ and $K_{2}\left(
2l+1,0,\frac{E}{\xi^{1/2}},0\right)  $ are combination coefficients\ and
$\xi^{1/4}r$.

Only the first term satisfies the boundary condition of bound states $\left.
u\left(  r\right)  \right\vert _{r\rightarrow\infty}\rightarrow0$, so the
coefficient of the second term must vanish, i.e.,%

\begin{equation}
K_{2}\left(  2l+1,0,\frac{E}{\xi^{1/2}},0\right)  =0, \label{eigenvalueham1}%
\end{equation}
where
\begin{align}
K_{2}\left(  2l+1,0,\frac{E}{\xi^{1/2}},0\right)   &  =\frac{\Gamma\left(
2l+2\right)  }{\Gamma\left(  l+\frac{1}{2}-\frac{E}{2\xi^{1/2}}\right)
\Gamma\left(  l+\frac{3}{2}+\frac{E}{2\xi^{1/2}}\right)  }\nonumber\\
&  \times J_{l+\frac{3}{2}+\frac{E}{2\xi^{1/2}}}\left(  l+\frac{1}{2}+\frac
{E}{2\xi^{1/2}},0,3l+\frac{3}{2}-\frac{E}{2\xi^{1/2}},0\right)  ,
\end{align}
with%
\begin{align}
&  J_{l+\frac{3}{2}+\frac{E}{2\xi^{1/2}}}\left(  l+\frac{1}{2}+\frac{E}%
{2\xi^{1/2}},0,3l+\frac{3}{2}-\frac{E}{2\xi^{1/2}},0\right) \nonumber\\
&  =\int_{0}^{\infty}x^{l+\frac{1}{2}+\frac{E}{2\xi^{1/2}}}e^{-x^{2}}N\left(
l+\frac{1}{2}+\frac{E}{2\xi^{1/2}},0,3l+\frac{3}{2}-\frac{E}{2\xi^{1/2}%
},0,x\right)  dx.
\end{align}
In the case of harmonic-oscillator potential, $K_{2}\left(  2l+1,0,\frac
{E}{\xi^{1/2}},0\right)  $ reduces to\cite{batola1982quelques}
\begin{equation}
K_{2}\left(  2l+1,0,\frac{E}{\xi^{1/2}},0\right)  =\frac{\Gamma\left(
l+\frac{3}{2}\right)  }{\Gamma\left(  \frac{3}{4}+\frac{l}{2}-\frac{E}%
{4\xi^{1/2}}\right)  }.
\end{equation}

\bigskip The zeros of $K_{2}\left(  2l+1,0,\frac{E}{\xi^{1/2}},0\right)  $
correspond to the bound state of the Coulomb potential, so the eigenvalue is
the singularities of $\Gamma\left(  \frac{3}{4}+\frac{l}{2}-\frac{E}%
{4\xi^{1/2}}\right)  $. Thus
\begin{equation}
\frac{3}{4}+\frac{l}{2}-\frac{E}{4\xi^{1/2}}=-n,\text{ \ }n=0,1,2\ldots,
\end{equation}
then
\begin{equation}
E=2\sqrt{\xi}\left(  2n+l+\frac{3}{2}\right)  ,\text{ \ }n=0,1,2\ldots.
\label{eigenvalueham}%
\end{equation}

The eigenfunction, by Eqs. (\ref{uexpham}) and (\ref{eigenvalueham1}), reads%
\begin{equation}
u_{l}\left(  r\right)  =A_{l}e^{-\frac{\sqrt{\xi}}{2}r^{2}}K_{1}\left(
2l+1,0,\frac{E}{\xi^{1/2}},0\right)  \left(  \xi^{1/4}r\right)  ^{\frac
{E}{2\xi^{1/2}}-1/2}\sum_{n\geq0}\frac{a_{n}}{\left(  \xi^{1/4}r\right)
^{n/2}}.
\end{equation}

\section{The solution of the Coulomb potential $U\left(  r\right)  =\frac{\xi
}{r}$ in terms of the Heun biconfluent function \label{Coulomb}}

In this appendix, for consistency with the solutions of other potentials, we
solve he Coulomb potential in terms of the Heun biconfluent function.

The radial equation of the Coulomb potential%

\begin{equation}
U\left(  r\right)  =\frac{\xi}{r}%
\end{equation}
reads%
\begin{equation}
\frac{d^{2}}{dr^{2}}u_{l}\left(  r\right)  +\left[  E-\frac{l\left(
l+1\right)  }{r^{2}}-\frac{\xi}{r}\right]  u_{l}\left(  r\right)  =0.
\end{equation}
Introducing $f_{l}\left(  z\right)  $ by
\begin{equation}
u_{l}\left(  z\right)  =A_{l}e^{-\frac{z^{2}}{2}}z^{2l+2}f_{l}\left(
z\right)  \label{coubuy}%
\end{equation}
with%
\begin{equation}
z=\left(  2r\right)  ^{1/2}\left(  -E\right)  ^{1/4}, \label{zrcoub}%
\end{equation}
where $A_{l}$ is a constant, we then arrive at an equation of $f_{l}\left(
z\right)  $%
\begin{equation}
f_{l}^{\prime\prime}\left(  z\right)  +\frac{-2{z}^{2}+4l+3}{z}f_{l}^{\prime
}\left(  z\right)  +\left(  -\frac{2\xi}{\sqrt{-E}}-4l-4\right)  f_{l}\left(
z\right)  =0 \label{eqfcoulomb}%
\end{equation}
This is a Biconfluent Heun equation \cite{ronveaux1995heun}.

The choice of the boundary condition has been discussed in Ref.
\cite{li2016exact}.

\subsection{The regular solution}

The regular solution is a solution satisfying the boundary condition at $r=0$
\cite{li2016exact}. The regular solution at $r=0$\ should satisfy the boundary
condition $\lim_{r\rightarrow0}u_{l}\left(  r\right)  /r^{l+1}=1$ for both
bound states and scattering states. In this section, we provide the regular
solution of Eq. (\ref{eqfcoulomb}).

The Biconfluent Heun equation (\ref{eqfcoulomb}) has two linearly independent
solutions \cite{ronveaux1995heun}%

\begin{align}
y_{l}^{\left(  1\right)  }\left(  z\right)   &  =N\left(  4l+2,0,-\frac{2\xi
}{\sqrt{-E}},0,z\right)  ,\\
y_{l}^{\left(  2\right)  }\left(  z\right)   &  =cN\left(  4l+2,0,-\frac{2\xi
}{\sqrt{-E}},0,z\right)  \ln z+\sum_{n\geq0}d_{n}z^{n-4l-2},
\end{align}
where%
\begin{equation}
c=\frac{1}{4l+2}\left\{  -d_{4l}\left(  -\frac{2\xi}{\sqrt{-E}}-4l\right)
\right\}
\end{equation}
is a constant with the coefficient $d_{\nu}$ given by the following recurrence
relation,%
\begin{align}
&  d_{-1}=0,\text{ \ }d_{0}=1,\\
&  \left(  \nu+2\right)  \left(  \nu-4l\right)  d_{\nu+2}+\left(  -\frac{2\xi
}{\sqrt{-E}}-2\left(  \nu+1\right)  +4l+2\right)  d_{\nu}=0,
\end{align}
and $N(\alpha,\beta,\gamma,\delta,z)$ is the biconfluent Heun function
\cite{ronveaux1995heun,slavyanov2000special,li2016exact}.

The biconfluent Heun function $N\left(  4l+2,0,-\frac{2\xi}{\sqrt{-E}%
},0,z\right)  $ has an expansion at $z=0$ \cite{ronveaux1995heun}:%

\begin{equation}
N\left(  4l+2,0,-\frac{2\xi}{\sqrt{-E}},0,z\right)  =\sum_{n\geq0}\frac{A_{n}%
}{\left(  4l+3\right)  _{n}}\frac{z^{n}}{n!},
\end{equation}
where the expansion coefficients is determined by the recurrence relation,%
\begin{align}
A_{0}  &  =1,\text{ \ }A_{1}=0,\\
A_{n+2}  &  =-\left(  n+1\right)  \left(  n+4l+3\right)  \left(  -\frac{2\xi
}{\sqrt{-E}}-4l-4-2n\right)  A_{n},
\end{align}
and $\left(  a\right)  _{n}=\Gamma\left(  a+n\right)  /\Gamma\left(  a\right)
$ is Pochhammer's symbol. Similar to the case of harmonic-oscillator
potential, the recurrence relation among the adjacent three terms reduces to a
recurrence relation between the adjacent two terms.

Only $y_{l}^{\left(  1\right)  }\left(  z\right)  $ satisfies the boundary
condition for the regular solution at $r=0$, so the radial eigenfunction reads%
\begin{align}
u_{l}\left(  z\right)   &  =A_{l}e^{-\frac{z^{2}}{2}}z^{2l+2}y_{l}^{\left(
1\right)  }\left(  z\right) \nonumber\\
&  =A_{l}e^{-\frac{z^{2}}{2}}z^{2l+2}N\left(  4l+2,0,-\frac{2\xi}{\sqrt{-E}%
},0,z\right)  ,
\end{align}
By Eq. (\ref{zrcoub}), we obtain the regular solution,%
\begin{equation}
u_{l}\left(  r\right)  =A_{l}e^{-\sqrt{-E}r}\left(  2\sqrt{-E}r\right)
^{l+1}N\left(  4l+2,0,-\frac{2\xi}{\sqrt{-E}},0,\left(  2r\right)
^{1/2}\left(  -E\right)  ^{1/4}\right)  . \label{regularcoub}%
\end{equation}

Similar to the case of harmonic-oscillator potential£¬by the relation
between the Heun function and the hypergeometric function (\ref{Heun1F1ham}),
the regular solution (\ref{regularcoub}) reduce to
\begin{equation}
u_{l}\left(  z\right)  =A_{l}e^{-\sqrt{-E}r}\left(  2\sqrt{-E}\right)
^{l+1}r^{l+1}\text{ }_{1}F_{1}\left(  l+1+\frac{\xi}{2\sqrt{-E}},2\left(
l+1\right)  ,2\sqrt{-E}r\right)  .
\end{equation}

\subsection{The irregular solution}

The irregular solution is a solution satisfying the boundary condition at
$r\rightarrow\infty$ \cite{li2016exact}. The boundary conditions for bound
states and scattering states at $r\rightarrow\infty$ are different.

The Biconfluent Heun equation (\ref{eqfcoulomb}) has two linearly independent
irregular solutions \cite{ronveaux1995heun}:%
\begin{align}
B_{l}^{+}\left(  4l+2,0,-\frac{2\xi}{\sqrt{-E}},0,z\right)   &  =z^{\frac
{1}{2}\left(  -\frac{2\xi}{\sqrt{-E}}-4l-4\right)  }\sum_{n\geq0}\frac{a_{n}%
}{z^{n}},\label{Bcoulomb}\\
H_{l}^{+}\left(  4l+2,0,-\frac{2\xi}{\sqrt{-E}},0,z\right)   &  =z^{-\frac
{1}{2}\left(  -\frac{2\xi}{\sqrt{-E}}+4l+4\right)  }e^{z^{2}}\sum_{n\geq
0}\frac{e_{n}}{z^{n}}, \label{Hcoulomb}%
\end{align}
with the expansion coefficients given by the recurrence relation%
\[
a_{0}=1,\text{ \ }a_{1}=0,
\]%
\begin{equation}
2\left(  n+2\right)  a_{n+2}+\left[  \frac{1}{4}\left(  -\frac{4\xi^{2}}%
{E}-\left(  4l+2\right)  ^{2}+4\right)  +\frac{2\xi}{\sqrt{-E}}+n\left(
n+2+\frac{2\xi}{\sqrt{-E}}\right)  \right]  a_{n}=0
\end{equation}
and%
\[
e_{0}=1,\text{ \ }e_{1}=0,
\]%
\begin{equation}
2\left(  n+2\right)  e_{n+2}-\left[  \frac{1}{4}\left(  -\frac{4\xi^{2}}%
{E}-\left(  4l+2\right)  ^{2}+4\right)  -\frac{2\xi}{\sqrt{-E}}+n\left(
n+2-\frac{2\xi}{\sqrt{-E}}\right)  \right]  e_{n}=0,
\end{equation}
where the recurrence relation among the adjacent three terms reduces to a
recurrence relation between the adjacent two terms.

\subsection{Eigenfunctions and eigenvalues}

To construct the solution, we first express the regular solution
(\ref{regularcoub}) as a linear combination of the two irregular solutions
(\ref{Bcoulomb}) and (\ref{Hcoulomb}).

The regular solution (\ref{regularcoub}), with the relation
\cite{ronveaux1995heun,li2016exact}%
\begin{align}
N\left(  4l+2,0,-\frac{2\xi}{\sqrt{-E}},0,z\right)   &  =K_{1}\left(
4l+2,0,-\frac{2\xi}{\sqrt{-E}},0\right)  B_{l}^{+}\left(  4l+2,0,-\frac{2\xi
}{\sqrt{-E}},0,z\right) \nonumber\\
&  +K_{2}\left(  4l+2,0,-\frac{2\xi}{\sqrt{-E}},0\right)  H_{l}^{+}\left(
4l+2,0,-\frac{2\xi}{\sqrt{-E}},0,z\right)
\end{align}
and the expansions (\ref{f12o1}) and (\ref{f22o1}), becomes%
\begin{align}
u_{l}\left(  r\right)   &  =A_{l}e^{-\sqrt{-E}r}\left(  2\sqrt{-E}r\right)
^{l+1}N\left(  4l+2,0,-\frac{2\xi}{\sqrt{-E}},0,\left(  2r\right)
^{1/2}\left(  -E\right)  ^{1/4}\right) \nonumber\\
&  =A_{l}e^{-\left(  \sqrt{-E}r+\frac{\xi}{2\sqrt{-E}}\ln2\sqrt{-E}r\right)
}K_{1}\left(  4l+2,0,-\frac{2\xi}{\sqrt{-E}},0\right)  \sum_{n\geq0}%
\frac{a_{n}}{\left(  2\sqrt{-E}r\right)  ^{n/2}}\nonumber\\
&  +A_{l}e^{\left(  \sqrt{-E}r+\frac{\xi}{2\sqrt{-E}}\ln2\sqrt{-E}r\right)
}K_{2}\left(  4l+2,0,-\frac{2\xi}{\sqrt{-E}},0\right)  \sum_{n\geq0}%
\frac{e_{n}}{\left(  2\sqrt{-E}r\right)  ^{n/2}} \label{uexpzcoulomb}%
\end{align}
where $K_{1}\left(  4l+2,0,-\frac{2\xi}{\sqrt{-E}},0\right)  $ and
$K_{2}\left(  4l+2,0,-\frac{2\xi}{\sqrt{-E}},0\right)  $ are combination
coefficients\cite{ronveaux1995heun}\ and $z=\left(  -2r\right)  ^{1/2}\left(
-E\right)  ^{1/4}$.

The boundary condition of bound states, $\left.  u\left(  r\right)
\right\vert _{r\rightarrow\infty}\rightarrow0$, requires that the coefficient
of the second term must vanish since this term diverges when $r\rightarrow
\infty$, i.e.,
\begin{equation}
K_{2}\left(  4l+2,0,-\frac{2\xi}{\sqrt{-E}},0\right)  =0,
\label{eigenvaluecoub1}%
\end{equation}
where%
\begin{align}
K_{2}\left(  4l+2,0,-\frac{2\xi}{\sqrt{-E}},0\right)   &  =\frac{\Gamma\left(
4l+3\right)  }{\Gamma\left(  2l+1+\frac{\xi}{\sqrt{-E}}\right)  \Gamma\left(
2l+2-\frac{\xi}{\sqrt{-E}}\right)  }\nonumber\\
&  \times J_{2l+2-\frac{\xi}{\sqrt{-E}}}\left(  2l+1-\frac{\xi}{\sqrt{-E}%
},0,6l+3+\frac{\xi}{\sqrt{-E}},0\right)
\end{align}
with%
\begin{align}
&  J_{2l+2-\frac{\xi}{\sqrt{-E}}}\left(  2l+1-\frac{\xi}{\sqrt{-E}%
},0,6l+3+\frac{\xi}{\sqrt{-E}},0\right) \nonumber\\
&  =\int_{0}^{\infty}x^{2l+1-\frac{\xi}{\sqrt{-E}}}e^{-x^{2}}N\left(
2l+1-\frac{\xi}{\sqrt{-E}},0,6l+3+\frac{\xi}{\sqrt{-E}},0,x\right)  dx.
\end{align}

In the case of Coulomb potential, $K_{2}\left(  4l+2,0,-\frac{2\xi}{\sqrt{-E}%
},0\right)  $ reduces to \cite{batola1982quelques}%
\begin{equation}
K_{2}\left(  4l+2,0,-\frac{2\xi}{\sqrt{-E}},0\right)  =\frac{\Gamma\left(
2l+2\right)  }{\Gamma\left(  1+l+\frac{\xi}{2\sqrt{-E}}\right)  }%
\end{equation}

While, as pointed above, the zeros on the positive imaginary of the
coefficient $K_{2}\left(  4l+2,0,-\frac{2\xi}{\sqrt{-E}},0\right)  $
correspond to the eigenvalues of bound states.

The zeros of $K_{2}\left(  4l+2,0,-\frac{2\xi}{\sqrt{-E}},0\right)  $
correspond to the bound state of the Coulomb potential, so the eigenvalue is
the singularities of $\Gamma\left(  1+l+\frac{\xi}{2\sqrt{-E}}\right)  $.
Thus
\begin{equation}
1+l+\frac{\xi}{2\sqrt{-E}}=-n,\text{ \ }n=0,1,2\ldots,
\end{equation}
then%
\begin{equation}
E=-\frac{\xi^{2}}{4}\left(  n+l+1\right)  ^{-2},\text{ \ }n=0,1,2,\ldots.
\label{eigenvaluecoub}%
\end{equation}

The eigenfunction, by Eqs. (\ref{uexpzcoulomb}) and (\ref{eigenvaluecoub1}),
reads%
\begin{equation}
u_{l}\left(  r\right)  =A_{l}\exp\left(  -\left(  \sqrt{-E}r+\frac{\xi}%
{2\sqrt{-E}}\ln2\sqrt{-E}r\right)  \right)  K_{1}\left(  4l+2,0,-\frac{2\xi
}{\sqrt{-E}},0\right)  \sum_{n\geq0}\frac{a_{n}}{\left(  2\sqrt{-E}r\right)
^{n/2}}.
\end{equation}












\providecommand{\href}[2]{#2}\begingroup\raggedright\endgroup


\end{document}